\documentclass[draftcls, onecolumn]{IEEEtran}
\usepackage{etex}
\usepackage{geometry}
\usepackage{exscale}
\usepackage{color}
\usepackage{epsfig,subfigure}
\usepackage{subfigure}
\usepackage{cite}
\usepackage{tikz}
\usepackage{pgfplots}
\usetikzlibrary{arrows}
\usetikzlibrary{automata}
\usetikzlibrary{calc}
\usepackage{nomencl}
\usepackage{amssymb}
\usepackage{rotating}
\usepackage{psfrag}
\usepackage{setspace}
\usepackage{amsmath}
\usepackage{amssymb}
\usepackage{graphicx}
\usepackage{caption}
\usepackage{tikz}
\usepackage{pgfplots}

\usepackage{booktabs}
\usepackage{floatflt}
\usepackage{epstopdf}
\usepackage{enumerate}

\usepackage{array}
\usepackage{psfrag}
\usepackage{multirow,hhline}
\usepackage{exscale}
\usepackage{color}
\usepackage{cite}
\usepackage{amsthm}
\usepackage{nicefrac}
\usetikzlibrary{arrows,shapes,positioning}

    \pgfplotsset{
    label style={anchor=near ticklabel},
    xlabel style={yshift=0cm},
    ylabel style={yshift=-.8cm}}

\usepackage{hyperref}

\newtheorem{theorem}{Theorem}
\newtheorem{lemma}{Lemma}

\newtheorem{remark}{Remark}

\newtheorem{proposition}{Proposition}

\newcommand{\X}{\boldsymbol{x}}
\newcommand{\Y}{\boldsymbol{y}}
\newcommand{\bS}{\boldsymbol{S}}

\newcommand{\U}{\boldsymbol{u}}

\newcommand{\WIone}{WI-1}
\newcommand{\WItwo}{WI-2}
\newcommand{\WIthree}{WI-3}
\newcommand{\WIThreea}{WI3a}
\newcommand{\WIThreeb}{WI3b}

\begin{document}
%\newgeometry{top=72pt,bottom=54pt,right=54pt,left=54pt}
\IEEEoverridecommandlockouts
%\addtolength{\topmargin}{4mm}
\title{{Cooperation for interference management: A GDoF perspective}}
\author{
\IEEEauthorblockN{Soheyl Gherekhloo, Anas Chaaban, and Aydin Sezgin}\\
\IEEEauthorblockA{Institute of Digital Communication Systems\\
Ruhr-Universit\"at Bochum, Germany\\
Email: {soheyl.gherekhloo@rub.de, anas.chaaban@kaust.edu.sa, aydin.sezgin@rub.de}}
\thanks{ This work was presented in part at IEEE
Allerton 2013 \cite{GherekhlooChaabanSezgin}.}}

\maketitle
\newcommand{\TheoreticalIRC}[0]{
\draw (-.5, 0.8) rectangle (-3.5,-0.8);
\node (t1) at (-.5,0) [inner sep=0] {};
\node (tx1) at (-2,0.5) [inner sep=0] {\large{Tx1}};
\node (W1) at (-5,0) [inner sep=0] {\large{$W_1$}};
\node (W1tl) at (-3.5,0) [inner sep=0] {};
\node (W1tR) at (-4.5,0) [inner sep=0] {};
\draw[->] (W1tR) to (W1tl);
\node (tx1f) at (-2,-.5) [inner sep=0] {{$X_1^n=f_1(W_1)$}};

\draw (-.5, -3.2) rectangle (-3.5,-4.8);
\node (t2) at (-0.5,-4) [inner sep=0] {};
\node (tx2) at (-2,-3.5) [inner sep=0] {\large{Tx2}};
\node (W2) at (-5,-4) [inner sep=0] {\large{$W_2$}};
\node (W2tl) at (-3.5,-4) [inner sep=0] {};
\node (W2tR) at (-4.5,-4) [inner sep=0] {};
\draw[->] (W2tR) to (W2tl);
\node (tx2f) at (-2,-4.5) [inner sep=0] {{$X_2^n=f_2(W_2)$}};

\draw (4.2, -1.2) rectangle (0.6,-2.8);
\node (Re) at (4.2,-2) [inner sep=0]{};
\node (relay) at (2.4,-1.6) [inner sep=0] {\large{Relay}};
\node (tx2f) at (2.4,-2.5) [inner sep=0] {$X_r[k]=f_{rk}(y_r^{k-1})$};
\node (relayin) at (0.6,-2) [inner sep=0] {};
\node (relayout) at (4.2,-2) [inner sep=0] {};

\draw (11.5, 0.8) rectangle (8.5,-0.8);
\node (r1) at (8.5,0) [inner sep=0] {};
\node (rx1) at (10,0.5) [inner sep=0] {\large{Rx1}};
\node (W1) at (13,0) [inner sep=0] {\large{$\hat W_1$}};
\node (W1Rl) at (11.5,0) [inner sep=0] {};
\node (W1RR) at (12.5,0) [inner sep=0] {};
\draw[->] (W1Rl) to (W1RR);
\node (tx1f) at (10,-.5) [inner sep=0] {{$\hat W_1=g_1(Y_1^{n})$}};

\draw (11.5, -3.2) rectangle (8.5,-4.8);
\node (r2) at (8.5,-4) [inner sep=0] {};
\node (rx2) at (10,-3.5) [inner sep=0] {\large{Rx2}};
\node (W2) at (13,-4) [inner sep=0] {\large{$\hat W_2$}};
\node (W2Rl) at (11.5,-4) [inner sep=0] {};
\node (W2RR) at (12.5,-4) [inner sep=0] {};
\draw[->] (W2Rl) to (W2RR);
\node (tx1f) at (10,-4.5) [inner sep=0] {{$\hat W_2=g_2(Y_2^{n})$}};

\draw[->] (t1) to (r1);
\draw[->] (t2) to (r2);
\draw[->] (relayout) to (r1);
\draw[->] (relayout) to (r2);
\draw[->] (t1) to (relayin);
\draw[->] (t2) to (relayin);
\draw[->, thick] (t1) to[in=120, out=0] (r2);
\draw[->, thick] (t2) to[in=-120, out=0] (r1);
}

\newcommand{\GaussianIRC}[0]{
\draw (-2.5, 0.5) rectangle (-1.5,-0.5);
\node (tx1) at (-2,0) [inner sep=0] {\large{Tx1}};

\node (t1) at (-.5,0) [inner sep=0] {};
\draw[->] (-1.5,0) to (t1);
\node (W1) at (-1,0.2) [inner sep=0] {{\large{$ x_1$}}};

%\node (t2) at (-.5,-4) [inner sep=0] {};
%\node (W1) at (-1.5,-4) [inner sep=0] {\large{$ x_2[k]$}};

\draw (-2.5, 0.5-4) rectangle (-1.5,-0.5-4);
\node (tx2) at (-2,-4) [inner sep=0] {\large{Tx2}};

\node (t2) at (-.5,0-4) [inner sep=0] {};
\draw[->] (-1.5,0-4) to (t2);
\node (W1) at (-1,0.2-4) [inner sep=0] {{\large{$ x_2$}}};

\node (relayin) at (0.6,-2) [inner sep=0] {\large{$\oplus$}};
\draw[->] (.8,-2) to (1.8,-2);
\draw (1.8, 0.5-2) rectangle (3.2,-0.5-2);
\node (relay) at (2.5,-2) [inner sep=0] {\large{Relay}};
\node (yr) at (1.3,-1.8) [inner sep=0] {\large{$y_r$}};
\draw[->] (3.2,-2) to (4.25,-2);
\node (yr) at (3.7,-1.8) [inner sep=0] {\large{$x_r$}};
\draw[->] (-.1,-2) to (.4,-2);
\node (zr) at (-.5,-2) [inner sep=0] {\large{$z_r$}};

%\node (relayin2) at (2.5,-2) [inner sep=0] {\large{$\rightarrow y_r[k] \text{ : }x_r[k]$}};
%\node (relayin3) at (-.5,-2) [inner sep=0] {\large{$z_r[k]\rightarrow $}};
\node (relayout) at (4.2,-2) [inner sep=0] {};

\node (r1) at (8.5,0) [inner sep=0] {};
\node (rx1) at (8.7,0) [inner sep=0] {\large{$\oplus$}};
\draw[->] (8.9,0) to (9.9,0);
\node (y1) at (9.4,0.2) [inner sep=0] {\large{$y_1$}};
\draw (9.9, 0.5) rectangle (10.9,-0.5);
\node (rx1) at (10.4,0) [inner sep=0] {\large{Rx1}};

\draw[->] (8.7,-1) to (8.7,-.2);
\node (z1) at (8.9,-1.2) [inner sep=0] {\large{$ z_1[k]$}};

%\node (r2) at (8.5,-4) [inner sep=0] {};
%\node (rx2) at (9.7,-4) [inner sep=0] {\large{$\oplus \rightarrow y_2[k]$}};

\node (r2) at (8.5,0-4) [inner sep=0] {};
\node (rx2) at (8.7,0-4) [inner sep=0] {\large{$\oplus$}};
\draw[->] (8.9,0-4) to (9.9,0-4);
\node (y2) at (9.4,0.2-4) [inner sep=0] {\large{$y_2$}};
\draw (9.9, 0.5-4) rectangle (10.9,-0.5-4);
\node (rx2) at (10.4,0-4) [inner sep=0] {\large{Rx2}};

\draw[->] (8.8,-3) to (8.8,-3.8);
\node (rx1) at (9,-2.7) [inner sep=0] {\large{$ z_2[k]$}};

\draw[->] (t1) to (r1);
\draw[->] (t2) to (r2);
\draw[->] (relayout) to (r1);
\draw[->] (relayout) to (r2);
\draw[->] (t1) to (relayin);
\draw[->] (t2) to (relayin);
\draw[->, thick] (t1) to[in=120, out=0] (r2);
\draw[->, thick] (t2) to[in=-120, out=0] (r1);

\draw [draw=white,fill=white,opacity=1,draw opacity=1] (4.6, 0.3) rectangle (5.3,-4.3);
\node (hd1) at (5,0) [inner sep=0] {\large{$ h_d$}};
\node (hd2) at (5,-4) [inner sep=0] {\large{$ h_d$}};
\node (hc1) at (5,-.8) [inner sep=0] {\large{$ h_c$}};
\node (hc2) at (5,-3.2) [inner sep=0] {\large{$ h_c$}};
\node (hr1) at (5,-1.7) [inner sep=0] {\large{$ h_r$}};
\node (hr2) at (5,-2.3) [inner sep=0] {\large{$ h_r$}};
\draw [draw=white,fill=white,opacity=1,draw opacity=1] (-1, -.5) rectangle (0,-1);
\node (hs1) at (0,-.7) [inner sep=0] {\large{$ h_s$}};
\draw [draw=white,fill=white,opacity=1,draw opacity=1] (-1, -3) rectangle (0,-3.5);
\node (hs1) at (0,-3.3) [inner sep=0] {\large{$ h_s$}};
}

\newcommand{\SCHEMEWIone}[0]{
%%%% horisonatal lines

\node (Zero) at (-4.5,-4) [inner sep=0] {$0$};
\draw[line width=1,-] (-4,-4) to (4,-4);

%%%%%%%%%%
\node (zerotx1rx1) at (-3,-4.3) [inner sep=0] {$\boldsymbol S^{q-n_d} \boldsymbol x_1[k]$};
\node (zerotx2rx1) at (0,-4.3) [inner sep=0] {$\boldsymbol S^{q-n_c}\boldsymbol x_2[k]$};
\node (zerotx3rx1) at (3.2,-4.3) [inner sep=0] {$\boldsymbol S^{q-n_r}\boldsymbol x_r[k]$};

%%%%%%% Relay -----> Rx1
\draw[line width=1,-] (2,3.25) to (4,3.25);
\node (zerotx1rx1) at (4.2,3.25) [inner sep=0] {$n_r$};

\draw[line width=1] (2.5,2.25) rectangle (3.5,3.25);
\node (zerotx1rx1) at (3,2.75) [inner sep=0] {$\boldsymbol{u}_{r,cf}$};

\draw[line width=1] (2.5,1.25) rectangle (3.5,2.25);
\node (zerotx1rx1) at (3,1.75) [inner sep=0] {$\boldsymbol{u}_{r,df1}$};

\draw[line width=1] (2.5,0.5) rectangle (3.5,1.25);
\node (zerotx1rx1) at (3,0.75) [inner sep=0] {$\boldsymbol{u}_{r,df2}$};

\draw[line width=1] (2.5,-0.5) rectangle (3.5,0.5);
\node (zerotx1rx1) at (3,0) [inner sep=0] {$\boldsymbol{0}$};

\draw[line width=1,<->] (3.7,-0.5) to (3.7,0.5);
\node (zerotx1rx1) at (4,0) [inner sep=0] {$\ell_2$};
\draw[line width=1,dashed,-] (0,-0.5) to (4,-0.5);

\draw[line width=1] (2.5,-2.5) rectangle (3.5,-0.5);
\node (zerotx1rx1) at (3,-1.5) [inner sep=0] {$\boldsymbol{u}_{r,cn}$};

\draw[line width=1] (2.5,-4) rectangle (3.5,-2.5);
\node (zerotx1rx1) at (3,-3) [inner sep=0] {$\boldsymbol{0}$};

\draw[line width=1,<->] (3.7,-4) to (3.7,-2.5);
\node (zerotx1rx1) at (4,-3.25) [inner sep=0] {$\ell_3$};
\draw[line width=1,dashed,-] (0,-2.5) to (4,-2.5);

%
%\draw[line width=1] (2.5,-4) rectangle (3.5,-3.5);
%\node (u1prx1) at (3,-3.75) [inner sep=0] {$\boldsymbol{u}_{2,p}$};

%%%%%%% Tx2 -----> Rx1
\draw[line width=1,-] (-1,0) to (1,0);
\node (zerotx1rx1) at (1.2,0) [inner sep=0] {$n_c$};

\draw[line width=1] (-0.5,-0.5) rectangle (0.5,0);
\node (zerotx1rx1) at (0,-0.25) [inner sep=0] {$\boldsymbol{0}$};

\draw[line width=1] (0,-1.5) rectangle (0.5,-0.5);
\node (zerotx1rx1) at (0.2,-1) [inner sep=0,rotate=90] {$\boldsymbol{u}_{2,df1}$};

\draw[line width=1] (-0.5,-2.5) rectangle (0.5,-0.5);
\node (zerotx1rx1) at (0,-2.3) [inner sep=0] {$\boldsymbol{u}_{2,cn}$};

\draw[line width=1] (-0.5,-3.5) rectangle (0.5,-2.5);
\node (zerotx1rx1) at (0,-3) [inner sep=0] {$\boldsymbol{u}_{2,cf}$};

\draw[line width=1] (-0.5,-4) rectangle (0.5,-3.5);
\node (u1prx1) at (0,-3.75) [inner sep=0] {$\boldsymbol{u}_{2,p}$};

%\draw[line width=1] (-0.5,-5) rectangle (0.5,-4.5);
%\node (u1cnf) at (0,-4.75) [inner sep=0] {$\boldsymbol{u}_{2,cnF}$};

%%%%%%% Tx1 -----> Rx1
\draw[line width=1,-] (-4,1) to (-2,1);
\node (zerotx1rx1) at (-1.8,1) [inner sep=0] {$n_d$};

\draw[line width=1,<->] (-3.7,0.5) to (-3.7,1);
\node (zerotx1rx1) at (-4,0.75) [inner sep=0] {$\ell_1$};
\draw[line width=1,dashed,-] (-4,0.5) to (4,0.5);

\draw[line width=1] (-3.5,0.5) rectangle (-2.5,1);
\node (zerotx1rx1) at (-3,0.75) [inner sep=0] {$\boldsymbol{0}$};

\draw[line width=1] (-3,-0.5) rectangle (-2.5,0.5);
\node (zerotx1rx1) at (-2.7,0) [inner sep=0,rotate=90] {$\boldsymbol{u}_{1,df1}$};

\draw[line width=1] (-3.5,-1.5) rectangle (-2.5,0.5);
\node (zerotx1rx1) at (-3,-1.3) [inner sep=0] {$\boldsymbol{u}_{1,cn}$};

\draw[line width=1] (-3.5,-2.5) rectangle (-2.5,-1.5);
\node (zerotx1rx1) at (-3,-2) [inner sep=0] {$\boldsymbol{u}_{1,cf}$};

\draw[line width=1] (-3.5,-3.5) rectangle (-2.5,-2.5);
\node (u1prx1) at (-3,-3) [inner sep=0] {$\boldsymbol{u}_{1,p}$};

\draw[line width=1] (-3.5,-4) rectangle (-2.5,-3.5);
\node (u1cnf) at (-3,-3.75) [inner sep=0] {$\boldsymbol{u}_{1,cnF}$};
}

\newcommand{\SCHEMEWIoneRelay}[0]{
%%%% horisonatal lines

\node (Zero) at (-4.5,-4) [inner sep=0] {$0$};
\draw[line width=1,-] (-4,-4) to (1,-4);

%%%%%%%%%%
\node (zerotx1rx1) at (-3,-4.3) [inner sep=0] {$\boldsymbol S^{q-n_s} \boldsymbol x_1[k]$};
\node (zerotx2rx1) at (0,-4.3) [inner sep=0] {$\boldsymbol S^{q-n_s}\boldsymbol x_2[k]$};

%%%%%%% Tx1 -----> relay
\draw[line width=1,-] (-4,1+2.25) to (-2,1+2.25);
\node (zerotx1rx1) at (-1.8,1+2.25) [inner sep=0] {$n_s$};

\draw[line width=1,<->] (-3.7,0.5+2.25) to (-3.7,1+2.25);
\node (zerotx1rx1) at (-4,0.75+2.25) [inner sep=0] {$\ell_1$};
%\draw[line width=1,dashed,-] (-4,0.5+2.25) to (4,0.5+2.25);

\draw[line width=1] (-3.5,0.5+2.25) rectangle (-2.5,1+2.25);
\node (zerotx1rx1) at (-3,0.75+2.25) [inner sep=0] {$\boldsymbol{0}$};

\draw[line width=1] (-3,-0.5+2.25) rectangle (-2.5,0.5+2.25);
\node (zerotx1rx1) at (-2.7,0+2.25) [inner sep=0,rotate=90] {$\boldsymbol{u}_{1,df1}$};

\draw[line width=1] (-3.5,-1.5+2.25) rectangle (-2.5,0.5+2.25);
\node (zerotx1rx1) at (-3,-1.3+2.25) [inner sep=0] {$\boldsymbol{u}_{1,cn}$};

\draw[line width=1] (-3.5,-2.5+2.25) rectangle (-2.5,-1.5+2.25);
\node (zerotx1rx1) at (-3,-2+2.25) [inner sep=0] {$\boldsymbol{u}_{1,cf}$};

\draw[line width=1] (-3.5,-3.5+2.25) rectangle (-2.5,-2.5+2.25);
\node (u1prx1) at (-3,-3+2.25) [inner sep=0] {$\boldsymbol{u}_{1,p}$};

\draw[line width=1] (-3.5,-4+2.25-1.5) rectangle (-2.5,-3.5+2.25);
\node (u1cnf) at (-3,-3.75+2.25-1) [inner sep=0] {$\boldsymbol{u}_{1,cnF}$};
\draw[line width=1] (-3.5,-4+2.25-1.5) rectangle (-2.5,-5.5+2.25-.75);
\node (u1df2) at (-3,-4+2.25-2) [inner sep=0] {$\boldsymbol{u}_{1,df2}$};

%%%%%%% Tx2 -----> relay
\draw[line width=1,-] (-1,1+2.25) to (1,1+2.25);
\node (zerotx1rx1) at (1.2,1+2.25) [inner sep=0] {$n_s$};

\draw[line width=1] (-.5,0.5+2.25) rectangle (.5,1+2.25);
\node (zerotx1rx1) at (0,0.75+2.25) [inner sep=0] {$\boldsymbol{0}$};

\draw[line width=1] (0,-0.5+2.25) rectangle (.5,0.5+2.25);
\node (zerotx1rx1) at (.25,0+2.25) [inner sep=0,rotate=90] {$\boldsymbol{u}_{2,df1}$};

\draw[line width=1] (-.5,-1.5+2.25) rectangle (.5,0.5+2.25);
\node (zerotx1rx1) at (0,-1.3+2.25) [inner sep=0] {$\boldsymbol{u}_{2,cn}$};

\draw[line width=1] (-.5,-2.5+2.25) rectangle (.5,-1.5+2.25);
\node (zerotx1rx1) at (0,-2+2.25) [inner sep=0] {$\boldsymbol{u}_{2,cf}$};

\draw[line width=1] (-.5,-3.5+2.25) rectangle (.5,-2.5+2.25);
\node (u1prx1) at (0,-3+2.25) [inner sep=0] {$\boldsymbol{u}_{2,p}$};

\draw[line width=1] (-.5,-4+2.25-1.5) rectangle (.5,-3.5+2.25);
\node (u1cnf) at (0,-3.75+2.25-1) [inner sep=0] {$\boldsymbol{u}_{2,cnF}$};
\draw[line width=1] (-.5,-4+2.25-1.5) rectangle (.5,-5.5+2.25-.75);
\node (u1df2) at (0,-4+2.25-2) [inner sep=0] {$\boldsymbol{u}_{2,df2}$};
}

\newcommand{\SCHEMEWItwoFigone}[0]{
\node (Zero) at (-4.5,-4) [inner sep=0] {$0$};
\draw[line width=1,-] (-4,-4) to (6,-4);

%%%%%%%%%%
\node (zerotx1rx1) at (-3,-4.3) [inner sep=0] {$\boldsymbol S^{q-n_d} \boldsymbol x_1[k]$};
\node (zerotx2rx1) at (0,-4.3) [inner sep=0] {$\boldsymbol S^{q-n_c}\boldsymbol x_2[k]$};
\node (zerotx3rx1) at (3.2,-4.3) [inner sep=0] {$\boldsymbol S^{q-n_r}\boldsymbol x_r[k]$};

%%%%%%% Relay -----> Rx1
\draw[line width=1,-] (2,3.25) to (4,3.25);
\node (zerotx1rx1) at (4.2,3.25) [inner sep=0] {$n_r$};

\draw[line width=1] (2.5,-1.5) rectangle (3.5,3.25);
\node (zerotx1rx1) at (3,0) [inner sep=0] {$\boldsymbol{0}$};

\draw[line width=1,-] (4,-1.5) to (2,-1.5);
\node (zerotx1rx1) at (4.2,-1.5) [inner sep=0] {$n_r^\prime$}; 

\draw[line width=1] (2.5,-1.5) rectangle (3.5,-4);
\node (zerotx1rx1) at (3,-3) [inner sep=0] {$\boldsymbol{x}^\prime_r[k]$};

%%%%%%% Tx2 -----> Rx1
\draw[line width=1,-] (-1,2.5) to (1,2.5);
\node (zerotx1rx1) at (1.2,2.5) [inner sep=0] {$n_c$};

\draw[line width=1] (-0.5,1.5) rectangle (0.5,2.5);
\node (zerotx1rx1) at (0,2) [inner sep=0] {$\boldsymbol{u}_{2,cm1}$};

\draw[line width=1] (-0.5,0.5) rectangle (0.5,1.5);
\node (zerotx1rx1) at (0,1) [inner sep=0] {$\boldsymbol{0}$};

\draw[line width=1] (-0.5,-0.5) rectangle (0.5,0.5);
\node (zerotx1rx1) at (0,0) [inner sep=0] {$\vdots$};

\draw[line width=1] (-0.5,-1.5) rectangle (0.5,-0.5);
\node (zerotx1rx1) at (0,-1) [inner sep=0] {$\boldsymbol{u}_{2,cmL}$};

\draw[line width=1] (-0.5,-1.5) rectangle (0.5,-2.5);
\node (zerotx1rx1) at (0,-2) [inner sep=0] {$\boldsymbol{0}$};

\draw[line width=1,-] (-1,-2.5) to (1,-2.5);
\node (zerotx1rx1) at (1.2,-2.5) [inner sep=0] {$n_c^\prime$};

\draw[line width=1] (-0.5,-1.5) rectangle (0.5,-4);
\node (zerotx1rx1) at (0,-3) [inner sep=0] {$\boldsymbol{x}_2^\prime[k]$};

%%%%%%% Tx1 -----> Rx1

%\draw[line width=1,-] (-4,1) to (-2,1);
%\node (zerotx1rx1) at (-1.8,1) [inner sep=0] {$n_d$};

\draw[line width=1,-] (-4,3.5) to (-2,3.5);
\node (zerotx1rx1) at (-1.8,3.5) [inner sep=0] {$n_d$};

\draw[line width=1] (-3.5,3.5) rectangle (-2.5,2.5);
\node (zerotx1rx1) at (-3,3) [inner sep=0] {$\boldsymbol{u}_{1,cm1}$};

\draw[line width=1] (-3.5,2.5) rectangle (-2.5,1.5);
\node (zerotx1rx1) at (-3,2) [inner sep=0] {$\boldsymbol{0}$};

\draw[line width=1] (-3.5,1.5) rectangle (-2.5,0.5);
\node (zerotx1rx1) at (-3,1) [inner sep=0] {$\vdots$};

\draw[line width=1] (-3.5,0.5) rectangle (-2.5,-0.5);
\node (zerotx1rx1) at (-3,0) [inner sep=0] {$\boldsymbol{u}_{1,cmL}$};
\draw[line width=1] (-3.5,-0.5) rectangle (-2.5,-1.5);
\node (zerotx1rx1) at (-3,-1) [inner sep=0] {$\boldsymbol{0}$};

\draw[line width=1,-] (-4,-1.5) to (-2,-1.5);
\node (zerotx1rx1) at (-1.8,-1.5) [inner sep=0] {$n_d^\prime$};

\draw[line width=1] (-3.5,0.5) rectangle (-2.5,-4);
\node (zerotx1rx1) at (-3,-2.5) [inner sep=0] {$\boldsymbol{x}_1^\prime[k]$};}

\newcommand{\SCHEMEWItwoFigoneAC}[0]{
\node (Zero) at (-4.5,-4) [inner sep=0] {$0$};
\draw[line width=1,-] (-4,-4) to (6,-4);

%%%%%%%%%%
\node (zerotx1rx1) at (-3,-4.3) [inner sep=0] {$\boldsymbol S^{q-n_d} \boldsymbol x_1[k]$};
\node (zerotx2rx1) at (0,-4.3) [inner sep=0] {$\boldsymbol S^{q-n_c}\boldsymbol x_2[k]$};
\node (zerotx3rx1) at (3.2,-4.3) [inner sep=0] {$\boldsymbol S^{q-n_r}\boldsymbol x_r[k]$};

%%%%%%% Relay -----> Rx1
\draw[line width=1,-] (2,-.75) to (4,-.75);
\node (zerotx1rx1) at (4.2,-1) [inner sep=0] {$n_r$};

\draw[line width=1] (2.5,-1.5) rectangle (3.5,-.75);
\node (zerotx1rx1) at (3,-1) [inner sep=0] {$\boldsymbol{0}$};

\draw[line width=1,-] (4,-1.5) to (2,-1.5);
\node (zerotx1rx1) at (4.2,-1.5) [inner sep=0] {$n_r^\prime$}; 

\draw[line width=1] (2.5,-1.5) rectangle (3.5,-4);
\node (zerotx1rx1) at (3,-3) [inner sep=0] {$\boldsymbol{x}^\prime_r[k]$};

%%%%%%% Tx2 -----> Rx1
\draw[line width=1,-] (-1,-1) to (1,-1);
\node (zerotx1rx1) at (1.2,-1) [inner sep=0] {$n_c$};

\draw[line width=1] (-0.5,-2.5) rectangle (0.5,-1);
\node (zerotx1rx1) at (0,-1.75) [inner sep=0] {$\boldsymbol{u}_{2,cm}$};

\draw[line width=1,-] (-1,-2.5) to (1,-2.5);
\node (zerotx1rx1) at (1.2,-2.5) [inner sep=0] {$n_c^\prime$};

\draw[line width=1] (-0.5,-2.5) rectangle (0.5,-4);
\node (zerotx1rx1) at (0,-3) [inner sep=0] {$\boldsymbol{x}_2^\prime[k]$};

%%%%%%% Tx1 -----> Rx1

\draw[line width=1,-] (-4,0) to (-2,0);
\node (zerotx1rx1) at (-1.8,0) [inner sep=0] {$n_d$};

\draw[line width=1] (-3.5,-1.5) rectangle (-2.5,0);
\node (zerotx1rx1) at (-3,-.75) [inner sep=0] {$\boldsymbol{u}_{1,cm}$};

\draw[line width=1,-] (-4,-1.5) to (-2,-1.5);
\node (zerotx1rx1) at (-1.8,-1.5) [inner sep=0] {$n_d^\prime$};

\draw[line width=1] (-3.5,-1.5) rectangle (-2.5,-4);
\node (zerotx1rx1) at (-3,-2.5) [inner sep=0] {$\boldsymbol{x}_1^\prime[k]$};}

\newcommand{\SCHEMEWItwoFigtwo}[0]{
%%%% horisonatal lines
\node (Zero) at (-4.5,-4) [inner sep=0] {$0$};
\draw[line width=1,-] (-4,-4) to (6,-4);

%%%%%%%%%%
\node (zerotx1rx1) at (-3,-4.3) [inner sep=0] {$\boldsymbol x_1^\prime[k]$};
\node (zerotx2rx1) at (0,-4.3) [inner sep=0] {$\boldsymbol x^\prime_2[k]$};
\node (zerotx3rx1) at (3.2,-4.3) [inner sep=0] {$\boldsymbol x^\prime_r[k]$};

%%%%%%% Relay -----> Rx1
\draw[line width=1,-] (2,1.5) to (4,1.5);
\node (zerotx1rx1) at (4.2,1.5) [inner sep=0] {$n_r^\prime$};

\draw[line width=1] (2.5,0) rectangle (3.5,1.5);
\node (zerotx1rx1) at (3,0.75) [inner sep=0] {$\boldsymbol{0}$};

\draw[line width=1] (2.5,-0.5) rectangle (3.5,0);
\node (zerotx1rx1) at (3,-0.25) [inner sep=0] {$\boldsymbol{u}_{r,cf1}$};

\draw[line width=1] (2.5,-1.5) rectangle (3.5,-0.5);
\node (zerotx1rx1) at (3,-1) [inner sep=0] {$\boldsymbol{0}$};

\draw[line width=1] (2.5,-2) rectangle (3.5,-1.5);
\node (zerotx1rx1) at (3,-1.75) [inner sep=0] {$\boldsymbol{u}_{r,cf2}$};

\draw[line width=1] (2.5,-2.5) rectangle (3.5,-2);
\node (zerotx1rx1) at (3,-2.25) [inner sep=0] {$\boldsymbol{u}_{r,cn}$};

\draw[line width=1] (2.5,-4) rectangle (3.5,-2.5);
\node (zerotx1rx1) at (3,-3) [inner sep=0] {$\boldsymbol{0}$};

%
%\draw[line width=1] (2.5,-4) rectangle (3.5,-3.5);
%\node (u1prx1) at (3,-3.75) [inner sep=0] {$\boldsymbol{u}_{2,p}$};

%%%%%%% Tx2 -----> Rx1
%\draw[line width=1] (-0.5,1.5) rectangle (0.5,4);
%\node (zerotx1rx1) at (-0,3.75) [inner sep=0] {$\boldsymbol{u}_{1,cm1}$};
%\node (zerotx1rx1) at (-0,3.25) [inner sep=0] {$\boldsymbol{0}$};
%\node (zerotx1rx1) at (-0,3) [inner sep=0] {$ \vdots$};
%\node (zerotx1rx1) at (-0,2.25) [inner sep=0] {$\boldsymbol{u}_{1,cmL}$};
%\node (zerotx1rx1) at (-0,1.75) [inner sep=0] {$\boldsymbol{0}$};
\draw[line width=1,-] (-1,-2) to (1,-2);
\node (zerotx1rx1) at (1.2,-2) [inner sep=0] {$n_c^\prime$};

\draw[line width=1] (-0.5,-2.5) rectangle (0.5,-2);
\node (zerotx1rx1) at (0,-2.25) [inner sep=0] {$\boldsymbol{u}_{2,cn}$};

\draw[line width=1] (-0.5,-3.5) rectangle (0.5,-2.5);
\node (zerotx1rx1) at (0,-3) [inner sep=0] {$\boldsymbol{u}_{2,cf}$};

\draw[line width=1] (-0.5,-4) rectangle (0.5,-3.5);
\node (zerotx1rx1) at (0,-3.75) [inner sep=0] {$\boldsymbol{0}$};

%%%%%%% Tx1 -----> Rx1

%\draw[line width=1] (-3.5,1.5) rectangle (-2.5,4);
%\node (zerotx1rx1) at (-3,3.75) [inner sep=0] {$\boldsymbol{u}_{1,cm1}$};
%\node (zerotx1rx1) at (-3,3.25) [inner sep=0] {$\boldsymbol{0}$};
%\node (zerotx1rx1) at (-3,3) [inner sep=0] {$ \vdots$};
%\node (zerotx1rx1) at (-3,2.25) [inner sep=0] {$\boldsymbol{u}_{1,cmL}$};
%\node (zerotx1rx1) at (-3,1.75) [inner sep=0] {$\boldsymbol{0}$};

\draw[line width=1,<->] (-3.8,1.5) to (-3.8,-2);
\node (zerotx1rx1) at (-4.1,-0.25) [inner sep=0] {$n_s^\prime$};

\draw[line width=1,-] (-4,1.5) to (-2,1.5);
\node (zerotx1rx1) at (-1.8,1.5) [inner sep=0] {$n_d^\prime$};

\draw[line width=1] (-3.5,1) rectangle (-2.5,1.5);
\node (zerotx1rx1) at (-3,1.25) [inner sep=0] {$\boldsymbol{u}_{1,cn}$};

\draw[line width=1] (-3.5,0) rectangle (-2.5,1);
\node (zerotx1rx1) at (-3,0.5) [inner sep=0] {$\boldsymbol{u}_{1,cf}$};

\draw[line width=1] (-3.5,-0.5) rectangle (-2.5,0);
\node (zerotx1rx1) at (-3,-0.3) [inner sep=0] {$\boldsymbol{0}$};

\draw[line width=1] (-3.5,-1.5) rectangle (-2.5,-0.5);
\node (zerotx1rx1) at (-3,-1) [inner sep=0] {$\boldsymbol{u}_{1,p1}$};

\draw[line width=1] (-3.5,-2) rectangle (-2.5,-1.5);
\node (u1prx1) at (-3,-1.75) [inner sep=0] {$\boldsymbol{u}_{1,cnF}$};

\draw[line width=1] (-3.5,-4) rectangle (-2.5,-2);
\node (u1prx1) at (-3,-3) [inner sep=0] {$\boldsymbol{u}_{1,p2}$};}

\newcommand{\SCHEMEWItwoFigtwoAC}[0]{
%%%% horisonatal lines
\node (Zero) at (-4.5,-4) [inner sep=0] {$0$};
\draw[line width=1,-] (-4,-4) to (6,-4);

%%%%%%%%%%
\node (zerotx1rx1) at (-3,-4.3) [inner sep=0] {$\boldsymbol x_1^\prime[k]$};
\node (zerotx2rx1) at (0,-4.3) [inner sep=0] {$\boldsymbol x^\prime_2[k]$};
\node (zerotx3rx1) at (3.2,-4.3) [inner sep=0] {$\boldsymbol x^\prime_r[k]$};

%%%%%%% Relay -----> Rx1
\draw[line width=1,-] (2,0) to (4,0);
\node (zerotx1rx1) at (4.2,0) [inner sep=0] {$n_r^\prime$};

\draw[line width=1] (2.5,-0.5) rectangle (3.5,0);
\node (zerotx1rx1) at (3,-0.25) [inner sep=0] {$\boldsymbol{u}_{r,cf1}$};

\draw[line width=1] (2.5,-1.5) rectangle (3.5,-0.5);
\node (zerotx1rx1) at (3,-1) [inner sep=0] {$\boldsymbol{0}$};
\draw[line width=1,<->] (2.25,-1.5) to (2.25,-.5);
\node (zerotx1rx1) at (2,-1) [inner sep=0] {$\ell_3$};

\draw[line width=1] (2.5,-2) rectangle (3.5,-1.5);
\node (zerotx1rx1) at (3,-1.75) [inner sep=0] {$\boldsymbol{u}_{r,cf2}$};

\draw[line width=1] (2.5,-2.5) rectangle (3.5,-2);
\node (zerotx1rx1) at (3,-2.25) [inner sep=0] {$\boldsymbol{u}_{r,cn}$};

\draw[line width=1] (2.5,-4) rectangle (3.5,-2.5);
\node (zerotx1rx1) at (3,-3) [inner sep=0] {$\boldsymbol{0}$};

%
%\draw[line width=1] (2.5,-4) rectangle (3.5,-3.5);
%\node (u1prx1) at (3,-3.75) [inner sep=0] {$\boldsymbol{u}_{2,p}$};

%%%%%%% Tx2 -----> Rx1
%\draw[line width=1] (-0.5,1.5) rectangle (0.5,4);
%\node (zerotx1rx1) at (-0,3.75) [inner sep=0] {$\boldsymbol{u}_{1,cm1}$};
%\node (zerotx1rx1) at (-0,3.25) [inner sep=0] {$\boldsymbol{0}$};
%\node (zerotx1rx1) at (-0,3) [inner sep=0] {$ \vdots$};
%\node (zerotx1rx1) at (-0,2.25) [inner sep=0] {$\boldsymbol{u}_{1,cmL}$};
%\node (zerotx1rx1) at (-0,1.75) [inner sep=0] {$\boldsymbol{0}$};
\draw[line width=1,-] (-1,-2) to (1,-2);
\node (zerotx1rx1) at (1.2,-2) [inner sep=0] {$n_c^\prime$};
\newcommand{\btwentygthirty}[0]{
\begin{axis}[%
width=7cm,
height=4cm,
scale only axis,
xmin=0,
xmax=5,
xtick={0,0.5,1,2,3,4,5},
xticklabels={{0},{$\frac{1}{2}$},{$1$},{2},{3},{4},{5}},
xlabel={$\alpha$},
ymin=1,
ymax=4.5,
ytick={1,2,3},
yticklabels={{1},{2},{3}},
ylabel={$d$}
]

\draw[<->] (axis cs:1.3, 1.3) -- (axis cs:1.3,3.3 );
\node[above] at (axis cs:1.2,2){$\beta$};
\node[above] at (axis cs:2.3,3.2){$\gamma+\beta-\alpha$};
\node[above] at (axis cs:2.5,2.9){$\gamma$};
\node[above] at (axis cs:3.5,3.5){$\alpha$};
\node[above] at (axis cs:4.5,3.95){$2\beta$};
%\node[above] at (axis cs:2.2,3.2){$\max\{\alpha,\beta\}+(\gamma-\alpha)^+$};

\addplot [color=blue,solid,line width=1.0pt,forget plot]
  table[row sep=crcr]{0	3\\
0.005005005005005	3\\
0.01001001001001	3\\
0.015015015015015	3\\
0.02002002002002	3\\
0.025025025025025	3\\
0.03003003003003	3\\
0.035035035035035	3\\
0.04004004004004	3\\
0.045045045045045	3\\
0.0500500500500501	3\\
0.0550550550550551	3\\
0.0600600600600601	3\\
0.0650650650650651	3\\
0.0700700700700701	3\\
0.0750750750750751	3\\
0.0800800800800801	3\\
0.0850850850850851	3\\
0.0900900900900901	3\\
0.0950950950950951	3\\
0.1001001001001	3\\
0.105105105105105	3\\
0.11011011011011	3\\
0.115115115115115	3\\
0.12012012012012	3\\
0.125125125125125	3\\
0.13013013013013	3\\
0.135135135135135	3\\
0.14014014014014	3\\
0.145145145145145	3\\
0.15015015015015	3\\
0.155155155155155	3\\
0.16016016016016	3\\
0.165165165165165	3\\
0.17017017017017	3\\
0.175175175175175	3\\
0.18018018018018	3\\
0.185185185185185	3\\
0.19019019019019	3\\
0.195195195195195	3\\
0.2002002002002	3\\
0.205205205205205	3\\
0.21021021021021	3\\
0.215215215215215	3\\
0.22022022022022	3\\
0.225225225225225	3\\
0.23023023023023	3\\
0.235235235235235	3\\
0.24024024024024	3\\
0.245245245245245	3\\
0.25025025025025	3\\
0.255255255255255	3\\
0.26026026026026	3\\
0.265265265265265	3\\
0.27027027027027	3\\
0.275275275275275	3\\
0.28028028028028	3\\
0.285285285285285	3\\
0.29029029029029	3\\
0.295295295295295	3\\
0.3003003003003	3\\
0.305305305305305	3\\
0.31031031031031	3\\
0.315315315315315	3\\
0.32032032032032	3\\
0.325325325325325	3\\
0.33033033033033	3\\
0.335335335335335	3\\
0.34034034034034	3\\
0.345345345345345	3\\
0.35035035035035	3\\
0.355355355355355	3\\
0.36036036036036	3\\
0.365365365365365	3\\
0.37037037037037	3\\
0.375375375375375	3\\
0.38038038038038	3\\
0.385385385385385	3\\
0.39039039039039	3\\
0.395395395395395	3\\
0.4004004004004	3\\
0.405405405405405	3\\
0.41041041041041	3\\
0.415415415415415	3\\
0.42042042042042	3\\
0.425425425425425	3\\
0.43043043043043	3\\
0.435435435435435	3\\
0.44044044044044	3\\
0.445445445445445	3\\
0.45045045045045	3\\
0.455455455455455	3\\
0.46046046046046	3\\
0.465465465465465	3\\
0.47047047047047	3\\
0.475475475475475	3\\
0.48048048048048	3\\
0.485485485485485	3\\
0.49049049049049	3\\
0.495495495495495	3\\
0.500500500500501	3\\
0.505505505505506	3\\
0.510510510510511	3\\
0.515515515515516	3\\
0.520520520520521	3\\
0.525525525525526	3\\
0.530530530530531	3\\
0.535535535535536	3\\
0.540540540540541	3\\
0.545545545545546	3\\
0.550550550550551	3\\
0.555555555555556	3\\
0.560560560560561	3\\
0.565565565565566	3\\
0.570570570570571	3\\
0.575575575575576	3\\
0.580580580580581	3\\
0.585585585585586	3\\
0.590590590590591	3\\
0.595595595595596	3\\
0.600600600600601	3\\
0.605605605605606	3\\
0.610610610610611	3\\
0.615615615615616	3\\
0.620620620620621	3\\
0.625625625625626	3\\
0.630630630630631	3\\
0.635635635635636	3\\
0.640640640640641	3\\
0.645645645645646	3\\
0.650650650650651	3\\
0.655655655655656	3\\
0.660660660660661	3\\
0.665665665665666	3\\
0.670670670670671	3\\
0.675675675675676	3\\
0.680680680680681	3\\
0.685685685685686	3\\
0.690690690690691	3\\
0.695695695695696	3\\
0.700700700700701	3\\
0.705705705705706	3\\
0.710710710710711	3\\
0.715715715715716	3\\
0.720720720720721	3\\
0.725725725725726	3\\
0.730730730730731	3\\
0.735735735735736	3\\
0.740740740740741	3\\
0.745745745745746	3\\
0.750750750750751	3\\
0.755755755755756	3\\
0.760760760760761	3\\
0.765765765765766	3\\
0.770770770770771	3\\
0.775775775775776	3\\
0.780780780780781	3\\
0.785785785785786	3\\
0.790790790790791	3\\
0.795795795795796	3\\
0.800800800800801	3\\
0.805805805805806	3\\
0.810810810810811	3\\
0.815815815815816	3\\
0.820820820820821	3\\
0.825825825825826	3\\
0.830830830830831	3\\
0.835835835835836	3\\
0.840840840840841	3\\
0.845845845845846	3\\
0.850850850850851	3\\
0.855855855855856	3\\
0.860860860860861	3\\
0.865865865865866	3\\
0.870870870870871	3\\
0.875875875875876	3\\
0.880880880880881	3\\
0.885885885885886	3\\
0.890890890890891	3\\
0.895895895895896	3\\
0.900900900900901	3\\
0.905905905905906	3\\
0.910910910910911	3\\
0.915915915915916	3\\
0.920920920920921	3\\
0.925925925925926	3\\
0.930930930930931	3\\
0.935935935935936	3\\
0.940940940940941	3\\
0.945945945945946	3\\
0.950950950950951	3\\
0.955955955955956	3\\
0.960960960960961	3\\
0.965965965965966	3\\
0.970970970970971	3\\
0.975975975975976	3\\
0.980980980980981	3\\
0.985985985985986	3\\
0.990990990990991	3\\
0.995995995995996	3\\
1.001001001001	3.001001001001\\
1.00600600600601	3.00600600600601\\
1.01101101101101	3.01101101101101\\
1.01601601601602	3.01601601601602\\
1.02102102102102	3.02102102102102\\
1.02602602602603	3.02602602602603\\
1.03103103103103	3.03103103103103\\
1.03603603603604	3.03603603603604\\
1.04104104104104	3.04104104104104\\
1.04604604604605	3.04604604604605\\
1.05105105105105	3.05105105105105\\
1.05605605605606	3.05605605605606\\
1.06106106106106	3.06106106106106\\
1.06606606606607	3.06606606606607\\
1.07107107107107	3.07107107107107\\
1.07607607607608	3.07607607607608\\
1.08108108108108	3.08108108108108\\
1.08608608608609	3.08608608608609\\
1.09109109109109	3.09109109109109\\
1.0960960960961	3.0960960960961\\
1.1011011011011	3.1011011011011\\
1.10610610610611	3.10610610610611\\
1.11111111111111	3.11111111111111\\
1.11611611611612	3.11611611611612\\
1.12112112112112	3.12112112112112\\
1.12612612612613	3.12612612612613\\
1.13113113113113	3.13113113113113\\
1.13613613613614	3.13613613613614\\
1.14114114114114	3.14114114114114\\
1.14614614614615	3.14614614614615\\
1.15115115115115	3.15115115115115\\
1.15615615615616	3.15615615615616\\
1.16116116116116	3.16116116116116\\
1.16616616616617	3.16616616616617\\
1.17117117117117	3.17117117117117\\
1.17617617617618	3.17617617617618\\
1.18118118118118	3.18118118118118\\
1.18618618618619	3.18618618618619\\
1.19119119119119	3.19119119119119\\
1.1961961961962	3.1961961961962\\
1.2012012012012	3.2012012012012\\
1.20620620620621	3.20620620620621\\
1.21121121121121	3.21121121121121\\
1.21621621621622	3.21621621621622\\
1.22122122122122	3.22122122122122\\
1.22622622622623	3.22622622622623\\
1.23123123123123	3.23123123123123\\
1.23623623623624	3.23623623623624\\
1.24124124124124	3.24124124124124\\
1.24624624624625	3.24624624624625\\
1.25125125125125	3.25125125125125\\
1.25625625625626	3.25625625625626\\
1.26126126126126	3.26126126126126\\
1.26626626626627	3.26626626626627\\
1.27127127127127	3.27127127127127\\
1.27627627627628	3.27627627627628\\
1.28128128128128	3.28128128128128\\
1.28628628628629	3.28628628628629\\
1.29129129129129	3.29129129129129\\
1.2962962962963	3.2962962962963\\
1.3013013013013	3.3013013013013\\
1.30630630630631	3.30630630630631\\
1.31131131131131	3.31131131131131\\
1.31631631631632	3.31631631631632\\
1.32132132132132	3.32132132132132\\
1.32632632632633	3.32632632632633\\
1.33133133133133	3.33133133133133\\
1.33633633633634	3.33633633633634\\
1.34134134134134	3.34134134134134\\
1.34634634634635	3.34634634634635\\
1.35135135135135	3.35135135135135\\
1.35635635635636	3.35635635635636\\
1.36136136136136	3.36136136136136\\
1.36636636636637	3.36636636636637\\
1.37137137137137	3.37137137137137\\
1.37637637637638	3.37637637637638\\
1.38138138138138	3.38138138138138\\
1.38638638638639	3.38638638638639\\
1.39139139139139	3.39139139139139\\
1.3963963963964	3.3963963963964\\
1.4014014014014	3.4014014014014\\
1.40640640640641	3.40640640640641\\
1.41141141141141	3.41141141141141\\
1.41641641641642	3.41641641641642\\
1.42142142142142	3.42142142142142\\
1.42642642642643	3.42642642642643\\
1.43143143143143	3.43143143143143\\
1.43643643643644	3.43643643643644\\
1.44144144144144	3.44144144144144\\
1.44644644644645	3.44644644644645\\
1.45145145145145	3.45145145145145\\
1.45645645645646	3.45645645645646\\
1.46146146146146	3.46146146146146\\
1.46646646646647	3.46646646646647\\
1.47147147147147	3.47147147147147\\
1.47647647647648	3.47647647647648\\
1.48148148148148	3.48148148148148\\
1.48648648648649	3.48648648648649\\
1.49149149149149	3.49149149149149\\
1.4964964964965	3.4964964964965\\
1.5015015015015	3.4984984984985\\
1.50650650650651	3.49349349349349\\
1.51151151151151	3.48848848848849\\
1.51651651651652	3.48348348348348\\
1.52152152152152	3.47847847847848\\
1.52652652652653	3.47347347347347\\
1.53153153153153	3.46846846846847\\
1.53653653653654	3.46346346346346\\
1.54154154154154	3.45845845845846\\
1.54654654654655	3.45345345345345\\
1.55155155155155	3.44844844844845\\
1.55655655655656	3.44344344344344\\
1.56156156156156	3.43843843843844\\
1.56656656656657	3.43343343343343\\
1.57157157157157	3.42842842842843\\
1.57657657657658	3.42342342342342\\
1.58158158158158	3.41841841841842\\
1.58658658658659	3.41341341341341\\
1.59159159159159	3.40840840840841\\
1.5965965965966	3.4034034034034\\
1.6016016016016	3.3983983983984\\
1.60660660660661	3.39339339339339\\
1.61161161161161	3.38838838838839\\
1.61661661661662	3.38338338338338\\
1.62162162162162	3.37837837837838\\
1.62662662662663	3.37337337337337\\
1.63163163163163	3.36836836836837\\
1.63663663663664	3.36336336336336\\
1.64164164164164	3.35835835835836\\
1.64664664664665	3.35335335335335\\
1.65165165165165	3.34834834834835\\
1.65665665665666	3.34334334334334\\
1.66166166166166	3.33833833833834\\
1.66666666666667	3.33333333333333\\
1.67167167167167	3.32832832832833\\
1.67667667667668	3.32332332332332\\
1.68168168168168	3.31831831831832\\
1.68668668668669	3.31331331331331\\
1.69169169169169	3.30830830830831\\
1.6966966966967	3.3033033033033\\
1.7017017017017	3.2982982982983\\
1.70670670670671	3.29329329329329\\
1.71171171171171	3.28828828828829\\
1.71671671671672	3.28328328328328\\
1.72172172172172	3.27827827827828\\
1.72672672672673	3.27327327327327\\
1.73173173173173	3.26826826826827\\
1.73673673673674	3.26326326326326\\
1.74174174174174	3.25825825825826\\
1.74674674674675	3.25325325325325\\
1.75175175175175	3.24824824824825\\
1.75675675675676	3.24324324324324\\
1.76176176176176	3.23823823823824\\
1.76676676676677	3.23323323323323\\
1.77177177177177	3.22822822822823\\
1.77677677677678	3.22322322322322\\
1.78178178178178	3.21821821821822\\
1.78678678678679	3.21321321321321\\
1.79179179179179	3.20820820820821\\
1.7967967967968	3.2032032032032\\
1.8018018018018	3.1981981981982\\
1.80680680680681	3.19319319319319\\
1.81181181181181	3.18818818818819\\
1.81681681681682	3.18318318318318\\
1.82182182182182	3.17817817817818\\
1.82682682682683	3.17317317317317\\
1.83183183183183	3.16816816816817\\
1.83683683683684	3.16316316316316\\
1.84184184184184	3.15815815815816\\
1.84684684684685	3.15315315315315\\
1.85185185185185	3.14814814814815\\
1.85685685685686	3.14314314314314\\
1.86186186186186	3.13813813813814\\
1.86686686686687	3.13313313313313\\
1.87187187187187	3.12812812812813\\
1.87687687687688	3.12312312312312\\
1.88188188188188	3.11811811811812\\
1.88688688688689	3.11311311311311\\
1.89189189189189	3.10810810810811\\
1.8968968968969	3.1031031031031\\
1.9019019019019	3.0980980980981\\
1.90690690690691	3.09309309309309\\
1.91191191191191	3.08808808808809\\
1.91691691691692	3.08308308308308\\
1.92192192192192	3.07807807807808\\
1.92692692692693	3.07307307307307\\
1.93193193193193	3.06806806806807\\
1.93693693693694	3.06306306306306\\
1.94194194194194	3.05805805805806\\
1.94694694694695	3.05305305305305\\
1.95195195195195	3.04804804804805\\
1.95695695695696	3.04304304304304\\
1.96196196196196	3.03803803803804\\
1.96696696696697	3.03303303303303\\
1.97197197197197	3.02802802802803\\
1.97697697697698	3.02302302302302\\
1.98198198198198	3.01801801801802\\
1.98698698698699	3.01301301301301\\
1.99199199199199	3.00800800800801\\
1.996996996997	3.003003003003\\
2.002002002002	3\\
2.00700700700701	3\\
2.01201201201201	3\\
2.01701701701702	3\\
2.02202202202202	3\\
2.02702702702703	3\\
2.03203203203203	3\\
2.03703703703704	3\\
2.04204204204204	3\\
2.04704704704705	3\\
2.05205205205205	3\\
2.05705705705706	3\\
2.06206206206206	3\\
2.06706706706707	3\\
2.07207207207207	3\\
2.07707707707708	3\\
2.08208208208208	3\\
2.08708708708709	3\\
2.09209209209209	3\\
2.0970970970971	3\\
2.1021021021021	3\\
2.10710710710711	3\\
2.11211211211211	3\\
2.11711711711712	3\\
2.12212212212212	3\\
2.12712712712713	3\\
2.13213213213213	3\\
2.13713713713714	3\\
2.14214214214214	3\\
2.14714714714715	3\\
2.15215215215215	3\\
2.15715715715716	3\\
2.16216216216216	3\\
2.16716716716717	3\\
2.17217217217217	3\\
2.17717717717718	3\\
2.18218218218218	3\\
2.18718718718719	3\\
2.19219219219219	3\\
2.1971971971972	3\\
2.2022022022022	3\\
2.20720720720721	3\\
2.21221221221221	3\\
2.21721721721722	3\\
2.22222222222222	3\\
2.22722722722723	3\\
2.23223223223223	3\\
2.23723723723724	3\\
2.24224224224224	3\\
2.24724724724725	3\\
2.25225225225225	3\\
2.25725725725726	3\\
2.26226226226226	3\\
2.26726726726727	3\\
2.27227227227227	3\\
2.27727727727728	3\\
2.28228228228228	3\\
2.28728728728729	3\\
2.29229229229229	3\\
2.2972972972973	3\\
2.3023023023023	3\\
2.30730730730731	3\\
2.31231231231231	3\\
2.31731731731732	3\\
2.32232232232232	3\\
2.32732732732733	3\\
2.33233233233233	3\\
2.33733733733734	3\\
2.34234234234234	3\\
2.34734734734735	3\\
2.35235235235235	3\\
2.35735735735736	3\\
2.36236236236236	3\\
2.36736736736737	3\\
2.37237237237237	3\\
2.37737737737738	3\\
2.38238238238238	3\\
2.38738738738739	3\\
2.39239239239239	3\\
2.3973973973974	3\\
2.4024024024024	3\\
2.40740740740741	3\\
2.41241241241241	3\\
2.41741741741742	3\\
2.42242242242242	3\\
2.42742742742743	3\\
2.43243243243243	3\\
2.43743743743744	3\\
2.44244244244244	3\\
2.44744744744745	3\\
2.45245245245245	3\\
2.45745745745746	3\\
2.46246246246246	3\\
2.46746746746747	3\\
2.47247247247247	3\\
2.47747747747748	3\\
2.48248248248248	3\\
2.48748748748749	3\\
2.49249249249249	3\\
2.4974974974975	3\\
2.5025025025025	3\\
2.50750750750751	3\\
2.51251251251251	3\\
2.51751751751752	3\\
2.52252252252252	3\\
2.52752752752753	3\\
2.53253253253253	3\\
2.53753753753754	3\\
2.54254254254254	3\\
2.54754754754755	3\\
2.55255255255255	3\\
2.55755755755756	3\\
2.56256256256256	3\\
2.56756756756757	3\\
2.57257257257257	3\\
2.57757757757758	3\\
2.58258258258258	3\\
2.58758758758759	3\\
2.59259259259259	3\\
2.5975975975976	3\\
2.6026026026026	3\\
2.60760760760761	3\\
2.61261261261261	3\\
2.61761761761762	3\\
2.62262262262262	3\\
2.62762762762763	3\\
2.63263263263263	3\\
2.63763763763764	3\\
2.64264264264264	3\\
2.64764764764765	3\\
2.65265265265265	3\\
2.65765765765766	3\\
2.66266266266266	3\\
2.66766766766767	3\\
2.67267267267267	3\\
2.67767767767768	3\\
2.68268268268268	3\\
2.68768768768769	3\\
2.69269269269269	3\\
2.6976976976977	3\\
2.7027027027027	3\\
2.70770770770771	3\\
2.71271271271271	3\\
2.71771771771772	3\\
2.72272272272272	3\\
2.72772772772773	3\\
2.73273273273273	3\\
2.73773773773774	3\\
2.74274274274274	3\\
2.74774774774775	3\\
2.75275275275275	3\\
2.75775775775776	3\\
2.76276276276276	3\\
2.76776776776777	3\\
2.77277277277277	3\\
2.77777777777778	3\\
2.78278278278278	3\\
2.78778778778779	3\\
2.79279279279279	3\\
2.7977977977978	3\\
2.8028028028028	3\\
2.80780780780781	3\\
2.81281281281281	3\\
2.81781781781782	3\\
2.82282282282282	3\\
2.82782782782783	3\\
2.83283283283283	3\\
2.83783783783784	3\\
2.84284284284284	3\\
2.84784784784785	3\\
2.85285285285285	3\\
2.85785785785786	3\\
2.86286286286286	3\\
2.86786786786787	3\\
2.87287287287287	3\\
2.87787787787788	3\\
2.88288288288288	3\\
2.88788788788789	3\\
2.89289289289289	3\\
2.8978978978979	3\\
2.9029029029029	3\\
2.90790790790791	3\\
2.91291291291291	3\\
2.91791791791792	3\\
2.92292292292292	3\\
2.92792792792793	3\\
2.93293293293293	3\\
2.93793793793794	3\\
2.94294294294294	3\\
2.94794794794795	3\\
2.95295295295295	3\\
2.95795795795796	3\\
2.96296296296296	3\\
2.96796796796797	3\\
2.97297297297297	3\\
2.97797797797798	3\\
2.98298298298298	3\\
2.98798798798799	3\\
2.99299299299299	3\\
2.997997997998	3\\
3.003003003003	3.003003003003\\
3.00800800800801	3.00800800800801\\
3.01301301301301	3.01301301301301\\
3.01801801801802	3.01801801801802\\
3.02302302302302	3.02302302302302\\
3.02802802802803	3.02802802802803\\
3.03303303303303	3.03303303303303\\
3.03803803803804	3.03803803803804\\
3.04304304304304	3.04304304304304\\
3.04804804804805	3.04804804804805\\
3.05305305305305	3.05305305305305\\
3.05805805805806	3.05805805805806\\
3.06306306306306	3.06306306306306\\
3.06806806806807	3.06806806806807\\
3.07307307307307	3.07307307307307\\
3.07807807807808	3.07807807807808\\
3.08308308308308	3.08308308308308\\
3.08808808808809	3.08808808808809\\
3.09309309309309	3.09309309309309\\
3.0980980980981	3.0980980980981\\
3.1031031031031	3.1031031031031\\
3.10810810810811	3.10810810810811\\
3.11311311311311	3.11311311311311\\
3.11811811811812	3.11811811811812\\
3.12312312312312	3.12312312312312\\
3.12812812812813	3.12812812812813\\
3.13313313313313	3.13313313313313\\
3.13813813813814	3.13813813813814\\
3.14314314314314	3.14314314314314\\
3.14814814814815	3.14814814814815\\
3.15315315315315	3.15315315315315\\
3.15815815815816	3.15815815815816\\
3.16316316316316	3.16316316316316\\
3.16816816816817	3.16816816816817\\
3.17317317317317	3.17317317317317\\
3.17817817817818	3.17817817817818\\
3.18318318318318	3.18318318318318\\
3.18818818818819	3.18818818818819\\
3.19319319319319	3.19319319319319\\
3.1981981981982	3.1981981981982\\
3.2032032032032	3.2032032032032\\
3.20820820820821	3.20820820820821\\
3.21321321321321	3.21321321321321\\
3.21821821821822	3.21821821821822\\
3.22322322322322	3.22322322322322\\
3.22822822822823	3.22822822822823\\
3.23323323323323	3.23323323323323\\
3.23823823823824	3.23823823823824\\
3.24324324324324	3.24324324324324\\
3.24824824824825	3.24824824824825\\
3.25325325325325	3.25325325325325\\
3.25825825825826	3.25825825825826\\
3.26326326326326	3.26326326326326\\
3.26826826826827	3.26826826826827\\
3.27327327327327	3.27327327327327\\
3.27827827827828	3.27827827827828\\
3.28328328328328	3.28328328328328\\
3.28828828828829	3.28828828828829\\
3.29329329329329	3.29329329329329\\
3.2982982982983	3.2982982982983\\
3.3033033033033	3.3033033033033\\
3.30830830830831	3.30830830830831\\
3.31331331331331	3.31331331331331\\
3.31831831831832	3.31831831831832\\
3.32332332332332	3.32332332332332\\
3.32832832832833	3.32832832832833\\
3.33333333333333	3.33333333333333\\
3.33833833833834	3.33833833833834\\
3.34334334334334	3.34334334334334\\
3.34834834834835	3.34834834834835\\
3.35335335335335	3.35335335335335\\
3.35835835835836	3.35835835835836\\
3.36336336336336	3.36336336336336\\
3.36836836836837	3.36836836836837\\
3.37337337337337	3.37337337337337\\
3.37837837837838	3.37837837837838\\
3.38338338338338	3.38338338338338\\
3.38838838838839	3.38838838838839\\
3.39339339339339	3.39339339339339\\
3.3983983983984	3.3983983983984\\
3.4034034034034	3.4034034034034\\
3.40840840840841	3.40840840840841\\
3.41341341341341	3.41341341341341\\
3.41841841841842	3.41841841841842\\
3.42342342342342	3.42342342342342\\
3.42842842842843	3.42842842842843\\
3.43343343343343	3.43343343343343\\
3.43843843843844	3.43843843843844\\
3.44344344344344	3.44344344344344\\
3.44844844844845	3.44844844844845\\
3.45345345345345	3.45345345345345\\
3.45845845845846	3.45845845845846\\
3.46346346346346	3.46346346346346\\
3.46846846846847	3.46846846846847\\
3.47347347347347	3.47347347347347\\
3.47847847847848	3.47847847847848\\
3.48348348348348	3.48348348348348\\
3.48848848848849	3.48848848848849\\
3.49349349349349	3.49349349349349\\
3.4984984984985	3.4984984984985\\
3.5035035035035	3.5035035035035\\
3.50850850850851	3.50850850850851\\
3.51351351351351	3.51351351351351\\
3.51851851851852	3.51851851851852\\
3.52352352352352	3.52352352352352\\
3.52852852852853	3.52852852852853\\
3.53353353353353	3.53353353353353\\
3.53853853853854	3.53853853853854\\
3.54354354354354	3.54354354354354\\
3.54854854854855	3.54854854854855\\
3.55355355355355	3.55355355355355\\
3.55855855855856	3.55855855855856\\
3.56356356356356	3.56356356356356\\
3.56856856856857	3.56856856856857\\
3.57357357357357	3.57357357357357\\
3.57857857857858	3.57857857857858\\
3.58358358358358	3.58358358358358\\
3.58858858858859	3.58858858858859\\
3.59359359359359	3.59359359359359\\
3.5985985985986	3.5985985985986\\
3.6036036036036	3.6036036036036\\
3.60860860860861	3.60860860860861\\
3.61361361361361	3.61361361361361\\
3.61861861861862	3.61861861861862\\
3.62362362362362	3.62362362362362\\
3.62862862862863	3.62862862862863\\
3.63363363363363	3.63363363363363\\
3.63863863863864	3.63863863863864\\
3.64364364364364	3.64364364364364\\
3.64864864864865	3.64864864864865\\
3.65365365365365	3.65365365365365\\
3.65865865865866	3.65865865865866\\
3.66366366366366	3.66366366366366\\
3.66866866866867	3.66866866866867\\
3.67367367367367	3.67367367367367\\
3.67867867867868	3.67867867867868\\
3.68368368368368	3.68368368368368\\
3.68868868868869	3.68868868868869\\
3.69369369369369	3.69369369369369\\
3.6986986986987	3.6986986986987\\
3.7037037037037	3.7037037037037\\
3.70870870870871	3.70870870870871\\
3.71371371371371	3.71371371371371\\
3.71871871871872	3.71871871871872\\
3.72372372372372	3.72372372372372\\
3.72872872872873	3.72872872872873\\
3.73373373373373	3.73373373373373\\
3.73873873873874	3.73873873873874\\
3.74374374374374	3.74374374374374\\
3.74874874874875	3.74874874874875\\
3.75375375375375	3.75375375375375\\
3.75875875875876	3.75875875875876\\
3.76376376376376	3.76376376376376\\
3.76876876876877	3.76876876876877\\
3.77377377377377	3.77377377377377\\
3.77877877877878	3.77877877877878\\
3.78378378378378	3.78378378378378\\
3.78878878878879	3.78878878878879\\
3.79379379379379	3.79379379379379\\
3.7987987987988	3.7987987987988\\
3.8038038038038	3.8038038038038\\
3.80880880880881	3.80880880880881\\
3.81381381381381	3.81381381381381\\
3.81881881881882	3.81881881881882\\
3.82382382382382	3.82382382382382\\
3.82882882882883	3.82882882882883\\
3.83383383383383	3.83383383383383\\
3.83883883883884	3.83883883883884\\
3.84384384384384	3.84384384384384\\
3.84884884884885	3.84884884884885\\
3.85385385385385	3.85385385385385\\
3.85885885885886	3.85885885885886\\
3.86386386386386	3.86386386386386\\
3.86886886886887	3.86886886886887\\
3.87387387387387	3.87387387387387\\
3.87887887887888	3.87887887887888\\
3.88388388388388	3.88388388388388\\
3.88888888888889	3.88888888888889\\
3.89389389389389	3.89389389389389\\
3.8988988988989	3.8988988988989\\
3.9039039039039	3.9039039039039\\
3.90890890890891	3.90890890890891\\
3.91391391391391	3.91391391391391\\
3.91891891891892	3.91891891891892\\
3.92392392392392	3.92392392392392\\
3.92892892892893	3.92892892892893\\
3.93393393393393	3.93393393393393\\
3.93893893893894	3.93893893893894\\
3.94394394394394	3.94394394394394\\
3.94894894894895	3.94894894894895\\
3.95395395395395	3.95395395395395\\
3.95895895895896	3.95895895895896\\
3.96396396396396	3.96396396396396\\
3.96896896896897	3.96896896896897\\
3.97397397397397	3.97397397397397\\
3.97897897897898	3.97897897897898\\
3.98398398398398	3.98398398398398\\
3.98898898898899	3.98898898898899\\
3.99399399399399	3.99399399399399\\
3.998998998999	3.998998998999\\
4.004004004004	4\\
4.00900900900901	4\\
4.01401401401401	4\\
4.01901901901902	4\\
4.02402402402402	4\\
4.02902902902903	4\\
4.03403403403403	4\\
4.03903903903904	4\\
4.04404404404404	4\\
4.04904904904905	4\\
4.05405405405405	4\\
4.05905905905906	4\\
4.06406406406406	4\\
4.06906906906907	4\\
4.07407407407407	4\\
4.07907907907908	4\\
4.08408408408408	4\\
4.08908908908909	4\\
4.09409409409409	4\\
4.0990990990991	4\\
4.1041041041041	4\\
4.10910910910911	4\\
4.11411411411411	4\\
4.11911911911912	4\\
4.12412412412412	4\\
4.12912912912913	4\\
4.13413413413413	4\\
4.13913913913914	4\\
4.14414414414414	4\\
4.14914914914915	4\\
4.15415415415415	4\\
4.15915915915916	4\\
4.16416416416416	4\\
4.16916916916917	4\\
4.17417417417417	4\\
4.17917917917918	4\\
4.18418418418418	4\\
4.18918918918919	4\\
4.19419419419419	4\\
4.1991991991992	4\\
4.2042042042042	4\\
4.20920920920921	4\\
4.21421421421421	4\\
4.21921921921922	4\\
4.22422422422422	4\\
4.22922922922923	4\\
4.23423423423423	4\\
4.23923923923924	4\\
4.24424424424424	4\\
4.24924924924925	4\\
4.25425425425425	4\\
4.25925925925926	4\\
4.26426426426426	4\\
4.26926926926927	4\\
4.27427427427427	4\\
4.27927927927928	4\\
4.28428428428428	4\\
4.28928928928929	4\\
4.29429429429429	4\\
4.2992992992993	4\\
4.3043043043043	4\\
4.30930930930931	4\\
4.31431431431431	4\\
4.31931931931932	4\\
4.32432432432432	4\\
4.32932932932933	4\\
4.33433433433433	4\\
4.33933933933934	4\\
4.34434434434434	4\\
4.34934934934935	4\\
4.35435435435435	4\\
4.35935935935936	4\\
4.36436436436436	4\\
4.36936936936937	4\\
4.37437437437437	4\\
4.37937937937938	4\\
4.38438438438438	4\\
4.38938938938939	4\\
4.39439439439439	4\\
4.3993993993994	4\\
4.4044044044044	4\\
4.40940940940941	4\\
4.41441441441441	4\\
4.41941941941942	4\\
4.42442442442442	4\\
4.42942942942943	4\\
4.43443443443443	4\\
4.43943943943944	4\\
4.44444444444444	4\\
4.44944944944945	4\\
4.45445445445445	4\\
4.45945945945946	4\\
4.46446446446446	4\\
4.46946946946947	4\\
4.47447447447447	4\\
4.47947947947948	4\\
4.48448448448448	4\\
4.48948948948949	4\\
4.49449449449449	4\\
4.4994994994995	4\\
4.5045045045045	4\\
4.50950950950951	4\\
4.51451451451451	4\\
4.51951951951952	4\\
4.52452452452452	4\\
4.52952952952953	4\\
4.53453453453453	4\\
4.53953953953954	4\\
4.54454454454454	4\\
4.54954954954955	4\\
4.55455455455455	4\\
4.55955955955956	4\\
4.56456456456456	4\\
4.56956956956957	4\\
4.57457457457457	4\\
4.57957957957958	4\\
4.58458458458458	4\\
4.58958958958959	4\\
4.59459459459459	4\\
4.5995995995996	4\\
4.6046046046046	4\\
4.60960960960961	4\\
4.61461461461461	4\\
4.61961961961962	4\\
4.62462462462462	4\\
4.62962962962963	4\\
4.63463463463463	4\\
4.63963963963964	4\\
4.64464464464464	4\\
4.64964964964965	4\\
4.65465465465465	4\\
4.65965965965966	4\\
4.66466466466466	4\\
4.66966966966967	4\\
4.67467467467467	4\\
4.67967967967968	4\\
4.68468468468468	4\\
4.68968968968969	4\\
4.69469469469469	4\\
4.6996996996997	4\\
4.7047047047047	4\\
4.70970970970971	4\\
4.71471471471471	4\\
4.71971971971972	4\\
4.72472472472472	4\\
4.72972972972973	4\\
4.73473473473473	4\\
4.73973973973974	4\\
4.74474474474474	4\\
4.74974974974975	4\\
4.75475475475475	4\\
4.75975975975976	4\\
4.76476476476476	4\\
4.76976976976977	4\\
4.77477477477477	4\\
4.77977977977978	4\\
4.78478478478478	4\\
4.78978978978979	4\\
4.79479479479479	4\\
4.7997997997998	4\\
4.8048048048048	4\\
4.80980980980981	4\\
4.81481481481481	4\\
4.81981981981982	4\\
4.82482482482482	4\\
4.82982982982983	4\\
4.83483483483483	4\\
4.83983983983984	4\\
4.84484484484484	4\\
4.84984984984985	4\\
4.85485485485485	4\\
4.85985985985986	4\\
4.86486486486486	4\\
4.86986986986987	4\\
4.87487487487487	4\\
4.87987987987988	4\\
4.88488488488488	4\\
4.88988988988989	4\\
4.89489489489489	4\\
4.8998998998999	4\\
4.9049049049049	4\\
4.90990990990991	4\\
4.91491491491491	4\\
4.91991991991992	4\\
4.92492492492492	4\\
4.92992992992993	4\\
4.93493493493493	4\\
4.93993993993994	4\\
4.94494494494494	4\\
4.94994994994995	4\\
4.95495495495495	4\\
4.95995995995996	4\\
4.96496496496496	4\\
4.96996996996997	4\\
4.97497497497497	4\\
4.97997997997998	4\\
4.98498498498498	4\\
4.98998998998999	4\\
4.99499499499499	4\\
5	4\\
};
\addplot [color=red,dashed,line width=1.0pt,forget plot]
  table[row sep=crcr]{0	2\\
0.005005005005005	1.98998998998999\\
0.01001001001001	1.97997997997998\\
0.015015015015015	1.96996996996997\\
0.02002002002002	1.95995995995996\\
0.025025025025025	1.94994994994995\\
0.03003003003003	1.93993993993994\\
0.035035035035035	1.92992992992993\\
0.04004004004004	1.91991991991992\\
0.045045045045045	1.90990990990991\\
0.0500500500500501	1.8998998998999\\
0.0550550550550551	1.88988988988989\\
0.0600600600600601	1.87987987987988\\
0.0650650650650651	1.86986986986987\\
0.0700700700700701	1.85985985985986\\
0.0750750750750751	1.84984984984985\\
0.0800800800800801	1.83983983983984\\
0.0850850850850851	1.82982982982983\\
0.0900900900900901	1.81981981981982\\
0.0950950950950951	1.80980980980981\\
0.1001001001001	1.7997997997998\\
0.105105105105105	1.78978978978979\\
0.11011011011011	1.77977977977978\\
0.115115115115115	1.76976976976977\\
0.12012012012012	1.75975975975976\\
0.125125125125125	1.74974974974975\\
0.13013013013013	1.73973973973974\\
0.135135135135135	1.72972972972973\\
0.14014014014014	1.71971971971972\\
0.145145145145145	1.70970970970971\\
0.15015015015015	1.6996996996997\\
0.155155155155155	1.68968968968969\\
0.16016016016016	1.67967967967968\\
0.165165165165165	1.66966966966967\\
0.17017017017017	1.65965965965966\\
0.175175175175175	1.64964964964965\\
0.18018018018018	1.63963963963964\\
0.185185185185185	1.62962962962963\\
0.19019019019019	1.61961961961962\\
0.195195195195195	1.60960960960961\\
0.2002002002002	1.5995995995996\\
0.205205205205205	1.58958958958959\\
0.21021021021021	1.57957957957958\\
0.215215215215215	1.56956956956957\\
0.22022022022022	1.55955955955956\\
0.225225225225225	1.54954954954955\\
0.23023023023023	1.53953953953954\\
0.235235235235235	1.52952952952953\\
0.24024024024024	1.51951951951952\\
0.245245245245245	1.50950950950951\\
0.25025025025025	1.4994994994995\\
0.255255255255255	1.48948948948949\\
0.26026026026026	1.47947947947948\\
0.265265265265265	1.46946946946947\\
0.27027027027027	1.45945945945946\\
0.275275275275275	1.44944944944945\\
0.28028028028028	1.43943943943944\\
0.285285285285285	1.42942942942943\\
0.29029029029029	1.41941941941942\\
0.295295295295295	1.40940940940941\\
0.3003003003003	1.3993993993994\\
0.305305305305305	1.38938938938939\\
0.31031031031031	1.37937937937938\\
0.315315315315315	1.36936936936937\\
0.32032032032032	1.35935935935936\\
0.325325325325325	1.34934934934935\\
0.33033033033033	1.33933933933934\\
0.335335335335335	1.32932932932933\\
0.34034034034034	1.31931931931932\\
0.345345345345345	1.30930930930931\\
0.35035035035035	1.2992992992993\\
0.355355355355355	1.28928928928929\\
0.36036036036036	1.27927927927928\\
0.365365365365365	1.26926926926927\\
0.37037037037037	1.25925925925926\\
0.375375375375375	1.24924924924925\\
0.38038038038038	1.23923923923924\\
0.385385385385385	1.22922922922923\\
0.39039039039039	1.21921921921922\\
0.395395395395395	1.20920920920921\\
0.4004004004004	1.1991991991992\\
0.405405405405405	1.18918918918919\\
0.41041041041041	1.17917917917918\\
0.415415415415415	1.16916916916917\\
0.42042042042042	1.15915915915916\\
0.425425425425425	1.14914914914915\\
0.43043043043043	1.13913913913914\\
0.435435435435435	1.12912912912913\\
0.44044044044044	1.11911911911912\\
0.445445445445445	1.10910910910911\\
0.45045045045045	1.0990990990991\\
0.455455455455455	1.08908908908909\\
0.46046046046046	1.07907907907908\\
0.465465465465465	1.06906906906907\\
0.47047047047047	1.05905905905906\\
0.475475475475475	1.04904904904905\\
0.48048048048048	1.03903903903904\\
0.485485485485485	1.02902902902903\\
0.49049049049049	1.01901901901902\\
0.495495495495495	1.00900900900901\\
0.500500500500501	1.001001001001\\
0.505505505505506	1.01101101101101\\
0.510510510510511	1.02102102102102\\
0.515515515515516	1.03103103103103\\
0.520520520520521	1.04104104104104\\
0.525525525525526	1.05105105105105\\
0.530530530530531	1.06106106106106\\
0.535535535535536	1.07107107107107\\
0.540540540540541	1.08108108108108\\
0.545545545545546	1.09109109109109\\
0.550550550550551	1.1011011011011\\
0.555555555555556	1.11111111111111\\
0.560560560560561	1.12112112112112\\
0.565565565565566	1.13113113113113\\
0.570570570570571	1.14114114114114\\
0.575575575575576	1.15115115115115\\
0.580580580580581	1.16116116116116\\
0.585585585585586	1.17117117117117\\
0.590590590590591	1.18118118118118\\
0.595595595595596	1.19119119119119\\
0.600600600600601	1.2012012012012\\
0.605605605605606	1.21121121121121\\
0.610610610610611	1.22122122122122\\
0.615615615615616	1.23123123123123\\
0.620620620620621	1.24124124124124\\
0.625625625625626	1.25125125125125\\
0.630630630630631	1.26126126126126\\
0.635635635635636	1.27127127127127\\
0.640640640640641	1.28128128128128\\
0.645645645645646	1.29129129129129\\
0.650650650650651	1.3013013013013\\
0.655655655655656	1.31131131131131\\
0.660660660660661	1.32132132132132\\
0.665665665665666	1.33133133133133\\
0.670670670670671	1.32932932932933\\
0.675675675675676	1.32432432432432\\
0.680680680680681	1.31931931931932\\
0.685685685685686	1.31431431431431\\
0.690690690690691	1.30930930930931\\
0.695695695695696	1.3043043043043\\
0.700700700700701	1.2992992992993\\
0.705705705705706	1.29429429429429\\
0.710710710710711	1.28928928928929\\
0.715715715715716	1.28428428428428\\
0.720720720720721	1.27927927927928\\
0.725725725725726	1.27427427427427\\
0.730730730730731	1.26926926926927\\
0.735735735735736	1.26426426426426\\
0.740740740740741	1.25925925925926\\
0.745745745745746	1.25425425425425\\
0.750750750750751	1.24924924924925\\
0.755755755755756	1.24424424424424\\
0.760760760760761	1.23923923923924\\
0.765765765765766	1.23423423423423\\
0.770770770770771	1.22922922922923\\
0.775775775775776	1.22422422422422\\
0.780780780780781	1.21921921921922\\
0.785785785785786	1.21421421421421\\
0.790790790790791	1.20920920920921\\
0.795795795795796	1.2042042042042\\
0.800800800800801	1.1991991991992\\
0.805805805805806	1.19419419419419\\
0.810810810810811	1.18918918918919\\
0.815815815815816	1.18418418418418\\
0.820820820820821	1.17917917917918\\
0.825825825825826	1.17417417417417\\
0.830830830830831	1.16916916916917\\
0.835835835835836	1.16416416416416\\
0.840840840840841	1.15915915915916\\
0.845845845845846	1.15415415415415\\
0.850850850850851	1.14914914914915\\
0.855855855855856	1.14414414414414\\
0.860860860860861	1.13913913913914\\
0.865865865865866	1.13413413413413\\
0.870870870870871	1.12912912912913\\
0.875875875875876	1.12412412412412\\
0.880880880880881	1.11911911911912\\
0.885885885885886	1.11411411411411\\
0.890890890890891	1.10910910910911\\
0.895895895895896	1.1041041041041\\
0.900900900900901	1.0990990990991\\
0.905905905905906	1.09409409409409\\
0.910910910910911	1.08908908908909\\
0.915915915915916	1.08408408408408\\
0.920920920920921	1.07907907907908\\
0.925925925925926	1.07407407407407\\
0.930930930930931	1.06906906906907\\
0.935935935935936	1.06406406406406\\
0.940940940940941	1.05905905905906\\
0.945945945945946	1.05405405405405\\
0.950950950950951	1.04904904904905\\
0.955955955955956	1.04404404404404\\
0.960960960960961	1.03903903903904\\
0.965965965965966	1.03403403403403\\
0.970970970970971	1.02902902902903\\
0.975975975975976	1.02402402402402\\
0.980980980980981	1.01901901901902\\
0.985985985985986	1.01401401401401\\
0.990990990990991	1.00900900900901\\
0.995995995995996	1.004004004004\\
1.001001001001	1.001001001001\\
1.00600600600601	1.00600600600601\\
1.01101101101101	1.01101101101101\\
1.01601601601602	1.01601601601602\\
1.02102102102102	1.02102102102102\\
1.02602602602603	1.02602602602603\\
1.03103103103103	1.03103103103103\\
1.03603603603604	1.03603603603604\\
1.04104104104104	1.04104104104104\\
1.04604604604605	1.04604604604605\\
1.05105105105105	1.05105105105105\\
1.05605605605606	1.05605605605606\\
1.06106106106106	1.06106106106106\\
1.06606606606607	1.06606606606607\\
1.07107107107107	1.07107107107107\\
1.07607607607608	1.07607607607608\\
1.08108108108108	1.08108108108108\\
1.08608608608609	1.08608608608609\\
1.09109109109109	1.09109109109109\\
1.0960960960961	1.0960960960961\\
1.1011011011011	1.1011011011011\\
1.10610610610611	1.10610610610611\\
1.11111111111111	1.11111111111111\\
1.11611611611612	1.11611611611612\\
1.12112112112112	1.12112112112112\\
1.12612612612613	1.12612612612613\\
1.13113113113113	1.13113113113113\\
1.13613613613614	1.13613613613614\\
1.14114114114114	1.14114114114114\\
1.14614614614615	1.14614614614615\\
1.15115115115115	1.15115115115115\\
1.15615615615616	1.15615615615616\\
1.16116116116116	1.16116116116116\\
1.16616616616617	1.16616616616617\\
1.17117117117117	1.17117117117117\\
1.17617617617618	1.17617617617618\\
1.18118118118118	1.18118118118118\\
1.18618618618619	1.18618618618619\\
1.19119119119119	1.19119119119119\\
1.1961961961962	1.1961961961962\\
1.2012012012012	1.2012012012012\\
1.20620620620621	1.20620620620621\\
1.21121121121121	1.21121121121121\\
1.21621621621622	1.21621621621622\\
1.22122122122122	1.22122122122122\\
1.22622622622623	1.22622622622623\\
1.23123123123123	1.23123123123123\\
1.23623623623624	1.23623623623624\\
1.24124124124124	1.24124124124124\\
1.24624624624625	1.24624624624625\\
1.25125125125125	1.25125125125125\\
1.25625625625626	1.25625625625626\\
1.26126126126126	1.26126126126126\\
1.26626626626627	1.26626626626627\\
1.27127127127127	1.27127127127127\\
1.27627627627628	1.27627627627628\\
1.28128128128128	1.28128128128128\\
1.28628628628629	1.28628628628629\\
1.29129129129129	1.29129129129129\\
1.2962962962963	1.2962962962963\\
1.3013013013013	1.3013013013013\\
1.30630630630631	1.30630630630631\\
1.31131131131131	1.31131131131131\\
1.31631631631632	1.31631631631632\\
1.32132132132132	1.32132132132132\\
1.32632632632633	1.32632632632633\\
1.33133133133133	1.33133133133133\\
1.33633633633634	1.33633633633634\\
1.34134134134134	1.34134134134134\\
1.34634634634635	1.34634634634635\\
1.35135135135135	1.35135135135135\\
1.35635635635636	1.35635635635636\\
1.36136136136136	1.36136136136136\\
1.36636636636637	1.36636636636637\\
1.37137137137137	1.37137137137137\\
1.37637637637638	1.37637637637638\\
1.38138138138138	1.38138138138138\\
1.38638638638639	1.38638638638639\\
1.39139139139139	1.39139139139139\\
1.3963963963964	1.3963963963964\\
1.4014014014014	1.4014014014014\\
1.40640640640641	1.40640640640641\\
1.41141141141141	1.41141141141141\\
1.41641641641642	1.41641641641642\\
1.42142142142142	1.42142142142142\\
1.42642642642643	1.42642642642643\\
1.43143143143143	1.43143143143143\\
1.43643643643644	1.43643643643644\\
1.44144144144144	1.44144144144144\\
1.44644644644645	1.44644644644645\\
1.45145145145145	1.45145145145145\\
1.45645645645646	1.45645645645646\\
1.46146146146146	1.46146146146146\\
1.46646646646647	1.46646646646647\\
1.47147147147147	1.47147147147147\\
1.47647647647648	1.47647647647648\\
1.48148148148148	1.48148148148148\\
1.48648648648649	1.48648648648649\\
1.49149149149149	1.49149149149149\\
1.4964964964965	1.4964964964965\\
1.5015015015015	1.5015015015015\\
1.50650650650651	1.50650650650651\\
1.51151151151151	1.51151151151151\\
1.51651651651652	1.51651651651652\\
1.52152152152152	1.52152152152152\\
1.52652652652653	1.52652652652653\\
1.53153153153153	1.53153153153153\\
1.53653653653654	1.53653653653654\\
1.54154154154154	1.54154154154154\\
1.54654654654655	1.54654654654655\\
1.55155155155155	1.55155155155155\\
1.55655655655656	1.55655655655656\\
1.56156156156156	1.56156156156156\\
1.56656656656657	1.56656656656657\\
1.57157157157157	1.57157157157157\\
1.57657657657658	1.57657657657658\\
1.58158158158158	1.58158158158158\\
1.58658658658659	1.58658658658659\\
1.59159159159159	1.59159159159159\\
1.5965965965966	1.5965965965966\\
1.6016016016016	1.6016016016016\\
1.60660660660661	1.60660660660661\\
1.61161161161161	1.61161161161161\\
1.61661661661662	1.61661661661662\\
1.62162162162162	1.62162162162162\\
1.62662662662663	1.62662662662663\\
1.63163163163163	1.63163163163163\\
1.63663663663664	1.63663663663664\\
1.64164164164164	1.64164164164164\\
1.64664664664665	1.64664664664665\\
1.65165165165165	1.65165165165165\\
1.65665665665666	1.65665665665666\\
1.66166166166166	1.66166166166166\\
1.66666666666667	1.66666666666667\\
1.67167167167167	1.67167167167167\\
1.67667667667668	1.67667667667668\\
1.68168168168168	1.68168168168168\\
1.68668668668669	1.68668668668669\\
1.69169169169169	1.69169169169169\\
1.6966966966967	1.6966966966967\\
1.7017017017017	1.7017017017017\\
1.70670670670671	1.70670670670671\\
1.71171171171171	1.71171171171171\\
1.71671671671672	1.71671671671672\\
1.72172172172172	1.72172172172172\\
1.72672672672673	1.72672672672673\\
1.73173173173173	1.73173173173173\\
1.73673673673674	1.73673673673674\\
1.74174174174174	1.74174174174174\\
1.74674674674675	1.74674674674675\\
1.75175175175175	1.75175175175175\\
1.75675675675676	1.75675675675676\\
1.76176176176176	1.76176176176176\\
1.76676676676677	1.76676676676677\\
1.77177177177177	1.77177177177177\\
1.77677677677678	1.77677677677678\\
1.78178178178178	1.78178178178178\\
1.78678678678679	1.78678678678679\\
1.79179179179179	1.79179179179179\\
1.7967967967968	1.7967967967968\\
1.8018018018018	1.8018018018018\\
1.80680680680681	1.80680680680681\\
1.81181181181181	1.81181181181181\\
1.81681681681682	1.81681681681682\\
1.82182182182182	1.82182182182182\\
1.82682682682683	1.82682682682683\\
1.83183183183183	1.83183183183183\\
1.83683683683684	1.83683683683684\\
1.84184184184184	1.84184184184184\\
1.84684684684685	1.84684684684685\\
1.85185185185185	1.85185185185185\\
1.85685685685686	1.85685685685686\\
1.86186186186186	1.86186186186186\\
1.86686686686687	1.86686686686687\\
1.87187187187187	1.87187187187187\\
1.87687687687688	1.87687687687688\\
1.88188188188188	1.88188188188188\\
1.88688688688689	1.88688688688689\\
1.89189189189189	1.89189189189189\\
1.8968968968969	1.8968968968969\\
1.9019019019019	1.9019019019019\\
1.90690690690691	1.90690690690691\\
1.91191191191191	1.91191191191191\\
1.91691691691692	1.91691691691692\\
1.92192192192192	1.92192192192192\\
1.92692692692693	1.92692692692693\\
1.93193193193193	1.93193193193193\\
1.93693693693694	1.93693693693694\\
1.94194194194194	1.94194194194194\\
1.94694694694695	1.94694694694695\\
1.95195195195195	1.95195195195195\\
1.95695695695696	1.95695695695696\\
1.96196196196196	1.96196196196196\\
1.96696696696697	1.96696696696697\\
1.97197197197197	1.97197197197197\\
1.97697697697698	1.97697697697698\\
1.98198198198198	1.98198198198198\\
1.98698698698699	1.98698698698699\\
1.99199199199199	1.99199199199199\\
1.996996996997	1.996996996997\\
2.002002002002	2\\
2.00700700700701	2\\
2.01201201201201	2\\
2.01701701701702	2\\
2.02202202202202	2\\
2.02702702702703	2\\
2.03203203203203	2\\
2.03703703703704	2\\
2.04204204204204	2\\
2.04704704704705	2\\
2.05205205205205	2\\
2.05705705705706	2\\
2.06206206206206	2\\
2.06706706706707	2\\
2.07207207207207	2\\
2.07707707707708	2\\
2.08208208208208	2\\
2.08708708708709	2\\
2.09209209209209	2\\
2.0970970970971	2\\
2.1021021021021	2\\
2.10710710710711	2\\
2.11211211211211	2\\
2.11711711711712	2\\
2.12212212212212	2\\
2.12712712712713	2\\
2.13213213213213	2\\
2.13713713713714	2\\
2.14214214214214	2\\
2.14714714714715	2\\
2.15215215215215	2\\
2.15715715715716	2\\
2.16216216216216	2\\
2.16716716716717	2\\
2.17217217217217	2\\
2.17717717717718	2\\
2.18218218218218	2\\
2.18718718718719	2\\
2.19219219219219	2\\
2.1971971971972	2\\
2.2022022022022	2\\
2.20720720720721	2\\
2.21221221221221	2\\
2.21721721721722	2\\
2.22222222222222	2\\
2.22722722722723	2\\
2.23223223223223	2\\
2.23723723723724	2\\
2.24224224224224	2\\
2.24724724724725	2\\
2.25225225225225	2\\
2.25725725725726	2\\
2.26226226226226	2\\
2.26726726726727	2\\
2.27227227227227	2\\
2.27727727727728	2\\
2.28228228228228	2\\
2.28728728728729	2\\
2.29229229229229	2\\
2.2972972972973	2\\
2.3023023023023	2\\
2.30730730730731	2\\
2.31231231231231	2\\
2.31731731731732	2\\
2.32232232232232	2\\
2.32732732732733	2\\
2.33233233233233	2\\
2.33733733733734	2\\
2.34234234234234	2\\
2.34734734734735	2\\
2.35235235235235	2\\
2.35735735735736	2\\
2.36236236236236	2\\
2.36736736736737	2\\
2.37237237237237	2\\
2.37737737737738	2\\
2.38238238238238	2\\
2.38738738738739	2\\
2.39239239239239	2\\
2.3973973973974	2\\
2.4024024024024	2\\
2.40740740740741	2\\
2.41241241241241	2\\
2.41741741741742	2\\
2.42242242242242	2\\
2.42742742742743	2\\
2.43243243243243	2\\
2.43743743743744	2\\
2.44244244244244	2\\
2.44744744744745	2\\
2.45245245245245	2\\
2.45745745745746	2\\
2.46246246246246	2\\
2.46746746746747	2\\
2.47247247247247	2\\
2.47747747747748	2\\
2.48248248248248	2\\
2.48748748748749	2\\
2.49249249249249	2\\
2.4974974974975	2\\
2.5025025025025	2\\
2.50750750750751	2\\
2.51251251251251	2\\
2.51751751751752	2\\
2.52252252252252	2\\
2.52752752752753	2\\
2.53253253253253	2\\
2.53753753753754	2\\
2.54254254254254	2\\
2.54754754754755	2\\
2.55255255255255	2\\
2.55755755755756	2\\
2.56256256256256	2\\
2.56756756756757	2\\
2.57257257257257	2\\
2.57757757757758	2\\
2.58258258258258	2\\
2.58758758758759	2\\
2.59259259259259	2\\
2.5975975975976	2\\
2.6026026026026	2\\
2.60760760760761	2\\
2.61261261261261	2\\
2.61761761761762	2\\
2.62262262262262	2\\
2.62762762762763	2\\
2.63263263263263	2\\
2.63763763763764	2\\
2.64264264264264	2\\
2.64764764764765	2\\
2.65265265265265	2\\
2.65765765765766	2\\
2.66266266266266	2\\
2.66766766766767	2\\
2.67267267267267	2\\
2.67767767767768	2\\
2.68268268268268	2\\
2.68768768768769	2\\
2.69269269269269	2\\
2.6976976976977	2\\
2.7027027027027	2\\
2.70770770770771	2\\
2.71271271271271	2\\
2.71771771771772	2\\
2.72272272272272	2\\
2.72772772772773	2\\
2.73273273273273	2\\
2.73773773773774	2\\
2.74274274274274	2\\
2.74774774774775	2\\
2.75275275275275	2\\
2.75775775775776	2\\
2.76276276276276	2\\
2.76776776776777	2\\
2.77277277277277	2\\
2.77777777777778	2\\
2.78278278278278	2\\
2.78778778778779	2\\
2.79279279279279	2\\
2.7977977977978	2\\
2.8028028028028	2\\
2.80780780780781	2\\
2.81281281281281	2\\
2.81781781781782	2\\
2.82282282282282	2\\
2.82782782782783	2\\
2.83283283283283	2\\
2.83783783783784	2\\
2.84284284284284	2\\
2.84784784784785	2\\
2.85285285285285	2\\
2.85785785785786	2\\
2.86286286286286	2\\
2.86786786786787	2\\
2.87287287287287	2\\
2.87787787787788	2\\
2.88288288288288	2\\
2.88788788788789	2\\
2.89289289289289	2\\
2.8978978978979	2\\
2.9029029029029	2\\
2.90790790790791	2\\
2.91291291291291	2\\
2.91791791791792	2\\
2.92292292292292	2\\
2.92792792792793	2\\
2.93293293293293	2\\
2.93793793793794	2\\
2.94294294294294	2\\
2.94794794794795	2\\
2.95295295295295	2\\
2.95795795795796	2\\
2.96296296296296	2\\
2.96796796796797	2\\
2.97297297297297	2\\
2.97797797797798	2\\
2.98298298298298	2\\
2.98798798798799	2\\
2.99299299299299	2\\
2.997997997998	2\\
3.003003003003	2\\
3.00800800800801	2\\
3.01301301301301	2\\
3.01801801801802	2\\
3.02302302302302	2\\
3.02802802802803	2\\
3.03303303303303	2\\
3.03803803803804	2\\
3.04304304304304	2\\
3.04804804804805	2\\
3.05305305305305	2\\
3.05805805805806	2\\
3.06306306306306	2\\
3.06806806806807	2\\
3.07307307307307	2\\
3.07807807807808	2\\
3.08308308308308	2\\
3.08808808808809	2\\
3.09309309309309	2\\
3.0980980980981	2\\
3.1031031031031	2\\
3.10810810810811	2\\
3.11311311311311	2\\
3.11811811811812	2\\
3.12312312312312	2\\
3.12812812812813	2\\
3.13313313313313	2\\
3.13813813813814	2\\
3.14314314314314	2\\
3.14814814814815	2\\
3.15315315315315	2\\
3.15815815815816	2\\
3.16316316316316	2\\
3.16816816816817	2\\
3.17317317317317	2\\
3.17817817817818	2\\
3.18318318318318	2\\
3.18818818818819	2\\
3.19319319319319	2\\
3.1981981981982	2\\
3.2032032032032	2\\
3.20820820820821	2\\
3.21321321321321	2\\
3.21821821821822	2\\
3.22322322322322	2\\
3.22822822822823	2\\
3.23323323323323	2\\
3.23823823823824	2\\
3.24324324324324	2\\
3.24824824824825	2\\
3.25325325325325	2\\
3.25825825825826	2\\
3.26326326326326	2\\
3.26826826826827	2\\
3.27327327327327	2\\
3.27827827827828	2\\
3.28328328328328	2\\
3.28828828828829	2\\
3.29329329329329	2\\
3.2982982982983	2\\
3.3033033033033	2\\
3.30830830830831	2\\
3.31331331331331	2\\
3.31831831831832	2\\
3.32332332332332	2\\
3.32832832832833	2\\
3.33333333333333	2\\
3.33833833833834	2\\
3.34334334334334	2\\
3.34834834834835	2\\
3.35335335335335	2\\
3.35835835835836	2\\
3.36336336336336	2\\
3.36836836836837	2\\
3.37337337337337	2\\
3.37837837837838	2\\
3.38338338338338	2\\
3.38838838838839	2\\
3.39339339339339	2\\
3.3983983983984	2\\
3.4034034034034	2\\
3.40840840840841	2\\
3.41341341341341	2\\
3.41841841841842	2\\
3.42342342342342	2\\
3.42842842842843	2\\
3.43343343343343	2\\
3.43843843843844	2\\
3.44344344344344	2\\
3.44844844844845	2\\
3.45345345345345	2\\
3.45845845845846	2\\
3.46346346346346	2\\
3.46846846846847	2\\
3.47347347347347	2\\
3.47847847847848	2\\
3.48348348348348	2\\
3.48848848848849	2\\
3.49349349349349	2\\
3.4984984984985	2\\
3.5035035035035	2\\
3.50850850850851	2\\
3.51351351351351	2\\
3.51851851851852	2\\
3.52352352352352	2\\
3.52852852852853	2\\
3.53353353353353	2\\
3.53853853853854	2\\
3.54354354354354	2\\
3.54854854854855	2\\
3.55355355355355	2\\
3.55855855855856	2\\
3.56356356356356	2\\
3.56856856856857	2\\
3.57357357357357	2\\
3.57857857857858	2\\
3.58358358358358	2\\
3.58858858858859	2\\
3.59359359359359	2\\
3.5985985985986	2\\
3.6036036036036	2\\
3.60860860860861	2\\
3.61361361361361	2\\
3.61861861861862	2\\
3.62362362362362	2\\
3.62862862862863	2\\
3.63363363363363	2\\
3.63863863863864	2\\
3.64364364364364	2\\
3.64864864864865	2\\
3.65365365365365	2\\
3.65865865865866	2\\
3.66366366366366	2\\
3.66866866866867	2\\
3.67367367367367	2\\
3.67867867867868	2\\
3.68368368368368	2\\
3.68868868868869	2\\
3.69369369369369	2\\
3.6986986986987	2\\
3.7037037037037	2\\
3.70870870870871	2\\
3.71371371371371	2\\
3.71871871871872	2\\
3.72372372372372	2\\
3.72872872872873	2\\
3.73373373373373	2\\
3.73873873873874	2\\
3.74374374374374	2\\
3.74874874874875	2\\
3.75375375375375	2\\
3.75875875875876	2\\
3.76376376376376	2\\
3.76876876876877	2\\
3.77377377377377	2\\
3.77877877877878	2\\
3.78378378378378	2\\
3.78878878878879	2\\
3.79379379379379	2\\
3.7987987987988	2\\
3.8038038038038	2\\
3.80880880880881	2\\
3.81381381381381	2\\
3.81881881881882	2\\
3.82382382382382	2\\
3.82882882882883	2\\
3.83383383383383	2\\
3.83883883883884	2\\
3.84384384384384	2\\
3.84884884884885	2\\
3.85385385385385	2\\
3.85885885885886	2\\
3.86386386386386	2\\
3.86886886886887	2\\
3.87387387387387	2\\
3.87887887887888	2\\
3.88388388388388	2\\
3.88888888888889	2\\
3.89389389389389	2\\
3.8988988988989	2\\
3.9039039039039	2\\
3.90890890890891	2\\
3.91391391391391	2\\
3.91891891891892	2\\
3.92392392392392	2\\
3.92892892892893	2\\
3.93393393393393	2\\
3.93893893893894	2\\
3.94394394394394	2\\
3.94894894894895	2\\
3.95395395395395	2\\
3.95895895895896	2\\
3.96396396396396	2\\
3.96896896896897	2\\
3.97397397397397	2\\
3.97897897897898	2\\
3.98398398398398	2\\
3.98898898898899	2\\
3.99399399399399	2\\
3.998998998999	2\\
4.004004004004	2\\
4.00900900900901	2\\
4.01401401401401	2\\
4.01901901901902	2\\
4.02402402402402	2\\
4.02902902902903	2\\
4.03403403403403	2\\
4.03903903903904	2\\
4.04404404404404	2\\
4.04904904904905	2\\
4.05405405405405	2\\
4.05905905905906	2\\
4.06406406406406	2\\
4.06906906906907	2\\
4.07407407407407	2\\
4.07907907907908	2\\
4.08408408408408	2\\
4.08908908908909	2\\
4.09409409409409	2\\
4.0990990990991	2\\
4.1041041041041	2\\
4.10910910910911	2\\
4.11411411411411	2\\
4.11911911911912	2\\
4.12412412412412	2\\
4.12912912912913	2\\
4.13413413413413	2\\
4.13913913913914	2\\
4.14414414414414	2\\
4.14914914914915	2\\
4.15415415415415	2\\
4.15915915915916	2\\
4.16416416416416	2\\
4.16916916916917	2\\
4.17417417417417	2\\
4.17917917917918	2\\
4.18418418418418	2\\
4.18918918918919	2\\
4.19419419419419	2\\
4.1991991991992	2\\
4.2042042042042	2\\
4.20920920920921	2\\
4.21421421421421	2\\
4.21921921921922	2\\
4.22422422422422	2\\
4.22922922922923	2\\
4.23423423423423	2\\
4.23923923923924	2\\
4.24424424424424	2\\
4.24924924924925	2\\
4.25425425425425	2\\
4.25925925925926	2\\
4.26426426426426	2\\
4.26926926926927	2\\
4.27427427427427	2\\
4.27927927927928	2\\
4.28428428428428	2\\
4.28928928928929	2\\
4.29429429429429	2\\
4.2992992992993	2\\
4.3043043043043	2\\
4.30930930930931	2\\
4.31431431431431	2\\
4.31931931931932	2\\
4.32432432432432	2\\
4.32932932932933	2\\
4.33433433433433	2\\
4.33933933933934	2\\
4.34434434434434	2\\
4.34934934934935	2\\
4.35435435435435	2\\
4.35935935935936	2\\
4.36436436436436	2\\
4.36936936936937	2\\
4.37437437437437	2\\
4.37937937937938	2\\
4.38438438438438	2\\
4.38938938938939	2\\
4.39439439439439	2\\
4.3993993993994	2\\
4.4044044044044	2\\
4.40940940940941	2\\
4.41441441441441	2\\
4.41941941941942	2\\
4.42442442442442	2\\
4.42942942942943	2\\
4.43443443443443	2\\
4.43943943943944	2\\
4.44444444444444	2\\
4.44944944944945	2\\
4.45445445445445	2\\
4.45945945945946	2\\
4.46446446446446	2\\
4.46946946946947	2\\
4.47447447447447	2\\
4.47947947947948	2\\
4.48448448448448	2\\
4.48948948948949	2\\
4.49449449449449	2\\
4.4994994994995	2\\
4.5045045045045	2\\
4.50950950950951	2\\
4.51451451451451	2\\
4.51951951951952	2\\
4.52452452452452	2\\
4.52952952952953	2\\
4.53453453453453	2\\
4.53953953953954	2\\
4.54454454454454	2\\
4.54954954954955	2\\
4.55455455455455	2\\
4.55955955955956	2\\
4.56456456456456	2\\
4.56956956956957	2\\
4.57457457457457	2\\
4.57957957957958	2\\
4.58458458458458	2\\
4.58958958958959	2\\
4.59459459459459	2\\
4.5995995995996	2\\
4.6046046046046	2\\
4.60960960960961	2\\
4.61461461461461	2\\
4.61961961961962	2\\
4.62462462462462	2\\
4.62962962962963	2\\
4.63463463463463	2\\
4.63963963963964	2\\
4.64464464464464	2\\
4.64964964964965	2\\
4.65465465465465	2\\
4.65965965965966	2\\
4.66466466466466	2\\
4.66966966966967	2\\
4.67467467467467	2\\
4.67967967967968	2\\
4.68468468468468	2\\
4.68968968968969	2\\
4.69469469469469	2\\
4.6996996996997	2\\
4.7047047047047	2\\
4.70970970970971	2\\
4.71471471471471	2\\
4.71971971971972	2\\
4.72472472472472	2\\
4.72972972972973	2\\
4.73473473473473	2\\
4.73973973973974	2\\
4.74474474474474	2\\
4.74974974974975	2\\
4.75475475475475	2\\
4.75975975975976	2\\
4.76476476476476	2\\
4.76976976976977	2\\
4.77477477477477	2\\
4.77977977977978	2\\
4.78478478478478	2\\
4.78978978978979	2\\
4.79479479479479	2\\
4.7997997997998	2\\
4.8048048048048	2\\
4.80980980980981	2\\
4.81481481481481	2\\
4.81981981981982	2\\
4.82482482482482	2\\
4.82982982982983	2\\
4.83483483483483	2\\
4.83983983983984	2\\
4.84484484484484	2\\
4.84984984984985	2\\
4.85485485485485	2\\
4.85985985985986	2\\
4.86486486486486	2\\
4.86986986986987	2\\
4.87487487487487	2\\
4.87987987987988	2\\
4.88488488488488	2\\
4.88988988988989	2\\
4.89489489489489	2\\
4.8998998998999	2\\
4.9049049049049	2\\
4.90990990990991	2\\
4.91491491491491	2\\
4.91991991991992	2\\
4.92492492492492	2\\
4.92992992992993	2\\
4.93493493493493	2\\
4.93993993993994	2\\
4.94494494494494	2\\
4.94994994994995	2\\
4.95495495495495	2\\
4.95995995995996	2\\
4.96496496496496	2\\
4.96996996996997	2\\
4.97497497497497	2\\
4.97997997997998	2\\
4.98498498498498	2\\
4.98998998998999	2\\
4.99499499499499	2\\
5	2\\
};
\addplot [color=blue,line width=1.0pt,only marks,mark=o,mark options={solid},forget plot]
  table[row sep=crcr]{1	3.0005\\
};
\addplot [color=blue,line width=3.0pt,only marks,mark=o,mark options={solid},forget plot]
  table[row sep=crcr]{1	2\\
};

\end{axis}
}
%
% This file was created by matlab2tikz v0.4.7 running on MATLAB 8.3.
% Copyright (c) 2008--2014, Nico Schlömer <nico.schloemer@gmail.com>
% All rights reserved.
% Minimal pgfplots version: 1.3
% 
% The latest updates can be retrieved from
%   http://www.mathworks.com/matlabcentral/fileexchange/22022-matlab2tikz
% where you can also make suggestions and rate matlab2tikz.
% 
\newcommand{\bsixtygthirty}[0]{
\begin{axis}[%
width=7cm,
height=4cm,
scale only axis,
xmin=0,
xmax=5,
xtick={0,0.5,1,2,3,4,5},
xticklabels={{0},{$\frac{1}{2}$},{$1$},{2},{3},{4},{5}},
xlabel={$\alpha$},
ymin=1,
ymax=7,
ytick={1,2,3,4,5,6},
yticklabels={{1},{2},{3},{4},{5},{6}},
ylabel={$d$}
]

\node[above] at (axis cs:1.8,4.8){$\gamma+\alpha$};
\node[above] at (axis cs:4.5,5.9){$2\gamma$};
%\node[above] at (axis cs:2.2,3.2){$\max\{\alpha,\beta\}+(\gamma-\alpha)^+$};

\addplot [color=blue,solid,line width=1.0pt,forget plot]
  table[row sep=crcr]{0	4\\
0.005005005005005	4\\
0.01001001001001	4\\
0.015015015015015	4\\
0.02002002002002	4\\
0.025025025025025	4\\
0.03003003003003	4\\
0.035035035035035	4\\
0.04004004004004	4\\
0.045045045045045	4\\
0.0500500500500501	4\\
0.0550550550550551	4\\
0.0600600600600601	4\\
0.0650650650650651	4\\
0.0700700700700701	4\\
0.0750750750750751	4\\
0.0800800800800801	4\\
0.0850850850850851	4\\
0.0900900900900901	4\\
0.0950950950950951	4\\
0.1001001001001	4\\
0.105105105105105	4\\
0.11011011011011	4\\
0.115115115115115	4\\
0.12012012012012	4\\
0.125125125125125	4\\
0.13013013013013	4\\
0.135135135135135	4\\
0.14014014014014	4\\
0.145145145145145	4\\
0.15015015015015	4\\
0.155155155155155	4\\
0.16016016016016	4\\
0.165165165165165	4\\
0.17017017017017	4\\
0.175175175175175	4\\
0.18018018018018	4\\
0.185185185185185	4\\
0.19019019019019	4\\
0.195195195195195	4\\
0.2002002002002	4\\
0.205205205205205	4\\
0.21021021021021	4\\
0.215215215215215	4\\
0.22022022022022	4\\
0.225225225225225	4\\
0.23023023023023	4\\
0.235235235235235	4\\
0.24024024024024	4\\
0.245245245245245	4\\
0.25025025025025	4\\
0.255255255255255	4\\
0.26026026026026	4\\
0.265265265265265	4\\
0.27027027027027	4\\
0.275275275275275	4\\
0.28028028028028	4\\
0.285285285285285	4\\
0.29029029029029	4\\
0.295295295295295	4\\
0.3003003003003	4\\
0.305305305305305	4\\
0.31031031031031	4\\
0.315315315315315	4\\
0.32032032032032	4\\
0.325325325325325	4\\
0.33033033033033	4\\
0.335335335335335	4\\
0.34034034034034	4\\
0.345345345345345	4\\
0.35035035035035	4\\
0.355355355355355	4\\
0.36036036036036	4\\
0.365365365365365	4\\
0.37037037037037	4\\
0.375375375375375	4\\
0.38038038038038	4\\
0.385385385385385	4\\
0.39039039039039	4\\
0.395395395395395	4\\
0.4004004004004	4\\
0.405405405405405	4\\
0.41041041041041	4\\
0.415415415415415	4\\
0.42042042042042	4\\
0.425425425425425	4\\
0.43043043043043	4\\
0.435435435435435	4\\
0.44044044044044	4\\
0.445445445445445	4\\
0.45045045045045	4\\
0.455455455455455	4\\
0.46046046046046	4\\
0.465465465465465	4\\
0.47047047047047	4\\
0.475475475475475	4\\
0.48048048048048	4\\
0.485485485485485	4\\
0.49049049049049	4\\
0.495495495495495	4\\
0.500500500500501	4\\
0.505505505505506	4\\
0.510510510510511	4\\
0.515515515515516	4\\
0.520520520520521	4\\
0.525525525525526	4\\
0.530530530530531	4\\
0.535535535535536	4\\
0.540540540540541	4\\
0.545545545545546	4\\
0.550550550550551	4\\
0.555555555555556	4\\
0.560560560560561	4\\
0.565565565565566	4\\
0.570570570570571	4\\
0.575575575575576	4\\
0.580580580580581	4\\
0.585585585585586	4\\
0.590590590590591	4\\
0.595595595595596	4\\
0.600600600600601	4\\
0.605605605605606	4\\
0.610610610610611	4\\
0.615615615615616	4\\
0.620620620620621	4\\
0.625625625625626	4\\
0.630630630630631	4\\
0.635635635635636	4\\
0.640640640640641	4\\
0.645645645645646	4\\
0.650650650650651	4\\
0.655655655655656	4\\
0.660660660660661	4\\
0.665665665665666	4\\
0.670670670670671	4\\
0.675675675675676	4\\
0.680680680680681	4\\
0.685685685685686	4\\
0.690690690690691	4\\
0.695695695695696	4\\
0.700700700700701	4\\
0.705705705705706	4\\
0.710710710710711	4\\
0.715715715715716	4\\
0.720720720720721	4\\
0.725725725725726	4\\
0.730730730730731	4\\
0.735735735735736	4\\
0.740740740740741	4\\
0.745745745745746	4\\
0.750750750750751	4\\
0.755755755755756	4\\
0.760760760760761	4\\
0.765765765765766	4\\
0.770770770770771	4\\
0.775775775775776	4\\
0.780780780780781	4\\
0.785785785785786	4\\
0.790790790790791	4\\
0.795795795795796	4\\
0.800800800800801	4\\
0.805805805805806	4\\
0.810810810810811	4\\
0.815815815815816	4\\
0.820820820820821	4\\
0.825825825825826	4\\
0.830830830830831	4\\
0.835835835835836	4\\
0.840840840840841	4\\
0.845845845845846	4\\
0.850850850850851	4\\
0.855855855855856	4\\
0.860860860860861	4\\
0.865865865865866	4\\
0.870870870870871	4\\
0.875875875875876	4\\
0.880880880880881	4\\
0.885885885885886	4\\
0.890890890890891	4\\
0.895895895895896	4\\
0.900900900900901	4\\
0.905905905905906	4\\
0.910910910910911	4\\
0.915915915915916	4\\
0.920920920920921	4\\
0.925925925925926	4\\
0.930930930930931	4\\
0.935935935935936	4\\
0.940940940940941	4\\
0.945945945945946	4\\
0.950950950950951	4\\
0.955955955955956	4\\
0.960960960960961	4\\
0.965965965965966	4\\
0.970970970970971	4\\
0.975975975975976	4\\
0.980980980980981	4\\
0.985985985985986	4\\
0.990990990990991	4\\
0.995995995995996	4\\
1.001001001001	4.001001001001\\
1.00600600600601	4.00600600600601\\
1.01101101101101	4.01101101101101\\
1.01601601601602	4.01601601601602\\
1.02102102102102	4.02102102102102\\
1.02602602602603	4.02602602602603\\
1.03103103103103	4.03103103103103\\
1.03603603603604	4.03603603603604\\
1.04104104104104	4.04104104104104\\
1.04604604604605	4.04604604604605\\
1.05105105105105	4.05105105105105\\
1.05605605605606	4.05605605605606\\
1.06106106106106	4.06106106106106\\
1.06606606606607	4.06606606606607\\
1.07107107107107	4.07107107107107\\
1.07607607607608	4.07607607607608\\
1.08108108108108	4.08108108108108\\
1.08608608608609	4.08608608608609\\
1.09109109109109	4.09109109109109\\
1.0960960960961	4.0960960960961\\
1.1011011011011	4.1011011011011\\
1.10610610610611	4.10610610610611\\
1.11111111111111	4.11111111111111\\
1.11611611611612	4.11611611611612\\
1.12112112112112	4.12112112112112\\
1.12612612612613	4.12612612612613\\
1.13113113113113	4.13113113113113\\
1.13613613613614	4.13613613613614\\
1.14114114114114	4.14114114114114\\
1.14614614614615	4.14614614614615\\
1.15115115115115	4.15115115115115\\
1.15615615615616	4.15615615615616\\
1.16116116116116	4.16116116116116\\
1.16616616616617	4.16616616616617\\
1.17117117117117	4.17117117117117\\
1.17617617617618	4.17617617617618\\
1.18118118118118	4.18118118118118\\
1.18618618618619	4.18618618618619\\
1.19119119119119	4.19119119119119\\
1.1961961961962	4.1961961961962\\
1.2012012012012	4.2012012012012\\
1.20620620620621	4.20620620620621\\
1.21121121121121	4.21121121121121\\
1.21621621621622	4.21621621621622\\
1.22122122122122	4.22122122122122\\
1.22622622622623	4.22622622622623\\
1.23123123123123	4.23123123123123\\
1.23623623623624	4.23623623623624\\
1.24124124124124	4.24124124124124\\
1.24624624624625	4.24624624624625\\
1.25125125125125	4.25125125125125\\
1.25625625625626	4.25625625625626\\
1.26126126126126	4.26126126126126\\
1.26626626626627	4.26626626626627\\
1.27127127127127	4.27127127127127\\
1.27627627627628	4.27627627627628\\
1.28128128128128	4.28128128128128\\
1.28628628628629	4.28628628628629\\
1.29129129129129	4.29129129129129\\
1.2962962962963	4.2962962962963\\
1.3013013013013	4.3013013013013\\
1.30630630630631	4.30630630630631\\
1.31131131131131	4.31131131131131\\
1.31631631631632	4.31631631631632\\
1.32132132132132	4.32132132132132\\
1.32632632632633	4.32632632632633\\
1.33133133133133	4.33133133133133\\
1.33633633633634	4.33633633633634\\
1.34134134134134	4.34134134134134\\
1.34634634634635	4.34634634634635\\
1.35135135135135	4.35135135135135\\
1.35635635635636	4.35635635635636\\
1.36136136136136	4.36136136136136\\
1.36636636636637	4.36636636636637\\
1.37137137137137	4.37137137137137\\
1.37637637637638	4.37637637637638\\
1.38138138138138	4.38138138138138\\
1.38638638638639	4.38638638638639\\
1.39139139139139	4.39139139139139\\
1.3963963963964	4.3963963963964\\
1.4014014014014	4.4014014014014\\
1.40640640640641	4.40640640640641\\
1.41141141141141	4.41141141141141\\
1.41641641641642	4.41641641641642\\
1.42142142142142	4.42142142142142\\
1.42642642642643	4.42642642642643\\
1.43143143143143	4.43143143143143\\
1.43643643643644	4.43643643643644\\
1.44144144144144	4.44144144144144\\
1.44644644644645	4.44644644644645\\
1.45145145145145	4.45145145145145\\
1.45645645645646	4.45645645645646\\
1.46146146146146	4.46146146146146\\
1.46646646646647	4.46646646646647\\
1.47147147147147	4.47147147147147\\
1.47647647647648	4.47647647647648\\
1.48148148148148	4.48148148148148\\
1.48648648648649	4.48648648648649\\
1.49149149149149	4.49149149149149\\
1.4964964964965	4.4964964964965\\
1.5015015015015	4.5015015015015\\
1.50650650650651	4.50650650650651\\
1.51151151151151	4.51151151151151\\
1.51651651651652	4.51651651651652\\
1.52152152152152	4.52152152152152\\
1.52652652652653	4.52652652652653\\
1.53153153153153	4.53153153153153\\
1.53653653653654	4.53653653653654\\
1.54154154154154	4.54154154154154\\
1.54654654654655	4.54654654654655\\
1.55155155155155	4.55155155155155\\
1.55655655655656	4.55655655655656\\
1.56156156156156	4.56156156156156\\
1.56656656656657	4.56656656656657\\
1.57157157157157	4.57157157157157\\
1.57657657657658	4.57657657657658\\
1.58158158158158	4.58158158158158\\
1.58658658658659	4.58658658658659\\
1.59159159159159	4.59159159159159\\
1.5965965965966	4.5965965965966\\
1.6016016016016	4.6016016016016\\
1.60660660660661	4.60660660660661\\
1.61161161161161	4.61161161161161\\
1.61661661661662	4.61661661661662\\
1.62162162162162	4.62162162162162\\
1.62662662662663	4.62662662662663\\
1.63163163163163	4.63163163163163\\
1.63663663663664	4.63663663663664\\
1.64164164164164	4.64164164164164\\
1.64664664664665	4.64664664664665\\
1.65165165165165	4.65165165165165\\
1.65665665665666	4.65665665665666\\
1.66166166166166	4.66166166166166\\
1.66666666666667	4.66666666666667\\
1.67167167167167	4.67167167167167\\
1.67667667667668	4.67667667667668\\
1.68168168168168	4.68168168168168\\
1.68668668668669	4.68668668668669\\
1.69169169169169	4.69169169169169\\
1.6966966966967	4.6966966966967\\
1.7017017017017	4.7017017017017\\
1.70670670670671	4.70670670670671\\
1.71171171171171	4.71171171171171\\
1.71671671671672	4.71671671671672\\
1.72172172172172	4.72172172172172\\
1.72672672672673	4.72672672672673\\
1.73173173173173	4.73173173173173\\
1.73673673673674	4.73673673673674\\
1.74174174174174	4.74174174174174\\
1.74674674674675	4.74674674674675\\
1.75175175175175	4.75175175175175\\
1.75675675675676	4.75675675675676\\
1.76176176176176	4.76176176176176\\
1.76676676676677	4.76676676676677\\
1.77177177177177	4.77177177177177\\
1.77677677677678	4.77677677677678\\
1.78178178178178	4.78178178178178\\
1.78678678678679	4.78678678678679\\
1.79179179179179	4.79179179179179\\
1.7967967967968	4.7967967967968\\
1.8018018018018	4.8018018018018\\
1.80680680680681	4.80680680680681\\
1.81181181181181	4.81181181181181\\
1.81681681681682	4.81681681681682\\
1.82182182182182	4.82182182182182\\
1.82682682682683	4.82682682682683\\
1.83183183183183	4.83183183183183\\
1.83683683683684	4.83683683683684\\
1.84184184184184	4.84184184184184\\
1.84684684684685	4.84684684684685\\
1.85185185185185	4.85185185185185\\
1.85685685685686	4.85685685685686\\
1.86186186186186	4.86186186186186\\
1.86686686686687	4.86686686686687\\
1.87187187187187	4.87187187187187\\
1.87687687687688	4.87687687687688\\
1.88188188188188	4.88188188188188\\
1.88688688688689	4.88688688688689\\
1.89189189189189	4.89189189189189\\
1.8968968968969	4.8968968968969\\
1.9019019019019	4.9019019019019\\
1.90690690690691	4.90690690690691\\
1.91191191191191	4.91191191191191\\
1.91691691691692	4.91691691691692\\
1.92192192192192	4.92192192192192\\
1.92692692692693	4.92692692692693\\
1.93193193193193	4.93193193193193\\
1.93693693693694	4.93693693693694\\
1.94194194194194	4.94194194194194\\
1.94694694694695	4.94694694694695\\
1.95195195195195	4.95195195195195\\
1.95695695695696	4.95695695695696\\
1.96196196196196	4.96196196196196\\
1.96696696696697	4.96696696696697\\
1.97197197197197	4.97197197197197\\
1.97697697697698	4.97697697697698\\
1.98198198198198	4.98198198198198\\
1.98698698698699	4.98698698698699\\
1.99199199199199	4.99199199199199\\
1.996996996997	4.996996996997\\
2.002002002002	5.002002002002\\
2.00700700700701	5.00700700700701\\
2.01201201201201	5.01201201201201\\
2.01701701701702	5.01701701701702\\
2.02202202202202	5.02202202202202\\
2.02702702702703	5.02702702702703\\
2.03203203203203	5.03203203203203\\
2.03703703703704	5.03703703703704\\
2.04204204204204	5.04204204204204\\
2.04704704704705	5.04704704704705\\
2.05205205205205	5.05205205205205\\
2.05705705705706	5.05705705705706\\
2.06206206206206	5.06206206206206\\
2.06706706706707	5.06706706706707\\
2.07207207207207	5.07207207207207\\
2.07707707707708	5.07707707707708\\
2.08208208208208	5.08208208208208\\
2.08708708708709	5.08708708708709\\
2.09209209209209	5.09209209209209\\
2.0970970970971	5.0970970970971\\
2.1021021021021	5.1021021021021\\
2.10710710710711	5.10710710710711\\
2.11211211211211	5.11211211211211\\
2.11711711711712	5.11711711711712\\
2.12212212212212	5.12212212212212\\
2.12712712712713	5.12712712712713\\
2.13213213213213	5.13213213213213\\
2.13713713713714	5.13713713713714\\
2.14214214214214	5.14214214214214\\
2.14714714714715	5.14714714714715\\
2.15215215215215	5.15215215215215\\
2.15715715715716	5.15715715715716\\
2.16216216216216	5.16216216216216\\
2.16716716716717	5.16716716716717\\
2.17217217217217	5.17217217217217\\
2.17717717717718	5.17717717717718\\
2.18218218218218	5.18218218218218\\
2.18718718718719	5.18718718718719\\
2.19219219219219	5.19219219219219\\
2.1971971971972	5.1971971971972\\
2.2022022022022	5.2022022022022\\
2.20720720720721	5.20720720720721\\
2.21221221221221	5.21221221221221\\
2.21721721721722	5.21721721721722\\
2.22222222222222	5.22222222222222\\
2.22722722722723	5.22722722722723\\
2.23223223223223	5.23223223223223\\
2.23723723723724	5.23723723723724\\
2.24224224224224	5.24224224224224\\
2.24724724724725	5.24724724724725\\
2.25225225225225	5.25225225225225\\
2.25725725725726	5.25725725725726\\
2.26226226226226	5.26226226226226\\
2.26726726726727	5.26726726726727\\
2.27227227227227	5.27227227227227\\
2.27727727727728	5.27727727727728\\
2.28228228228228	5.28228228228228\\
2.28728728728729	5.28728728728729\\
2.29229229229229	5.29229229229229\\
2.2972972972973	5.2972972972973\\
2.3023023023023	5.3023023023023\\
2.30730730730731	5.30730730730731\\
2.31231231231231	5.31231231231231\\
2.31731731731732	5.31731731731732\\
2.32232232232232	5.32232232232232\\
2.32732732732733	5.32732732732733\\
2.33233233233233	5.33233233233233\\
2.33733733733734	5.33733733733734\\
2.34234234234234	5.34234234234234\\
2.34734734734735	5.34734734734735\\
2.35235235235235	5.35235235235235\\
2.35735735735736	5.35735735735736\\
2.36236236236236	5.36236236236236\\
2.36736736736737	5.36736736736737\\
2.37237237237237	5.37237237237237\\
2.37737737737738	5.37737737737738\\
2.38238238238238	5.38238238238238\\
2.38738738738739	5.38738738738739\\
2.39239239239239	5.39239239239239\\
2.3973973973974	5.3973973973974\\
2.4024024024024	5.4024024024024\\
2.40740740740741	5.40740740740741\\
2.41241241241241	5.41241241241241\\
2.41741741741742	5.41741741741742\\
2.42242242242242	5.42242242242242\\
2.42742742742743	5.42742742742743\\
2.43243243243243	5.43243243243243\\
2.43743743743744	5.43743743743744\\
2.44244244244244	5.44244244244244\\
2.44744744744745	5.44744744744745\\
2.45245245245245	5.45245245245245\\
2.45745745745746	5.45745745745746\\
2.46246246246246	5.46246246246246\\
2.46746746746747	5.46746746746747\\
2.47247247247247	5.47247247247247\\
2.47747747747748	5.47747747747748\\
2.48248248248248	5.48248248248248\\
2.48748748748749	5.48748748748749\\
2.49249249249249	5.49249249249249\\
2.4974974974975	5.4974974974975\\
2.5025025025025	5.5025025025025\\
2.50750750750751	5.50750750750751\\
2.51251251251251	5.51251251251251\\
2.51751751751752	5.51751751751752\\
2.52252252252252	5.52252252252252\\
2.52752752752753	5.52752752752753\\
2.53253253253253	5.53253253253253\\
2.53753753753754	5.53753753753754\\
2.54254254254254	5.54254254254254\\
2.54754754754755	5.54754754754755\\
2.55255255255255	5.55255255255255\\
2.55755755755756	5.55755755755756\\
2.56256256256256	5.56256256256256\\
2.56756756756757	5.56756756756757\\
2.57257257257257	5.57257257257257\\
2.57757757757758	5.57757757757758\\
2.58258258258258	5.58258258258258\\
2.58758758758759	5.58758758758759\\
2.59259259259259	5.59259259259259\\
2.5975975975976	5.5975975975976\\
2.6026026026026	5.6026026026026\\
2.60760760760761	5.60760760760761\\
2.61261261261261	5.61261261261261\\
2.61761761761762	5.61761761761762\\
2.62262262262262	5.62262262262262\\
2.62762762762763	5.62762762762763\\
2.63263263263263	5.63263263263263\\
2.63763763763764	5.63763763763764\\
2.64264264264264	5.64264264264264\\
2.64764764764765	5.64764764764765\\
2.65265265265265	5.65265265265265\\
2.65765765765766	5.65765765765766\\
2.66266266266266	5.66266266266266\\
2.66766766766767	5.66766766766767\\
2.67267267267267	5.67267267267267\\
2.67767767767768	5.67767767767768\\
2.68268268268268	5.68268268268268\\
2.68768768768769	5.68768768768769\\
2.69269269269269	5.69269269269269\\
2.6976976976977	5.6976976976977\\
2.7027027027027	5.7027027027027\\
2.70770770770771	5.70770770770771\\
2.71271271271271	5.71271271271271\\
2.71771771771772	5.71771771771772\\
2.72272272272272	5.72272272272272\\
2.72772772772773	5.72772772772773\\
2.73273273273273	5.73273273273273\\
2.73773773773774	5.73773773773774\\
2.74274274274274	5.74274274274274\\
2.74774774774775	5.74774774774775\\
2.75275275275275	5.75275275275275\\
2.75775775775776	5.75775775775776\\
2.76276276276276	5.76276276276276\\
2.76776776776777	5.76776776776777\\
2.77277277277277	5.77277277277277\\
2.77777777777778	5.77777777777778\\
2.78278278278278	5.78278278278278\\
2.78778778778779	5.78778778778779\\
2.79279279279279	5.79279279279279\\
2.7977977977978	5.7977977977978\\
2.8028028028028	5.8028028028028\\
2.80780780780781	5.80780780780781\\
2.81281281281281	5.81281281281281\\
2.81781781781782	5.81781781781782\\
2.82282282282282	5.82282282282282\\
2.82782782782783	5.82782782782783\\
2.83283283283283	5.83283283283283\\
2.83783783783784	5.83783783783784\\
2.84284284284284	5.84284284284284\\
2.84784784784785	5.84784784784785\\
2.85285285285285	5.85285285285285\\
2.85785785785786	5.85785785785786\\
2.86286286286286	5.86286286286286\\
2.86786786786787	5.86786786786787\\
2.87287287287287	5.87287287287287\\
2.87787787787788	5.87787787787788\\
2.88288288288288	5.88288288288288\\
2.88788788788789	5.88788788788789\\
2.89289289289289	5.89289289289289\\
2.8978978978979	5.8978978978979\\
2.9029029029029	5.9029029029029\\
2.90790790790791	5.90790790790791\\
2.91291291291291	5.91291291291291\\
2.91791791791792	5.91791791791792\\
2.92292292292292	5.92292292292292\\
2.92792792792793	5.92792792792793\\
2.93293293293293	5.93293293293293\\
2.93793793793794	5.93793793793794\\
2.94294294294294	5.94294294294294\\
2.94794794794795	5.94794794794795\\
2.95295295295295	5.95295295295295\\
2.95795795795796	5.95795795795796\\
2.96296296296296	5.96296296296296\\
2.96796796796797	5.96796796796797\\
2.97297297297297	5.97297297297297\\
2.97797797797798	5.97797797797798\\
2.98298298298298	5.98298298298298\\
2.98798798798799	5.98798798798799\\
2.99299299299299	5.99299299299299\\
2.997997997998	5.997997997998\\
3.003003003003	6\\
3.00800800800801	6\\
3.01301301301301	6\\
3.01801801801802	6\\
3.02302302302302	6\\
3.02802802802803	6\\
3.03303303303303	6\\
3.03803803803804	6\\
3.04304304304304	6\\
3.04804804804805	6\\
3.05305305305305	6\\
3.05805805805806	6\\
3.06306306306306	6\\
3.06806806806807	6\\
3.07307307307307	6\\
3.07807807807808	6\\
3.08308308308308	6\\
3.08808808808809	6\\
3.09309309309309	6\\
3.0980980980981	6\\
3.1031031031031	6\\
3.10810810810811	6\\
3.11311311311311	6\\
3.11811811811812	6\\
3.12312312312312	6\\
3.12812812812813	6\\
3.13313313313313	6\\
3.13813813813814	6\\
3.14314314314314	6\\
3.14814814814815	6\\
3.15315315315315	6\\
3.15815815815816	6\\
3.16316316316316	6\\
3.16816816816817	6\\
3.17317317317317	6\\
3.17817817817818	6\\
3.18318318318318	6\\
3.18818818818819	6\\
3.19319319319319	6\\
3.1981981981982	6\\
3.2032032032032	6\\
3.20820820820821	6\\
3.21321321321321	6\\
3.21821821821822	6\\
3.22322322322322	6\\
3.22822822822823	6\\
3.23323323323323	6\\
3.23823823823824	6\\
3.24324324324324	6\\
3.24824824824825	6\\
3.25325325325325	6\\
3.25825825825826	6\\
3.26326326326326	6\\
3.26826826826827	6\\
3.27327327327327	6\\
3.27827827827828	6\\
3.28328328328328	6\\
3.28828828828829	6\\
3.29329329329329	6\\
3.2982982982983	6\\
3.3033033033033	6\\
3.30830830830831	6\\
3.31331331331331	6\\
3.31831831831832	6\\
3.32332332332332	6\\
3.32832832832833	6\\
3.33333333333333	6\\
3.33833833833834	6\\
3.34334334334334	6\\
3.34834834834835	6\\
3.35335335335335	6\\
3.35835835835836	6\\
3.36336336336336	6\\
3.36836836836837	6\\
3.37337337337337	6\\
3.37837837837838	6\\
3.38338338338338	6\\
3.38838838838839	6\\
3.39339339339339	6\\
3.3983983983984	6\\
3.4034034034034	6\\
3.40840840840841	6\\
3.41341341341341	6\\
3.41841841841842	6\\
3.42342342342342	6\\
3.42842842842843	6\\
3.43343343343343	6\\
3.43843843843844	6\\
3.44344344344344	6\\
3.44844844844845	6\\
3.45345345345345	6\\
3.45845845845846	6\\
3.46346346346346	6\\
3.46846846846847	6\\
3.47347347347347	6\\
3.47847847847848	6\\
3.48348348348348	6\\
3.48848848848849	6\\
3.49349349349349	6\\
3.4984984984985	6\\
3.5035035035035	6\\
3.50850850850851	6\\
3.51351351351351	6\\
3.51851851851852	6\\
3.52352352352352	6\\
3.52852852852853	6\\
3.53353353353353	6\\
3.53853853853854	6\\
3.54354354354354	6\\
3.54854854854855	6\\
3.55355355355355	6\\
3.55855855855856	6\\
3.56356356356356	6\\
3.56856856856857	6\\
3.57357357357357	6\\
3.57857857857858	6\\
3.58358358358358	6\\
3.58858858858859	6\\
3.59359359359359	6\\
3.5985985985986	6\\
3.6036036036036	6\\
3.60860860860861	6\\
3.61361361361361	6\\
3.61861861861862	6\\
3.62362362362362	6\\
3.62862862862863	6\\
3.63363363363363	6\\
3.63863863863864	6\\
3.64364364364364	6\\
3.64864864864865	6\\
3.65365365365365	6\\
3.65865865865866	6\\
3.66366366366366	6\\
3.66866866866867	6\\
3.67367367367367	6\\
3.67867867867868	6\\
3.68368368368368	6\\
3.68868868868869	6\\
3.69369369369369	6\\
3.6986986986987	6\\
3.7037037037037	6\\
3.70870870870871	6\\
3.71371371371371	6\\
3.71871871871872	6\\
3.72372372372372	6\\
3.72872872872873	6\\
3.73373373373373	6\\
3.73873873873874	6\\
3.74374374374374	6\\
3.74874874874875	6\\
3.75375375375375	6\\
3.75875875875876	6\\
3.76376376376376	6\\
3.76876876876877	6\\
3.77377377377377	6\\
3.77877877877878	6\\
3.78378378378378	6\\
3.78878878878879	6\\
3.79379379379379	6\\
3.7987987987988	6\\
3.8038038038038	6\\
3.80880880880881	6\\
3.81381381381381	6\\
3.81881881881882	6\\
3.82382382382382	6\\
3.82882882882883	6\\
3.83383383383383	6\\
3.83883883883884	6\\
3.84384384384384	6\\
3.84884884884885	6\\
3.85385385385385	6\\
3.85885885885886	6\\
3.86386386386386	6\\
3.86886886886887	6\\
3.87387387387387	6\\
3.87887887887888	6\\
3.88388388388388	6\\
3.88888888888889	6\\
3.89389389389389	6\\
3.8988988988989	6\\
3.9039039039039	6\\
3.90890890890891	6\\
3.91391391391391	6\\
3.91891891891892	6\\
3.92392392392392	6\\
3.92892892892893	6\\
3.93393393393393	6\\
3.93893893893894	6\\
3.94394394394394	6\\
3.94894894894895	6\\
3.95395395395395	6\\
3.95895895895896	6\\
3.96396396396396	6\\
3.96896896896897	6\\
3.97397397397397	6\\
3.97897897897898	6\\
3.98398398398398	6\\
3.98898898898899	6\\
3.99399399399399	6\\
3.998998998999	6\\
4.004004004004	6\\
4.00900900900901	6\\
4.01401401401401	6\\
4.01901901901902	6\\
4.02402402402402	6\\
4.02902902902903	6\\
4.03403403403403	6\\
4.03903903903904	6\\
4.04404404404404	6\\
4.04904904904905	6\\
4.05405405405405	6\\
4.05905905905906	6\\
4.06406406406406	6\\
4.06906906906907	6\\
4.07407407407407	6\\
4.07907907907908	6\\
4.08408408408408	6\\
4.08908908908909	6\\
4.09409409409409	6\\
4.0990990990991	6\\
4.1041041041041	6\\
4.10910910910911	6\\
4.11411411411411	6\\
4.11911911911912	6\\
4.12412412412412	6\\
4.12912912912913	6\\
4.13413413413413	6\\
4.13913913913914	6\\
4.14414414414414	6\\
4.14914914914915	6\\
4.15415415415415	6\\
4.15915915915916	6\\
4.16416416416416	6\\
4.16916916916917	6\\
4.17417417417417	6\\
4.17917917917918	6\\
4.18418418418418	6\\
4.18918918918919	6\\
4.19419419419419	6\\
4.1991991991992	6\\
4.2042042042042	6\\
4.20920920920921	6\\
4.21421421421421	6\\
4.21921921921922	6\\
4.22422422422422	6\\
4.22922922922923	6\\
4.23423423423423	6\\
4.23923923923924	6\\
4.24424424424424	6\\
4.24924924924925	6\\
4.25425425425425	6\\
4.25925925925926	6\\
4.26426426426426	6\\
4.26926926926927	6\\
4.27427427427427	6\\
4.27927927927928	6\\
4.28428428428428	6\\
4.28928928928929	6\\
4.29429429429429	6\\
4.2992992992993	6\\
4.3043043043043	6\\
4.30930930930931	6\\
4.31431431431431	6\\
4.31931931931932	6\\
4.32432432432432	6\\
4.32932932932933	6\\
4.33433433433433	6\\
4.33933933933934	6\\
4.34434434434434	6\\
4.34934934934935	6\\
4.35435435435435	6\\
4.35935935935936	6\\
4.36436436436436	6\\
4.36936936936937	6\\
4.37437437437437	6\\
4.37937937937938	6\\
4.38438438438438	6\\
4.38938938938939	6\\
4.39439439439439	6\\
4.3993993993994	6\\
4.4044044044044	6\\
4.40940940940941	6\\
4.41441441441441	6\\
4.41941941941942	6\\
4.42442442442442	6\\
4.42942942942943	6\\
4.43443443443443	6\\
4.43943943943944	6\\
4.44444444444444	6\\
4.44944944944945	6\\
4.45445445445445	6\\
4.45945945945946	6\\
4.46446446446446	6\\
4.46946946946947	6\\
4.47447447447447	6\\
4.47947947947948	6\\
4.48448448448448	6\\
4.48948948948949	6\\
4.49449449449449	6\\
4.4994994994995	6\\
4.5045045045045	6\\
4.50950950950951	6\\
4.51451451451451	6\\
4.51951951951952	6\\
4.52452452452452	6\\
4.52952952952953	6\\
4.53453453453453	6\\
4.53953953953954	6\\
4.54454454454454	6\\
4.54954954954955	6\\
4.55455455455455	6\\
4.55955955955956	6\\
4.56456456456456	6\\
4.56956956956957	6\\
4.57457457457457	6\\
4.57957957957958	6\\
4.58458458458458	6\\
4.58958958958959	6\\
4.59459459459459	6\\
4.5995995995996	6\\
4.6046046046046	6\\
4.60960960960961	6\\
4.61461461461461	6\\
4.61961961961962	6\\
4.62462462462462	6\\
4.62962962962963	6\\
4.63463463463463	6\\
4.63963963963964	6\\
4.64464464464464	6\\
4.64964964964965	6\\
4.65465465465465	6\\
4.65965965965966	6\\
4.66466466466466	6\\
4.66966966966967	6\\
4.67467467467467	6\\
4.67967967967968	6\\
4.68468468468468	6\\
4.68968968968969	6\\
4.69469469469469	6\\
4.6996996996997	6\\
4.7047047047047	6\\
4.70970970970971	6\\
4.71471471471471	6\\
4.71971971971972	6\\
4.72472472472472	6\\
4.72972972972973	6\\
4.73473473473473	6\\
4.73973973973974	6\\
4.74474474474474	6\\
4.74974974974975	6\\
4.75475475475475	6\\
4.75975975975976	6\\
4.76476476476476	6\\
4.76976976976977	6\\
4.77477477477477	6\\
4.77977977977978	6\\
4.78478478478478	6\\
4.78978978978979	6\\
4.79479479479479	6\\
4.7997997997998	6\\
4.8048048048048	6\\
4.80980980980981	6\\
4.81481481481481	6\\
4.81981981981982	6\\
4.82482482482482	6\\
4.82982982982983	6\\
4.83483483483483	6\\
4.83983983983984	6\\
4.84484484484484	6\\
4.84984984984985	6\\
4.85485485485485	6\\
4.85985985985986	6\\
4.86486486486486	6\\
4.86986986986987	6\\
4.87487487487487	6\\
4.87987987987988	6\\
4.88488488488488	6\\
4.88988988988989	6\\
4.89489489489489	6\\
4.8998998998999	6\\
4.9049049049049	6\\
4.90990990990991	6\\
4.91491491491491	6\\
4.91991991991992	6\\
4.92492492492492	6\\
4.92992992992993	6\\
4.93493493493493	6\\
4.93993993993994	6\\
4.94494494494494	6\\
4.94994994994995	6\\
4.95495495495495	6\\
4.95995995995996	6\\
4.96496496496496	6\\
4.96996996996997	6\\
4.97497497497497	6\\
4.97997997997998	6\\
4.98498498498498	6\\
4.98998998998999	6\\
4.99499499499499	6\\
5	6\\
};
\addplot [color=red,dashed,line width=1.0pt,forget plot]
  table[row sep=crcr]{0	2\\
0.005005005005005	1.98998998998999\\
0.01001001001001	1.97997997997998\\
0.015015015015015	1.96996996996997\\
0.02002002002002	1.95995995995996\\
0.025025025025025	1.94994994994995\\
0.03003003003003	1.93993993993994\\
0.035035035035035	1.92992992992993\\
0.04004004004004	1.91991991991992\\
0.045045045045045	1.90990990990991\\
0.0500500500500501	1.8998998998999\\
0.0550550550550551	1.88988988988989\\
0.0600600600600601	1.87987987987988\\
0.0650650650650651	1.86986986986987\\
0.0700700700700701	1.85985985985986\\
0.0750750750750751	1.84984984984985\\
0.0800800800800801	1.83983983983984\\
0.0850850850850851	1.82982982982983\\
0.0900900900900901	1.81981981981982\\
0.0950950950950951	1.80980980980981\\
0.1001001001001	1.7997997997998\\
0.105105105105105	1.78978978978979\\
0.11011011011011	1.77977977977978\\
0.115115115115115	1.76976976976977\\
0.12012012012012	1.75975975975976\\
0.125125125125125	1.74974974974975\\
0.13013013013013	1.73973973973974\\
0.135135135135135	1.72972972972973\\
0.14014014014014	1.71971971971972\\
0.145145145145145	1.70970970970971\\
0.15015015015015	1.6996996996997\\
0.155155155155155	1.68968968968969\\
0.16016016016016	1.67967967967968\\
0.165165165165165	1.66966966966967\\
0.17017017017017	1.65965965965966\\
0.175175175175175	1.64964964964965\\
0.18018018018018	1.63963963963964\\
0.185185185185185	1.62962962962963\\
0.19019019019019	1.61961961961962\\
0.195195195195195	1.60960960960961\\
0.2002002002002	1.5995995995996\\
0.205205205205205	1.58958958958959\\
0.21021021021021	1.57957957957958\\
0.215215215215215	1.56956956956957\\
0.22022022022022	1.55955955955956\\
0.225225225225225	1.54954954954955\\
0.23023023023023	1.53953953953954\\
0.235235235235235	1.52952952952953\\
0.24024024024024	1.51951951951952\\
0.245245245245245	1.50950950950951\\
0.25025025025025	1.4994994994995\\
0.255255255255255	1.48948948948949\\
0.26026026026026	1.47947947947948\\
0.265265265265265	1.46946946946947\\
0.27027027027027	1.45945945945946\\
0.275275275275275	1.44944944944945\\
0.28028028028028	1.43943943943944\\
0.285285285285285	1.42942942942943\\
0.29029029029029	1.41941941941942\\
0.295295295295295	1.40940940940941\\
0.3003003003003	1.3993993993994\\
0.305305305305305	1.38938938938939\\
0.31031031031031	1.37937937937938\\
0.315315315315315	1.36936936936937\\
0.32032032032032	1.35935935935936\\
0.325325325325325	1.34934934934935\\
0.33033033033033	1.33933933933934\\
0.335335335335335	1.32932932932933\\
0.34034034034034	1.31931931931932\\
0.345345345345345	1.30930930930931\\
0.35035035035035	1.2992992992993\\
0.355355355355355	1.28928928928929\\
0.36036036036036	1.27927927927928\\
0.365365365365365	1.26926926926927\\
0.37037037037037	1.25925925925926\\
0.375375375375375	1.24924924924925\\
0.38038038038038	1.23923923923924\\
0.385385385385385	1.22922922922923\\
0.39039039039039	1.21921921921922\\
0.395395395395395	1.20920920920921\\
0.4004004004004	1.1991991991992\\
0.405405405405405	1.18918918918919\\
0.41041041041041	1.17917917917918\\
0.415415415415415	1.16916916916917\\
0.42042042042042	1.15915915915916\\
0.425425425425425	1.14914914914915\\
0.43043043043043	1.13913913913914\\
0.435435435435435	1.12912912912913\\
0.44044044044044	1.11911911911912\\
0.445445445445445	1.10910910910911\\
0.45045045045045	1.0990990990991\\
0.455455455455455	1.08908908908909\\
0.46046046046046	1.07907907907908\\
0.465465465465465	1.06906906906907\\
0.47047047047047	1.05905905905906\\
0.475475475475475	1.04904904904905\\
0.48048048048048	1.03903903903904\\
0.485485485485485	1.02902902902903\\
0.49049049049049	1.01901901901902\\
0.495495495495495	1.00900900900901\\
0.500500500500501	1.001001001001\\
0.505505505505506	1.01101101101101\\
0.510510510510511	1.02102102102102\\
0.515515515515516	1.03103103103103\\
0.520520520520521	1.04104104104104\\
0.525525525525526	1.05105105105105\\
0.530530530530531	1.06106106106106\\
0.535535535535536	1.07107107107107\\
0.540540540540541	1.08108108108108\\
0.545545545545546	1.09109109109109\\
0.550550550550551	1.1011011011011\\
0.555555555555556	1.11111111111111\\
0.560560560560561	1.12112112112112\\
0.565565565565566	1.13113113113113\\
0.570570570570571	1.14114114114114\\
0.575575575575576	1.15115115115115\\
0.580580580580581	1.16116116116116\\
0.585585585585586	1.17117117117117\\
0.590590590590591	1.18118118118118\\
0.595595595595596	1.19119119119119\\
0.600600600600601	1.2012012012012\\
0.605605605605606	1.21121121121121\\
0.610610610610611	1.22122122122122\\
0.615615615615616	1.23123123123123\\
0.620620620620621	1.24124124124124\\
0.625625625625626	1.25125125125125\\
0.630630630630631	1.26126126126126\\
0.635635635635636	1.27127127127127\\
0.640640640640641	1.28128128128128\\
0.645645645645646	1.29129129129129\\
0.650650650650651	1.3013013013013\\
0.655655655655656	1.31131131131131\\
0.660660660660661	1.32132132132132\\
0.665665665665666	1.33133133133133\\
0.670670670670671	1.32932932932933\\
0.675675675675676	1.32432432432432\\
0.680680680680681	1.31931931931932\\
0.685685685685686	1.31431431431431\\
0.690690690690691	1.30930930930931\\
0.695695695695696	1.3043043043043\\
0.700700700700701	1.2992992992993\\
0.705705705705706	1.29429429429429\\
0.710710710710711	1.28928928928929\\
0.715715715715716	1.28428428428428\\
0.720720720720721	1.27927927927928\\
0.725725725725726	1.27427427427427\\
0.730730730730731	1.26926926926927\\
0.735735735735736	1.26426426426426\\
0.740740740740741	1.25925925925926\\
0.745745745745746	1.25425425425425\\
0.750750750750751	1.24924924924925\\
0.755755755755756	1.24424424424424\\
0.760760760760761	1.23923923923924\\
0.765765765765766	1.23423423423423\\
0.770770770770771	1.22922922922923\\
0.775775775775776	1.22422422422422\\
0.780780780780781	1.21921921921922\\
0.785785785785786	1.21421421421421\\
0.790790790790791	1.20920920920921\\
0.795795795795796	1.2042042042042\\
0.800800800800801	1.1991991991992\\
0.805805805805806	1.19419419419419\\
0.810810810810811	1.18918918918919\\
0.815815815815816	1.18418418418418\\
0.820820820820821	1.17917917917918\\
0.825825825825826	1.17417417417417\\
0.830830830830831	1.16916916916917\\
0.835835835835836	1.16416416416416\\
0.840840840840841	1.15915915915916\\
0.845845845845846	1.15415415415415\\
0.850850850850851	1.14914914914915\\
0.855855855855856	1.14414414414414\\
0.860860860860861	1.13913913913914\\
0.865865865865866	1.13413413413413\\
0.870870870870871	1.12912912912913\\
0.875875875875876	1.12412412412412\\
0.880880880880881	1.11911911911912\\
0.885885885885886	1.11411411411411\\
0.890890890890891	1.10910910910911\\
0.895895895895896	1.1041041041041\\
0.900900900900901	1.0990990990991\\
0.905905905905906	1.09409409409409\\
0.910910910910911	1.08908908908909\\
0.915915915915916	1.08408408408408\\
0.920920920920921	1.07907907907908\\
0.925925925925926	1.07407407407407\\
0.930930930930931	1.06906906906907\\
0.935935935935936	1.06406406406406\\
0.940940940940941	1.05905905905906\\
0.945945945945946	1.05405405405405\\
0.950950950950951	1.04904904904905\\
0.955955955955956	1.04404404404404\\
0.960960960960961	1.03903903903904\\
0.965965965965966	1.03403403403403\\
0.970970970970971	1.02902902902903\\
0.975975975975976	1.02402402402402\\
0.980980980980981	1.01901901901902\\
0.985985985985986	1.01401401401401\\
0.990990990990991	1.00900900900901\\
0.995995995995996	1.004004004004\\
1.001001001001	1.001001001001\\
1.00600600600601	1.00600600600601\\
1.01101101101101	1.01101101101101\\
1.01601601601602	1.01601601601602\\
1.02102102102102	1.02102102102102\\
1.02602602602603	1.02602602602603\\
1.03103103103103	1.03103103103103\\
1.03603603603604	1.03603603603604\\
1.04104104104104	1.04104104104104\\
1.04604604604605	1.04604604604605\\
1.05105105105105	1.05105105105105\\
1.05605605605606	1.05605605605606\\
1.06106106106106	1.06106106106106\\
1.06606606606607	1.06606606606607\\
1.07107107107107	1.07107107107107\\
1.07607607607608	1.07607607607608\\
1.08108108108108	1.08108108108108\\
1.08608608608609	1.08608608608609\\
1.09109109109109	1.09109109109109\\
1.0960960960961	1.0960960960961\\
1.1011011011011	1.1011011011011\\
1.10610610610611	1.10610610610611\\
1.11111111111111	1.11111111111111\\
1.11611611611612	1.11611611611612\\
1.12112112112112	1.12112112112112\\
1.12612612612613	1.12612612612613\\
1.13113113113113	1.13113113113113\\
1.13613613613614	1.13613613613614\\
1.14114114114114	1.14114114114114\\
1.14614614614615	1.14614614614615\\
1.15115115115115	1.15115115115115\\
1.15615615615616	1.15615615615616\\
1.16116116116116	1.16116116116116\\
1.16616616616617	1.16616616616617\\
1.17117117117117	1.17117117117117\\
1.17617617617618	1.17617617617618\\
1.18118118118118	1.18118118118118\\
1.18618618618619	1.18618618618619\\
1.19119119119119	1.19119119119119\\
1.1961961961962	1.1961961961962\\
1.2012012012012	1.2012012012012\\
1.20620620620621	1.20620620620621\\
1.21121121121121	1.21121121121121\\
1.21621621621622	1.21621621621622\\
1.22122122122122	1.22122122122122\\
1.22622622622623	1.22622622622623\\
1.23123123123123	1.23123123123123\\
1.23623623623624	1.23623623623624\\
1.24124124124124	1.24124124124124\\
1.24624624624625	1.24624624624625\\
1.25125125125125	1.25125125125125\\
1.25625625625626	1.25625625625626\\
1.26126126126126	1.26126126126126\\
1.26626626626627	1.26626626626627\\
1.27127127127127	1.27127127127127\\
1.27627627627628	1.27627627627628\\
1.28128128128128	1.28128128128128\\
1.28628628628629	1.28628628628629\\
1.29129129129129	1.29129129129129\\
1.2962962962963	1.2962962962963\\
1.3013013013013	1.3013013013013\\
1.30630630630631	1.30630630630631\\
1.31131131131131	1.31131131131131\\
1.31631631631632	1.31631631631632\\
1.32132132132132	1.32132132132132\\
1.32632632632633	1.32632632632633\\
1.33133133133133	1.33133133133133\\
1.33633633633634	1.33633633633634\\
1.34134134134134	1.34134134134134\\
1.34634634634635	1.34634634634635\\
1.35135135135135	1.35135135135135\\
1.35635635635636	1.35635635635636\\
1.36136136136136	1.36136136136136\\
1.36636636636637	1.36636636636637\\
1.37137137137137	1.37137137137137\\
1.37637637637638	1.37637637637638\\
1.38138138138138	1.38138138138138\\
1.38638638638639	1.38638638638639\\
1.39139139139139	1.39139139139139\\
1.3963963963964	1.3963963963964\\
1.4014014014014	1.4014014014014\\
1.40640640640641	1.40640640640641\\
1.41141141141141	1.41141141141141\\
1.41641641641642	1.41641641641642\\
1.42142142142142	1.42142142142142\\
1.42642642642643	1.42642642642643\\
1.43143143143143	1.43143143143143\\
1.43643643643644	1.43643643643644\\
1.44144144144144	1.44144144144144\\
1.44644644644645	1.44644644644645\\
1.45145145145145	1.45145145145145\\
1.45645645645646	1.45645645645646\\
1.46146146146146	1.46146146146146\\
1.46646646646647	1.46646646646647\\
1.47147147147147	1.47147147147147\\
1.47647647647648	1.47647647647648\\
1.48148148148148	1.48148148148148\\
1.48648648648649	1.48648648648649\\
1.49149149149149	1.49149149149149\\
1.4964964964965	1.4964964964965\\
1.5015015015015	1.5015015015015\\
1.50650650650651	1.50650650650651\\
1.51151151151151	1.51151151151151\\
1.51651651651652	1.51651651651652\\
1.52152152152152	1.52152152152152\\
1.52652652652653	1.52652652652653\\
1.53153153153153	1.53153153153153\\
1.53653653653654	1.53653653653654\\
1.54154154154154	1.54154154154154\\
1.54654654654655	1.54654654654655\\
1.55155155155155	1.55155155155155\\
1.55655655655656	1.55655655655656\\
1.56156156156156	1.56156156156156\\
1.56656656656657	1.56656656656657\\
1.57157157157157	1.57157157157157\\
1.57657657657658	1.57657657657658\\
1.58158158158158	1.58158158158158\\
1.58658658658659	1.58658658658659\\
1.59159159159159	1.59159159159159\\
1.5965965965966	1.5965965965966\\
1.6016016016016	1.6016016016016\\
1.60660660660661	1.60660660660661\\
1.61161161161161	1.61161161161161\\
1.61661661661662	1.61661661661662\\
1.62162162162162	1.62162162162162\\
1.62662662662663	1.62662662662663\\
1.63163163163163	1.63163163163163\\
1.63663663663664	1.63663663663664\\
1.64164164164164	1.64164164164164\\
1.64664664664665	1.64664664664665\\
1.65165165165165	1.65165165165165\\
1.65665665665666	1.65665665665666\\
1.66166166166166	1.66166166166166\\
1.66666666666667	1.66666666666667\\
1.67167167167167	1.67167167167167\\
1.67667667667668	1.67667667667668\\
1.68168168168168	1.68168168168168\\
1.68668668668669	1.68668668668669\\
1.69169169169169	1.69169169169169\\
1.6966966966967	1.6966966966967\\
1.7017017017017	1.7017017017017\\
1.70670670670671	1.70670670670671\\
1.71171171171171	1.71171171171171\\
1.71671671671672	1.71671671671672\\
1.72172172172172	1.72172172172172\\
1.72672672672673	1.72672672672673\\
1.73173173173173	1.73173173173173\\
1.73673673673674	1.73673673673674\\
1.74174174174174	1.74174174174174\\
1.74674674674675	1.74674674674675\\
1.75175175175175	1.75175175175175\\
1.75675675675676	1.75675675675676\\
1.76176176176176	1.76176176176176\\
1.76676676676677	1.76676676676677\\
1.77177177177177	1.77177177177177\\
1.77677677677678	1.77677677677678\\
1.78178178178178	1.78178178178178\\
1.78678678678679	1.78678678678679\\
1.79179179179179	1.79179179179179\\
1.7967967967968	1.7967967967968\\
1.8018018018018	1.8018018018018\\
1.80680680680681	1.80680680680681\\
1.81181181181181	1.81181181181181\\
1.81681681681682	1.81681681681682\\
1.82182182182182	1.82182182182182\\
1.82682682682683	1.82682682682683\\
1.83183183183183	1.83183183183183\\
1.83683683683684	1.83683683683684\\
1.84184184184184	1.84184184184184\\
1.84684684684685	1.84684684684685\\
1.85185185185185	1.85185185185185\\
1.85685685685686	1.85685685685686\\
1.86186186186186	1.86186186186186\\
1.86686686686687	1.86686686686687\\
1.87187187187187	1.87187187187187\\
1.87687687687688	1.87687687687688\\
1.88188188188188	1.88188188188188\\
1.88688688688689	1.88688688688689\\
1.89189189189189	1.89189189189189\\
1.8968968968969	1.8968968968969\\
1.9019019019019	1.9019019019019\\
1.90690690690691	1.90690690690691\\
1.91191191191191	1.91191191191191\\
1.91691691691692	1.91691691691692\\
1.92192192192192	1.92192192192192\\
1.92692692692693	1.92692692692693\\
1.93193193193193	1.93193193193193\\
1.93693693693694	1.93693693693694\\
1.94194194194194	1.94194194194194\\
1.94694694694695	1.94694694694695\\
1.95195195195195	1.95195195195195\\
1.95695695695696	1.95695695695696\\
1.96196196196196	1.96196196196196\\
1.96696696696697	1.96696696696697\\
1.97197197197197	1.97197197197197\\
1.97697697697698	1.97697697697698\\
1.98198198198198	1.98198198198198\\
1.98698698698699	1.98698698698699\\
1.99199199199199	1.99199199199199\\
1.996996996997	1.996996996997\\
2.002002002002	2\\
2.00700700700701	2\\
2.01201201201201	2\\
2.01701701701702	2\\
2.02202202202202	2\\
2.02702702702703	2\\
2.03203203203203	2\\
2.03703703703704	2\\
2.04204204204204	2\\
2.04704704704705	2\\
2.05205205205205	2\\
2.05705705705706	2\\
2.06206206206206	2\\
2.06706706706707	2\\
2.07207207207207	2\\
2.07707707707708	2\\
2.08208208208208	2\\
2.08708708708709	2\\
2.09209209209209	2\\
2.0970970970971	2\\
2.1021021021021	2\\
2.10710710710711	2\\
2.11211211211211	2\\
2.11711711711712	2\\
2.12212212212212	2\\
2.12712712712713	2\\
2.13213213213213	2\\
2.13713713713714	2\\
2.14214214214214	2\\
2.14714714714715	2\\
2.15215215215215	2\\
2.15715715715716	2\\
2.16216216216216	2\\
2.16716716716717	2\\
2.17217217217217	2\\
2.17717717717718	2\\
2.18218218218218	2\\
2.18718718718719	2\\
2.19219219219219	2\\
2.1971971971972	2\\
2.2022022022022	2\\
2.20720720720721	2\\
2.21221221221221	2\\
2.21721721721722	2\\
2.22222222222222	2\\
2.22722722722723	2\\
2.23223223223223	2\\
2.23723723723724	2\\
2.24224224224224	2\\
2.24724724724725	2\\
2.25225225225225	2\\
2.25725725725726	2\\
2.26226226226226	2\\
2.26726726726727	2\\
2.27227227227227	2\\
2.27727727727728	2\\
2.28228228228228	2\\
2.28728728728729	2\\
2.29229229229229	2\\
2.2972972972973	2\\
2.3023023023023	2\\
2.30730730730731	2\\
2.31231231231231	2\\
2.31731731731732	2\\
2.32232232232232	2\\
2.32732732732733	2\\
2.33233233233233	2\\
2.33733733733734	2\\
2.34234234234234	2\\
2.34734734734735	2\\
2.35235235235235	2\\
2.35735735735736	2\\
2.36236236236236	2\\
2.36736736736737	2\\
2.37237237237237	2\\
2.37737737737738	2\\
2.38238238238238	2\\
2.38738738738739	2\\
2.39239239239239	2\\
2.3973973973974	2\\
2.4024024024024	2\\
2.40740740740741	2\\
2.41241241241241	2\\
2.41741741741742	2\\
2.42242242242242	2\\
2.42742742742743	2\\
2.43243243243243	2\\
2.43743743743744	2\\
2.44244244244244	2\\
2.44744744744745	2\\
2.45245245245245	2\\
2.45745745745746	2\\
2.46246246246246	2\\
2.46746746746747	2\\
2.47247247247247	2\\
2.47747747747748	2\\
2.48248248248248	2\\
2.48748748748749	2\\
2.49249249249249	2\\
2.4974974974975	2\\
2.5025025025025	2\\
2.50750750750751	2\\
2.51251251251251	2\\
2.51751751751752	2\\
2.52252252252252	2\\
2.52752752752753	2\\
2.53253253253253	2\\
2.53753753753754	2\\
2.54254254254254	2\\
2.54754754754755	2\\
2.55255255255255	2\\
2.55755755755756	2\\
2.56256256256256	2\\
2.56756756756757	2\\
2.57257257257257	2\\
2.57757757757758	2\\
2.58258258258258	2\\
2.58758758758759	2\\
2.59259259259259	2\\
2.5975975975976	2\\
2.6026026026026	2\\
2.60760760760761	2\\
2.61261261261261	2\\
2.61761761761762	2\\
2.62262262262262	2\\
2.62762762762763	2\\
2.63263263263263	2\\
2.63763763763764	2\\
2.64264264264264	2\\
2.64764764764765	2\\
2.65265265265265	2\\
2.65765765765766	2\\
2.66266266266266	2\\
2.66766766766767	2\\
2.67267267267267	2\\
2.67767767767768	2\\
2.68268268268268	2\\
2.68768768768769	2\\
2.69269269269269	2\\
2.6976976976977	2\\
2.7027027027027	2\\
2.70770770770771	2\\
2.71271271271271	2\\
2.71771771771772	2\\
2.72272272272272	2\\
2.72772772772773	2\\
2.73273273273273	2\\
2.73773773773774	2\\
2.74274274274274	2\\
2.74774774774775	2\\
2.75275275275275	2\\
2.75775775775776	2\\
2.76276276276276	2\\
2.76776776776777	2\\
2.77277277277277	2\\
2.77777777777778	2\\
2.78278278278278	2\\
2.78778778778779	2\\
2.79279279279279	2\\
2.7977977977978	2\\
2.8028028028028	2\\
2.80780780780781	2\\
2.81281281281281	2\\
2.81781781781782	2\\
2.82282282282282	2\\
2.82782782782783	2\\
2.83283283283283	2\\
2.83783783783784	2\\
2.84284284284284	2\\
2.84784784784785	2\\
2.85285285285285	2\\
2.85785785785786	2\\
2.86286286286286	2\\
2.86786786786787	2\\
2.87287287287287	2\\
2.87787787787788	2\\
2.88288288288288	2\\
2.88788788788789	2\\
2.89289289289289	2\\
2.8978978978979	2\\
2.9029029029029	2\\
2.90790790790791	2\\
2.91291291291291	2\\
2.91791791791792	2\\
2.92292292292292	2\\
2.92792792792793	2\\
2.93293293293293	2\\
2.93793793793794	2\\
2.94294294294294	2\\
2.94794794794795	2\\
2.95295295295295	2\\
2.95795795795796	2\\
2.96296296296296	2\\
2.96796796796797	2\\
2.97297297297297	2\\
2.97797797797798	2\\
2.98298298298298	2\\
2.98798798798799	2\\
2.99299299299299	2\\
2.997997997998	2\\
3.003003003003	2\\
3.00800800800801	2\\
3.01301301301301	2\\
3.01801801801802	2\\
3.02302302302302	2\\
3.02802802802803	2\\
3.03303303303303	2\\
3.03803803803804	2\\
3.04304304304304	2\\
3.04804804804805	2\\
3.05305305305305	2\\
3.05805805805806	2\\
3.06306306306306	2\\
3.06806806806807	2\\
3.07307307307307	2\\
3.07807807807808	2\\
3.08308308308308	2\\
3.08808808808809	2\\
3.09309309309309	2\\
3.0980980980981	2\\
3.1031031031031	2\\
3.10810810810811	2\\
3.11311311311311	2\\
3.11811811811812	2\\
3.12312312312312	2\\
3.12812812812813	2\\
3.13313313313313	2\\
3.13813813813814	2\\
3.14314314314314	2\\
3.14814814814815	2\\
3.15315315315315	2\\
3.15815815815816	2\\
3.16316316316316	2\\
3.16816816816817	2\\
3.17317317317317	2\\
3.17817817817818	2\\
3.18318318318318	2\\
3.18818818818819	2\\
3.19319319319319	2\\
3.1981981981982	2\\
3.2032032032032	2\\
3.20820820820821	2\\
3.21321321321321	2\\
3.21821821821822	2\\
3.22322322322322	2\\
3.22822822822823	2\\
3.23323323323323	2\\
3.23823823823824	2\\
3.24324324324324	2\\
3.24824824824825	2\\
3.25325325325325	2\\
3.25825825825826	2\\
3.26326326326326	2\\
3.26826826826827	2\\
3.27327327327327	2\\
3.27827827827828	2\\
3.28328328328328	2\\
3.28828828828829	2\\
3.29329329329329	2\\
3.2982982982983	2\\
3.3033033033033	2\\
3.30830830830831	2\\
3.31331331331331	2\\
3.31831831831832	2\\
3.32332332332332	2\\
3.32832832832833	2\\
3.33333333333333	2\\
3.33833833833834	2\\
3.34334334334334	2\\
3.34834834834835	2\\
3.35335335335335	2\\
3.35835835835836	2\\
3.36336336336336	2\\
3.36836836836837	2\\
3.37337337337337	2\\
3.37837837837838	2\\
3.38338338338338	2\\
3.38838838838839	2\\
3.39339339339339	2\\
3.3983983983984	2\\
3.4034034034034	2\\
3.40840840840841	2\\
3.41341341341341	2\\
3.41841841841842	2\\
3.42342342342342	2\\
3.42842842842843	2\\
3.43343343343343	2\\
3.43843843843844	2\\
3.44344344344344	2\\
3.44844844844845	2\\
3.45345345345345	2\\
3.45845845845846	2\\
3.46346346346346	2\\
3.46846846846847	2\\
3.47347347347347	2\\
3.47847847847848	2\\
3.48348348348348	2\\
3.48848848848849	2\\
3.49349349349349	2\\
3.4984984984985	2\\
3.5035035035035	2\\
3.50850850850851	2\\
3.51351351351351	2\\
3.51851851851852	2\\
3.52352352352352	2\\
3.52852852852853	2\\
3.53353353353353	2\\
3.53853853853854	2\\
3.54354354354354	2\\
3.54854854854855	2\\
3.55355355355355	2\\
3.55855855855856	2\\
3.56356356356356	2\\
3.56856856856857	2\\
3.57357357357357	2\\
3.57857857857858	2\\
3.58358358358358	2\\
3.58858858858859	2\\
3.59359359359359	2\\
3.5985985985986	2\\
3.6036036036036	2\\
3.60860860860861	2\\
3.61361361361361	2\\
3.61861861861862	2\\
3.62362362362362	2\\
3.62862862862863	2\\
3.63363363363363	2\\
3.63863863863864	2\\
3.64364364364364	2\\
3.64864864864865	2\\
3.65365365365365	2\\
3.65865865865866	2\\
3.66366366366366	2\\
3.66866866866867	2\\
3.67367367367367	2\\
3.67867867867868	2\\
3.68368368368368	2\\
3.68868868868869	2\\
3.69369369369369	2\\
3.6986986986987	2\\
3.7037037037037	2\\
3.70870870870871	2\\
3.71371371371371	2\\
3.71871871871872	2\\
3.72372372372372	2\\
3.72872872872873	2\\
3.73373373373373	2\\
3.73873873873874	2\\
3.74374374374374	2\\
3.74874874874875	2\\
3.75375375375375	2\\
3.75875875875876	2\\
3.76376376376376	2\\
3.76876876876877	2\\
3.77377377377377	2\\
3.77877877877878	2\\
3.78378378378378	2\\
3.78878878878879	2\\
3.79379379379379	2\\
3.7987987987988	2\\
3.8038038038038	2\\
3.80880880880881	2\\
3.81381381381381	2\\
3.81881881881882	2\\
3.82382382382382	2\\
3.82882882882883	2\\
3.83383383383383	2\\
3.83883883883884	2\\
3.84384384384384	2\\
3.84884884884885	2\\
3.85385385385385	2\\
3.85885885885886	2\\
3.86386386386386	2\\
3.86886886886887	2\\
3.87387387387387	2\\
3.87887887887888	2\\
3.88388388388388	2\\
3.88888888888889	2\\
3.89389389389389	2\\
3.8988988988989	2\\
3.9039039039039	2\\
3.90890890890891	2\\
3.91391391391391	2\\
3.91891891891892	2\\
3.92392392392392	2\\
3.92892892892893	2\\
3.93393393393393	2\\
3.93893893893894	2\\
3.94394394394394	2\\
3.94894894894895	2\\
3.95395395395395	2\\
3.95895895895896	2\\
3.96396396396396	2\\
3.96896896896897	2\\
3.97397397397397	2\\
3.97897897897898	2\\
3.98398398398398	2\\
3.98898898898899	2\\
3.99399399399399	2\\
3.998998998999	2\\
4.004004004004	2\\
4.00900900900901	2\\
4.01401401401401	2\\
4.01901901901902	2\\
4.02402402402402	2\\
4.02902902902903	2\\
4.03403403403403	2\\
4.03903903903904	2\\
4.04404404404404	2\\
4.04904904904905	2\\
4.05405405405405	2\\
4.05905905905906	2\\
4.06406406406406	2\\
4.06906906906907	2\\
4.07407407407407	2\\
4.07907907907908	2\\
4.08408408408408	2\\
4.08908908908909	2\\
4.09409409409409	2\\
4.0990990990991	2\\
4.1041041041041	2\\
4.10910910910911	2\\
4.11411411411411	2\\
4.11911911911912	2\\
4.12412412412412	2\\
4.12912912912913	2\\
4.13413413413413	2\\
4.13913913913914	2\\
4.14414414414414	2\\
4.14914914914915	2\\
4.15415415415415	2\\
4.15915915915916	2\\
4.16416416416416	2\\
4.16916916916917	2\\
4.17417417417417	2\\
4.17917917917918	2\\
4.18418418418418	2\\
4.18918918918919	2\\
4.19419419419419	2\\
4.1991991991992	2\\
4.2042042042042	2\\
4.20920920920921	2\\
4.21421421421421	2\\
4.21921921921922	2\\
4.22422422422422	2\\
4.22922922922923	2\\
4.23423423423423	2\\
4.23923923923924	2\\
4.24424424424424	2\\
4.24924924924925	2\\
4.25425425425425	2\\
4.25925925925926	2\\
4.26426426426426	2\\
4.26926926926927	2\\
4.27427427427427	2\\
4.27927927927928	2\\
4.28428428428428	2\\
4.28928928928929	2\\
4.29429429429429	2\\
4.2992992992993	2\\
4.3043043043043	2\\
4.30930930930931	2\\
4.31431431431431	2\\
4.31931931931932	2\\
4.32432432432432	2\\
4.32932932932933	2\\
4.33433433433433	2\\
4.33933933933934	2\\
4.34434434434434	2\\
4.34934934934935	2\\
4.35435435435435	2\\
4.35935935935936	2\\
4.36436436436436	2\\
4.36936936936937	2\\
4.37437437437437	2\\
4.37937937937938	2\\
4.38438438438438	2\\
4.38938938938939	2\\
4.39439439439439	2\\
4.3993993993994	2\\
4.4044044044044	2\\
4.40940940940941	2\\
4.41441441441441	2\\
4.41941941941942	2\\
4.42442442442442	2\\
4.42942942942943	2\\
4.43443443443443	2\\
4.43943943943944	2\\
4.44444444444444	2\\
4.44944944944945	2\\
4.45445445445445	2\\
4.45945945945946	2\\
4.46446446446446	2\\
4.46946946946947	2\\
4.47447447447447	2\\
4.47947947947948	2\\
4.48448448448448	2\\
4.48948948948949	2\\
4.49449449449449	2\\
4.4994994994995	2\\
4.5045045045045	2\\
4.50950950950951	2\\
4.51451451451451	2\\
4.51951951951952	2\\
4.52452452452452	2\\
4.52952952952953	2\\
4.53453453453453	2\\
4.53953953953954	2\\
4.54454454454454	2\\
4.54954954954955	2\\
4.55455455455455	2\\
4.55955955955956	2\\
4.56456456456456	2\\
4.56956956956957	2\\
4.57457457457457	2\\
4.57957957957958	2\\
4.58458458458458	2\\
4.58958958958959	2\\
4.59459459459459	2\\
4.5995995995996	2\\
4.6046046046046	2\\
4.60960960960961	2\\
4.61461461461461	2\\
4.61961961961962	2\\
4.62462462462462	2\\
4.62962962962963	2\\
4.63463463463463	2\\
4.63963963963964	2\\
4.64464464464464	2\\
4.64964964964965	2\\
4.65465465465465	2\\
4.65965965965966	2\\
4.66466466466466	2\\
4.66966966966967	2\\
4.67467467467467	2\\
4.67967967967968	2\\
4.68468468468468	2\\
4.68968968968969	2\\
4.69469469469469	2\\
4.6996996996997	2\\
4.7047047047047	2\\
4.70970970970971	2\\
4.71471471471471	2\\
4.71971971971972	2\\
4.72472472472472	2\\
4.72972972972973	2\\
4.73473473473473	2\\
4.73973973973974	2\\
4.74474474474474	2\\
4.74974974974975	2\\
4.75475475475475	2\\
4.75975975975976	2\\
4.76476476476476	2\\
4.76976976976977	2\\
4.77477477477477	2\\
4.77977977977978	2\\
4.78478478478478	2\\
4.78978978978979	2\\
4.79479479479479	2\\
4.7997997997998	2\\
4.8048048048048	2\\
4.80980980980981	2\\
4.81481481481481	2\\
4.81981981981982	2\\
4.82482482482482	2\\
4.82982982982983	2\\
4.83483483483483	2\\
4.83983983983984	2\\
4.84484484484484	2\\
4.84984984984985	2\\
4.85485485485485	2\\
4.85985985985986	2\\
4.86486486486486	2\\
4.86986986986987	2\\
4.87487487487487	2\\
4.87987987987988	2\\
4.88488488488488	2\\
4.88988988988989	2\\
4.89489489489489	2\\
4.8998998998999	2\\
4.9049049049049	2\\
4.90990990990991	2\\
4.91491491491491	2\\
4.91991991991992	2\\
4.92492492492492	2\\
4.92992992992993	2\\
4.93493493493493	2\\
4.93993993993994	2\\
4.94494494494494	2\\
4.94994994994995	2\\
4.95495495495495	2\\
4.95995995995996	2\\
4.96496496496496	2\\
4.96996996996997	2\\
4.97497497497497	2\\
4.97997997997998	2\\
4.98498498498498	2\\
4.98998998998999	2\\
4.99499499499499	2\\
5	2\\
};
\addplot [color=blue,line width=1.0pt,only marks,mark=o,mark options={solid},forget plot]
  table[row sep=crcr]{1	4.0005\\
};
\addplot [color=blue,line width=3.0pt,only marks,mark=o,mark options={solid},forget plot]
  table[row sep=crcr]{1	3\\
};

\end{axis}
}
\begin{abstract}
The impact of cooperation on interference management is investigated by studying an elemental wireless
network, the so called symmetric interference relay channel (IRC), from a generalized degrees of
freedom (GDoF) perspective. 
This is motivated by the fact that the deployment of relays is considered as
a remedy to overcome the bottleneck of current systems in terms of achievable rates. 
The focus of this work is on the regime in which the interference link is weaker than the source-relay link in the IRC.
Our approach towards studying the GDoF goes
through the capacity analysis of the linear deterministic IRC (LD-IRC). 
New upper bounds on the sum-capacity of the LD-IRC based on genie-aided approaches are established. 
These upper bounds together with some existing upper
bounds are achieved by using four novel transmission schemes. 
Extending the upper bounds and the transmission schemes to the Gaussian case, the GDoF of
the Gaussian IRC is characterized for the aforementioned regime. This completes the GDoF results
available in the literature for the symmetric GDoF.
It is shown that in the strong interference regime, in contrast to the IC, the GDoF is not a monotonically
increasing function of the interference level. 
\end{abstract}

\section{Introduction} 
The demand for higher rates already exceeds today the supply of rates and thus resulting in a spectrum deficit. This is mainly due to an exclusive bandwidth allocation to each user. Hence, each user communicates over its own bandwidth without neither causing interference to other users nor be interfered by the other users. This is referred to as interference avoidance. 
In general, this approach is suboptimal and one possible solution is indeed to eliminate the constraint of exclusive bandwidth usage. Doing so results in the flexibility of spectrum usage, however it does not come for free. In more details, we obtain a system in which in addition to the additive Gaussian noise, users have to cope with the interference due to the concurrent transmissions and receptions.
Thus, there is a need for sophisticated interference management, i.e., the handling of the communication traffic in a way such that the burden caused by the interference is limited.
A key element to achieve this is to deploy additional nodes, referred to as relay nodes. 
The relay(s) then cooperate(s) with the transmitters and receivers and coordinate(s) in a distributed way the communication. The goal of this paper is to determine the fundamental limits of such a distributed coordination and cooperation. 

To be more concrete, we add a causal, full-duplex relay node to a two user interference channel (IC) in which two transmitters want to communicate with their desired receiver. The obtained setup is known as an interference relay channel (IRC), which has been first introduced in \cite{SahinErkip}. 
Obviously, the achievable sum-rate of the IRC cannot be worse than that of the IC. However, in order to benefit from the relay, sophisticated relaying strategies, which mitigate the impact of interference, are required. 
In \cite{TianYener_PotentRelayJournal}, an achievable sum-rate for the IRC based on compress-and-forward relaying strategy is studied. 
Moreover, special cases of IRC in which some channels are absent are studied in \cite{MaricDaboraGoldsmith_IT, BassiPiantanidaYang_ISIT,
ZahaviDabora_ISIT2014}. 
%Recently, the diversity-multiplexing tradeoff region based on different relaying strategies has been derived in \cite{ZahaviZhangMaricDaboraGoldsmithCui_IT_Trans}.
%%%%%%%%%%%%%%%%%%%%%%%%%%%%%%%%%%%%%%%%%%%%%%%%55
%In \cite{MaricDaboraGoldsmith_IT}, the benefit of interference forwarding is studied for the special case of IRC when the relay can only observe one of the transmitters. 
%Moreover, the constant gap optimality result for this special IRC is determined in \cite{BassiPiantanidaYang_ISIT}. 
The IRC has been analyzed for the variant with a cognitive relay in \cite{RiniTuninettiDevroyeGoldsmith, RiniTuninettiDevroye_ISIT,
RiniTuninettiDevroye_ITW,
SridharanVishwanathJafarShamai}. 
While in \cite{RiniTuninettiDevroyeGoldsmith}, the capacity of the IRC with a cognitive relay has been derived for very strong interference, in \cite{RiniTuninettiDevroye_ISIT} the capacity of this setup has been characterized within a constant gap for the case when no interference link is present. However, the capacity of IRC is still an open problem.

In order to obtain some progress on this research frontier, we study the so-called generalized degrees of freedom (GDoF) for the symmetric Gaussian IRC. This metric serves as an approximation of the capacity and has been used in several works. Essentially, it determines the number of available interference free streams in each channel use, where each stream has a capacity of a reference point-to-point channel (P2P). One of the most fundamental GDoF result is given in \cite{EtkinTseWang} in which the authors completed the GDoF characterization for the IC. 

The GDoF of the symmetric Gaussian IRC has been characterized in \cite{ChaabanSezgin_IT_IRC} for the case when the source-relay channel is weaker than the interference channel. It has been shown that relaying can increase the GDoF performance of the IRC, although it was shown in \cite{CadambeJafar_ImpactOfRelays} that relaying cannot increase the degrees of freedom (DoF). 
However, the GDOF of remaining regime, i.e., source-relay channel is stronger than the interference channel, was an open
problem until now. 
In this work, we study the GDoF of the IRC for the complementary regime and settle this problem. Hence, this work completes the characterization of the GDoF for the symmetric Gaussian IRC. 

To characterize the GDoF for the IRC, we consider first the deterministic model \cite{AvestimehrDiggaviTse_IT}.
Essentially, the Gaussian IRC is modeled as a linear deterministic interference relay channel (LD-IRC) in which the relationship between inputs and outputs of the channel is deterministic. This makes analyzing the LD-IRC simpler than the Gaussian IRC. Generally, solving the capacity for the deterministic channel can lead to the GDoF for the original Gaussian channel. This has been shown for several setups such as IC \cite{EtkinTseWang} and X-channel \cite{HuangJafarCadambe}. In this work, first we characterize the capacity of the LD-IRC. Next, we extend the result to the Gaussian IRC and derive the GDoF of the Gaussian IRC. 

To characterize the capacity of the LD-IRC, upper bounds and lower bounds for the capacity of the LD-IRC are needed. Besides two new upper bounds borrowed from \cite{ChaabanSezgin_IT_IRC}, two cut-set bounds and four new upper bounds based on genie aided approach are used. All bounds are capacity-tight and required for the characterization of the capacity of the LD-IRC. 
On the other hand, four transmission schemes are introduced. 
The proposed schemes are based on rate splitting. The optimal rate splitting which achieves the upper bound is provided for each scheme. Hence, these schemes achieve the sum-capacity upper bounds in the whole regime in which the source-relay link is stronger than the interference link. 
The proposed schemes are combination of common and private signaling with different relaying strategies. While in previous work \cite{ChaabanSezgin_IT_IRC}, only classical decode-and-forward \cite{Cover} and compute-and-forward \cite{NazerGastpar} were required to characterize the GDoF, in this work we need also cooperative interference neutralization \cite{ChaabanSezginTuninetti_Butterfly_Network} in addition to those relaying strategies. This is mainly due to the fact that in this work the source-relay link is stronger than the interference link and hence by providing some additional information to the relay which is not received at the destinations, the relay is able to neutralize the interference partially. Roughly speaking, using this relaying strategy, the causal relay is operating, to some extent, like a cognitive (Non-causal) relay.

As aforementioned, the upper bounds established in the LD-IRC are extended to the Gaussian model.
They provide us upper bounds on the GDoF performance for the IRC. Next, we apply the proposed schemes for the LD-IRC and obtain lower bounds for the GDoF of Gaussian IRC.
%channel to obtain the GDOF result of the Gaussian case
%Moreover, by splitting each link in the Gaussian IRC into several sub-channels with specific rate, the transmission schemes for the LD-IRC can be extended to the Gaussian IRC. Notice that each sub-channel in the Gaussian IRC can be used as a bit level in the LD-IRC. Hence, by using this approach, we need to only find the optimal sub-channel allocation instead of solving a power allocation problem. The solution of this problem can be directly obtained from the LD-IRC. 

The rest of the paper is organized as follows. In Section \ref{section:system_model}, we introduce the system model. 
The main results, i.e., the capacity of the LD-IRC and the GDoF of the Gaussian IRC, are presented in Section \ref{sec:Main_result}. 
The remainder of the paper is devoted to prove the results.
The upper bounds on the capacity of the LD-IRC are presented in Section
\ref{sec:UBLD-IRC}. In Section \ref{Sec:AchievabilityLDIRC}, the transmission schemes are presented for the LD-IRC. The upper bounds on the GDoF of the IRC are presented in Section \ref{sec:UBGaussian-IRC}. Finally, in Section \ref{sec:GDoFAchievingSchemes}, the extension of the transmission schemes to the Gaussian case is explained.

\emph{Notation}: Throughout the paper, we use $\mathbb{F}_2$ to denote the binary field and $\oplus$ to denote the modulo 2 addition. We use normal lower-case, normal upper-case, boldface lower-case, and boldface upper-case letters to denote scalars, scalar random variables, vectors, and matrices, respectively. $\boldsymbol{X}_{[a:b]}$ denotes the matrix formed by the $a$-th to $b$-th rows of a matrix $\mathbf{X}$, and the vector $\X_{[a:b]}$ is defined similarly. We write $X\sim\mathcal{N}(0,P)$ to indicate that the random variable $X$ is normal distributed with zero mean and variance $P$. Bern$(a)$ is a Bernoulli distribution with probability $a$. Furthermore, we define $x^{n}$ as the length-$n$ sequence $(x[1],\cdots,x[n])$. The vector $\boldsymbol{0}_{q}$ denotes the zero-vector of length $q$, the matrix $\boldsymbol{I}_q$ is the $q\times q$ identity matrix, {the matrix $\boldsymbol{0}_{l , m }$ represents the $l\times m$ zero matrix,} and $\boldsymbol{x}^{T}$ denotes the transposition of a {vector} $\boldsymbol{x}$. 
Moreover, we define the functions $C(x)$ and $C^+(x)$ as
\begin{align}
%C(x) &= \frac{1}{2}\log(1+x)\notag \\ 
%C^+(x) &= \max\lbrace0,C(x)	\rbrace. \notag
C(x) = \nicefrac{1}{2}\log(1+x), \quad
C^+(x) = \left(C(x)\right)^+, %\notag
\end{align}
where $(x)^+ = \max\lbrace0,x	\rbrace$. In this work, we suppose that all logarithms are binary unless the base of logarithm is given.
We say that a set of random variables is i.i.d if its components are independent and identically distributed.

\section{System Model}
\label{section:system_model}
We consider a network which consists of a two user interference channel with a causal and full-duplex relay as shown in Fig.~\ref{fig:GaussianIRC}. The transmission is performed in $n$ channel uses. While the transmitters are active in the all channel uses, the relay is active in the last $n-1$ channel uses due to the causality constraint.

%\begin{figure}[h]
%\centering
%\begin{tikzpicture}[scale=0.8]
%\TheoreticalIRC
%\end{tikzpicture}
%\caption{Information theoretic model of IRC.}
%\label{fig:SI_Model_1}
%\end{figure}

Transmitter $i$ (TX$i$), $i \in \{1,2\}$, has a message $w_i$, which is a realization of a random variable $W_i$ uniformly distributed over the set $\mathcal W_i \triangleq \lbrace 1,\ldots,\left\lfloor 2^{nR_i}\right\rfloor \rbrace$ for its respective receiver (RX$i$). The messages of the Tx's are assumed to be independent from each other. Using an encoding function $f_i$, the message $w_i$ is encoded into a length $n$ codeword $x_i^{n} \in \mathbb R^{n}$, satisfying a power constraint
\begin{align}
\frac{1}{n}\sum_{k=1}^{n}\mathbb{E}[x_{i}[k]^2] = P_i \leq P, \quad i \in \{1,2\}.
\end{align}
Then, the $k$th symbol $x_i[k]$, $k=1,\ldots,n$, is transmitted in the $k$th channel use.

At the end of the $k$th channel use, the relay has $y_r[k]$ which is given by
\begin{align}
y_{r}[k] &= h_{s} x_{1}[k] + h_{s} x_{2}[k] + z_{r}[k],\label{eq1:rec_relay_Ga}
\end{align}
where $h_s$ represents the real-positive channel gain value of the source-relay channel, and $z_{r}[k]$ is a realization of an i.i.d. $\mathcal{N}(0,1)$ random variable which represents the additive white Gaussian noise (AWGN) at the relay. {The relay re-encodes $y_r[1], \ldots, y_r[k]$ using a function $f_{rk}$ into $x_r[k+1]$, and sends the symbol $x_r[k+1]$ in the $(k+1)$th channel use.} Since the transmit signal of the relay is generated from its received signal in previous channel uses, the relay is inactive in channel use $k=1$, i.e., $x_r[1]=0$. Furthermore, the relay signal satisfies the power constraint
\begin{align}
\frac{1}{n}\sum_{k= 1}^{n}\mathbb{E}[x_r[k]^2] = P_r \leq P.
\end{align}

The destinations wait until end of the $n$th channel use, at which Rx$j$ has received the sequence $y_{j}^{n}$, where $y_j[k]$ is given by 
\begin{align}
y_{j}[k] &=  h_d x_{j}[k] + h_c x_{l}[k] + h_{r}  x_{r}[k] + z_{j}[k], 
\label{eq1:rec_rx_Ga}
\end{align}
where $l$ is the index of the undesired Tx ($j\neq l$) and $h_d$, $h_c$ and $h_r$ represent the real-positive channel gain value of the desired, interference and relay-destination channels, respectively, and the noise $z_j$ is a realization of an $\mathcal N(0,1)$ random variable. The AWGN at the receivers and the relay are independent of each other. 
Moreover, it is assumed that all channel values are perfectly known at all nodes. By using a decoding function $g_j$, Rx$j$ decodes $\hat w_j$, i.e., $\hat w_j=g_j(y_j^{n})$. The messages sets, encoding functions, and decoding functions constitute a code for the channel which is denoted as a $(n; 2^{nR_1} ; 2^{nR_2})$ code. Such a code induces an average error probability $\mathbb{P}^{(n)}$ defined as    
\begin{align}
\label{ERROR}
      \mathbb{P}^{(n)}=\dfrac{1}
      {2^{nR_\Sigma}}\sum_{\boldsymbol{w}\in\mathcal{W}_1
      \times\mathcal{W}_2}\mathrm{Prob}
      \left(E\right),
\end{align}
where $R_\Sigma=R_1+R_2$, $\boldsymbol{w}=(w_1,w_2)$, and $E$ is the error event $\hat w_i\neq w_i$ for some $i\in\{1,2\}$. Reliable communication is said to take place if this error probability can be made {arbitrarily} small by increasing $n$. The achievability of a rate tuple $(R_1,R_2)$ is defined as the existence of a coding scheme which achieves reliable communication with these rates. In other words, a rate tuple $(R_1,R_2)$ is said to be achievable if there exists a sequence of $(n,2^{nR_1},2^{nR_2})$ codes such that $\mathbb{P}^{(n)}\to0$ as $n\to\infty$. The set of all achievable rate tuples is the capacity region of the IRC denoted by $\mathcal{C}$. In this paper, we focus on the sum-capacity defined as the maximum achievable sum-rate, i.e.,
   \begin{align}\label{sumcap}
      {C}_{\Sigma}=\max_{(R_1,R_2)\in\mathcal{C}}R_\Sigma.
   \end{align}

\begin{figure}
\centering
\begin{tikzpicture}[scale=0.8]
\GaussianIRC 
%\TheoreticalIRC
\end{tikzpicture}
\caption{System model of the symmetric Gaussian IRC.}
\label{fig:GaussianIRC}
\end{figure}
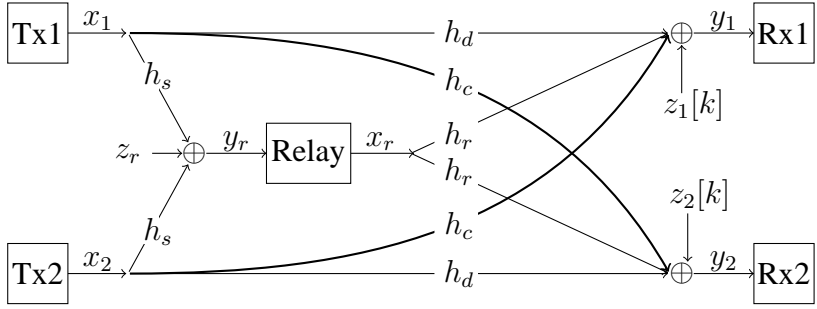

We consider the interference limited scenario, and hence, we assume that the received signal power from each node is larger than the noise variance, i.e.,
\begin{align}
\label{InterferenceLimited}
\min\{h_d^2,h_c^2,h_s^2,h_r^2\}P>1.
\end{align}
For readability, we use the following parameters
\begin{align}
m_d=\frac{1}{2}\log(Ph_d^2), \quad m_c=\frac{1}{2}\log(Ph_c^2), \quad m_r=\frac{1}{2}\log(Ph_r^2), \quad m_s=\frac{1}{2}\log(Ph_s^2).
\end{align}
Furthermore, since the focus of the paper is on the GDoF of the IRC, for a fixed $h_d^2$ we define $\alpha$, $\beta$, and $\gamma$ as
\begin{align}
\alpha = \frac{m_c}{m_d}, \quad \beta  = \frac{m_r}{m_d},  \quad \gamma = \frac{m_s}{m_d}. \label{eq:param_def}
\end{align} 
Then, the GDoF of the Gaussian IRC, $d(\alpha, \beta , \gamma)$ is defined as
\begin{align}
d(\alpha, \beta , \gamma) = \lim_{m_d\rightarrow \infty} \frac{C_{\mathrm{G},\Sigma}(\alpha,\beta,\gamma)}{m_d}, \label{eq:2_1}
\end{align}
where $C_{\mathrm{G},\Sigma}$ represents the sum-capacity of the Gaussian IRC\footnote{For a fixed $h_d^2$, we can write $m_d\to\infty$ is equivalent to $P\to \infty$.}. Our approach towards the GDoF analysis of Gaussian IRC starts with the \textit{linear-deterministic (LD) approximation} of the wireless network introduced by Avestimehr  \textit{et al.} in~\cite{AvestimehrDiggaviTse_IT}. Next, we introduce the linear deterministic IRC (LD-IRC).

\subsection{Linear Deterministic Model}
The Gaussian IRC shown in Fig.~\ref{fig:GaussianIRC} can be approximated by the LD model as follows. An input symbol at Tx$i$ is given by a binary vector $\X_i\in\mathbb{F}_2^q$ where $q=\max\{n_d,n_c,n_r,n_s\}$ and the integers $n_d$, $n_c$, $n_r$, and $n_s$ represent the channel strength and they are defined as follows
\begin{align}
      n_d=\left\lfloor m_d\right\rfloor,\quad n_c=\left\lfloor m_c\right\rfloor,\quad n_r=\left\lfloor m_r\right\rfloor,\quad n_s=\left\lfloor  m_s\right\rfloor.
      \label{eq:LD_ITC_channel_def}
\end{align}
In the $k$th channel use, where $k=1,\ldots,n$, the output signal vector $\Y_r[k]$ at the relay is given by the following deterministic function of the inputs 
\begin{align}
\Y_r[k]=\bS^{q-n_s}\X_1[k] \oplus \bS^{q-n_s}\X_2[k], \label{eq1:rec_relay_LD}
\end{align}
where $\bS\in\mathbb{F}_2^{q\times q}$ is a down-shift matrix defined as 
\begin{align}\label{Sm}
\bS=\begin{pmatrix}
\boldsymbol{0}_{q-1}^{T} & 0\\ 
\boldsymbol{I}_{q-1} & \boldsymbol{0}_{q-1}
\end{pmatrix}.
\end{align}
The relay generates a signal vector $\X_r[k+1]$ at the end of the $k$th channel use from $\Y_r[1],\ldots,\Y_r[k]$ ($\X_r[1]=\boldsymbol 0_q$). The output signal vector $\Y_j[k]$ at Rx$j$ is given by the following deterministic functions of the inputs
\begin{align}
\Y_{j}[k] &= \bS^{q-n_d}  \X_{j}[k] \oplus \bS^{q-n_c} \X_{l}[k] \oplus \bS^{q-n_r}  \X_{r}[k],
\label{eq1:rec_rx_LD}
\end{align}
where $j\neq l$. The input-output equations \eqref{eq1:rec_relay_LD} and \eqref{eq1:rec_rx_LD} approximate the input-output equations of the Gaussian IRC given in \eqref{eq1:rec_relay_Ga} and \eqref{eq1:rec_rx_Ga} in the high SNR regime, respectively. We denote the sum-capacity of the LD-IRC by $C_{\mathrm{det},\Sigma}$. 
%Next, we study the sum-capacity of the LD-IRC in the regime of channel parameters where $n_c<n_s$, whereas the other channel parameters are arbitrary.

In the next section, we present the main results of the paper, which are a complete characterization of the sum-capacity of the LD-IRC, and the GDoF of the Gaussian one.

\section{Summary of the Main Results}
\label{sec:Main_result}
{The main results of the paper are the characterization of the sum-capacity of the LD-IRC, and the approximate sum-capacity of the Gaussian IRC given in terms of GDoF. The sum-capacity of the LD-IRC serves as a stepping stone towards the GDoF of the Gaussian counterpart. Therefore, in this paper we start by studying the LD-IRC, and in the process, we gather insights that are later used in the Gaussian case. The sum-capacity of the LD-IRC is summarized in the following subsection.}

\subsection{Sum-Capacity of the LD-IRC}
\label{sec:main_result_LD-IRC}
The channel parameter space of the IRC can be split into two regimes, $n_s\leq n_c$ and $n_s>n_c$. The first regime corresponds to cases where the source-relay channels are weaker than the cross channels which has been studied in \cite{ChaabanThesis}. We summarize the result of \cite{ChaabanThesis} here for completeness.

\begin{theorem}[weaker source-relay channels LD-IRC \cite{ChaabanThesis}]
The sum-capacity of the LD-IRC with $n_s\leq n_c$ is given by
\begin{align}
\label{SumCapLDIRC}
C_{\mathrm{det},\Sigma}=\min\left\{\begin{array}{c}
2\max\{n_d,n_r\}\\
2\max\{n_d,n_s\}\\
\max\{n_d,n_c,n_r\}+\max\{n_d,n_c\}-n_c\\
2\max\{n_d,n_c\}-n_c+n_s\\
2\max\{n_c,n_r,n_d-n_c\}\\
2\max\{n_c,n_d+n_s-n_c\}\end{array}\right\}.
\end{align}
\end{theorem}

{The remaining regime was left open in \cite{ChaabanThesis}. In this paper, we provide the complete solution  of the problem by closing the whole regime where the source-relay channels are stronger than the cross channels, i.e., $n_s>n_c$. The sum-capacity in this regime is given in the following theorem.}

\begin{theorem}[stronger source-relay channels LD-IRC]
\label{Theorem:capacity_LD_IRC}
The sum-capacity of the LD-IRC with $n_s>n_c$ is given by
\begin{align}
C_{\mathrm{det},\Sigma} = \min \begin{Bmatrix} 
2\max\{n_d,n_r\} \\
\max\{n_d,n_c,n_r\} + \max\{n_d,n_c\}-\left(n_c-[n_s-\max\{n_d,n_c\}]^+\right)^+\\
n_r+2\max\{n_d,n_c\}-n_c\\
2\max\{n_s,n_r+n_s-n_c,n_d-n_c\}\\
\max\{n_d,n_c\}+\max\{n_d,n_s\}\\
2\max\{n_c,n_r+(n_d-n_c)^+\}\\
\end{Bmatrix}
\label{eq:capacity_LD_IRC}
\end{align}
if $n_c\neq n_d$, and by
\begin{align}
\label{eq:capacity_LD_IRC_nceqnd}
C_{\mathrm{det},\Sigma} = \max\{n_d,\min\{n_r,n_s\}\}
\end{align}
otherwise.
\end{theorem}
The proof of this theorem for the special case $n_c=n_d$ is simple. In fact, the converse for this case follows from cut-set bounds, while its achievability follows by either ignoring the relay to achieve $n_d$, or using decode-forward at the relay to achieve $\min\{n_s,n_r\}$. The converse and achievability of the remaining case ($n_c\neq n_d$) are more involved, as they require genie-aided upper bounds, and transmission schemes which are based on different relay strategies. Details of the converse and achievability are presented in Sections \ref{sec:UBLD-IRC} and \ref{Sec:AchievabilityLDIRC}, respectively.

\subsection{GDoF Analysis of the Gaussian IRC}
\label{sec:main_result_Gaussian-IRC}
Using the capacity result to the LD-IRC and extending this for the Gaussian case, we obtain the GDoF characterization of the Gaussian IRC. {Similar to above, the GDoF of the case where the source-relay channels are weaker than the cross channels was characterized in \cite{ChaabanSezgin_IT_IRC}. In the language of GDoF, this corresponds to $\gamma\leq \alpha$. Basically, the GDoF in this case can be obtained from \eqref{SumCapLDIRC} by replacing $n_d$, $n_c$, $n_r$, and $n_s$ by $1$, $\alpha$, $\beta$, and $\gamma$, respectively. Using similar replacement in the statement of Theorem \ref{Theorem:capacity_LD_IRC} gives the GDoF of the remaining case $\gamma >\alpha$.} The following theorem presents this result.

\begin{theorem}[stronger source-relay channels Gaussian IRC]
\label{Theorem:GDoF_IRC}
The GDoF of the IRC with $\gamma>\alpha $ is given by
\begin{align}
d = \min \begin{Bmatrix} 
2\max\{1,\beta\} \\
\max\{1,\alpha,\beta\}+\max\{1,\alpha\}-\left(\alpha-[\gamma-\max\{1,\alpha\}]^+\right)^+ \\
\beta+2\max\{1,\alpha\}-\alpha\\
2\max\{\gamma,\beta+\gamma-\alpha,1-\alpha\}\\
\max\{1,\alpha\}+\max\{1,\gamma\}\\
2\max\{\alpha,\beta+(1-\alpha)^+\}
\end{Bmatrix}.
\end{align}
if $\alpha\neq1$, and by
\begin{align}
d = \max\{1,\min\{\beta,\gamma\}\}
\end{align}
otherwise.
\end{theorem}
{The proofs of the converse and achievability of this theorem are given in Sections \ref{sec:UBGaussian-IRC} and \ref{sec:GDoFAchievingSchemes}. The proofs are based on the insights obtained from the LD-IRC.}

Before we continue, we recall the GDoF for the IC given in \cite{EtkinTseWang}. 
\begin{lemma}[ETW \cite{EtkinTseWang}]
The GDoF for the IC is given by
\begin{align}
d&= \min\{\max\{2-2\alpha,2\alpha\},2-\alpha\}, \quad  \alpha \leq 1 \\
d&= \min\{\alpha,2\}, \phantom{\{2-2\alpha,2\alpha\},2-\alpha\} \quad}  \alpha > 1.
\end{align}
\end{lemma}

\subsection{Discussion}
Before going into details of the proof of Theorems \ref{Theorem:capacity_LD_IRC} and \ref{Theorem:GDoF_IRC}, we discuss the GDoF result and the obtained gain by using the relay. To this end, we study the GDoF for the IRC for different ranges of source-relay channel strength.
%%%%%%
\begin{figure}
\centering
\subfigure [] {
\begin{tikzpicture}[scale=1]
\bonegseven
\end{tikzpicture}
\label{Fig:bonegseven}
}
\subfigure [] {
\begin{tikzpicture}[scale=1]
\bfourgseven
\end{tikzpicture}
\label{Fig:bfourgseven2}
}
\subfigure [] {
\begin{tikzpicture}[scale=1]
\bsevengseven
\end{tikzpicture}
\label{Fig:bsevengseven}
}

\caption{The GDoF of the IRC (solid-line) for different values of $\beta$ while $\gamma= 0.7$. The dashed line represents the GDoF of the IC. (a) $\beta=0.1$, (b) $\beta=0.4$, (c) $\beta = 0.7$. }
\label{Fig:gseven}
\end{figure}
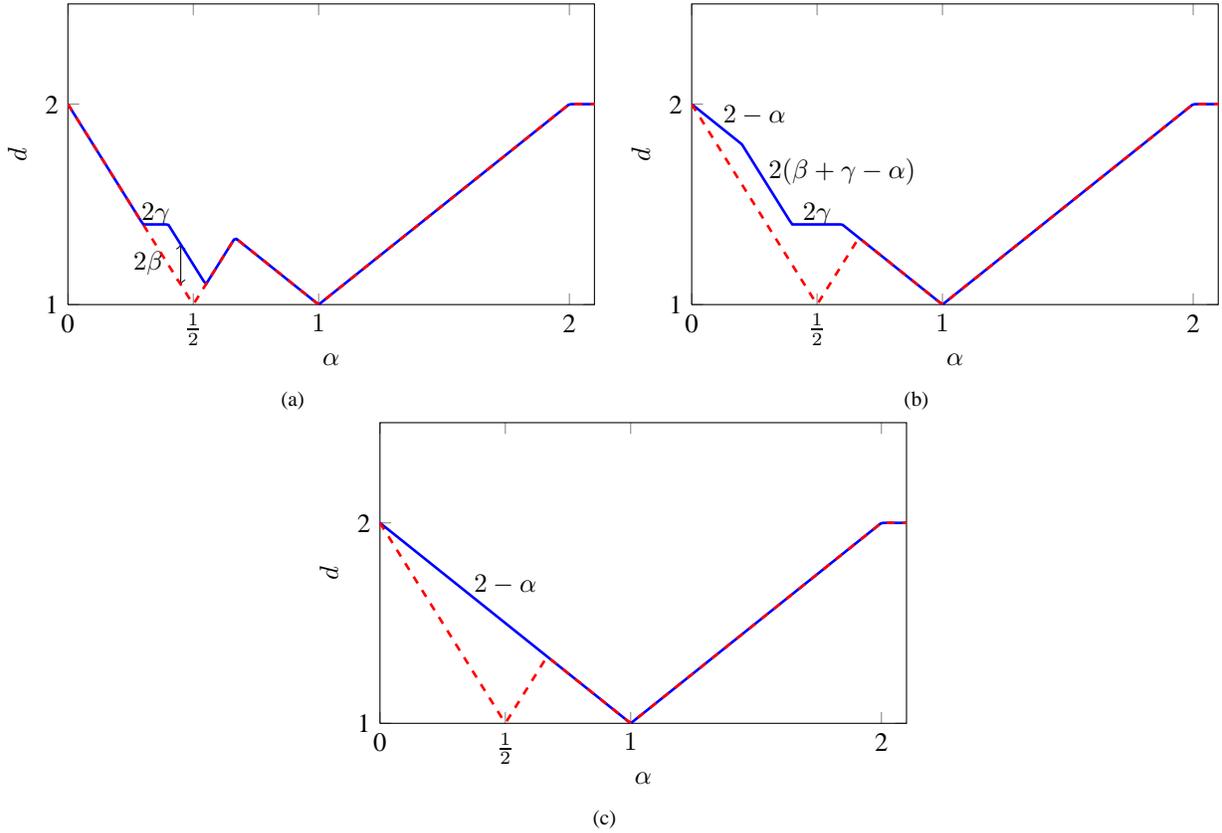
First, consider the case where the source-relay channel is weak ($\gamma<1$) and the relay-destination channel is weaker than the source-relay channel ($\beta<\gamma$).
The GDoF for the IRC with $\gamma=0.7$ and different relay-destination channel strength is illustrated in Fig. \ref{Fig:gseven}. 
As it is shown in this figure, the gain obtained by the relay is evident in the weak interference regime ($\alpha<1$). Interestingly, despite the weak ingoing and outgoing links of the relay, the relay can increase the GDoF. 
While in Fig. \ref{Fig:bonegseven}, we cannot benefit from relay for $\alpha\leq 1-\gamma$, in Fig. \ref{Fig:bfourgseven2} and \ref{Fig:bsevengseven}, relay can increase the GDoF even for very small values of $\alpha$. 
It is worth noting that in the case shown in Fig. \ref{Fig:bsevengseven}, the slop of GDoF is $-1$  with respect to $\alpha$ for $\alpha<\frac{2}{3}$, while in Fig. \ref{Fig:bonegseven} the slop changes from $-2$ to $0$ and then the GDoF of the IRC and that of the IC are parallel to each other. Now, consider the case shown in Fig. \ref{Fig:bfourgseven2} for values of $\alpha<\frac{2}{3}$. In this case, the slop of GDoF is $-1$ for $\alpha$'s in the beginning and the end of the interval $[0,\frac{2}{3}]$. However, in between, the GDoF behaviour is similar to that of shown in Fig. \ref{Fig:bonegseven}. This is due to the fact that in this case, we use a scheme which is a combination of the schemes used in other two extreme cases.
To understand the GDoF behaviour shown in Fig. \ref{Fig:gseven}, we explain briefly the main idea of the transmission scheme. 
To benefit from the relay, we need to use superposition block Markov encoding \cite{CoverElgamal}, \cite{Willems}. 
Using this encoding, some ``future" signals are provided to the relay. These future signals can be used at the relay in the next channel uses. Hence, the relay can operate partially as a cognitive relay. 
For the case in which $\gamma<1$, the future signal sent by Tx$i$ is received at relay and Rx$i$. 
However, by using backward decoding at the receiver side (details are provided in the following sections), Rx$i$ knows this future signal sent by Tx$i$ and thus, it removes the interference caused by that future signal. 
Doing this, some interference free dimension will be available at the Rx's. 
%Notice that in very weak interference regime, the size of the interference free dimension increases versus $\alpha$.
This interference free dimension can be either used by the undesired Tx to send some common signal as in Han-Kobayashi scheme \cite{HanKobayashi} or by the relay to provide some additional information to the receivers. 
%Regardless of how this interference free dimension is used, the additional signals provided to the receivers have to be received in this interference free dimension. 
While for the first choice (when common signals are provided), the interference channel has to be strong enough, for the second choice, the relay-destination link needs to be sufficiently strong. In the case shown in Fig. \ref{Fig:bonegseven}, the relay-destination channel is weak. Hence, relay cannot use this interference free dimension. 
However, for sufficiently strong interference, the interference free dimension will be accessible for common signaling. Thus, first for $1-\gamma<\alpha$, we can benefit from the interference free dimension. On the other hand, in Fig. \ref{Fig:bsevengseven}, the relay-destination channel is sufficiently strong. Hence, in this case, the interference free dimension will be used by the relay for providing some additional signal to the receivers. In this case, $\beta$ is so large that the relay can forward more information when the interference gets stronger as long as $\alpha<\frac{2}{3}$. Therefore, in whole regime $\alpha<\frac{2}{3}$, the GDoF performance of the IRC is higher than that of the IC. 
Now, consider the case shown in Fig. \ref{Fig:bfourgseven2}. 
In this case, the relay-destination channel is weaker than the previous case. Hence, the relay cannot benefit from increasing the interference channel in whole regime $\alpha<\frac{2}{3}$. Due to this, for very low values of $\alpha$ ($\alpha<2(\beta+\gamma-1)$), the GDoF behaviour is similar to the case shown in Fig. \ref{Fig:bsevengseven}. On the other hand, when interference channel is strong enough ($\beta<\alpha$), we benefit from common signaling.

Generally, for each setup, there is a regime where the GDoF optimal scheme is to send only private signals. 
As long as the setup operates in this regime, the interference has to be completely ignored and the stronger the interference, the worse is the GDoF performance. This regime has been characterized for the IC in \cite{EtkinTseWang} and it is given by $\alpha\leq\frac{1}{2}$. As it is shown in Fig. \ref{Fig:gseven}, relaying shrinks this regime. In the extreme case, when relay-destination channel is sufficiently strong, this regime is completely vanished (cf. Fig. \ref{Fig:bsevengseven}).

\begin{figure}
\centering
\begin{tikzpicture}[scale=1]
\bfiftheengseven
\end{tikzpicture}
\caption{The GDoF of the IRC (solid-line) for $\beta = 1.5$ and $\gamma= 0.7$. The dashed line represents the GDoF of the IC.}
\label{Fig:gsevenbfiftheen}
\end{figure}
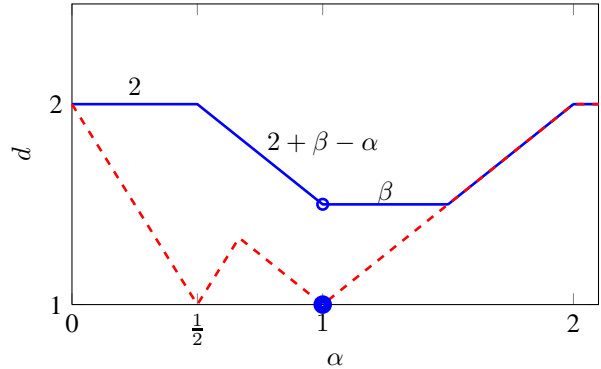
Now, we consider the case that the source-relay channel is weaker than the direct channel ($\gamma<1$) and the relay-destination channel is stronger than the direct channel ($1<\beta$). The GDoF result for this case is shown in Fig. \ref{Fig:gsevenbfiftheen}. In this case, the relay is so strong that it can forward some additional signals which are received at the receivers without any interference from other users as long as $\alpha<\beta$. This is shown in Fig. \ref{Fig:gsevenbfiftheen} where $\beta =1.5$.
\begin{figure}
\centering
\subfigure [] {
\begin{tikzpicture}[scale=1]
\btwogthirty
\end{tikzpicture}
\label{Fig:btwogthirty}
}
\subfigure [] {
\begin{tikzpicture}[scale=1]
\bfiftheengthirty
\end{tikzpicture}
\label{Fig:bfiftheengthirty}
}
\subfigure [] {
\begin{tikzpicture}[scale=1]
\btwentygthirty
\end{tikzpicture}
\label{Fig:btwentygthirty}
}
\subfigure [] {
\begin{tikzpicture}[scale=1]
\bsixtygthirty
\end{tikzpicture}
\label{Fig:bsixtygthirty}
}
\caption{The GDoF of the IRC (solid-line) for $\gamma=3$. The dashed line represents the GDoF of the IC. (a) $\beta = 0.2$, (b) $\beta = 1.5$, (c) $\beta = 2$, (d) $\beta=6$.}
\label{Fig:gthirty}
\end{figure}
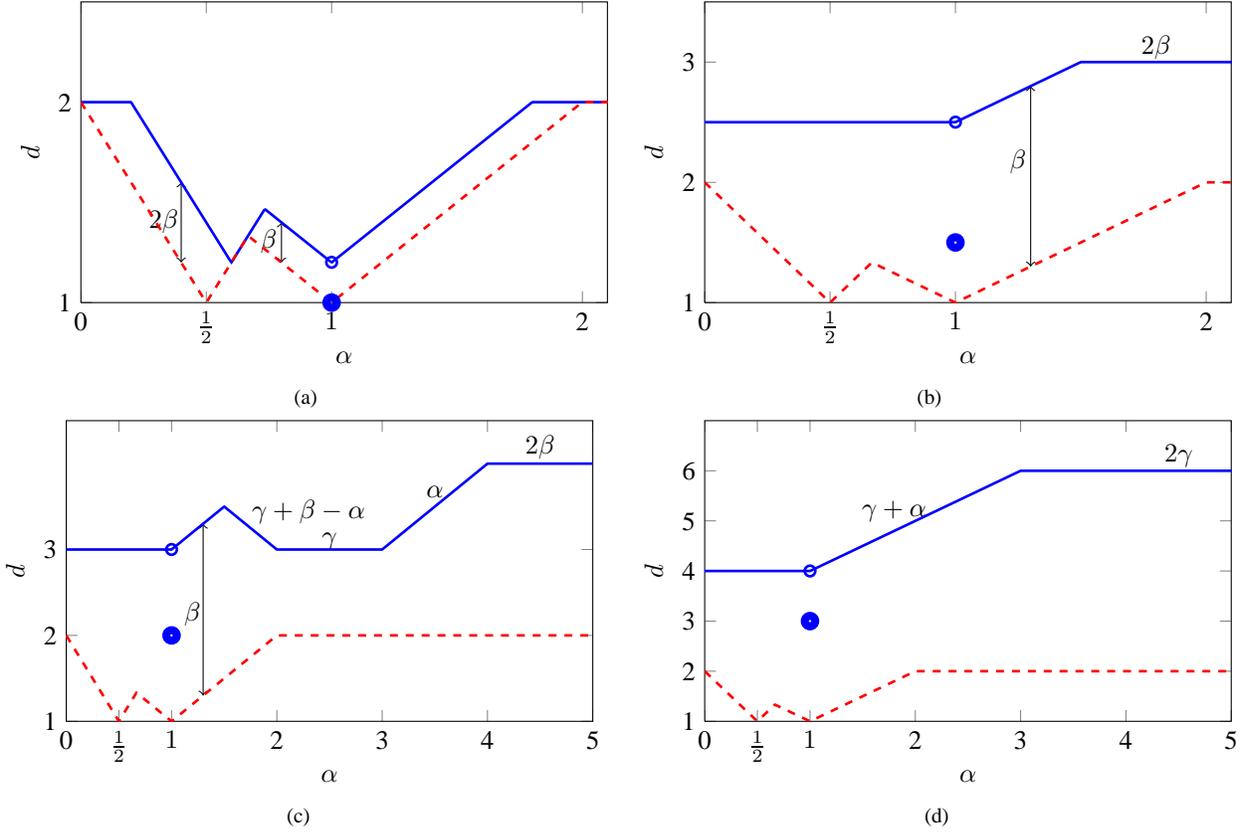

Finally, consider the case that the source-relay link $\gamma$ is strong. In Fig. \ref{Fig:gthirty}, the GDoF for the IRC with $\gamma=3$ is illustrated. In this case, the observation at the relay is so good that the relay performs as a cognitive relay. Therefore, the larger is the relay-destination channel, the more capable is the relay to neutralize the interference on air (see. Figures \ref{Fig:btwogthirty}, \ref{Fig:bfiftheengthirty}). 
In Fig. \ref{Fig:btwentygthirty}, the relay-destination link is large enough to neutralize the interference. However, at point $\alpha = \gamma/2$, the interference channel becomes so strong that the capacity of the source-relay channel is not sufficiently large for providing enough information to the relay. Hence, at this point the increase of the GDoF versus $\alpha$ stops.  
An interesting observation in Fig. \ref{Fig:btwentygthirty} is that the GDoF can decrease versus the interference channel strength in strong interference regime. This is in contrast to the IC \cite{EtkinTseWang} and X-channel \cite{HuangJafarCadambe} with strong interference where the GDoF is a nondecreasing function of $\alpha$.
The GDoF for the IRC for the case that the relay-destination channel is very strong is illustrated in Fig. \ref{Fig:bsixtygthirty}.
In this case, the bottleneck of the transmission will be the source-relay channel. This is shown in Fig. \ref{Fig:bsixtygthirty}, where $\beta=6$. 
In this case, the relay-destination channel is so strong that it is able to forward all its observation to the destinations without overlapping with other signals. In other words, we have a complete cooperation between relay and destinations since all received signal at the relay is available at the receivers. Due to this, the receiver is able to cancel the interference completely as long as the capacity of the source-relay channel is larger than the capacity of the interference channel ($\alpha <\gamma$).

{Next, we discuss the LD-IRC. We start by presenting the upper bounds which provide the converse of Theorem \ref{Theorem:capacity_LD_IRC}.}

\section{Upper Bounds for the LD-IRC}
\label{sec:UBLD-IRC}
{A standard bounding approach that can be applied for the IRC is the cut-set bound \cite{CoverThomas}. In addition to the cut-set bounds, further bounds that can be tighter than the cut-set bounds in some cases can be derived by using genie-aided methods. Those techniques are used in this paper for developing upper bounds that coincide with the statement of Theorem \ref{Theorem:capacity_LD_IRC}. These upper bounds are given in the following lemma.}

\begin{lemma}
\label{Lemma:new_bounds}
The sum-capacity of the LD-IRC is upper bounded by
\begin{align}
\label{eq:new_UB6_IRC_det}
C_{\mathrm{det},\Sigma} \leq & \max\{n_d,\min\{n_r,n_s\}\}  \quad \text{ for } n_d = n_c\\
\label{eq:new_UB4_IRC_det}
C_{\mathrm{det},\Sigma} \leq & \max\{n_d,n_c,n_r\} + \max\{n_d,n_c\}\\
\label{eq:new_UB1_IRC_det}
C_{\mathrm{det},\Sigma} \leq & n_r+2\max\{n_d,n_c\}-n_c\\
\label{eq:new_UB3_IRC_det}
C_{\mathrm{det},\Sigma} \leq & \max\{n_d,n_c\}+\max\{n_d,n_s\}\\
\label{eq:new_UB2_IRC_det}
C_{\mathrm{det},\Sigma} \leq & 2\max\{n_c,n_r,n_d-\max\{n_c,n_s\}\}+2(n_s-n_c)^+\\
\label{eq:new_UB5_IRC_det}
C_{\mathrm{det},\Sigma} \leq & 2\max\{n_c,n_r+n_d-n_c\}  \quad \text{ for } n_r\leq n_c\leq n_d.
\end{align}
\end{lemma}
\begin{proof}
Details of the proof of this lemma are given in Appendix \ref{app:UB_LD_IRC_proof}. Shortly, the first and second bounds are derived from the cut-set bounds. The remaining bounds are derived using genie-aided methods. The bounds \eqref{eq:new_UB3_IRC_det} and \eqref{eq:new_UB2_IRC_det} are tightened versions of the upper bounds given in~\cite[Theorems 3 and 4]{ChaabanSezgin_IT_IRC}, tightened for the case where $n_c<n_s$. Finally, the bound \eqref{eq:new_UB5_IRC_det} is inspired by a similar upper bound obtained for the IC\footnote{The bound given by $C_{\mathrm{det-IC},\Sigma}\leq 2\max\{n_d-n_c,n_c\}$ for the LD-IC.} \cite{EtkinTseWang, BreslerTse}.
\end{proof}

Note that similar to the X channel \cite{HuangJafarCadambe} and the $K$-user IC \cite{CadambeJafar_KUserIC}, the capacity of the LD-IRC has a discontinuity at $n_c=n_d$, i.e., if the cross channel is equal in strength to the direct channel. 

{In addition to these new bounds,} some upper bounds are borrowed from \cite{ChaabanSezgin_IT_IRC,ChaabanThesis}. These bounds are given in the following lemma.

\begin{lemma}(\cite{ChaabanSezgin_IT_IRC})
\label{Lemma:known_UB1_det}
The sum capacity of the LD-IRC is upper bounded as follows
\begin{align}
\label{eq:old_UB1_IRC_det}
C_{\mathrm{det},\Sigma} &\leq 2\max\{n_d,n_r\} \\
\label{eq:old_UB2_IRC_det}
C_{\mathrm{det},\Sigma} &\leq \max\{n_d,n_r,n_c\}+\max\{n_d,n_c\}-n_c+(n_s-\max\{n_d,n_c\})^+.
\end{align}
\end{lemma}
{The first of these bounds is in fact a cut-set bound, while the second is a genie-aided bound. The proof of these bounds can be found in \cite{ChaabanSezgin_IT_IRC}.}

Now, we need to show that the upper bounds \eqref{eq:new_UB4_IRC_det}-\eqref{eq:old_UB2_IRC_det} coincide with the capacity given in Theorem \ref{Theorem:capacity_LD_IRC}. First, it is clear that an upper bound is obtained by taking the minimum of all available upper bounds. At this point, in order to allow a direct comparison between the bounds, we need to get rid of the condition associated with the bound \eqref{eq:new_UB5_IRC_det}. First, we note that the first condition given by $n_r\leq n_c$ can be dropped since if this condition is not satisfied, i.e., if $n_r>n_c$ with $n_c\leq n_d$, then we have
\begin{align*}
2\max\{n_c,n_r+n_d-n_c\} &= 2(n_r+n_d-n_c) \\
&> n_r+2n_d-n_c\\
&= n_r+2\max\{n_d,n_c\}-n_c,
\end{align*}
%$$2\max\{n_c,n_r+n_d-n_c\}\geq n_r+2\max\{n_d,n_c\}-n_c,$$ 
which makes the bound \eqref{eq:new_UB5_IRC_det} redundant given \eqref{eq:new_UB1_IRC_det}. Thus, we can write  \eqref{eq:new_UB5_IRC_det} as
\begin{align}
C_{\mathrm{det},\Sigma} \leq & 2\max\{n_c,n_r+n_d-n_c\}  \quad \text{ if } n_c\leq n_d.
\end{align}
Furthermore, if $n_c>n_d$, then by replacing $n_r+n_d-n_c$ in this bound by $n_r+(n_d-n_c)^+$, we get the bound 
\begin{align}
\label{eq:new_UB5_IRC_det_2}
C_{\mathrm{det},\Sigma} \leq & 2\max\{n_c,n_r+(n_d-n_c)^+\}.
\end{align}
This bound holds since
\begin{align*}
2\max\{n_c,n_r+(n_d-n_c)^+\}  &= 2n_r  \\
&> n_r+n_c \\ 
&= n_r+ 2\max\{n_d,n_c\}-n_c,
\end{align*}
which shows that the bound \eqref{eq:new_UB5_IRC_det} is redundant given \eqref{eq:new_UB1_IRC_det} in this case. Thus, we can replace the bound \eqref{eq:new_UB5_IRC_det} by \eqref{eq:new_UB5_IRC_det_2}.

{Now, we can compare the upper bounds with the bounds in \eqref{eq:capacity_LD_IRC}. The first term in \eqref{eq:capacity_LD_IRC} is \eqref{eq:old_UB1_IRC_det}. The second term in \eqref{eq:capacity_LD_IRC} is the minimum between \eqref{eq:new_UB4_IRC_det} and \eqref{eq:old_UB2_IRC_det}. The third term is \eqref{eq:new_UB1_IRC_det}. The fourth term is \eqref{eq:new_UB2_IRC_det} evaluated for $n_c<n_s$. The fifth term is \eqref{eq:new_UB3_IRC_det}, and the last term is \eqref{eq:new_UB5_IRC_det_2}. Finally, the upper bound for the special case $n_c=n_d$ \eqref{eq:capacity_LD_IRC_nceqnd} is given by \eqref{eq:new_UB6_IRC_det}.}

This completes the proof of the converse of Theorem \ref{Theorem:capacity_LD_IRC}. Next, we propose a transmission scheme that achieves the sum-capacity of the LD-IRC.

\section{Sum-capacity Achieving Schemes}
\label{Sec:AchievabilityLDIRC}
{The goal of this section is to introduce transmission schemes which achieve the sum-capacity in Theorem \ref{Theorem:capacity_LD_IRC}.} To this end, different schemes are proposed to cover different operating regimes of the IRC. While only one scheme is enough to achieve the sum-capacity in the strong interference (SI) regime ($n_c>n_d$), three schemes are required to complete the characterization of the capacity in the weak interference (WI) regime ($n_c<n_d$). In addition, one simple scheme is required for the special case where the cross channels and the desired ones are equally strong ($n_d=n_c$). {These schemes are described in the following paragraphs. But before we describe the schemes in detail, we describe the building blocks of these schemes.}

\subsection{Building blocks}
\label{Sec:BuildingBlocks}
The transmission schemes we propose are based on the 
\begin{itemize}
\item private and common signaling approach 
of the Han-Kobayashi scheme \cite{HanKobayashi}, \cite{EtkinTseWang}, \cite{BreslerTse} 
\end{itemize}
in addition to three relaying strategies 
\begin{itemize}
\item compute-forward~(CF)~\cite{NazerGastpar}
\item decode-forward~(DF)~\cite{Cover}
\item cooperative interference neutralization~(CN)~\cite{ChaabanSezginTuninetti}.
\end{itemize}
%A simple description of common and private signaling can be found in \cite{EtkinTseWang, BreslerTse}.
Next, we introduce the three relaying schemes. 

\subsubsection{Compute-Forward (CF)}
In CF, Tx$i$ sends $u_{i,cf}[k]$ in the $k$th channel use ($k=1,\ldots,n-1$). At the end of the $k$th channel use, the relay decodes $u_{r,cf}[k+1] = u_{1,cf}[k]\oplus u_{2,cf}[k]$, and sends it in channel use $k+1$. Now, consider the decoding at the receiver side. We explain the decoding only for Rx$1$, since Rx$2$ does it similarly. Rx$1$ waits until the end of transmission block $n$. Then it performs backward decoding starting with the $n$th channel use, where only the relay is active. Rx$1$ decodes $u_{r,cf}[n]=u_{1,cf}[n-1]\oplus u_{2,cf}[n-1]$ in the $n$th channel use. Then, in channel use $(n-1)$, depending on the channel parameters, Rx$1$ has two possibilities:
\begin{itemize}
\item If interference is weak, i.e., $n_c<n_d$, Rx$1$ decodes $u_{1,cf}[n-1]$, {then it adds it to $u_{r,cf}[n]$ to extract $u_{2,cf}[n-1]$. Next, it subtracts the contribution of $u_{2,cf}[n-1]$ from the received signal, and decodes $u_{r,cf}[n-1]$.}
\item If interference is strong, i.e., $n_d<n_c$, Rx$1$ decodes $u_{2,cf}[n-1]$ first, {then it adds it to $u_{r,cf}[n]$ to extract $u_{1,cf}[n-1]$. Next, it subtracts the contribution of $u_{1,cf}[n-1]$ from the received signal, and decodes $u_{r,cf}[n-1]$.}
\end{itemize}
Next, Rx$1$ proceeds backwards until reaching the first channel use. Since the relay is silent in the first channel use, Rx$1$ decodes
\begin{itemize}
\item $u_{1,cf}[1]$ if $n_c<n_d$,
\item $u_{2,cf}[1]$ if $n_d<n_c$.
\end{itemize}
Note that in both cases, each receiver obtains the CF signals sent by both Tx's, and thus, the CF signals can be interpreted as common relayed signals.

\subsubsection{Decode-Forward (DF)}
In DF, Tx$i$ sends $u_{i,df}[k]$ in the $k$th channel use ($k=1,\ldots,n-1$).
The relay decodes both $u_{1,df}[k]$ and $u_{2,df}[k]$ as in a multiple access channel (MAC) in the $k$th channel use. Then, the relay forwards these two signals in channel use $k+1$. The sent signal by the relay is $u_{r,df}[k+1]$.
%Let $u_{r,df}[k+1]$ denote the sum $u_{1,df}[k]\oplus u_{2,df}[k]$. 
Similar to CF, the Rx's wait until the end of the $n$th channel use and then they start with backward decoding. First, Rx$i$ processes the received signal in the $n$th channel use. As the Tx's are silent in the $n$th channel use, Rx$i$ decodes only the relay signal, i.e. $u_{r,df}[n]$. Thus, each Rx decodes $u_{1,df}[n-1]$ and $u_{2,df}[n-1]$ as in the MAC channel. Then, Rx$i$ starts processing the received signal in channel use $(n-1)$. It removes the interference caused by $u_{1,df}[n-1]$ and $u_{2,df}[n-1]$, since they are both already known at Rx$i$ (decoded in the $n$th channel use). Then, it decodes the relay signal $u_{r,df}[n-1]$ and obtains $u_{1,df}[n-2]$ and $u_{2,df}[n-2]$. Rx$i$ proceeds backwards until the second channel use, where $u_{r,df}[2]$ is decoded, and $u_{1,df}[1]$ and $u_{2,df}[1]$ are extracted. As a result, each receiver has $u_{i,df}[k]$ for $i=1,2$ and $k=1,\ldots,n-1$.

\subsubsection{Cooperative interference neutralization (CN)}
In this scheme, the relay transmits its signal in such a way that the interference is completely neutralized at the Rx's. To do this, Tx$i$ uses block-Markov encoding \cite{Willems} by sending both $u_{i,cn}[k]$ and $u_{i,cn}[k-1]$ in the $k$th channel use ($k=2,\ldots,n-1$). In the first and the $n$th channel use, Tx$i$ sends $u_{i,cn}[1]$ and $u_{i,cn}[n-1]$, respectively. Note that while in CF and DF, the Tx's are silent in $n$th channel use, in CN, the Tx's are active in $n$th channel use. 
Similar to CF, the relay decodes the sum of the signals as follows. At the end of the first channel use, the relay decodes $u_{1,cn}[1]\oplus u_{2,cn}[1]$. Note that in the second channel use, this sum is received again at the relay due to the block-Markov encoding. Therefore, the relay removes the interference caused by this sum, and then decodes $u_{1,cn}[2]\oplus u_{2,cn}[2]$. Proceeding this way, the relay knows $u_{1,cn}[k]\oplus u_{2,cn}[k]$ at the end of the $k$th channel use, $k=1,\ldots,n-1$. This known sum can be used in the next channel use for interference neutralization as follows. The relay sends $u_{r,cn}[k] = u_{1,cn}[k-1]\oplus u_{2,cn}[k-1]$ in the $k$th channel use such that it overlaps with the transmitters' signal $u_{1,cn}[k-1]$ and $u_{2,cn}[k-1]$ at Rx2 and Rx1, respectively. Similar to CF and DF, the receivers use backward decoding. Rx$1$ first processes the received signal in channel use $n$, where it receives $u_{r,cn}[n]\oplus u_{2,cn}[n-1]=u_{1,cn}[n-1]$ as a superposition of the signals from Tx$2$ and the relay. This superposition of the relay signal and undesired signal neutralizes the interference and provides the receiver the desired signal as the aggregate.
In addition to this, Rx$1$ receives another copy of $u_{1,cn}[n-1]$ which is sent by Tx$1$. Depending on whether the IRC operates in the weak interference regime or the strong interference regime, Rx$1$ decodes $u_{1,cn}[n-1]$ either from Tx$1$ or from the superposition of the signals from relay and Tx$2$. Thus, Rx$1$ gets $u_{1,cn}[n-1]$. The contribution of the other received instance of $u_{1,cn}[n-1]$ can be removed from the received signal afterwards if needed. Then Rx$1$ proceeds to channel use $n-1$ where it removes the contribution of $u_{1,cn}[n-1]$, and then decodes $u_{1,cn}[n-2]$ similar to channel use $n$. Rx$1$ proceeds similarly until the first channel use, and hence gets $u_{1,cn}[k]$ for $k=1,\ldots,n-1$.

Next, we introduce the achievable schemes for the LD-IRC which are combinations of private and common signaling with the three relaying strategies presented above.

\subsection{Scheme \WIone{}} 
\label{Sec:WI1}
The first scheme is developed for the WI regime {$n_c<n_d$}, and it performs optimally in several cases in this regime.
{We summarize the performance of the scheme in terms of achievalbe sum-rate in the following proposition.
\begin{proposition}\label{prop:RateWI1}
The achievable sum-rate with the scheme \WIone{} for the IRC is given by
\begin{align}\label{eq:R_sum_UB_A1}
R_{\Sigma,\text{\WIone{}}} = 
\begin{cases}
 \min\{n_s+n_d,n_r+n_s-n_c\}  & n_c< n_d\leq n_s\leq n_r \\ % \label{eq:R_sum_UB_A1} \\
\min\{2n_d,n_r+n_d-n_c\}  & n_c\leq n_s\leq n_d\leq n_r \\ %\label{eq:R_sum_UB_A2} \\
\min\{n_r+n_d,n_r+n_s-n_c\}  & n_c< n_d\leq n_r\leq n_s \\ % \label{eq:R_sum_UB_B1} \\
\min\{2n_d,n_s+n_d-n_c\}  & n_c\leq n_r\leq n_d\leq n_s  %\label{eq:R_sum_UB_B2} 
\end{cases}.
\end{align}
{This proves the achievability of Theorem \ref{Theorem:capacity_LD_IRC} within the four regimes.}
Note that these sum-rate expressions coincide with the upper bounds given in Lemmas \ref{Lemma:new_bounds} and \ref{Lemma:known_UB1_det}. 
\end{proposition}
}
The rest of the subsection is devoted to the proof of the proposition~\ref{prop:RateWI1}.

 In this transmission scheme, in addition to private signaling, CN, DF, and CF relaying strategies are used.
\subsubsection*{Encoding at transmitters} In the $k$th channel use, Tx$1$ transmits $\boldsymbol{x}_1[k]$ given by
\begin{align}
\boldsymbol{x}_1[k] = 
\begin{bmatrix}
\boldsymbol{0}_{\ell_1} \\ \boldsymbol{u}_{1,cn1}[k-1] \oplus \boldsymbol{u}_{1,df1}[k] \\ \boldsymbol{u}_{1,cn2}[k-1] \\ \boldsymbol{u}_{1,cf}[k] \\ \boldsymbol{u}_{1,p}[k]\\
\boldsymbol{u}_{1,cn}[k] \\\boldsymbol{u}_{1,df2}[k]\\ \boldsymbol{0}_{s}
\end{bmatrix}, k=1,\ldots,n,
\end{align}
where the subscript $p$ refers to the private signal vector, {and $\boldsymbol{0}_{s}$ is used to complete the length of $\X_1[k]$ to $q$ which is equal to $\max\{n_r,n_s\}$ in these cases {\eqref{eq:R_sum_UB_A1}}.} The vectors $\boldsymbol{u}_{i,cn1}[0]$, $\boldsymbol{u}_{i,cn2}[0]$, $\boldsymbol{u}_{i,df1}[n]$, $\boldsymbol{u}_{i,cf}[n]$, $\boldsymbol{u}_{i,p}[n]$, $\boldsymbol{u}_{i,cn1}[n]$, $\boldsymbol{u}_{i,cn2}[n]$, and $\boldsymbol{u}_{i,df2}[n]$ are zero vectors for $i\in\{1,2\}$.
%where $n$ represents the number of transmission slots and the subscripts $p$, $cn$, $df$, and $cf$ refer to private, CN, DF, and CF signal vectors, respectively.
Moreover, the length of the vectors $\boldsymbol{u}_{i,df1}$, $\boldsymbol{u}_{i,df2}$, $\boldsymbol{u}_{i,cf}$, and $\boldsymbol{u}_{i,p}$ are $2R_{df1}$, $2R_{df2}$, $R_{cf}$, and $R_{p}$, respectively.

Note that the length of DF signal vector is twice as many information bits it contains. In particular, if $\tilde \U_{i,df1}$ denotes the information bits of $df1$ from Tx$i$ with length $R_{df1}$, then $ \U_{1,df1}$ is constructed as follows $\U_{1,df1}=\begin{bmatrix}
\tilde \U_{1,df1}^T & \boldsymbol{0}_{R_{df1}}^T \end{bmatrix}^T$. Similarly, $\U_{2,df1}=\begin{bmatrix}
\boldsymbol{0}_{R_{df1}}^T & \tilde \U_{2,df1}^T \end{bmatrix}^T$. Therefore, in $\U_{1,df1} \oplus \U_{2,df1}$, the information bits of the DF signal vectors from both Tx's do not overlap.

It is worth mentioning that in $\X_1[k]$, the CN signal vector with time index $[k]$ is desired at the relay and the CN signal vector with time index $[k-1]$ is neutralized at undesired Rx. 
Moreover, we define $\boldsymbol{u}_{1,cn}$ as follows $\boldsymbol{u}_{1,cn} = [\boldsymbol{u}_{1,cn1}^T,\boldsymbol{u}_{1,cn2}^T]^T$, where $\boldsymbol{u}_{1,cn1}$ is added to $\boldsymbol{u}_{1,df1}[k]$ before transmission. 
The length of the vectors $\boldsymbol{u}_{1,cn1}$ and $\boldsymbol{u}_{1,cn2}$ are $R_{cn1} = 2R_{df1}$ and $R_{cn2}$, respectively. Therefore, the length of the vector $\boldsymbol{u}_{1,cn}$ which is $R_{cn}$ is larger than the length of $\boldsymbol{u}_{1,df1}[k]$. Therefore, we need to keep in mind that 
\begin{align}
2R_{df1} \leq R_{cn}. \label{eq:rate_cons_Tx_A1_A2}
\end{align}
Obviously, $\boldsymbol{x}_1[k]$ can be generated as long as
\begin{align}
2R_{cn} + R_{cf} + R_{p} + 2R_{df2} + \ell_1 \leq q
\end{align}
\iffalse
It is worth to mention that the time slot $k=1$ is used as the initialization phase, in which Tx$1$ sends
\begin{align}
\boldsymbol{x}_1[1] = 
\begin{bmatrix}
\boldsymbol{0}_{\ell_1} \\ \boldsymbol{u}_{1,df1}[1] \\ \boldsymbol{0}_{R_{cn2}} \\ \boldsymbol{u}_{1,cf}[1] \\ \boldsymbol{u}_{1,p}[1]\\
\boldsymbol{u}_{1,cn}[1] \\\boldsymbol{u}_{1,df2}[1]
\end{bmatrix}.
\end{align}
Moreover, in the time slot $n$, Tx$1$ sends
\begin{align}
\boldsymbol{x}_1[n] = 
\begin{bmatrix}
\boldsymbol{0}_{\ell_1} \\ \boldsymbol{u}_{1,cn1}[n-1]  \\ \boldsymbol{u}_{1,cn2}[n-1] \\ \boldsymbol{0}_{R_{cf}} \\ \boldsymbol{0}_{R_p}\\
\boldsymbol{0}_{R_{cn}} \\\boldsymbol{0}_{R_{df2}}
\end{bmatrix}.
\end{align}
\fi
Similarly, Tx'2 transmits $\X_2[k]$.

\subsubsection*{Decoding at the relay} 
The relay receives in $k$th channel use the superposition of the signal vectors transmitted from both Tx's as shown in Fig. \ref{Fig:RE:A_1_2_Rx1_a}.
Notice that in this figure the time index of all signals is $[k]$ except $\boldsymbol{u}_{i,cn}$ which has a time index of $[k-1]$.
Supposing that the decoding in time slot $k-1$ is done successfully at the relay, the relay knows $\U_{1,cn}[k-1]\oplus \U_{2,cn}[k-1]$ in time slot $k$. Hence, it removes the interference caused by $\U_{1,cn}[k-1]\oplus \U_{2,cn}[k-1]$ before decoding process in time slot $k$.
In time slot $k$, the relay wants to decode the sum of the CN signal vectors $\U_{1,cn}[k] \oplus \U_{2,cn}[k]$, the sum of the CF signal vectors $\U_{1,cf}[k] \oplus \U_{2,cf}[k]$, and the DF signal vectors $\U_{1,df1}[k]$, $\U_{1,df2}[k]$, $\U_{2,df1}[k]$, and $\U_{2,df2}[k]$. For successful decoding at the relay, the following constraints must be satisfied
\begin{align}
\ell_1 + R_{cf} + R_p + 2(R_{cn} + R_{df2}) \leq n_s &\quad \text{ if } \max\{R_{cn},R_{df2}\}>0 
\label{eq:rate_cons_relay_A1_A2_1}\\
\ell_1 + R_{cf}  \leq n_s &\quad \text{ if } \max\{R_{cn},R_{df2}\}=0.
 \label{eq:rate_cons_relay_A1_A2}
\end{align}
Note that the case distinction in \eqref{eq:rate_cons_relay_A1_A2_1} and \eqref{eq:rate_cons_relay_A1_A2} is due to the fact that the relay does not need to decode the sum of private signal vectors if no CN and DF signal vectors are transmitted.

\subsubsection*{Encoding at the relay}
As the previous (with time index $k-1$) CN, CF, and DF signal vectors are available at the relay in the $k$th channel use, the relay constructs the following signal vector at the time slot $k=2,\ldots,n$
\begin{align*}
\boldsymbol{x}_r[k] = 
\begin{bmatrix}
\boldsymbol u_{r,cf}[k]\\
\boldsymbol u_{r,df1}[k]\\
\boldsymbol u_{r,df2}[k]\\
\boldsymbol{0}_{\ell_2}\\
\boldsymbol u_{r,cn}[k]\\
\boldsymbol{0}_{\ell_3}\\
\boldsymbol{0}_{r}
\end{bmatrix},
\end{align*}
where we define $\U_{r,R}[k] = \U_{1,R}[k-1] \oplus \U_{2,R}[k-1]$, with $R\in\{cf,df1,df2,cn\}$ for readability, and where $\ell_2$ {is chosen so that $\U_{r,df2}$ does not overlap with $\U_{1,cn}$, and} $\ell_3$ is chosen so that at the receiver side, the CN signal vector from the undesired transmitter is aligned with the CN signal vector sent by the relay. {Moreover, $\boldsymbol{0}_{r}$ is used to complete the length of $\X_r[k]$ to $q$.}

\subsubsection*{Decoding at the receiver side} 
We explain the decoding at Rx$1$ since the decoding at Rx$2$ is similar. Rx$1$ waits until end of the $n$th channel use. Then it starts with the backward decoding. Supposing that Rx$1$ decodes $\Y_1[n]$ successfully, it knows
\begin{itemize}
\item $\U_{r,cf}[n]=\U_{1,cf}[n-1]\oplus \U_{2,cf}[n-1]$,
\item $\U_{r,df1}[n]\rightarrow \tilde \U_{1,df1}[n-1], \tilde \U_{2,df1}[n-1]$,
\item $\U_{r,df2}[n]\rightarrow \tilde \U_{1,df2}[n-1], \tilde \U_{2,df2}[n-1]$, 
\item $\U_{1,cn}[n-1]$.
\end{itemize}
Next, it starts decoding $\Y_1[n-1]$. The received signal vector at Rx$1$ in time slot $2\leq k\leq n-1$ is shown in Fig.~\ref{Fig:RE:A_1_2_Rx1_b}. 
\begin{figure}
%\vspace{-2.5%cm}
\centering
\subfigure [] {
%\resizebox{0.4\textwidth}{!}{%
\begin{tikzpicture}[->,>=stealth',shorten >=1pt,auto,node distance=3cm,thick,scale=1]
 \SCHEMEWIoneRelay
\end{tikzpicture}                
\label{Fig:RE:A_1_2_Rx1_a}
%}
}
\subfigure [] {
\begin{tikzpicture}[->,>=stealth',shorten >=1pt,auto,node distance=3cm,thick,scale=1]
\SCHEMEWIone 
\end{tikzpicture} 
\label{Fig:RE:A_1_2_Rx1_b}               
}
\caption{(a) The received signal vector at the relay. (b) The received signal at Rx$1$. Both plots are based on the proposed transmission scheme \WIone{} in time slot $k$, where $k=2,\ldots,n-1$. Here, $\U_{i,cnF}$ denotes $\U_{i,cn}[k]$, while $\U_{i,cn}$ represents $\U_{i,cn}[k-1]$. The time index of all remaining signals is $[k]$.}
\label{Fig:RE:A_1_2_Rx1}
\end{figure}

Since Rx$1$ knows $\tilde\U_{1,df1}[n-1]$ and $\tilde\U_{2,df1}[n-1]$, it can remove the interference caused by $\U_{1,df1}[n-1]$, $\U_{2,df1}[n-1]$ completely. Next, it decodes the interference free received bits from the relay. In order to avoid an overlap between the CF and DF signals from relay with the signal transmitted by the desired Tx, the top-most $\ell_1$ bits of the signal vectors from Tx's are set to zero. 
Rx$1$ can decode $\U_{r,cf}[n-1]$, $\U_{r,df1}[n-1]$, and $\U_{r,df2}[n-1]$ as long as
\begin{align}
R_{cf} + 2(R_{df1} + R_{df2}) \leq (n_r - n_d + \ell_1)^+. \label{eq:rate_const_Rx_Re_A1_A2}
\end{align} 
Since this scheme is proposed for the WI regime ($n_c<n_d$), Rx1 decodes next the top-most $n_d-n_c$ desired CN bits (i.e., $\U_{1,cn}[n-2]$) interference free. Moreover, the relay signal vector neutralizes the undesired CN signal vector, and replaces it by $ \U_{1,cn}[n-2]$, i.e., $\U_{r,cn}[n-1] \oplus \U_{2,cn}[n-2] = \U_{1,cn}[n-2]$. Notice that this neutralization is possible if
\begin{align}
n_c-\ell_1 \leq n_r \quad \text{if } 0<R_{cn}. \label{eq:rate_const4_Rx_Re_A1_A2}
\end{align}
Note that the overlap between the two received versions of $\U_{1,cn}[n-2]$ (from TX$1$ and the relay) can be removed by decoding the bits successively as long as $n_c\neq n_d$, {which is satisfied in the WI regime where $n_c<n_d$)}. After decoding the CN signal vector, Rx$1$ observes the top most $n_d-n_c$ bits of the desired CF signal vector (i.e., $\U_{1,cf}[n-1]$) interference free. Moreover, due to the backward decoding, the sum of the CF signal vectors (i.e., $\U_{r,cf}[n] = \U_{1,cf}[n-1]\oplus \U_{2,cf}[n-1]$) is known at Rx1. Therefore, Rx1 can reconstruct the top-most $n_d-n_c$ bits of the undesired CF bits and remove their interference. Similar to decoding the CN signal vector, decoding of the CF signal vector is also accomplished in a successive manner, which is possible as long as $n_c\neq n_d$. After decoding the CF signal vector, Rx1 decodes the private signal vector. It can be done reliably as long as 
\begin{align}
R_{p} \leq n_d -n_c.  \label{eq:rate_const2_Rx_Re_A1_A2}
\end{align}
The other required rate constraint is that the desired CN, CF, and private signal vectors ($\U_{1,cn}[n-2]$, $\U_{1,cf}[n-1]$, $\U_{1,p}[n-1]$) need to be observed at Rx$1$ and without any overlap with the CN signal vector corresponding to the next time slot (denoted by subscript $cnF$ in Fig. \ref{Fig:RE:A_1_2_Rx1}). Therefore, we write
\begin{align}
\ell_1 + R_{cn}+R_{cf} +R_p + (R_{cn} + 2R_{df2}-(n_s-n_d)^+)^+ \leq n_d. \label{eq:rate_const3_Rx_Re_A1_A2} 
\end{align}
%{(Isn't $\ell_1 + R_{cn}+R_{cf} +R_p\leq n_d$ enough?)}.
Note that the expression $(R_{cn} + 2R_{df2}-(n_s-n_d)^+)^+$ is the length of the signal vector which is received at Rx's but it is not decoded at the receiver side (e.g., part of $\U_{1,cnF}$ in Fig.~\ref{Fig:RE:A_1_2_Rx1_b} which is received at Rx$1$.).

Using this scheme, we achieve the sum-rate
\begin{align}
nR_{\Sigma,\text{\WIone{}}} = 2(n-1)(R_{cn}+R_{df1}+R_{df2}+R_{cf}+R_p) \label{eq:R_Sigma_Scheme1}
\end{align}
Note that the term $(n-1)$ is due to the fact that the last channel use is used for neutralization. In other words, no additional information bits are transmitted in this channel use.
Now, by dividing the expression \eqref{eq:R_Sigma_Scheme1} by $n$ and letting $n\to\infty$, we obtain the following sum-rate
\begin{align}
R_{\Sigma,\text{\WIone{}}} = 2(R_{cn}+R_{df1}+R_{df2}+R_{cf}+R_p).
\end{align}
The next goal is to maximize the sum-rate under the conditions in \eqref{eq:rate_cons_Tx_A1_A2}-\eqref{eq:rate_const3_Rx_Re_A1_A2}. Hence, we obtain the following optimization problem.
\begin{align}
\max \quad & R_{\Sigma,\text{\WIone{}}}\notag \\
\text{s.t.} \quad & \text{\eqref{eq:rate_cons_Tx_A1_A2}-\eqref{eq:rate_const3_Rx_Re_A1_A2} are satisfied}\\ \notag 
& R_{cn},R_{df1},R_{df2},R_{cf},R_p,\ell_1\geq 0\notag
\end{align}
%The solution of this optimization problem is presented in Appendix \ref{app:Rate_alloc_WI_S1}.
{The solution of this optimization problem is given in Table \ref{Tab:WI1}, with an achievable sum-rate as given in {\eqref{eq:R_sum_UB_A1}}. This proves the achievability of Theorem \ref{Theorem:capacity_LD_IRC} for the four regimes corresponding to this scheme.}

\begin{table*}
\centering
\begin{tabular}{|c||c|c|c|c|}
\hline
\footnotesize Regime & \footnotesize$n_c< n_d\leq n_s\leq n_r$ & \footnotesize$n_c\leq n_s \leq n_d \leq n_r$ & \footnotesize$n_c< n_d \leq n_r \leq n_s$ & \footnotesize$n_c \leq n_r \leq n_d \leq n_s$\\\hline
$\ell_1$ & 0 & $[R_{cf} - (n_r-n_d)]^+$& 0 & 0\\\hline
$R_{cn}$ & $\frac{n_s-R_{cf}- R_p-2R_{df2}}{2}$ & 0 & \scriptsize $\frac{\min\{2n_c,\max\{n_s+n_c-n_r,2n_s-2n_d\}\}}{2}$ & $\frac{\min\{2n_c,n_s+n_c-n_d\}}{2}$ \\\hline
$R_{df1}$ & $\frac{\min\{R_{cn},n_r-n_d-R_{cf}-2R_{df2}\}}{2}$ & 0 & $\frac{(n_r-n_d -|n_s-n_d-n_c|)^+}{2}$ & 0 \\\hline
$R_{df2}$ & $\frac{[n_s-n_d-n_c]^+}{2}$ & 0 & $\frac{\min\{(n_s-n_d-n_c)^+,n_r-n_d\}}{2}$ & 0 \\\hline
$R_{cf}$ & \scriptsize $[(n_c+n_d-n_s)^+-2(n_d+n_c-n_r)^+]^+$ & $\frac{\min\{2n_c,n_r-n_d+n_c\}}{2}$& \footnotesize $\min\{(n_d+n_c-n_s)^+,n_r-n_d\}$ & 0\\\hline
$R_p$ & $n_d-n_c$ & $n_d-n_c$& $n_d-n_c$ & $n_d-n_c$\\\hline
%$R_\Sigma$ & \footnotesize $\min\{n_s+n_d,n_r+n_s-n_c\}$ & \footnotesize $\min\{2n_d,n_r+n_d-n_c\}$&  \footnotesize $\min\{n_r+n_d,n_r+n_s-n_c\}$ & \footnotesize $\min\{2n_d,n_s+n_d-n_c\}$\\\hline
\end{tabular}
\caption{Rate allocation parameters for the scheme \WIone{}.}
\label{Tab:WI1}
\end{table*}

\subsection{Scheme \WItwo{}} 
\label{sec:SchemeWI2}
Scheme \WIone{}  above does not achieve the sum-capacity of the LD-IRC for the whole WI regime. In what follows, we introduce another scheme for the WI regime which achieves the sum-capacity of the LD-IRC in parts of the WI regime that are not characterized by scheme \WIone{}. This scheme is called scheme \WItwo{}. The performance of this scheme is summarized in the following proposition.
\begin{proposition}\label{prop:RateWI2}
The scheme \WItwo{} achieves the following sum-rate for the IRC
\begin{align}
R_{\Sigma,\text{\WItwo{}}} &= \min\begin{Bmatrix}
2n_d-n_c, 2\max\{n_r+n_s-n_c,n_d-n_c\}
\end{Bmatrix},  
\label{eq:R_sum_LB_A3}
\end{align}
in two regimes, the first of which is described by
\begin{align}
n_c\leq n_s\leq n_r \leq n_d, \label{eq:regimeWI2_opt1}
\end{align}
and the second by
\begin{align}
n_c\leq n_r\leq n_s \leq n_d-\frac{n_c}{2}. \label{eq:regimeWI2_opt2}
\end{align}
This proves the achievability of Theorem \ref{Theorem:capacity_LD_IRC} within these two regimes.
Notice that this sum-rate expression coincides with the upper bounds given in Lemmas \ref{Lemma:new_bounds} and \ref{Lemma:known_UB1_det}. 
\end{proposition}
In what follows the proof of this proposition is presented.
Notice that in both regimes where \WItwo{} is optimal (given in \eqref{eq:regimeWI2_opt1} and \eqref{eq:regimeWI2_opt2}), we have $q=n_d$. While both the \WIone{} and \WItwo{} schemes use CN and CF and private signaling, the difference between the two schemes is that \WIone{} uses DF while \WItwo{} uses common signaling instead. In what follows, we present this scheme in detail. 

\subsubsection*{Encoding at transmitters}
Tx$1$ constructs its transmit signal as follows
\begin{align}
\boldsymbol{x}_1[k] = 
\begin{bmatrix}
\boldsymbol{u}_{1,cm}[k] \\
\boldsymbol{u}_{1,cn}[k-1] \\
\boldsymbol{u}_{1,cf}[k] \\
\boldsymbol{0}_{\ell_1} \\
\boldsymbol{u}_{1,p1}[k] \\
\boldsymbol{u}_{1,cn}[k] \\
\boldsymbol{u}_{1,p2}[k]
\end{bmatrix},   k=1,\ldots,n,
\end{align}
where $\boldsymbol{u}_{1,cm}$ represents the common signal vector with length $2R_{cm}$, and where $\ell_1$ will be specified later. Here, the private signals $\boldsymbol{u}_{1,p1}$, $\boldsymbol{u}_{1,p2}$, the CN signal $\boldsymbol{u}_{1,cn}$, and the CF signal $\boldsymbol{u}_{1,cf}$ have lengths $R_{p1}$, $R_{p2}$, $R_{cn}$, and $R_{cf}$, respectively. As described in Section \ref{Sec:BuildingBlocks}, the signals $\boldsymbol{u}_{1,cn}[0]$, $\boldsymbol{u}_{1,cm}[n]$, $\boldsymbol{u}_{1,cf}[n]$, $\boldsymbol{u}_{1,p1}[n]$, $\boldsymbol{u}_{1,cn}[n]$, and $\boldsymbol{u}_{1,p2}[n]$ are zero vectors. Tx$2$ generates $\X_2[k]$ similarly. Clearly, this works only if
\begin{align}
\label{Scheme2Cond1}
2R_{cm}+2R_{cn}+R_{cf}+\ell_1+R_{p1}+R_{p2}\leq q=n_d.
\end{align}

\subsubsection*{Decoding at the relay}
The relay receives the sum of the top-most $n_s$ bits transmitted by Tx's. In channel use $k$, the relay wants to decode the following sums of signals $\U_{1,cf}[k]\oplus \U_{2,cf}[k]$ and $\U_{1,cn}[k]\oplus \U_{2,cn}[k]$. This is possible if the relay observes the CF and CN signals, which requires following constraint
\begin{align}
\label{Scheme2Cond2}
2R_{cm} + 2R_{cn}+ R_{cf} + \ell_1+R_{p1} \leq n_s.
\end{align}
Therefore, at the end of the $k$th channel use, the relay knows $\U_{r,cf}[k+1]=\U_{1,cf}[k]\oplus \U_{2,cf}[k]$ and $\U_{r,cn}[k+1]=\U_{1,cn}[k]\oplus \U_{2,cn}[k]$.

\subsubsection*{Encoding at the relay} 
In channel use $k$, the relay sends the following signal vector
\begin{align*}
\boldsymbol{x}_r[k] = 
\begin{bmatrix}
\boldsymbol 0_{\ell_2} \\
\U_{r,cf1}[k]\\
\boldsymbol{0}_{\ell_3}\\
\U_{r,cf2}[k]\\ 
\boldsymbol 0_{r_1}
\end{bmatrix}\oplus
\begin{bmatrix}
\boldsymbol 0_{\ell_5} \\
\boldsymbol{u}_{r,cn1}[k]\\
\boldsymbol 0_{r_2}
\end{bmatrix},
\end{align*}
where $\boldsymbol 0_{r_1}$ and $\boldsymbol 0_{r_2}$ are zero vectors which insure that the length of $\X_r$ is $q$ (which equals $n_d$ in this regime), $\U_{r,cf1}[k]$ consists of the top-most $\ell_4$ bits ($\ell_4$ is to be determined) of $\boldsymbol{u}_{r,cf}[k]$ and $\U_{r,cf2}[k]$ consists of remaining bits, and $\U_{r,cn1}[k]$ consists of the top-most $\ell_6$ bits of $\boldsymbol{u}_{r,cn}[k]$, where $\ell_6$ is the number of bits of $\U_{2,cn}[k]$ that appear as interference at Rx$1$, i.e., $\ell_6=\min\{R_{cn},(n_c-2R_{cm})^+\}$. Thus, $\U_{r,cf1}[k]=(\boldsymbol{u}_{r,cf}[k])_{[1:\ell_4]}$, $\U_{r,cf2}[k]=(\boldsymbol{u}_{r,cf}[k])_{[\ell_4+1:R_{cf}]}$, and $\U_{r,cn1}[k]=(\boldsymbol{u}_{r,cn}[k])_{[1:\ell_6]}$. The signals above ($\U_{r,cf1}$, $\U_{r,cf2}$, and $\boldsymbol{u}_{r,cn1}$) will be received by receivers if 
\begin{align}
\label{Scheme2Cond3}
\ell_2+R_{cf}+\ell_3\leq n_r,\\
\label{Scheme2Cond4}
\ell_5+\ell_6\leq n_r.
\end{align}
The parameters $\ell_2$ to $\ell_6$ should be chosen in such a way that facilitates the decoding at the destinations, as we shall see next.

\subsubsection*{Decoding at the receiver side} 
Now, we describe decoding at Rx$1$. The decoding at Rx$2$ is done similarly. Similar to scheme \WIone{}, in this scheme, the destinations use backward decoding. Rx$1$ waits until the end of $n$th channel use. Assuming that decoding $\Y_{1}[n]$ is done successfully, Rx$1$ knows
\begin{itemize}
\item $\U_{1,cn}[n-1]$
\item $\U_{r,cf1}[n]=\U_{1,cf1}[n-1]\oplus \U_{2,cf1}[n-1]$
\item $\U_{r,cf2}[n]=\U_{1,cf2}[n-1]\oplus \U_{2,cf2}[n-1]$.
\end{itemize}
Next, Rx$1$ start with processing $\Y_{1}[n-1]$ which is shown in Fig. \ref{Fig:RE:A_3_Rx1}. First, Rx$1$ decodes the common signals as in an IC \cite{BreslerTse}. Thus, we treat the IRC at this stage as an IC, by treating the CN, CF, P1, and P2 signals as noise. Furthermore, to make sure that $\U_{r,cf1}[n-1]$ and $\U_{r,cf1}[n-2]$ do not interfere with the desired common signal, we require
\begin{align}
\label{Scheme2Cond5}
n_r-\ell_2 \leq n_d-2R_{cm} \quad \text{if } R_{cf}>0.
\end{align}
\begin{remark}
\label{remark:WI_2_1}
As we discussed earlier, the CN signal sent by the relay needs to be received aligned with the CN signal sent by Tx$2$ to facilitate interference neutralization at Rx$1$. Moreover, since in weak interference regime $n_c<n_d$, $\U_{2,cn}[n-1]$ is received on lower level than $\U_{1,cn}[n-1]$. Now, due to \eqref{Scheme2Cond1}, $\U_{1,cn}[n-1]$ cannot overlap the common signal $\U_{1,cm}[n-1]$ and hence, $\U_{2,cn}[n-1]$ and $\U_{r,cn1}[n-1]$ cannot do either. 
\end{remark}
Under this condition, we get an IC with $n_{d,IC}=2R_{cm}$ and $n_{c,IC}=(n_c-n_d+2R_{cm})^+$ which is an IC with weak interference. The decoding of both common signals at the receivers in this IC is done in a MAC fashion, achieving a rate of~\cite{BreslerTse} $$\min\left\{\frac{n_{d,IC}}{2},n_{c,IC}\right\}.$$
Therefore, under the aforementioned conditions \eqref{Scheme2Cond1},\eqref{Scheme2Cond5}, the common rate $\min\left\{R_{cm},(n_c-n_d+2R_{cm})^+\right\}$ is achieved. 

\begin{figure}
\centering
\begin{tikzpicture}[->,>=stealth',shorten >=1pt,auto,node distance=3cm,thick,scale=1]
%%%% horisonatal lines
\SCHEMEWItwoFigoneAC
\end{tikzpicture}
\caption{The received signal at Rx$1$ in the $k$th time slot ($2\leq k \leq n-1$) based on scheme \WItwo{}. The time index of common signals is $[k]$.}
\label{Fig:RE:A_3_Rx1}
\end{figure}

After removing the common signal vectors from the received signal, Rx$1$ observes a superposition of $\boldsymbol{x}_1^\prime$, $\boldsymbol{x}_2^\prime$, and $\boldsymbol{x}_r^\prime$ shown in Fig. \ref{Fig:RE:A_3_Rx1_after_remov_Comm}, where we define 
\begin{align*}
%\label{eq:channel_after_com_remov_1}
n_d^\prime &= n_d-2R_{cm} \\ 
%\label{eq:channel_after_com_remov_2} 
n_c^\prime &= n_c-2R_{cm} \\
%\label{eq:channel_after_com_remov_3} \\
% \label{eq:channel_after_com_remov_4}
n_r^\prime &= n_r- \ell_2.
\end{align*}
At this point, we need to make sure that $n_d^\prime$, $n_c^\prime$, $n_r^\prime$, and $n_s^\prime$ are all non-negative. The terms $n_d'$, $n_s'$, and $n_r'$ are clearly non-negative due to \eqref{Scheme2Cond1}, \eqref{Scheme2Cond2}, \eqref{Scheme2Cond3}, and \eqref{Scheme2Cond4}. To guarantee that $n_c'$ is non-negative, we require
\begin{align}
\label{Scheme2Cond6}
2R_{cm}\leq n_c.
\end{align}

\begin{figure}
\centering
\begin{tikzpicture}[->,>=stealth',shorten >=1pt,auto,node distance=3cm,thick,scale=1]
\SCHEMEWItwoFigtwoAC
\end{tikzpicture}
\caption{The received signal at Rx$1$ in the $k$th channel use ($2\leq k\leq n-1$) after removing the common signal vectors based on scheme \WItwo{}. Here, $\U_{i,cnF}$ denotes $\U_{i,cn}[k]$, while $\U_{i,cn}$ represents $\U_{i,cn}[k-1]$. The time index of all remaining signals is $[k]$.}
\label{Fig:RE:A_3_Rx1_after_remov_Comm}
\end{figure}
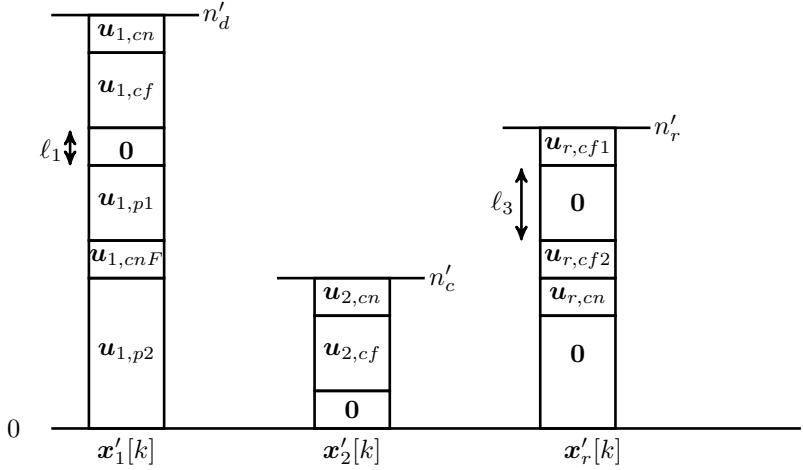

At this step, we are able to specify $\ell_1$, $\ell_2$, $\ell_3$, $\ell_4$, and $\ell_5$. We do this while describing the decoding at Rx$1$ step by step. First, we need to avoid interference between the relay signal $\X_r'$ and $\U_{1,cn}$ and $\U_{1,cf}$. 
Thus, we need $n_r'\leq n_d'-R_{cn}-R_{cf}$ leading to
\begin{align}
n_r-\ell_2 \leq n_d-2R_{cm}-R_{cn}-R_{cf} \quad \text{ if } 0<R_{cf}. \label{Scheme2Cond7}
%\min\{\ell_2,\ell_5\}\geq n_r-n_d+2R_{cm}+R_{cn}+R_{cf}.
\end{align}
\begin{remark}
Notice that if $R_{cf}=0$, $\X_r'$ contains only CN signal. As we discussed earlier, CN relaying strategy can neutralize the interference at Rx$1$ as long as the relay CN signal is aligned with the CN signal from Tx$2$. Moreover, due to the condition of WI regime ($n_c<n_d$), the CN signal from Tx$1$ is received at Rx$1$ on higher level than that of Tx$1$. Thus, by using successive decoding (discussed below), the interference caused by the CN signals from Tx$2$ and relay are removed. Hence, the interference caused by the CN signal in $\X_r'$ does not need to be considered for reliable decoding of $\U_{1,cn}$ and $\U_{1,cf}$.
\end{remark}
Notice that \eqref{Scheme2Cond7} is stricter than \eqref{Scheme2Cond5}. Thus, we set
\begin{align}
\label{Scheme2Cond7p}
\begin{cases}
\ell_2 & = (n_r-n_d+2R_{cm}+R_{cn}+R_{cf})^+ \quad \text{ if } 0<R_{cf1} \\ 
\ell_2 &= (n_r-n_d+2R_{cm}+R_{cn} + R_{cf} + \ell_1 +R_{p1})^+ \quad \text{ if } R_{cf1}=0 \text{ and } 0<R_{cf2}  \\ 
\ell_2 &= n_r \quad \text{otherwise,}
\end{cases}
\end{align}
%Thus, we set
%\begin{align}
%\label{Scheme2Cond7p}
%\begin{cases}
%\ell_2 & = n_r-n_d+2R_{cm}+R_{cn}+R_{cf} \quad \text{ if } 0<R_{cf,} \\ 
%\ell_2 &= 0\quad \text{otherwise,}
%\end{cases}
%\end{align}
which satisfies both \eqref{Scheme2Cond5} and \eqref{Scheme2Cond7}. Since $n_c<n_d$, then Rx$1$ can decode the first bit of $\U_{1,cn}[n-2]$. If we align $\U_{2,cn}$ and $\U_{r,cn}$ at Rx$1$, then the interference of $\U_{2,cn}$ and $\U_{r,cn}$ is neutralized (as in Section \ref{Sec:BuildingBlocks}) since Rx$1$ receives $\U_{r,cn}[n-1]\oplus \U_{2,cn}[n-2]=\U_{1,cn}[n-2]$. This alignment is possible if the following two conditions are satisfied. Firstly, Rx$1$ has to receive $\U_{r,cn}$ which requires condition \eqref{Scheme2Cond4}. Secondly, the CN signal from the relay and Tx$2$ has to be aligned at Rx$1$. This is possible as long as $n_r-\ell_5=n_c'$ and hence
\begin{align}
\label{Scheme2Cond8}
\ell_5=n_r-n_c+2R_{cm}.
\end{align}
Note that $\ell_5$ is always positive since $n_c\leq n_r$ in this regime. After decoding the first bit of $\U_{1,cn}[n-2]$, Rx$1$ removes the contribution of the first bit of $\U_{r,cn}[n-1]\oplus \U_{2,cn}[n-2]$. Then, Rx$1$ decodes the second bit of $\U_{1,cn}[n-1]$ received from Tx$1$, and cancels the interference of the second bit of $\U_{r,cn}[n-1]\oplus \U_{2,cn}[n-2]$. It continues this way until all bits of $\U_{1,cn}[n-1]$ are decoded, and all bits of $\U_{r,cn}[n-1]\oplus \U_{2,cn}[n-2]$ are cancelled. 

Next, Rx$1$ decodes the first bit of $\U_{1,cf}[n-1]$ interference free. 
Then, it uses this bit in combination with $\U_{r,cf1}[n]$ and $\U_{r,cf2}[n]$ (decoded in the $n$th channel use) to extract the first bit of the interference signal $\U_{2,cf}[n-1]$ and subtract its contribution from the received signal. Then it proceeds to decode the second bit of $\U_{1,cf}[n-1]$. Decoding proceeds this way until all bits of $\U_{1,cf}[n-1]$ are decoded and all bits of $\U_{2,cf}[n-1]$ are cancelled. Note that at this point there is no interference left from Tx$2$ if the signals $\U_{2,p1}[n-1]$, $\U_{2,cn}[n-1]$, and $\U_{2,p2}$ are received below the noise floor at Rx$1$, i.e.,
\begin{align}
\label{Scheme2Cond9}
n_c-2R_{cm}-R_{cn}-R_{cf}-\ell_1\leq 0.
\end{align}
Under this condition, Rx$1$ decodes $\U_{r,cf1}$ which is received by Rx$1$ if \eqref{Scheme2Cond3} holds, and is interference free if
\begin{align}
\label{Scheme2Cond10}
\ell_1\geq \ell_4.
\end{align}
Then it decodes $\U_{1,p1}$ which is also received interference free since
\begin{align}
\label{Scheme2Cond11}
\begin{cases}
\ell_3= R_{p1} &\quad \text{ if } 0<R_{cf1} \\ 
\ell_3=0  &\quad \text{ if } 0=R_{cf1}
\end{cases}.
\end{align}
Notice that for the case in which $R_{cf1}=0$, \eqref{Scheme2Cond7p} guarantees that $\U_{1,p1}$ is received interference free.
Now, Rx$1$ wants to decode $\U_{r,cf2}[n-1]$. To do this, it first removes the contribution of $\U_{1,cn}[n-1]$ (denoted by $\U_{1,cnF}$ in Fig. \ref{Fig:RE:A_3_Rx1_after_remov_Comm}) from the received signal, which is possible since Rx$1$ has decoded $\U_{1,cn}[n-1]$ in channel use $n$. After removing $\U_{1,cn}[n-1]$, Rx$1$ observes $\U_{r,cf2}[n-1]$ interference free if
\begin{align}
\label{Scheme2Cond12}
\ell_4=(R_{cf}-R_{cn})^+.
\end{align}
Under this condition, Rx$1$ can decode $\U_{r,cf2}[n-1]$. Finally, Rx$1$ decodes $\U_{1,p2}$ interference free. 

As a result, this scheme achieves
\begin{align}
nR_{\Sigma,\text{\WItwo{}}} = 2(n-1) \left( \min\{R_{cm},(n_c-n_d+2R_{cm})^+\}+R_{cn} + R_{cf} + R_{p1} + R_{p2}\right).
\end{align}
Dividing this expression by $n$ and letting $n\to \infty$, we obtain the sum-rate 
\begin{align}
R_{\Sigma,\text{\WItwo{}}} = 2\left(\min\{R_{cm},n_c-n_d+2R_{cm}\}+R_{cn} + R_{cf} + R_{p1} + R_{p2}\right).
\end{align}
Therefore, the optimal achievable sum-rate of scheme 2 is obtained by solving the following optimization problem 
\begin{align*}
\max  \quad &R_{\Sigma,\text{\WItwo{}}}\\
\text{s.t.} \quad & \text{\eqref{Scheme2Cond1}-\eqref{Scheme2Cond12}}\\
& R_{cm},R_{cn},R_{cf},R_{p1},R_{p2},\ell_1,\ell_2,\ell_3,\ell_4,\ell_5,\ell_6\geq 0.
\end{align*}
The solution of this optimization problem is presented in Tables \ref{Tab:WI21} and Tables \ref{Tab:WI22}, with $$\ell_1=n_d-2R_{cm}-2R_{cn}-R_{cf}-R_{p1}-R_{p2}.$$

\begin{table*}
\centering
\begin{tabular}{|c||c|c|c|c|}
\hline
 Regime &  $n_r+n_s\leq n_d$ & $n_d< n_r+n_s \leq n_d+\frac{n_c}{2}$ & $n_d<\min\{n_r+n_s-\frac{n_c}{2},\frac{3}{2}n_c\}$ & $\frac{3}{2}n_c\leq n_d < n_r+n_s-\frac{n_c}{2}$ \\\hline
$R_{cm}$ & 0 & 0 & $\frac{n_c}{2}$ & 0 \\\hline
$R_{cn}$ & 0 & $n_s-\max\{n_c,n_d-n_r\}$ & 0 & $n_s-n_c- (n_s-\frac{3}{2}n_c)^+$ \\\hline
$R_{cf}$ & 0 & $n_r-n_d+n_s$ & 0 & $\frac{n_c}{2}$ \\\hline
$R_{p1}$ & 0 & $\max\{n_c,n_d-n_r\}-n_c$ & $n_s-n_c$&$(n_s-\frac{3}{2}n_c)^+$ \\\hline
$R_{p2}$ & $n_d-n_c$ & $n_d-n_s$ & $n_d-n_s$&$n_d-n_s$ \\\hline
\end{tabular}
\caption{Rate allocation parameters for the scheme \WItwo{} in the regime where $n_c\leq n_s\leq n_r \leq n_d$.}
\label{Tab:WI21}
\end{table*}

\begin{table*}
\centering
\begin{tabular}{|c||c|c|c|}
\hline
 Regime &  $n_r+n_s \leq n_d$ & $n_d< n_r+n_s\leq n_d+\frac{n_c}{2}$ & $n_d+\frac{n_c}{2}<n_r+n_s$ \\\hline
$R_{cm}$ & 0 & 0 & 0 \\\hline
$R_{cn}$ & 0 & $n_s-n_c-\max\{0,2n_s-2n_c+n_r-n_d\}$ & $n_s-n_c-(n_s-\frac{3}{2}n_c)^+$ \\\hline
$R_{cf}$ & 0 & $n_r-n_d+n_s$ & $\frac{n_c}{2}$ \\\hline
$R_{p1}$ & 0 & $\max\{0,2n_s-2n_c+n_r-n_d\}$ & $(n_s-\frac{3}{2}n_c)^+$ \\\hline
$R_{p2}$ & $n_d-n_c$ & $n_d-n_s$ & $n_d-n_s$ \\\hline
\end{tabular}
\caption{Rate allocation parameters for the scheme \WItwo{} in the regime where $n_c\leq n_r \leq n_s\leq n_d-\frac{n_c}{2}$.}
\label{Tab:WI22}
\end{table*}

The given rate allocation satisfies constraints \eqref{Scheme2Cond1}-\eqref{Scheme2Cond12} and achieves the sum-rate in \eqref{eq:R_sum_LB_A3}. As a result, this proves the achievability of Theorem \ref{Theorem:capacity_LD_IRC} for the two regimes $n_c\leq n_s\leq n_r \leq n_d$ and $n_c\leq n_r \leq n_s\leq n_d-\frac{n_c}{2}$.

\subsection{Scheme \WIthree{}} 
Note that scheme \WIone{} and \WItwo{} do not cover all possible regimes with weak interference $n_c<n_d$ and with a source-relay channel stronger than the cross channel $n_c<n_s$ (condition of Theorem \ref{Theorem:capacity_LD_IRC}). In particular, three possibilities remain given by (i) $n_c\leq n_r\leq n_s\leq n_d< n_s+\frac{n_c}{2}$, (ii) $n_r\leq n_c< n_d\leq n_s$, and (iii) $n_r\leq n_c\leq n_s\leq n_d$. Next, we present the last scheme which covers these three cases and completes the proof of the achievability of Theorem \ref{Theorem:capacity_LD_IRC} for the WI regime. This scheme is called \WIthree{}. Its achievable sum-rate is presented in the following proposition. 

\begin{proposition}\label{prop:RateWI3}
The achievable sum-rate with the scheme \WIthree{} for the IRC is given by
\begin{align}\label{eq:R_sum_UB_A1}
R_{\Sigma,\text{\WIthree{}}} = 
\begin{cases}
\min\{n_s+n_d-n_c,2n_d+n_r-n_c,2\max\{n_c,n_r+n_d-n_c\}\}  & \text{ if }n_r\leq n_c < n_d \leq n_s \\
\min\{2n_d-n_c,2\max\{n_s,n_d-n_c\},2\max\{n_c,n_r+n_d-n_c\}\}  & \text{ if }n_r\leq n_c \leq n_s \leq n_d  \\
\min\{2n_d-n_c,2\max\{n_r+n_s-n_c,n_d-n_c\}\}  & \text{ if }n_c\leq n_r \leq n_s \leq n_d < n_s+\frac{n_c}{2}.
\end{cases}
\end{align}
{This proves the achievability of Theorem \ref{Theorem:capacity_LD_IRC} within the three regimes.}
Note that these sum-rate expressions coincide with the upper bounds given in Lemmas \ref{Lemma:new_bounds} and \ref{Lemma:known_UB1_det}. 
\end{proposition}
%Scheme \WIthree{} achieves the following sum-rates 
%\begin{align}
%R_{\Sigma,\text{\WIthree{}}} &= \min\{n_s+n_d-n_c,2n_d+n_r-n_c,2\max\{n_c,n_r+n_d-n_c\}\}  & \text{ if }n_r\leq n_c \leq n_d \leq n_s & \label{eq:R_sum_UB_B3} \\
%R_{\Sigma,\text{\WIthree{}}} &= \min\{2n_d-n_c,2\max\{n_s,n_d-n_c\},2\max\{n_c,n_r+n_d-n_c\}\}  & \text{ if }n_r\leq n_c \leq n_s \leq n_d & \label{eq:R_sum_UB_B4} \\
%R_{\Sigma,\text{\WIthree{}}} &= \min\{2n_d-n_c,2\max\{n_r+n_s-n_c,n_d-n_c\}\}  & \text{ if }n_c\leq n_r \leq n_s \leq n_d < n_s+\frac{n_c}{2} \label{eq:R_sum_UB_B5_2} 
%\end{align}
%It can be easily verified that the achievable sum-rates in \eqref{eq:R_sum_UB_B3}-\eqref{eq:R_sum_UB_B5_2} coincide the upper bounds in Lemmas \ref{Lemma:new_bounds} and \ref{Lemma:known_UB1_det}.
In what follows, we present scheme \WIthree{} in details. This transmission scheme is a combination of common and private signaling in addition to CN. 

\subsubsection*{Encoding at transmitters}
In $k$th channel use, Tx$1$ constructs the following signal vectors
\begin{align}
\boldsymbol{x}_1[k] = 
\begin{bmatrix}
\boldsymbol{u}_{1,cm1}[k] \\
\boldsymbol{u}_{1,cm2}[k] \\
\boldsymbol{u}_{1,cn1}[k-1] \\
\boldsymbol{u}_{1,cn2}[k-1] \\
\boldsymbol{u}_{1,p1}[k] \\
\boldsymbol{0}_{\ell_1^u} \\
\boldsymbol{u}_{1,cn1}[k] \\
\boldsymbol{0}_{\ell_1^d} \\
\boldsymbol{u}_{1,cn3}[k-1] \\
\boldsymbol{u}_{1,p2}[k] \\
\boldsymbol{u}_{1,cn2}[k] \\
\boldsymbol{u}_{1,cn3}[k]\\
\boldsymbol{0}_{s}
\end{bmatrix},  k=1,\ldots,n,
\end{align}
where  $s$ is chosen so that the length of $\X_1$ is $q$.  The length of vectors $\boldsymbol{u}_{1,a}$ is $R_a$, where $a\in\{cm2,cn1,cn2,cn3,p1,p2\}$. 
The length of the zero vector, i.e., $\ell_1^u$ and  $\ell_1^d$, will be chosen later in such a way that facilitates reliable decoding. We define $$\ell_1 = \ell_1^u+\ell_1^d. $$
The vector $\boldsymbol{u}_{1,cm1}$ with rate $R_{cm1}$ has a length of $\ell_{cm1}$. 
%\begin{remark}
%In this scheme, $\boldsymbol{u}_{i,cn2^d}$ does not appear with $\boldsymbol{u}_{i,cn1}$ or $\boldsymbol{u}_{i,p1}$ at the same time. Therefore, we write
%\begin{align}
%R_{cn2^d} = 0 \quad \text{ if } R_{cn1} > 0  \text{ or } R_{p1} >0. \label{eq:Cond_assump_Scheme3}
%\end{align}
%\end{remark}
We further set $\boldsymbol{u}_{1,cn1}[0]$, $\boldsymbol{u}_{1,cn2}[0]$, $\boldsymbol{u}_{1,cn3}[0]$, $\boldsymbol{u}_{1,cm1}[n]$ , $\boldsymbol{u}_{1,cm2}[n]$, $\boldsymbol{u}_{1,cn1}[n]$, $\boldsymbol{u}_{1,cn2}[n]$, $\boldsymbol{u}_{1,cn3}[n]$, $\boldsymbol{u}_{1,p1}[n]$, and $\boldsymbol{u}_{1,p2}[n]$ to be zero vectors. Similarly, Tx$2$ constructs $\X_2[k]$. This construction requires
\begin{align}
\ell_{cm1} + R_{cm2} + 2R_{cn1} + 2R_{cn2} + 2R_{cn3} + R_{p1} + R_{p2}+ \ell_1  \leq q. \label{Scheme1Condnd}
\end{align}
%\begin{remark}
%In this scheme, the signal vector $\boldsymbol{u}_{i,cn3}$ does not appear with $\boldsymbol{u}_{i,cn1}$ or $\boldsymbol{u}_{i,p1}$ at the same time. Therefore, we have following condition
%\begin{align}
%R_{cn3} = 0 \qquad \text{ if } R_{cn1} > 0 \text{ or }  R_{p1} > 0
%\end{align}
%\end{remark}

\subsubsection*{Decoding at the relay}
The relay receives the top-most $n_s$ bits of the transmitted signal vectors. In channel use $k$, the relay decodes $\U_{r,c}[k+1]=\U_{1,c}[k]\oplus \U_{2,c}[k]$, where $c\in\{cn1,cn2,cn3\}$. To enable this decoding, it is required to set the length of the transmitted signals in such a way that the relay is able to observe the desired signals $\U_{r,c}[k+1]$. To write the constraint for reliability decoding of $\U_{r,c}[k+1]$, we need to distinguish between several cases. 
This case distinction is necessary since the relay does not need to decode the sum of the private and common signal vectors when there is no CN signal vector in the lower level. The necessary constraint for decoding the CN signal vectors at the relay is given as follows 
\begin{align}
\begin{cases}
\ell_{cm1} + R_{cm2} + 2R_{cn1} + 2R_{cn2} + 2R_{cn3} + R_{p1} + R_{p2}+ \ell_1  \leq n_s & \text{ if } R_{cn2} \neq 0 \text{ or } R_{cn3} \neq 0 \\ 
\ell_{cm1} + R_{cm2} + 2R_{cn1} +  R_{p1} + \ell_1^u \leq n_s & \text{ if } R_{cn2} = R_{cn3} = 0  \text{ and } 0<R_{cn1}
%\\
%0 \leq n_s & \text{ if } R_{cn2} = R_{cn3} = R_{cn1} =0
\end{cases}.
\label{eq:rate_const_Re_B3_1}
\end{align}

%{Notice that this condition is tighter than \eqref{Scheme1Condnd} for regimes \eqref{eq:R_sum_UB_B3}-\eqref{eq:R_sum_UB_B5_2}.}

\subsubsection*{Encoding at the relay} After decoding the CN signals, the relay generates 
\begin{align*}
\boldsymbol{x}_r[k] = 
\begin{bmatrix}
\boldsymbol{0}_{\ell_2} \\
\boldsymbol{u}_{r,cn1}[k]\\
\boldsymbol{u}_{r,cn2}[k]\\
\boldsymbol{0}_{\ell_3}\\
\boldsymbol{u}_{r,cn3}[k]\\
\boldsymbol{0}_{r}
\end{bmatrix}
\end{align*}
in channel use $k\in\{2,\cdots, n\}$, where $\ell_2$ and $\ell_3$ will be chosen later, and $r$ is chosen so that the length of $\X_r$ is $q$. The signal $\boldsymbol{u}_{r,cn1}$, $\boldsymbol{u}_{r,cn2}$, and $\boldsymbol{u}_{r,cn3}$ will be received by Rx$1$ if
\begin{align}
\label{Scheme3Cond1}
R_{cn1}+R_{cn2}+R_{cn3}+\ell_2+\ell_3\leq n_r.
\end{align}

\subsubsection*{Decoding at the receiver side} 
Here, we only discuss the decoding process at Rx$1$ since decoding at Rx$2$ is done similarly. Similar to previous schemes, receivers use backward decoding. Hence, Rx$1$ waits until the end of $n$th channel use. Assuming that decoding $\Y_1[n]$ is done successfully, Rx$1$ knows
\begin{itemize}
\item $\U_{1,cn1}[n-1]$
\item $\U_{1,cn2}[n-1]$
\item $\U_{1,cn3}[n-1]$.
\end{itemize}
Next, it starts with processing $\Y_{1}[n-1]$. To do this, first, it decodes the first common signals of both transmitters $\boldsymbol{u}_{1,cm1}[n-1]$ and $\boldsymbol{u}_{2,cm1}[n-1]$ in a MAC fashion, while treating the remaining signals as noise. Decoding these signals requires that $\X_r$ does not interfere with the desired common signal, thus
\begin{align}
\label{Scheme3Cond2}
n_r-\ell_2\leq n_d-\ell_{cm1} \quad \text{ if } \ell_{cm1}>0.
\end{align} 
Similar to decoding the common signals in the \WItwo{} scheme (cf. Section \ref{sec:SchemeWI2}), this decoding is done as in a symmetric IC with channels $n_{d,IC}=\ell_{cm1}$ and $n_{c,IC}=n_c-n_d+\ell_{cm1}$. Decoding the two common signals can be done at a rate of $\min\{\frac{n_{d,IC}}{2},n_{c,IC}\}$ leading to the achievability of the following common rate
\begin{align}
R_{cm1}  = \min\left\{\frac{\ell_{cm1}}{2},(n_c-n_d+\ell_{cm1})^+\right\} \label{eq:Common_cond_3}
\end{align}
After removing the common signal vectors $\U_{1,cm1}[n-1]$ and $\U_{2,cm1}[n-1]$ and the top-most $\ell_2$ zeros sent by the relay, Rx$1$ observes the superposition of $n_d'$ bits from Tx$1$, $n_c'$ bits from Tx$2$ and $n_r'$ bits from the relay. 
%Similarly, the relay observation after removing the sum of the common signal vector $\U_{1,cm1}[n-1]\oplus \U_{2,cm1}[n-1]$ contains $n_s'$ bits. 
These parameters are defined as follows
\begin{align*}
n_d^\prime &= n_d - \ell_{cm1} \\ 
n_c^\prime &= (n_c - \ell_{cm1})^+ \\ 
%n_s^\prime &= (n_s - \ell_{cm1})^+ \\ 
n_r^\prime &= n_r-\ell_2.
\end{align*}
%Moreover, we define a further parameter $q'=\max\{n_d',n_c',n_r',n_s'\}$ which is equal to $\max\{n_d',n_s'\}$
%due to $n_c<n_d$, and \eqref{Scheme3Cond2}.

Now, for the sake of simplicity, we distinguish between two cases and explain the decoding for each case separately. 
\begin{itemize}
\item Scheme \WIThreea{}: In this case the signal vector $\U_{i,cn3}$ does not appear. Hence, we have $R_{cn3} = 0$ and we set $\ell_3=0$. The received signal vector at Rx$1$ after removing $\U_{i,cm1}[n-1]$ is illustrated for this case in Fig. \ref{Fig:RE:B_3_Rx1_after_remov_Comm}.
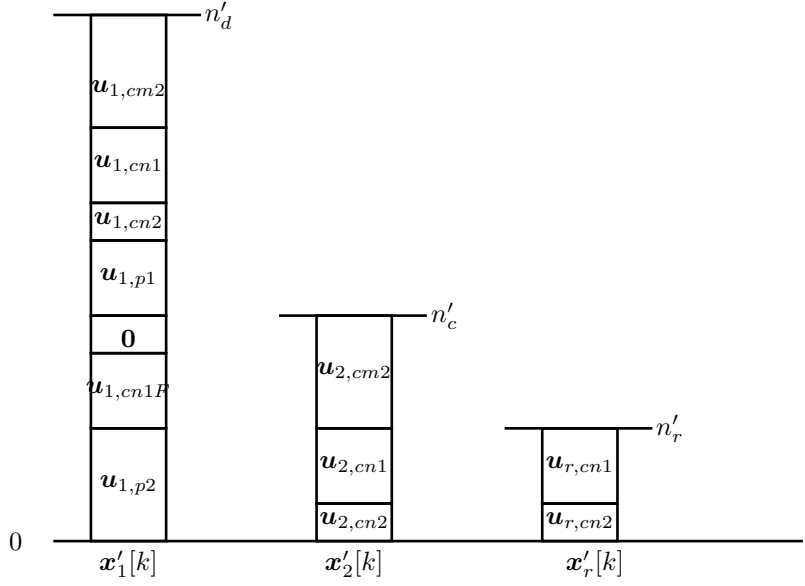
\begin{figure}[t]
\centering
\begin{tikzpicture}[->,>=stealth',shorten >=1pt,auto,node distance=3cm,thick,scale=1]
\YprimeSchemeThreeTwoa
\end{tikzpicture}
\caption{The received signal at Rx1 based on scheme \WIThreea{} in the $k$th channel use ($2\leq k\leq n-1$) after removing the common signal vectors ${\U}_{1,cm1}[k]$ and ${\U}_{2,cm1}[k]$. Note that $\U_{1,cn1F} = \U_{1,cn1}[k]$, $\U_{i,cn1} = \U_{i,cn1}[k-1]$, and $\U_{i,cn2} = \U_{i,cn2}[k-1]$ for $i\in\{1,2\}$. The time index of all other signal vectors is $[k]$.}
\label{Fig:RE:B_3_Rx1_after_remov_Comm}
\end{figure}
%First, Rx$1$ removes the interference causes by $\U_{1,cn1}[n-1]$ ($\U_{1,cnF}$ in Fig. \ref{Fig:RE:B_3_Rx1_after_remov_Comm}) since this signal vector is already known at Rx$1$ from decoding in channel use $n$.
To guarantee that Rx$1$ receives the common signal vector $\U_{1,cm2}[n-1]$ the CN signal vectors $\U_{1,cn1}[n-2]$, $\U_{1,cn2}[n-2]$, and the private signal vectors $\U_{1,p1}[n-1]$, $\U_{1,p2}[n-1]$ without any overlap with each other, we write
\begin{align}
\ell_{cm1}+R_{cm2} + 2R_{cn1} + R_{cn2} + R_{p1} + \ell_1 + R_{p2}\leq n_d.
\end{align}
First, Rx$1$ decodes $\boldsymbol{u}_{1,cm2}[n-1]$, $\boldsymbol{u}_{1,cn1}[n-2]$, $\boldsymbol{u}_{1,cn2}[n-2]$, and $\boldsymbol{u}_{1,p1}[n-1]$. These signals are received interference free as long as
\begin{align}
\label{eq:cond_scheme3_rx_2}
R_{cm2} + R_{cn1} + R_{cn2} + R_{p1} &\leq n_d^\prime - n_c^\prime, \\
\label{Scheme3Cond4}
R_{cm2} + R_{cn1} + R_{cn2} + R_{p1} &\leq n_d^\prime - n_r^\prime. 
\end{align}
{Notice that condition \eqref{Scheme3Cond4} is tighter than \eqref{Scheme3Cond2}.} 
{Now, Rx$1$ is ready to remove the contribution of the $\U_{2,cn1}[n-2]$ and $\U_{2,cn2}[n-2]$. To enable this, we require that the CN signal vectors of Tx$2$ and the relay are aligned at Rx$1$. This alignment is possible if}
\begin{align}
\label{Scheme3Cond5}
n_c^\prime - R_{cm2} = n_r^\prime \text{ if } R_{cn1} \neq 0 \text{ or } R_{cn2} \neq 0.
\end{align} 
From this condition, we obtain 
\begin{align}
\ell_2 = \begin{cases}
 n_r - n_c'+R_{cm2} & \text{ if } R_{cn1} \neq 0 \text{ or } R_{cn2} \neq 0 \\
 n_r & \text{ if } R_{cn1} = 0 \text{ and } R_{cn2} = 0
\end{cases}.
\end{align}
Under the condition in \eqref{Scheme3Cond5}, interference neutralization takes place as shown in Section \ref{Sec:BuildingBlocks}, and Rx$1$ receives $\U_{1,cn1}[n-2]$ and  $\U_{1,cn2}[n-2]$ (or parts thereof) as an aggregate of the CN signals from Tx$2$ and the relay. Since Rx$1$ has already decoded $\U_{1,cn1}[n-2]$ and $\U_{1,cn2}[n-2]$, aggregate interference from these signals can be removed. This solves the problem of the CN interference. 

Since Rx$1$ has decoded $\U_{1,cn1}[n-1]$ in channel use $n$, it removes the contribution of this signal ($\U_{1,cn1F}$ in Fig. \ref{Fig:RE:B_3_Rx1_after_remov_Comm}) from its received signal. 
Next, Rx$1$ decodes $\U_{2,cm2}[n-1]$. This can be done reliably as long as 
\begin{align}
n_d' - R_{cm2} - 2R_{cn1}-R_{cn2}-R_{p1}-\ell_1 \leq n_c' - R_{cm2} \quad \text{if} \quad R_{cm2}> 0 .
\end{align}

Finally, Rx$1$ decodes $\U_{1,p2}[n-1]$. To guarantee that this signal vector is received interference free, following condition needs to be satisfied
%\begin{align}
%\label{Scheme3Cond7}
%R_{cm2} +R_{p2} &\leq n_c^\prime \quad \text{if} \quad R_{cm2}> 0 .
%\end{align}
%To avoid the interference from the signals $\U_{2,p1}[n-1]$, $\U_{2,p2}[n-1]$, $\U_{2,cn1}[n-1]$, and $\U_{2,cn2}[n-1]$, we set 
\begin{align}
\begin{cases}
n_c' -R_{cm2}-R_{cn1} - R_{cn2} = 0 & \text{ if } 0<R_{p1} \\ 
n_c' -R_{cm2}-R_{cn1} - R_{cn2} - \ell_1^u \leq 0 & \text{ if } 0=R_{p1} \\ 
\end{cases}
\label{Scheme3Cond6}
\end{align}
%Moreover, to avoid interference from $\U_{1,cn2}[n-1]$, we set
%\begin{align}
%\ell_1 = n_d-\ell_{cm1} -2R_{cn1}-R_{cn2}-R_{p1}-n_c'.\label{Scheme3Cond8}
%\end{align}
As a result this scheme achieves
\begin{align}
nR_{\Sigma,\text{\WIThreea{}}} = 2(n-1)\left(R_{cm1}+ R_{cm2}+R_{cn1}+R_{cn2}+R_{p1}+R_{p2}\right).
\end{align}
By dividing this expression by $n$ and letting $n\to\infty$, we obtain
\begin{align}
R_{\Sigma,\text{\WIThreea{}}} = 2\left(R_{cm1}+ R_{cm2}+R_{cn1}+R_{cn2}+R_{p1}+R_{p2}\right).
\end{align}
This sum-rate has to be maximized under the conditions in 
\eqref{eq:rate_const_Re_B3_1}-\eqref{Scheme3Cond6}, which can be formulated as the following optimization problem
\begin{align}
\max\quad & R_{\Sigma,\text{\WIThreea{}}} \label{eq:Scheme_Opt_B3a} \\
\text{s.t.}\quad &\text{\eqref{eq:rate_const_Re_B3_1}-\eqref{Scheme3Cond6} are satisfied} \notag\\
& \ell_{cm1},R_{cm2},R_{cn1},R_{cn2},R_{p1},R_{p2},\ell_1^u,\ell_1^d\geq 0. \notag
\end{align}
\item Scheme \WIThreeb{}: Compared to previous case, here $0<R_{cn3}$ while $R_{cn1} = R_{p1} = 0$. The received signal vector after removing the $\U_{1,cm1}$ and $\U_{2,cm1}$ is illustrated in Fig. \ref{Fig:RE:B_3b_Rx1_after_remov_Comm}. To guarantee that Rx$1$ receives the common signal vector $\U_{1,cm2}$, the CN signal vectors $\U_{1,cn2}$, $\U_{1,cn3}$, and private signal vector $\U_{1,p2}$, we write
\begin{align}
\ell_{cm1}+R_{cm2} + R_{cn2} + \ell_1 + R_{cn3} + R_{p2} \leq n_d.
\label{eq:rate_const_Re_B3b_1}
\end{align}
 
First, Rx$1$ decodes $\U_{1,cm2}[n-1]$ and $\U_{1,cn2}[n-2]$. To guarantee that these signals are received interference free, we write
\begin{align}
R_{cm2} + R_{cn2} \leq n_d' -n_c' \\
R_{cm2} + R_{cn2} \leq n_d' -n_r'. 
\end{align}
\begin{figure}[t]
\centering
\begin{tikzpicture}[->,>=stealth',shorten >=1pt,auto,node distance=3cm,thick,scale=1]
\YprimeSchemeThreeTwob
\end{tikzpicture}
\caption{The received signal at Rx1 based on scheme \WIThreeb{} in the $k$th channel use ($2\leq k\leq n-1$) after removing the common signal vectors ${\U}_{1,cm1}[k]$ and ${\U}_{2,cm1}[k]$. Note that $\U_{i,cn2} = \U_{i,cn2}[k-1]$, and $\U_{i,cn3} = \U_{i,cn3}[k-1]$ for $i\in\{1,2\}$. The time index of all other signal vectors is $[k]$.}
\label{Fig:RE:B_3b_Rx1_after_remov_Comm}
\end{figure}
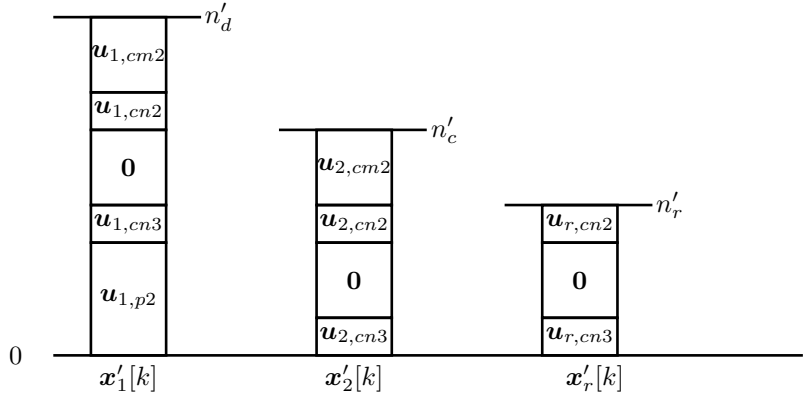
Next Rx$1$ decodes $\U_{2,cm2}[n-1]$. This can be done reliably as long as 
\begin{align}
n_d' -R_{cm2}-R_{cn2}-\ell_1 \leq n_c' -R_{cm2} \quad \text{ if } R_{cm2}\neq 0.\label{eq:rate_const_Re_B3b_7} 
\end{align}

Next, Rx$1$ removes the interference of $\U_{1,cn2}[n-2]$ caused by Tx$2$ and the relay (i.e., $\U_{2,cn2}[n-2]\oplus \U_{r,cn2}[n-1]$). This can be done as long as the CN signal vectors $\U_{2,cn2}[n-2]$ and $\U_{r,cn2}[n-1]$ are aligned. Therefore, we write
\begin{align}
n_r' = n_c' - R_{cm2}  \quad \text{ if } R_{cn2}\neq 0.
\label{eq:rate_const_Re_B3b_5}
\end{align}

Next, Rx$1$ decodes the top-most bit of $\U_{1,cn3}[n-2]$ (sent by Tx$1$) without any overlap with $\U_{2,cn3}[n-2]\oplus \U_{r,cn3}[n-1]=\U_{1,cn3}[n-2]$. This is possible since $n_c<n_d$.
After that Rx$1$ removes the interference caused by the top-most bit of $\U_{2,cn3}[n-2]\oplus \U_{r,cn3}[n-1]$ and then it decodes the second bit of $\U_{1,cn3}[n-2]$. It continues this decoding until the whole vector $\U_{1,cn3}[n-2]$ is decoded. This can be done reliably as long as the interference caused by $\U_{2,cn3}[n-2]$ can be neutralized by the relay CN signal $\U_{r,cn3}[n-1]$. Hence, we write
\begin{align}
n_r' = R_{cn3}\quad & \text{ if } R_{cn2} = 0. 
\label{eq:rate_const_Re_B3b_8} 
\end{align}
From the conditions in \eqref{eq:rate_const_Re_B3b_5} and \eqref{eq:rate_const_Re_B3b_8}, we can fix $\ell_2$ and $\ell_3$ as follows
\begin{align}
\begin{cases}
\ell_2 = n_r- n_c'+R_{cm2}  &\text{ and } \ell_3 = \ell_1 \quad \text{ if } R_{cn2} \neq 0 \\ 
\ell_2 = n_r- R_{cn3} & \text{ and } \ell_3 = 0 \quad \text{ if } R_{cn2} = 0 \text{ and } R_{cn3}\neq 0 \\
\ell_2 = n_r & \text{ and } \ell_3 = 0 \quad \text{otherwise}.
\end{cases}
\end{align}
Finally, Rx$1$ decodes $\U_{1p2}[n-1]$. To guarantee that $\U_{1p2}[n-1]$ is received interference free, following condition is required
\begin{align}
R_{p2} \leq n_d' -n_c'. \label{eq:rate_const_Re_B3b_9} 
\end{align}

Now, Rx$1$ has completed the decoding its desired signals in channel use $n-1$. It proceeds backwards till channel use 1. Therefore, the sum-rate is given as follows
\begin{align}
nR_{\Sigma,\text{\WIThreeb{}}} = 2(n-1)\left(R_{cm1}+ R_{cm2}+R_{cn2}+R_{cn3}+R_{p2}\right).
\end{align}
By dividing the expression by $n$ and letting $n\to\infty$, we obtain
\begin{align}
R_{\Sigma,\text{\WIThreeb{}}} = 2\left(R_{cm1}+ R_{cm2}+R_{cn2}+R_{cn3}+R_{p2}\right).
\end{align}
This sum-rate has to be maximized under the conditions in \eqref{eq:rate_const_Re_B3_1}-\eqref{eq:Common_cond_3},  \eqref{eq:rate_const_Re_B3b_1}-\eqref{eq:rate_const_Re_B3b_7}, and \eqref{eq:rate_const_Re_B3b_9} which can be formulated as following optimization problem
\begin{align}
\max\quad & R_{\Sigma,\text{\WIThreeb{}}} \label{eq:Scheme_Opt_B3b} \\
\text{s.t.}\quad &\text{\eqref{eq:rate_const_Re_B3_1}-\eqref{eq:Common_cond_3},  \eqref{eq:rate_const_Re_B3b_1}-\eqref{eq:rate_const_Re_B3b_7}, and \eqref{eq:rate_const_Re_B3b_9} are satisfied} \notag \\
& \ell_{cm1},R_{cm2},R_{cn2},R_{cn3},R_{p1},R_{p2},\ell_1^u,\ell_1^d,\ell_2,\ell_3\geq 0. \notag
\end{align}
\end{itemize}
The optimal parameters for the optimization problems \eqref{eq:Scheme_Opt_B3a} and \eqref{eq:Scheme_Opt_B3b} (with $\ell_1^u$ and $\ell_1^d$) are given in Table \ref{Tab:WI3_1}-\ref{Tab:WI3_3}. Using these optimal values, the sum-rate in \eqref{eq:R_sum_UB_A1} is achieved.
As a result, this scheme together with Schemes \WIone{} and \WItwo{} proves the achievability of Theorem \ref{Theorem:capacity_LD_IRC} for the WI regime.

\begin{table*}
\centering
\begin{tabular}{|c||c|c|c|c|}
\hline
\multirow{2}{*}{Regime} & \multicolumn{2}{c|}{$n_s+n_c\leq n_d+2n_r$}  &  \multicolumn{2}{c|}{$n_s+n_c>n_d+2n_r$}  \\ \cline{2-5}
&$3n_c \leq n_s+n_d $ &  $3n_c>n_s+n_d$ & $n_d\leq 2n_c$ & $n_d>2n_c$  \\ \hline
$\ell_{cm1}$ & $0$ & $n_c+n_d-n_s$ &  0 & 0 \\ \hline
$R_{cm2}$    & $\frac{n_c+n_d-n_s}{2}$& $0$ & $n_c-n_r$ & $n_c-n_r$ \\ \hline
$R_{cn1}$    &  {$\frac{1}{2}\min\{(3n_d-3n_c-n_s)^+,2R_{cm2}\}$} & 0 & $(n_r-2n_c+2n_d-n_s)^+$ & \footnotesize{$\min\{n_r,n_c-n_r\}$}  \\ \hline
$R_{cn2}$    & \footnotesize{$\min\{(n_d-n_c-R_{cm2}-R_{cn1})^+,n_s-n_d \}$} & \footnotesize{$\min\{n_d-n_c,n_s-n_d\}$} & \footnotesize{$\min\{n_s-n_d,n_d-2n_c+n_r\}$} & $(2n_r-n_c)^+$ \\ \hline
$R_{p1}$     &  $(n_d-2n_c)^+$ & $(2n_d-n_c-n_s)^+$ & 0 & $n_d-2n_c$   \\ \hline
$R_{cn3}$    &  $(n_s-n_d-R_{cn2})$ & $(n_c+n_s-2n_d)^+$ & $(n_c-n_d+n_r)^+$& 0 \\ \hline
$R_{p2}$     &  \multicolumn{2}{c|}{$\min\{n_c-R_{cm2}-R_{cn3},n_d-n_c-R_{p1}-R_{cn1}\}$} & $\min\{n_d-n_c,n_r\}$ & $n_r$   \\ \hline
%$R_\Sigma$ &  \multicolumn{2}{c|}{$n_s+n_d-n_c$}  &  \multicolumn{2}{c|}{$2(n_r+n_d-n_c)$}  \\ \hline
\end{tabular}
\caption{Rate allocation parameters for the scheme \WIthree{} in the regime where $n_r\leq n_c< n_d \leq n_s$ and $n_r+n_d-n_c>n_c$. In these regimes, $\ell_1^u     =R_{cm2}-R_{cn1}$, $\ell_1^d =0$ and $\ell_1 = \ell_1^u$.}
\label{Tab:WI3_1}
\end{table*}

\begin{table*}
\centering
\begin{tabular}{|c||c|c|c|}
\hline
Regime & $n_s< \min\{n_r+n_d,3n_c-n_d\}$ & \footnotesize{$3n_c \leq \min\{n_s+n_d,2n_d+n_r\} $} & $n_d\leq \min\{n_s-n_r,\frac{3n_c-n_r}{2}\}$ \\ \hline
$\ell_{cm1}$ & $n_c+n_d-n_s$ & 0 & $n_c-n_r$ \\ \hline
$R_{cm2}$    &  $0$ & $n_d-n_c$ & $0$ \\ \hline
$R_{cn1}$    &  0 & 0 & $0$\\ \hline
$R_{cn2}$    &  $\min\{n_d-n_c,n_s-n_d\}$ & 0 & $n_r-(n_c+n_r-n_d)^+$\\ \hline
$R_{p1}$     &  $(2n_d-n_c-n_s)^+$ & 0 & $(n_d-n_r-n_c)^+$  \\ \hline
$R_{cn3}$    &  $(n_c+n_s-2n_d)^+$ & $(3n_c-2n_d)^+$ & $(n_c+n_r-n_d)^+$  \\ \hline
$R_{p2}$     & \footnotesize{$\min\{n_c-R_{cm2}-R_{cn3},n_d-n_c-R_{p1}-R_{cn1}\}$} & \footnotesize{$(2n_c-n_d)-(3n_c-2n_d)^+$} & $(n_d-n_c)-R_{p1}$  \\ \hline
%$R_\Sigma$ & $n_s+n_d-n_c$  & $2n_c$ & $2n_d+n_r-n_c$ \\ \hline
\end{tabular}
\caption{Rate allocation parameters for the scheme \WIthree{} in the regime where $n_r\leq n_c\leq n_d \leq n_s$ and $n_r+n_d-n_c\leq n_c$. In these regimes, $\ell_1^u     =R_{cm2}-R_{cn1}$, $\ell_1^d =0$ and $\ell_1 = \ell_1^u$.}
\label{Tab:WI3_2}
\end{table*}
\begin{table*}
\centering
\begin{tabular}{|c||c|c|c|c|c|}
\hline
\multirow{3}{*}{Regime} & \multirow{2}{*}{$n_s\leq n_d-n_c$} & \multicolumn{4}{c|}{$n_c-n_r \leq n_d-n_c< n_s$} \\ \cline{3-6}
& &\multicolumn{2}{c|}{\footnotesize{$2(n_d-n_s) \leq n_c\leq 2n_r$}} &\multirow{2}{*}{\footnotesize{$n_s\leq \min \{n_d-\frac{n_c}{2},n_r+n_d-n_c\}$}}& \multirow{2}{*}{\footnotesize{$\max\{n_r+n_d-n_s,2n_r\} \leq n_c$}}  \\ \cline{3-4}
 &  & $2n_d\leq 3n_c$  & $2n_d>3n_c$  & &  \\ 
 \hline
$\ell_{cm1}$ & 0 &$n_c$ & 0  & 0 & 0 \\ \hline
$R_{cm2}$    & 0 & 0 & $\frac{n_c}{2}$ & $n_d-n_s$ & $n_c-n_r$ \\ \hline
$R_{cn1}$    & 0 & 0 & \footnotesize{$n_d-\frac{3n_c}{2}-R_{p1}$} & {$\min\{n_s-n_c,n_c+n_s-n_d\}$} & $n_r-(2n_c-n_d)^+$ \\ \hline
$R_{cn2}$    & 0 & 0 & 0  & 0 & 0\\ \hline
$R_{p1}$     & 0 &$n_d-n_c$& $(n_d-2n_c)^+$  & $(n_d-2n_c)^+$   & $(n_d-2n_c)^+$ \\ \hline
%$R_{cn3}$    & 0 & 0 & 0 & 0 & 0 \\ \hline
$R_{p2}$     & $n_d-n_c$ &0 & $\frac{n_c}{2}$& $n_c+n_s-n_d$ &  $n_r$  \\ \hline
$\ell_1^u$     & 0 &\multicolumn{4}{c|}{$(2n_c-n_d)^+$} \\ \hline
$\ell_1^d$     & $n_c$ &\multicolumn{4}{c|}{$R_{cm2}-R_{cn1}-\ell_1^u$} \\ \hline 
%$R_\Sigma$   & $2(n_d-n_c)$ & \multicolumn{2}{c|}{$2n_d-n_c$}& $2n_s$ &  $2(n_r+n_d-n_c)$ \\ \hline
\end{tabular}
\caption{Rate allocation parameters for the scheme \WIthree{} in the regime where $n_r\leq n_c\leq n_s \leq n_d$ and $ n_c \leq n_r+n_d-n_c$ or $n_s \leq n_d-n_c$. In this cases, $R_{cn3} = 0$. Hence, only scheme \WIThreea{} is used in this case.}
\label{Tab:WI3_3}
\end{table*}

\begin{table*}
\centering
\begin{tabular}{|c||c|c|c|c|}
\hline
\multirow{2}{*}{Regime} & \multicolumn{2}{c|}{{$n_r \leq n_c\leq n_s \leq n_d<\min\{2n_c-n_r,n_s+n_c\}$}}&\multicolumn{2}{c|}{$n_c \leq n_r\leq n_s \leq n_d \leq n_s+ \frac{n_c}{2}$}  \\ \cline{2-5}
& $2n_d\leq 3n_c$ & $2n_d>3n_c$ & $2n_d\leq 3n_c$ & $2n_d>3n_c$ \\ \hline
$\ell_{cm1}$ & $n_c$ & $0$ & $n_c$ & 0 \\ \hline
$R_{cm2}$    & 0 & $2n_c-n_d$ & 0 & $\frac{n_c}{2}$ \\ \hline
$R_{cn1}$    & 0 & $0$ & 0 & $\min\{n_d-\frac{3n_c}{2}, \frac{n_c}{2}\}$ \\ \hline
%$R_{cn2}$    & 0 & $0$ & 0 & 0 \\ \hline
$R_{p1}$     & $n_d-n_c$ & 0 & $(n_d-n_c)$ & $(n_d-2n_c)^+$ \\ \hline
%$R_{cn3}$    & 0 & $0$ & 0 & 0 \\ \hline
$R_{p2}$     & 0 & $n_d-n_c$ & 0 & $\frac{n_c}{2}$ \\ \hline
$\ell_1^u$     &\multicolumn{1}{c|}{$R_{cm2}-R_{cn1}$} & $n_d-n_c$  & $0$ & $(2n_c-n_d)^+$ \\ \hline 
$\ell_1^d$     &\multicolumn{1}{c|}{$0$} & $0$ & \multicolumn{2}{c|}{$R_{cm2}-R_{cn1}-\ell_1^u$} \\ \hline 
%$R_\Sigma$   & $2n_d-n_c$ & $2n_c$ & \multicolumn{2}{c|}{$2n_d-n_c$} \\ \hline
\end{tabular}
\caption{Rate allocation parameters for the scheme \WIthree{}. In this case, $R_{cn2}=0 = R_{cn3} = 0$.}
\label{Tab:WI3_3}
\end{table*}

\subsection{Scheme SI} 
Finally, we present the scheme which is optimal is strong interference regime $n_d<n_c$. The achievable sum-rate of this scheme is summarized in following proposition. 

\begin{proposition}\label{prop:RateSI}
The achievable sum-rate with the scheme SI for the IRC is given by
\begin{align}
R_\Sigma = \min\{2\max\{n_d,n_r\},\max\{n_r,n_c\}+(n_s-n_c),n_s+n_c,n_c+n_r\}, \label{eq:sum_rate_SI}
\end{align}
This proves the achievability of Theorem \ref{Theorem:capacity_LD_IRC} within strong interference regime.
Note that this sum-rate expression coincides with the upper bounds given in Lemmas \ref{Lemma:new_bounds} and \ref{Lemma:known_UB1_det}. 
\end{proposition} 
In what follows we present scheme SI. In this scheme, we use CN, CF, and DF relaying strategies in addition to private and common signaling. 
\subsubsection*{Encoding at transmitters} Suppose that Tx$1$ constructs in the $k$th channel use $\X_1[k]$ as follows
\begin{align}
\X_1[k] = \begin{bmatrix}
\U_{1,cm1}[k] \\ \U_{1,cm2}[k] \\
\U_{1,cf1}[k]\\ \boldsymbol{0}_{\ell_1} \\ \U_{1,cf2}[k]\\ \U_{1,cn1}[k-1]\oplus \U_{1,df1}[k]\\ \U_{1,cn2}[k-1]  \\ \U_{1,df2}[k] \\  \U_{1,cn}[k] \\ \boldsymbol{0}_s
\end{bmatrix},k=1,\ldots,n,
\end{align}
the signal vectors $\U_{1,cn1}[0]$, $\U_{1,cn2}[0]$, ${\U}_{1,cm1}[n]$, $\U_{1,cm2}[n]$, $\U_{1,cf1}$, $\U_{1,cf2}[n]$, $\U_{1,df1}[n]$, $\U_{1,df2}[n]$, and $\U_{1,cn}[n]$ are zero vectors. 
The signal vector $\U_{1,cn}$ is defined as $\begin{bmatrix}
\U_{1,cn1}^T & \U_{1,cn2}^T
\end{bmatrix}^T$.
Similar to previous schemes, subscripts $cm$, $cn$, $cf$, and $df$ represent common, CN, CF, and DF signal vectors, respectively. Moreover, the length of the zero vector, i.e., $\ell_1$ is set later in such a way that facilitates reliable decoding. 
Note that the length of signal vector $\U_{1,a}$ is $R_a$, where $a\in\{cm2,cf1,cf2,cn1,cn2\}$ and the length of signal vectors $\U_{1,df1}$, and $\U_{1,df2}$ are $2R_{df1}$, and $2R_{df2}$, respectively. The signal vector $\U_{1,cm1}$ with rate $R_{cm1}$ has a length of $\ell_{cm1}$. 
The DF signal vectors are generated in a similar way as in Scheme \WIone{}. Therefore, $\U_{1,df1}=\begin{bmatrix}
\tilde{\U}_{1,df1}^T & \boldsymbol{0}_{R_{df1}}^T \end{bmatrix}^T$ and $\U_{2,df1}=\begin{bmatrix} \boldsymbol{0}_{R_{df1}}^T &
\tilde{\U}_{2,df1}^T \end{bmatrix}^T$, where $\tilde{\U}_{i,df1}$, $i\in\{1,2\}$ is a signal vector which contains $R_{df1}$ information bits of Tx$i$. Note that the vectors $\U_{1,cn1}$ and $\U_{1,df1}$ have the same length, hence, we have $R_{cn1} = 2R_{df1}$. This construction needs to satisfy 
\begin{align}
\ell_{cm1} + R_{cm2} + R_{cf1} + R_{cf2} + 2R_{cn1} + 2R_{cn2} + 2R_{df2}  + \ell_1 \leq q. \label{eq:SI_1}
\end{align}
\subsubsection*{Decoding at the relay}
The relay receives the sum of the top-most $n_s$ bits sent by Tx's. Supposing that the decoding at the relay is done reliably in time slot $k-1$, the relay knows $\U_{1,cn}[k-1]\oplus \U_{2,cn}[k-1]$ in the beginning of time slot $k$. Hence, it remove the interference caused by this sum before decoding process in the $k$th channel use.
In channel use $k$, the relay decodes $\U_{r,c}[k+1] = \U_{1,c}[k] \oplus \U_{2,c}[k]$, where $c=\{cf1,cf2,df1,df2,cn1,cn2\}$. This can be done reliably as long as the source-relay link is sufficiently strong. This can be formulated as follows
\begin{align}
\ell_{cm1}+ R_{cm2} + R_{cf1} + \ell_1 + R_{cf2} + 2R_{cn1} + 2R_{cn2} +2R_{df2} \leq n_s. \label{eq:Sch4_cond1}
\end{align}
Therefore, at the end of $k$th channel use, the relay knows $\U_{r,c}[k+1]$, where $c=\{cf1,cf2,df1,df2,cn1,cn2\}$.

\subsubsection*{Encoding at the relay}
In the $k$th channel use ($2 \leq k\leq n$), the relay constructs the following signal vector 
\begin{align*}
\X_r[k]=\begin{bmatrix}
\boldsymbol{0}_{\ell_2} \\ 
\boldsymbol{u}_{r,df1}[k] \\ \boldsymbol{u}_{r,df2}[k]\\  \boldsymbol{u}_{r,cf1}[k]\\ \boldsymbol{u}_{r,cf2}[k]\\ \boldsymbol{0}_{\ell_3}\\ \boldsymbol{u}_{r,cn1}[k] \\ \boldsymbol{u}_{r,cn2}[k]\\ \boldsymbol{0}_{r}
\end{bmatrix},
\end{align*}
where $r$ is chosen so that the length of $\X_r[k]$ is $q$.
Note that $\U_{r,c}[k+1] = \U_{1,c}[k-1]\oplus \U_{2,c}[k-1]$, with $c\in\{df1,df2,cf1,cf2,cn1,cn2\}$. 
It is worth mentioning that $\U_{1,df1}[k] \oplus \U_{2,df1}[k] = \begin{bmatrix}
\tilde{\U}_{1,df1}^T[k] & \tilde{\U}_{2,df1}^T[k]
\end{bmatrix}^T$ and $\U_{1,df2}[k] \oplus \U_{2,df2}[k] = \begin{bmatrix}
\tilde{\U}_{1,df2}^T[k] & \tilde{\U}_{2,df2}^T[k]
\end{bmatrix}^T$.
\subsubsection*{Decoding at the receiver side} Here, we present the decoding only for Rx$1$ since decoding at Rx$2$ is similar. Rx1 waits until the end of $n$th channel use. Next, it starts with the backward decoding.  Supposing that the decoding the received signal vector in the $n$th channel use is done successfully, Rx$1$ obtains
\begin{itemize}
\item $\U_{r,df1}[n] \rightarrow \tilde{\U}_{1,df1}[n-1], \tilde{\U}_{2,df1}[n-1]$
\item $\U_{r,df2}[n] \rightarrow \tilde{\U}_{1,df2}[n-1], \tilde{\U}_{2,df2}[n-1]$
\item $\U_{r,cf1}[n]$
\item $\U_{r,cf2}[n]$
\item $\U_{1,cn}[n-1]$.
\end{itemize}
Next, Rx$1$ starts decoding $\Y_1[n-1]$. It decodes first $\U_{1,cm1}[n-1]$ and $\U_{2,cm1}[n-1]$ as in the MAC while ignoring the remaining signals. This can be done successfully as long as the relay signal does not cause any interference. Hence, we write
\begin{align}
n_r-\ell_2 \leq n_c -\ell_{cm1} \quad \text{ if } 0<\ell_{cm1}. \label{eq:SI_5}
\end{align}
For decoding these common signal vectors, we consider an IC with the channels $n_{c,IC}= \ell_{cm1}$ and $n_{d,IC} = n_d - n_c + \ell_{cm1}$. The common signal vectors $\U_{1,cm1}$ and $\U_{2,cm1}$ can be decoded reliably as long as their length does not exceed $\min \lbrace \frac{n_{c,IC}}{2},n_{d,IC}\rbrace$. Hence, we write
\begin{align}
R_{cm1} = \min \left\lbrace\frac{\ell_{cm1}}{2},(n_d-n_c+ \ell_{cm1})^+\right\rbrace.
\end{align}
After removing the common signal vectors $\U_{1,cm1}$ and $\U_{2,cm1}$ from the received signal, Rx$1$ observes a superposition of $\X_1'$, $\X_2'$, and $\X_r'$ shown in Fig. \ref{Fig:B_4_Rx1_after_remov_Comm}, where we define
\begin{align*}
n_d^\prime &= (n_d - \ell_{cm1})^+ \\ 
n_c^\prime &= n_c - \ell_{cm1} \\ 
%n_s^\prime &= (n_s - \ell_{cm1})^+ \\ 
n_r^\prime &= n_r-\ell_2.
\end{align*}
%Due to the condition in \eqref{eq:SI_5} and since we focus on the case that $n_d<n_c\leq n_s$, we obtain $q' = n_s$.

\begin{figure}[t]
\centering
\begin{tikzpicture}[->,>=stealth',shorten >=1pt,auto,node distance=3cm,thick,scale=1]
\YprimeSchemeFour
\end{tikzpicture}
\caption{The received signal vector at Rx1 in the $k$th channel use ($2<k<n-1$) after removing the signal vectors ${\U}_{1,cm1}[k]$ and ${\U}_{2,cm2}[k]$. Note that $\U_{2,cn1}$ and $\U_{2,cn2}$ represent $\U_{2,cn1}[k-1]$ and $\U_{2,cn2}[k-1]$, respectively. Moreover, the time index of all remaining signals is $[k]$. }
\label{Fig:B_4_Rx1_after_remov_Comm}
\end{figure}
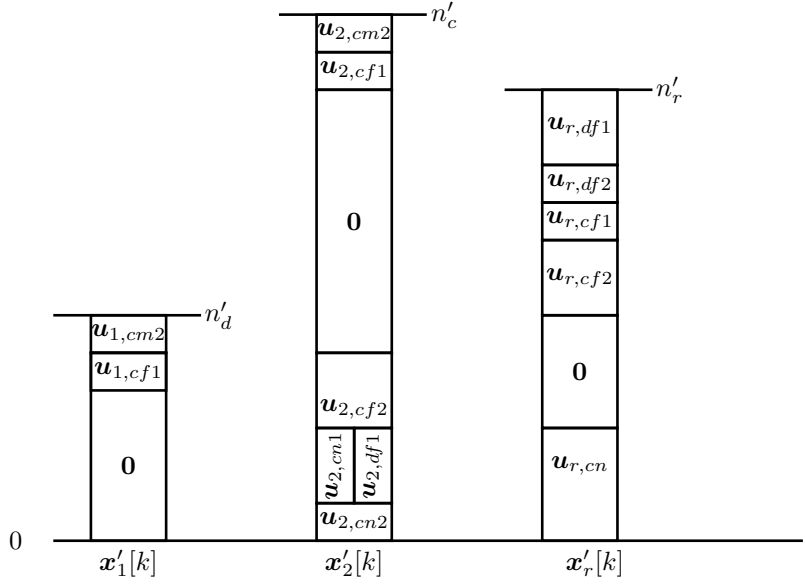
Next, Rx$1$ removes the interference caused by ${\U}_{1,df1}[n-1]$, $\U_{2,df1}[n-1]$, ${\U}_{1,df2}[n-1]$, and ${\U}_{2,df2}[n-1]$, since they are known at Rx$1$ due to the backward decoding. 
Then, it decodes $\U_{2,cm2}[n-1]$ and $\U_{2,cf1}[n-1]$. To do this following constraints are required
\begin{align}
R_{cm2} + R_{cf1} & \leq (n_c^\prime-n_r^\prime)^+ \label{eq:sch4_cond_3} \\ 
R_{cm2} + R_{cf1} & \leq n_c^\prime - n_d^\prime. \label{eq:sch4_cond_4}
\end{align}
Now, Rx$1$ constructs $\U_{1,cf1}[n-1]$ by adding $\U_{r,cf}[n]$ and $\U_{2,cf1}[n-1]$ which are both known at Rx$1$. Then it removes the interference caused by $\U_{1,cf1}[n-1]$ from the received signal. 

Next, Rx$1$ decodes $\U_{r,df1}[n-1]$, $\U_{r,df2}[n-1]$, $\U_{r,cf1}[n-1]$, and $\U_{r,cf2}[n-1]$. This can be done successfully as long as 
\begin{align}
\underbrace{2R_{df1}}_{R_{cn1}} + 2R_{df2} + R_{cf1} + R_{cf2} &\leq (n_r' -n_d')^+ \quad \text{ if } 0<R_{cm2} \\ 
2R_{df1} + 2R_{df2} + R_{cf1} + R_{cf2} & \leq [n_r' - (n_c' - R_{cm2} - R_{cf1} - \ell_1 )]^+ . \label{eq:SI_9}
\end{align}
\begin{remark}
Note that for the case that $R_{cm2} = 0$, an overlap between the relay CF signals and $\U_{1,cf2}[n-1]$ is avoided using the condition in \eqref{eq:SI_9} and since $n_d'<n_c'$.
\end{remark}
Next, Rx$1$ decodes $\U_{1,cm2}[n-1]$. To do this the following constraint needs to be satisfied
\begin{align}R_{cm2} & \leq [n_d' - (n_c' - R_{cm2} - R_{cf1} - \ell_1 )]^+.
\end{align}
Since $n_d'<n_c'$, Rx$1$ decodes the top-most bit of $\U_{2,cf2}[n-1]$ which is received without any interference from $\U_{1,cf2}[n-1]$. Then, it makes modulo 2 sum of this bit with the top-most bit of $\U_{r,cf}[n]$ (this is already known from decoding in $n$th channel use) to construct the top-most bit of $\U_{1,cf2}[n-1]$. Then it removes the interference caused by this bit. Rx$1$ repeats this decoding process until the whole signal vector $\U_{2,cf2}[n-1]$ is decoded. Therefore, Rx$1$ obtains $\U_{1,cf2}[n-1]$. 

Next, Rx$1$ decodes the top-most bit of $\U_{2,cn1}[n-2]\oplus\U_{r,cn1}[n-1]$. Again this bit is received on higher level than $\U_{1,cn1}[n-2]$ (which is sent from Tx$1$) since $n_d'<n_c'$. Hence, Rx$1$ can decode  the top-most bit of $\U_{1,cn1}[n-2]=\U_{2,cn1}[n-2]\oplus\U_{r,cn1}[n-1]$ and remove the interference caused this bit which is also sent by Tx$1$. Then it decodes the remaining bits of $\U_{2,cn1}[n-2]\oplus\U_{r,cn1}[n-1]$ similarly. Doing this, the whole CN signal vector is decoded as long as the CN signal vector from Tx$2$ is received at Rx$1$ and this is aligned with that of the relay. Hence, this constraint needs to be satisfied
\begin{align}
n_c' - R_{cm2} - R_{cf1} - \ell_1-R_{cf2}-R_{cn1} -R_{cn2} &\geq  0 \\ 
R_{cn1}+R_{cn2} &\leq n_r'.
\end{align}

To guarantee that the interference from $\U_{2,cn}[n-1]$ is not received at Rx$1$, we have
\begin{align}
n_c' - R_{cm2} - R_{cf1} - \ell_1-R_{cf2}-R_{cn1} -R_{cn2} - 2R_{df2} &\leq  0. \label{eq:SI_14}
\end{align}
At this step, we can set the parameters $\ell_2$ and $\ell_3$ as follows 
\begin{align}
%\ell_2 &= (n_r-n_c' + R_{cm2}+R_{cf1})^+ \\ 
\ell_2 &= [n_r-(n_c'-\ell_1+R_{cf2}+2R_{df1}+ 2R_{df2})]^+ \\ 
\ell_3 &= n_r' -2R_{df1} - 2R_{df2}-R_{cf1}-R_{cf2} -R_{cn1} - R_{cn2} .
\end{align}
By using this scheme, we achieve the sum-rate  
\begin{align}
n R_{\Sigma,\text{SI}} = 2(n-1)(R_{cm1} + R_{cm2} + R_{cf1}+ R_{df1}+R_{df2}+R_{cf2}+R_{cn1}+R_{cn2}).
\end{align}
Dividing the expression by $n$ and letting $n\to \infty$, we obtain the following sum-rate
\begin{align}
R_{\Sigma,\text{SI}} = 2(R_{cm1} + R_{cm2} + R_{cf1}+ R_{df1}+R_{df2}+R_{cf2}+R_{cn1}+R_{cn2}).
\end{align}
This sum-rate has to be maximized under the constraints in \eqref{eq:SI_1}-\eqref{eq:SI_14}. This is formulated as follows
\begin{align}
\max & \quad  R_{\Sigma,\text{SI}}  \label{eq:Scheme4_Opt} \\
\text{s.t.} & \quad \text{ \eqref{eq:SI_1}-\eqref{eq:SI_14} are satisfied} \notag\\
& \quad R_{cm1} , R_{cm2} , R_{cf1}, R_{df1},R_{df2},R_{cf2},R_{cn1},R_{cn2}, \ell_1, \ell_2 ,\ell_3\geq 0 \notag
\end{align}
The optimal parameters are given in Table \ref{Tab:SI_1} and \ref{Tab:SI_2}. Using these parameters, the sum-rate given in \eqref{eq:sum_rate_SI} is achieved. This shows the achievability of Theorem \ref{Theorem:capacity_LD_IRC} for the strong interference regime ($n_d<n_c$).

\begin{table*}
\centering
\begin{tabular}{|c||c|c|c|c|c|}
\hline
\multirow{2}{*}{Regime} & \multicolumn{3}{c|}{$n_c \leq n_r$} & \multicolumn{2}{c|}{$n_d\leq n_r \leq n_c$}   \\ \cline{2-6}
& $\max\{2n_c,n_r\}\leq n_s$ & $\max\{n_r,n_s\} \leq 2n_c$ & $\max\{n_s,2n_c\} \leq n_r$ & $n_r\leq \frac{n_s}{2}$ & $n_r >\frac{n_s}{2}$  \\ \hline
$\ell_{cm1}$ & 0 & 0 & 0 & 0 & 0\\ \hline
$R_{cm2}$ & 0 & 0 & 0 & 0 & 0\\ \hline
$R_{cf1}$ & 0 & 0 & 0 & $(n_r-n_s+n_c)^+$ & $n_c-n_r$ \\ \hline
$R_{cf2}$ & 0 & $\min\{\frac{n_r+n_c-n_s}{2},2n_c-n_s\}$ & $(2n_c-n_s)^+$ & 0 & $n_r-\frac{n_s}{2}$ \\ \hline
$R_{df1}$ & $\frac{\min\{n_r-n_c,n_c\}}{2}$ & $\frac{(n_r+n_s-3n_c)^+}{2}$ & $\min\{\frac{n_s-n_c}{2},\frac{n_c}{2}\}$ & 0 & 0 \\ \hline
$R_{df2}$ & $\frac{(n_r-2n_c)^+}{2}$ & 0& $\frac{(n_s-2n_c)^+}{2}$ & 0 & 0 \\ \hline
$R_{cn1}$ & \multicolumn{5}{c|}{ $2R_{df1}$}  \\ \hline
$R_{cn2}$ & $(2n_c-n_r)^+$ & $(n_s-n_c)-R_{cn1}$& 0 & $\min\{n_r,n_s-n_c\}$ &  $n_s-n_c$ \\ \hline
$\ell_1$ & 0 & $\frac{(3n_c-n_r-n_s)^+}{2}$ & 0 & $n_c-n_r$ & $n_c-\frac{n_s}{2}$ \\ \hline
$R_{\Sigma}$ & $n_r+n_c$ & $n_r+n_s-n_c$ & $n_s+n_c$ & $2n_r$& $n_s$ \\ \hline
\end{tabular}
\caption{Rate allocation parameters for the scheme SI when $n_d \leq n_r$. }
\label{Tab:SI_1}
\end{table*}

\begin{table*}
\centering
\begin{tabular}{|c||c|c|c|}
\hline
\multirow{2}{*}{Regime} & \multicolumn{3}{c|}{$n_r < n_d$} \\ \cline{2-4}
& $n_d \leq \min\{\frac{n_s}{2},\frac{n_r+n_c}{2}\}$  & $n_s\leq \min\{n_r+n_c,2n_d\}$  &  $n_r+n_c\leq \min\{n_s,2n_d\}$  \\ \hline
$\ell_{cm1}$ & 0  & $2n_c-n_s$  & $n_c-n_r$ \\ \hline
$R_{cm2}$ & $n_d - \min\{n_r,n_s-n_c\}$  & 0  & 0\\ \hline
$R_{cf1}$ & 0 & 0 & 0\\ \hline
$R_{cf2}$ & 0 & 0 & 0 \\ \hline
$R_{df1}$ & 0 & 0  & 0 \\ \hline
$R_{df2}$ & 0 & 0 & 0 \\ \hline
$R_{cn1}$ & \multicolumn{3}{c|}{$2R_{df1}$} \\ \hline
$R_{cn2}$ & $\min\{n_r,n_s-n_c\}$ & $n_s-n_c$ & $n_r$ \\ \hline
$\ell_1$  & $n_c-n_d$ & 0  & 0 \\ \hline
$R_{\Sigma}$ & $2n_d$ & $n_s$  & $n_r+n_c$ \\ \hline
\end{tabular}
\caption{Rate allocation parameters for the scheme SI when $n_r < n_d$. }
\label{Tab:SI_2}
\end{table*}

\subsection{Scheme II} Until now, the achievability of Theorem \ref{Theorem:capacity_LD_IRC} for the weak interference regime ($n_c<n_d$) and strong interference regime ($n_d<n_c$) has been shown. In what follows, we present a scheme which is optimal for the intermediate interference regime (II), in which $n_c=n_d$. The achievable sum-rate using this scheme is presented in the following proposition. 
\begin{proposition}\label{prop:RateII}
The achievable sum-rate with the scheme II for the IRC is given by
\begin{align}
R_\Sigma &=  \max\{n_d,\min\{n_r,n_s\}\}    \text{  if  }   n_c=n_d . \label{eq:R_sum_LB_SW} 
\end{align}
This proves the achievability of Theorem \ref{Theorem:capacity_LD_IRC} when $n_c=n_d$.
Note that this sum-rate expression coincides with the upper bounds given in Lemma \ref{Lemma:new_bounds}. 
\end{proposition} 
In what follows, we present this scheme in details.
Compared to the previous schemes, in which both transmitters send in all channel uses, in this scheme, only one Tx is active. 
One can use time division multiplexing access (TDMA) and assign the first $\frac{n}{2}$ channel uses to Tx$1$ and the second $\frac{n}{2}$ channel uses to Tx$2$, to achieve the same individual rate at both users. 
However, since in this work, we are interested in the achievable sum-rate and not in the individual rate, we explain the scheme for the case that Tx$2$ is inactive in all channel uses. 

\subsubsection*{Encoding at transmitters} Tx$1$ generates in the $k$th channel use, the following signal vector
\begin{align}
\X_1[k]=\begin{bmatrix}
\U_{1,cm}[k] \\ \U_{1,df}[k] \\ \boldsymbol{0}_s
\end{bmatrix}, \quad k=1,\ldots,n,
\end{align}
where $\U_{1,cm}$ and $\U_{1,df}$ represent the common and DF signal vectors, respectively. The length of the zero vector $s$ is chosen such that the length of $\X_1[k]$ is $q$. The length of signal vectors $\U_{1,cm}[k]$ and $\U_{1,df}[k]$ are $R_{cm}$ and $R_{df}$, respectively, where 
\begin{align}
R_{cm} & = n_d \label{eq:SchemeII_1}\\
R_{df} & = \min\{(n_s-n_d)^+,(n_r-n_d)^+\} \label{eq:SchemeII_2}
\end{align}
%The length of the signal vectors $\U_{1,cm}[k]$ and $\U_{1,df}[k]$ are $n_d$ and $\min\{(n_s-n_d)^+,(n_r-n_d)^+\}$,\footnote{Since in this scheme, Tx$2$ is off, Tx$1$ does not share any bit level in $\U_{1,df}[k]$ with Tx$2$.} respectively. 
Moreover, $\U_{1,cm}[n]$ and $\U_{1,df}[n]$ are both zero vectors. Notice that Tx$2$ is silent.
\subsubsection*{Decoding at the relay}
In channel use $k=1,\ldots,n-1$, the relay receives the top-most $n_s$ bits of $\X_1[k]$. Note that the length of $\U_{1,df}[k]$ is chosen such that the relay is able to receive all bits in $\U_{1,df}[k]$. Therefore, the relay knows $\U_{1,df}[k]$ in the $(k+1)$th channel use.
\subsubsection*{Encoding at the relay}
In channel use $k=2,\ldots,n$, the relay sends
\begin{align}
\X_r[k] =  \begin{bmatrix}
\U_{r,df}[k] \\ \boldsymbol{0}_{r} \\ 
\end{bmatrix},
\end{align} 
where $\U_{r,df}[k] = \U_{1,df}[k-1]$ and $r$ is chosen such that the length of $\X_r[k]$ is $q$.
\subsubsection*{Decoding at the receiver side}
Now, we need to show that Rx$1$ is able to decode $\U_{1,cm}[k]$ and $\U_{1,df}[k]$ for all $k=1,\ldots,n-1$. To do this, we use backward decoding. Rx$1$ waits until the end of $n$th channel use. Supposing that decoding in the $n$th channel use is done successfully, Rx$1$ obtains $\U_{r,df}[n]$. Therefore, it knows $\U_{1,df}[n-1]$ before it starts processing the received signal vector in the $(n-1)$th channel use. In the $(n-1)$th channel use, Rx$1$ receives 
\begin{align}
\Y_1[n-1] &= \bS^{q-n_d} \X_{1}[n-1] \oplus \bS^{q-n_r} \X_{r}[n-1]\\
%\Y_1[n-1] 
&= \begin{bmatrix} \boldsymbol{0}_{q-\max\{n_d,n_r\}}\\
\U_{r,df}[n-1] \\ \boldsymbol{0}_{\ell_1} \\ \U_{1,cm}[n-1]
\end{bmatrix},
\end{align}
where $\ell_1= (n_r - n_d - R_{df})^+ $. Since $ 0\leq \ell_1$, the signal vectors $\U_{r,df}[n-1]$ and $\U_{1,cm}[n-1]$ are received without an overlap at Rx$1$. Hence, Rx$1$ decodes both these signal vectors and obtains $\U_{1,df}[n-2]$ and $\U_{1,cm}[n-1]$. Doing this decoding backward until the first channel use, Rx$1$ receives $\U_{1,cm}[k]$ and $\U_{1,df}[k]$ for all $k=1,\ldots,n-1$. Therefore, by using this scheme, we achieve
\begin{align}
R_{\Sigma,II} &= n_d+\min\{(n_s-n_d)^+,(n_r-n_d)^+\} \\ 
&= \max\{n_d, \min \{n_r,n_s\}\}
\end{align}
which shows the achievable sum-rate given in Proposition \ref{prop:RateII}. 
This scheme completes the achievability of Theorem \ref{Theorem:capacity_LD_IRC} together with schemes \WIone, \WItwo, \WIthree, and SI. 

Until now, we characterized the sum-capacity for the LD-IRC when $n_c<n_s$. The rest of the paper is dedicated to the GDoF characterization for the Gaussian case. 

\section{Upper Bounds for the Gaussian IRC}
\label{sec:UBGaussian-IRC}
In this section, we prove the converse of Theorem \ref{Theorem:GDoF_IRC}. To do this, we use the insights obtained from bounding the capacity of the LD-IRC in Section \ref{sec:UBLD-IRC} to establish the upper bounds for the GDoF of the Gaussian case. 
Before presenting the upper bounds on the capacity of Gaussian IRC, we present a lemma which will be required for establishing one of the upper bounds. 
\begin{lemma}\label{Lemma:new_UB4_2_Gauss}
If $ \Gamma^n = \frac{h_c}{\sqrt{Ph_r^2}}X_i^n+U_i^n$ and $\Delta^n=h_cX_i^n+h_rX_r^n+Z_j^n$, where $i,j\in\{1,2\}$, $i\neq j$, and $U_i\sim \mathcal{N}(0,1)$ is i.i.d. over the time and independent from all other random variables, then $h(\Gamma^n)-h(\Delta^n|W_j)$ is upper bounded as follows
\begin{align}
h\left(\Gamma^n\right) - h\left(\Delta^n|W_j\right) \leq n C\left(2+ \frac{h_c^2}{(h_c-h_r)^2}\right).
\end{align}
\end{lemma}
\begin{proof}
The proof is given in Appendix \ref{app:proof_Lemma_additional}.
\end{proof}
Now, we present the upper bounds on the capacity of the Gaussian IRC in the following lemma.
\begin{lemma}
\label{Lemma:UB_Gauss}
The GDoF of the Gaussian IRC is upper bounded by 
\begin{align}
\label{eq:new_UB1_IRC_Gauss}
d \leq & \max\{1,\min\{\beta,\gamma\}\}  \quad \text{ if } \alpha = 1\\
\label{eq:new_UB2_IRC_Gauss}
d \leq & \max\{1,\alpha,\beta\} + \max\{1,\alpha\}\\
\label{eq:new_UB3_IRC_Gauss}
d \leq & \beta+2\max\{1,\alpha\}-\alpha\\
\label{eq:new_UB4_IRC_Gauss}
d \leq & \max\{1,\alpha\}+\max\{1,\gamma\}\\
\label{eq:new_UB5_IRC_Gauss}
d \leq & 2\max\{\alpha,\beta,1-\max\{\alpha,\gamma\}\}+2(\gamma-\alpha)^+\\
\label{eq:new_UB6_IRC_Gauss}
d \leq & 2\max\{\alpha,\beta+1-\alpha\}  \quad \text{ if } \beta \leq \alpha \leq 1.
\end{align}
\end{lemma}
\begin{proof}
%The proof of this Lemma is given in Appendix \ref{app:known_UB_Gauss_proof}.
While the first two upper bounds are cut-set bounds, the remaining bounds are established by using genie-aided methods. The bounds given in \eqref{eq:new_UB4_IRC_Gauss} and \eqref{eq:new_UB5_IRC_Gauss} are inspired by similar bounds presented in \cite[Theorems 3,4]{ChaabanSezgin_IT_IRC} which are tightened for the case where $\alpha<\gamma$. The complete proof of this lemma is given in Appendix \ref{app:known_UB_Gauss_proof}.
\end{proof}

In addition to the upper bounds in Lemma \ref{Lemma:UB_Gauss}, some upper bounds are borrowed from \cite{ChaabanThesis,ChaabanSezgin_IT_IRC}. In the following lemma, we present these bounds. 
\begin{lemma}
\label{Lemma:known_UB1_Gauss}
(\cite{ChaabanSezgin_IT_IRC}) The GDoF of the Gaussian IRC is upper bounded by 
\begin{align}
d & \leq 2 \max\{1,\beta\} \label{eq:old_UB1_IRC_Gauss} \\ 
d & \leq \max\{1,\beta,\alpha\}+\max\{1,\alpha\}-\alpha+(\gamma-\max\{1,\alpha\})^+. \label{eq:old_UB2_IRC_Gauss}
\end{align}
\end{lemma}
While the first upper bound is a cut-set bound, the second upper bound is derived by using the genie-aided method. The proof of these bounds are given in \cite{ChaabanSezgin_IT_IRC}.

Now, we need to show that the minimum of the upper bounds in \eqref{eq:new_UB1_IRC_Gauss}-\eqref{eq:old_UB2_IRC_Gauss}  coincide with the GDoF expression in Theorem \ref{Theorem:GDoF_IRC}. This can be shown similar to the linear deterministic case by keeping in mind that the channel parameters $n_d$, $n_c$, $n_r$, and $n_s$ in the LD-IRC are equivalent to $1$, $\alpha$, $\beta$, and $\gamma$ in the Gaussian IRC, respectively.

%%%%%%%%%%%%%%%%%%%%%%%%%%%%%%%%%%%%%%%%%%%%%%%
%%%%%%%%%%%%%%%%%%%%%%%%%%%%%%%%%%%%%%%%%%%%%%%
%%%%%%%%%%%%%%%%%%%%%%%%%%%%%%%%%%%%%%%%%%%%%%%
%%%%%%%%%%%%%%%%%%%%%%%%%%%%%%%%%%%%%%%%%%%%%%%
%%%%%%%%%%%%%%%%%%%%%%%%%%%%%%%%%%%%%%%%%%%%%%%

\section{GDoF Achieving Schemes}
\label{sec:GDoFAchievingSchemes}
In this section, we show the achievability of the GDoF in Theorem \ref{Theorem:GDoF_IRC}. This will be done by extending the achievablity schemes presented for the LD-IRC to the Gaussian case in a similar manner as in \cite{JafarVishwanath,HuangJafarCadambe}. To do this, we decompose the Gaussian channel into $N$ sub-channels which can be accessed at the receiver side successively by using successive decoding. 
This decomposition reduces the rate allocation problem to a sub-channel allocation problem which can be solved as in the LD-IRC.
 In what follows, we present the idea of the channel decomposition for the point-to-point (P2P) channel. Then, we present the transmission scheme for the Gaussian IRC. At the end, we present the the strategies used in the Gaussian IRC over the sub-channels.
\subsubsection{Point-to-point channel}
Consider a received signal over a point-to-point channel in $n$ channel uses $y^n=x^n+z^n$, where $x$, $y$, and $z$ are the inputs with a power constraint $P$ ($1<P$), output, and AWGN with unit variance, respectively. By decomposing the channel into $N$ sub-channels, the received signal $y$ can be rewritten as $y=\sum_{\ell=1}^N x_\ell + z$, where the power of $x_\ell$ is $\delta^\ell - \delta^{\ell-1}$ and its rate is $R_\ell$, where $\log \delta = \frac{1}{N}\log P$. Note that the power constraint is satisfied since $\sum_{\ell=1}^N \delta^\ell - \delta^{\ell-1} = P-1 < P$. 
Notice that the signal in the $\ell$th sub-channel is received on a higher power level than in the $(\ell-1)$th sub-channel.
Therefore, by doing successive decoding at the receiver, the receiver decodes $x_\ell^n$ while $x_{\ell-1}^n$, $\ldots$, $x_1^n$ are treated as noise. Hence, the following rate is achievable

\begin{align}
R_\ell &\leq \frac{1}{2}\log\left(1+\frac{p_\ell}{1+p_1+p_2+\ldots+p_{\ell-1}}\right) \notag \\ 
&= \frac{1}{2}\log\left(1+\frac{\delta^\ell-\delta^{\ell-1}
}{\delta^{\ell-1}}\right) \notag \\ 
&=\frac{1}{2N}\log(P).
\end{align}
Using this for all sub-channels, we obtain the sum-rate $\frac{1}{2}\log(P)$ which is approximately equal to the capacity of the P2P channel in the high SNR regime.

\subsubsection{Gaussian IRC} Now, we want to describe the transmission scheme for the Gaussian IRC. Suppose that Tx$1$ wants to send a message $W_1(b)$ to Rx$1$ in block $b$, where $b=1,\ldots,B$. To do this, Tx$1$ uses a nested-lattice codebook (\cite{NazerGastpar}, \cite{WilsonNarayananPfisterSprintson}, and \cite{ErezZamir}) $(\Lambda_{f},\Lambda_{c})$ with rate $R_s$ and unit power,
 to generate the codewords $x_{1,\ell}^n(b)$, where $\ell=1,\ldots,N$ and $\Lambda_{f}$, $\Lambda_{c}$ represent the fine, coarse lattices, respectively. The codeword $x_{1,\ell}^n(b)$ is given as follows
\begin{align}
x_{1,\ell}^n(b) = \sqrt{P_\ell}\left[\left(\lambda_{1,\ell}(b) - d_{1,\ell} \right) \mod \Lambda_{c}\right],
\end{align}
where $\lambda_{1,\ell}(b) \in \Lambda_{f}$ and $d_{1,\ell}$ is an $n$-dimensional random dither vector uniformly distributed over the fundamental Voronoi region $\mathcal{V}(\Lambda_{c})$. Note that the dither vector is assumed to be known at all nodes. 
Tx$1$ sends in the $\ell$th sub-channel $x_{1,\ell}^n(b)$. Hence, $x_1^n(b) = \sum_{\ell=1}^N x_{1,\ell}^n(b)$. 
Similar to the P2P channel, the power of $x_{1,\ell}^n(b)$ is $p_\ell = \delta^\ell - \delta^{\ell-1}$ and its rate is $R_{1,\ell}$, where $\log(\delta) = \frac{1}{N}\log(P)$. Note that, the transmit power by Tx$1$ satisfies the power constraint $P$. Tx$1$ can decide whether it sends over the $\ell$th sub-channel or not. Hence, $R_{1,\ell}\in\{0,R_s\}$, where $R_s$ represents the maximum achievable rate by using a sub-channel. The same is done by Tx$2$. 
Note that both Tx's use the same coarse and fine lattices for generating the code-words. 

Now, consider the relay side. The received signal at the relay in block $b$ is given by 
\begin{align}
y_r^n(b) = y_r^{\prime n}(b) + \sum_{\ell=1}^{N-N_s} h_s[x_{1,\ell}^n(b) + x_{2,\ell}^n(b)] + z_r^n(b),
\end{align} 
where $y_r^{\prime n}$ is the part which is received at the relay higher than the noise level. This part is the sum of the transmitted signals by both Tx's in the top-most $N_s$ sub-channels. Hence, we can write
\begin{align}
y_r^{\prime n}(b) = \sum_{\ell^\prime=1}^{N_s} y_{r,\ell^{\prime}}^{\prime n}(b)= \sum_{\ell=N-N_s+1}^N h_s[x_{1,\ell}^n(b) + x_{2,\ell}^n(b)].
\end{align}
To obtain $N_s$, consider the $(N-N_s+1)$th sub-channel. The signal in this sub-channel is received at the relay on the lowest power level which is still higher than the noise level. Therefore, we write
\begin{align}
1<\delta^{(N-N_s)}h_s^2. 
\end{align}
By solving this inequality, we obtain $N_s\leq \frac{\log(Ph_s^2)}{\log\delta}$. Since $N_s$ is the maximum number of the sub-channels received at the relay higher than the noise level, we obtain 
$N_s= \left\lfloor \frac{\log(Ph_s^2)}{\log(\delta)} \right\rfloor$.

In each block $b$, the relay decodes $y_r^{\prime n}(b)$. To do this, it decodes first the received signal in the highest sub-channel, i.e., ${y_{r,N_s}^{\prime n}(b)=h_s(x_{1,N}^n+x_{2,N}^n)}$. Hence, it decodes first the sum ${\lambda_{1,N}(b)+\lambda_{2,N}(b) \mod \Lambda_c}$ while it treats the interference caused by the lower sub-channels, i.e., $h_s(x_{1,\ell}^n(b) + x_{2,\ell}^n(b))$ ($\ell=1,\ldots,N-1$) as noise. 
After decoding ${\lambda_{1,N}(b)+\lambda_{2,N}(b) \mod \Lambda_c}$ successfully, the relay constructs ${x_{1,N}(b) + x_{2,N}(b)}$ as shown in \cite{Nazer_IZS2012}. Then it removes the interference caused by $h_s(x_{1,N}(b) + x_{2,N}(b))$. Next, it decodes $y_{r,(N_s-1)}^{\prime n}(b)$ by treating all signals received in lower sub-channels as noise. Proceeding this decoding successively, relay completes decoding $y_r^{\prime n}(b)$. Generally, the relay is able to decode the sum $\lambda_{1,\ell}(b)+\lambda_{2,\ell}(b) \mod \Lambda_c$ for all $\ell\in\{N-N_s+1,\ldots,N\}$, as long as \cite{NazerGastpar}
\begin{align}
R_{1,\ell}, R_{2,\ell}\leq R_s &\leq \frac{1}{2}\log\left(\frac{1}{2}+\frac{h_s^2p_\ell}{1+2h_s^2(p_{\ell-1}+p_{\ell-2}+\ldots+p_{1} )}\right)\\
&=\frac{1}{2}\log\left(\frac{1}{2}+\frac{h_s^2(\delta^{\ell}-\delta^{\ell-1})}{1+2h_s^2(\delta^{\ell-1}-1)}\right). \label{eq:LB_Gaussian_3}
\end{align}
The condition in \eqref{eq:LB_Gaussian_3} is written for the worst case which is the case when both transmitters share all the sub-channels. Suppose that the $\ell$th sub-channel is used only by one of the transmitters or some sub-channels from the first to the $(\ell-1)$th one are not used by both Tx's. 
Then, the rate constraint will be looser than that of in \eqref{eq:LB_Gaussian_3}. 
Now, using the fact that $1-2h_s^2< 1 \leq h_s^2\delta^{\ell-1}$ for all $\ell\in\{N-N_s+1,\ldots,N\}$, we tighten the condition in \eqref{eq:LB_Gaussian_3} as follows
\begin{align}
R_{1,\ell}, R_{2,\ell}\leq R_s &\leq \frac{1}{2}\log\left(\frac{1}{3}+\frac{h_s^2(\delta^{\ell}-\delta^{\ell-1})}{3h_s^2 \delta^{\ell-1}}\right) \\
&= \frac{1}{2}\log\left(\frac{\delta}{3}\right).
 \label{eq:LB_Gaussian_4}
\end{align}
After decoding the received signal in block $b$, the relay generates $x_{r,\ell}^n(b+1)$, $\ell=1,\ldots,N$ which is sent over the $\ell$th sub-channel in the $(b+1)$th block. This signal is generated by using the nested-lattice codebook $(\Lambda_f,\Lambda_c)$ as follows
\begin{align}
x_{r,\ell}^n(b+1) = \sqrt{P_{r,\ell}} \left[(\lambda_{r,\ell}(b+1)-d_{r,\ell})\mod \Lambda_c\right],
\end{align}
where $\lambda_{r,\ell}(b+1) \in \Lambda_{f}$ and $d_{r,\ell}$ is an $n$-dimensional random dither vector uniformly distributed over the fundamental Voronoi region $\mathcal{V}(\Lambda_{c})$.
The power and the rate of $x_{r,\ell}^n(b+1)$ is $p_{r,\ell}\leq p_\ell$ and $R_{r,\ell}$, respectively. Similar to encoding at the transmitter side, relay can decide whether it uses the $\ell$th sub-channel or not. Therefore, $R_{r,\ell}\in \{0,R_s\}$. The sent signal in the $(b+1)$th block by the relay is $x_{r}^n(b+1) = \sum_{\ell=1}^N x_{r,\ell}^n(b+1)$. The transmit power of the relay is given by $P^\prime = \sum_{\ell=1}^N p_{r,\ell}\leq P$.

Now, consider the decoding at the receiver side.  Here, we explain the decoding at Rx$1$. The same is done by Rx$2$. Rx$1$ starts with decoding at the end of block $B$. 
In the block $B$, Rx$1$ receives $y_1^n(B) = h_dx_1^n(B) + h_cx_2^n(B) + h_r x_r^n(B)+z_1^n(B)$. The signal which is received at Rx$1$ higher than noise power level is given by
\begin{align}
y_1^{\prime n}(B) = \sum_{\ell^\prime=1}^{N_m} y_{1,\ell^{\prime}}^{\prime n}(B)=
\sum_{\ell=N-N_d+1}^{N} h_d x_{1,\ell}^n(B) + \sum_{\ell=N-N_c+1}^{N} h_c x_{2,\ell}^n(B) + 
\sum_{\ell=N-N_r+1}^{N} h_rx_{r,\ell}^n(B),
\end{align}
where $N_d$, $N_c$, and $N_r$ are the number of sub-channels which are received at Rx$1$ from Tx$1$, Tx$2$, and the relay higher than the noise level, respectively. Moreover, $N_m$ is number of sub-channels which are observed at Rx$1$ higher than the noise level. Therefore, we obtain $N_d = \left\lfloor \frac{\log(Ph_d^2)}{\log(\delta)} \right\rfloor$, $N_c = \left\lfloor \frac{\log(Ph_c^2)}{\log(\delta)} \right\rfloor$, $N_r = \left\lfloor \frac{\log(Ph_r^2)}{\log(\delta)} \right\rfloor$, and $N_m=\max\{N_d,N_c,N_r\}$. 
This can be shown similar to obtaining $N_s$.
\begin{remark}
To guarantee that the sub-channels used by both Tx's are aligned at Rx$1$, we choose number of sub-channels $N$ such that it exists an $\ell\in\{1,\ldots,N\}$ where
\begin{align}
\begin{cases}
Ph_c^2 = h_d^2 \delta^\ell &\quad \text{if } h_c^2<h_d^2 \\
Ph_d^2 = h_c^2 \delta^\ell &\quad \text{if } h_d^2<h_c^2 
\end{cases}.
\end{align}
For aligning the sub-channels used by the relay and Tx's at Rx$1$, the relay needs to reduce its transmit power to $P^\prime$ given as follows
\begin{align}
P^\prime = \frac{\delta^{N_r}}{h_r^2} \leq P.
\end{align}
Notice that reducing the transmit power of the relay from $P$ to $P^\prime$ does not change the number of sub-channels received from relay at the Rx's over the noise level, i.e., $N_r=\left\lfloor \frac{\log(P'h_r^2)}{\log(\delta)}\right\rfloor$.
\end{remark}

Decoding the received signal in block $B$ is done at Rx$1$ in a successive manner. This is started with decoding $y_{1,N_m}^{\prime n}(B)$. After doing this, Rx$1$ removes the interference caused by $y_{1,N_m}^{\prime n}(B)$ and decodes $y_{1,N_m-1}^{\prime n}(B)$. This successive decoding is proceeded until end of decoding $y_{1,1}^{\prime n}(B)$. 
To write the rate constraint for successful decoding of $y_{1,\ell^\prime}^{\prime n}(B)$ for all $\ell^\prime \in  \{1,\ldots,N_m\}$, we consider the worst case which can occur. This is when for all $\ell^\prime \in \{1,\ldots,N_m\}$, $y^{\prime n}_{1,\ell^\prime}(B)$ contains three signals which are from Tx$1$, Tx$2$, and the relay. 
Suppose that Rx$1$ wants to decode $y^{\prime n}_{1,\ell^\prime}(B) = h_d x^n_{1,\ell_1}(B)+h_c x^n_{2,\ell_2}(B)+h_r x^n_{r,\ell_r}(B)$, where $ \ell^\prime \in\{1,\ldots,N_m\}$ and $\ell_1$, $\ell_2$, and $\ell_r$ are the index of the sub-channels used by Tx$1$, Tx$2$, and the relay which are received aligned at Rx$1$. Therefore, $h_d^2p_{\ell_1} = h_c^2 p_{\ell_2}=h_r^2p_{r,\ell_r} = p_{\ell^\prime}$, where $p_{\ell^\prime} = \delta^{\ell^\prime}-\delta^{\ell^\prime-1}$. Rx$1$ is able to decode ${h_d \lambda_{1,\ell_1}(B) + h_c \lambda_{2,\ell_2}(B) + h_r \lambda_{r,\ell_r}(B) \mod \Lambda_c}$ as long as \cite{NazerGastpar}
{\begin{align}
R_{\ell_1}, R_{\ell_2}, R_{r,\ell_r}  \leq R_s & \leq \frac{1}{2}\log\left(\frac{1}{3}+\frac{p_{\ell^\prime}}{1+3(p_{1} + p_2 + \ldots, p_{\ell^\prime-1}) + 3}\right) \\
& = \frac{1}{2}\log\left(\frac{1}{3}+\frac{\delta^{\ell^\prime}-\delta^{\ell^\prime-1}}{1+3(\delta^{\ell^\prime-1}-1) + 3}\right) \\
&=\frac{1}{2}\log\left(\frac{1}{3}+\frac{\delta^{\ell^\prime}-\delta^{\ell^\prime-1}}{1+3\delta^{\ell^\prime-1}}\right) \label{eq:LB_Gaussian_11}
\end{align}}
Since for all $\ell^\prime \in \{1,\ldots,N_m\}$, $1\leq \delta^{\ell^\prime-1}$, we tighten the condition in \eqref{eq:LB_Gaussian_11} and obtain 
\begin{align}
R_{\ell_1}, R_{\ell_2}, R_{r,\ell_r}  \leq R_s & \leq \frac{1}{2}\log\left(\frac{1}{4}+\frac{\delta^{\ell^\prime}-\delta^{\ell^\prime-1}}{4\delta^{\ell^\prime-1}}\right)\\
&= \frac{1}{2}\log\left(\frac{\delta}{4}\right).
 \label{eq:LB_Gaussian_12}
\end{align}
By considering both conditions in \eqref{eq:LB_Gaussian_4} and \eqref{eq:LB_Gaussian_12}, we conclude that the maximum achievable rate using one sub-channel is given by 
\begin{align}
R_s & =\frac{1}{2}\log\left(\frac{\delta}{4}\right). \label{eq:LB_Gaussian_13}
\end{align}
Note that $\delta$ has to be larger than $4$, which is equivalent to $4<P^\frac{1}{N}$. This is always satisfied when $P\to \infty$.

After decoding $y_1^n(B)$, Rx$1$ decodes the received signal in the block $(B-1)$, i.e., $y_1^n(B-1)$. Due to the backward decoding, Rx$1$ knows parts of $y_1^n(B-1)$ a priori. Hence, it removes first the interference caused by these parts before decoding $y_1^n(B-1)$. Then, it starts decoding $y_1^n(B-1)$ similar to $y_1^n(B)$. The backward decoding proceeds until the end of decoding $y_1^n(1)$.

Now, we discuss different strategies which can be used over each sub-channel. Consider the $\ell^\prime$th sub-channel at the receiver side in block $b$, i.e. $y^{\prime n}_{1,\ell^\prime}(b) = h_d x^n_{1,\ell_1}(b)+h_c x^n_{2,\ell_2}(b)+h_r x^n_{r,\ell_r}(b)$. After cancelling the interference caused by the a priori known signals (known due to the backward decoding or successive interference cancellation), this sub-channel can be used only for one of the following cases.
\begin{itemize}
\item Common signaling: Let suppose that Tx$1$ uses the $\ell$th sub-channel for sending the common signal. Since both receivers have to be able to decode this signal, this sub-channel has to be received at both receivers over the noise level. Therefore, we have
\begin{align}
N-\min\{N_d,N_c\}<\ell.
\end{align}
\item Private signaling: Compared to common signal, only the desired Rx needs to be able to decode the private signal. Therefore, if Tx$i$ ($i\in\{1,2\}$) sends over $\ell$th sub-channel a private signal, then the following condition needs to be satisfied for reliable decoding of the private signal at Rx$i$.
\begin{align}
N- N_d<\ell.
\end{align}
\item CF signaling: In CF signaling, relay is also involved in the communication. This is done as follows. Suppose that Tx's use $\ell$th sub-channel for transmitting the CF signal. This sub-channel must be in the top-most $N_s$ sub-channels. Unless the relay does not observe this sub-channel over the noise level. Therefore, we write
\begin{align}
N-N_s<\ell.
\end{align}
Using nested lattice code, the relay decodes in block $b$ the sum $\lambda_{1,\ell}(b) + \lambda_{2,\ell}(b)\mod \Lambda_c$. 
Next, the relay encodes this sum into $x_{r,\ell_r}(b+1)$ and sends it over $\ell_r$th sub channel in the next block. At the receiver side, Rx$1$ receives in block $b$, over sub-channel $N_d-(N-\ell)$, $N_c-(N-\ell)$, and $N_r-(N-\ell_r)$ the CF signal from Tx$1$, Tx$2$, and the relay, respectively. 
As we mentioned, before Rx$1$ starts decoding with the last block to the first one. Suppose that the decoding of received signal in block $B,B-1,\ldots, b+1$ is done successfully. 
Therefore, Rx$1$ knows $\lambda_{1,\ell}(b)+\lambda_{2,\ell}(b)\mod\Lambda_c$, since this sum is sent by the relay in the $(b+1)$th block. 
Using this sum, Rx$1$ can obtain the CF signals from both Tx's if it decodes either $\lambda_{1,\ell}(b)$ or $\lambda_{2,\ell}(b)$.
Depending on the channel strength, Rx$1$ decodes the CF signal which is received in the higher sub-channel and reconstructs the other one. 
For instance, suppose that the desired channel is stronger than the undesired channel ($N_c<N_d$). Then, Rx$1$ obtains $\lambda_{1,\ell}(b)$. Using $\lambda_{1}(b)+\lambda_{2}(b)\mod\Lambda_c$ which is known at the Rx$1$, it obtains $\lambda_{2,\ell}(b)$. Knowing $\lambda_{2,\ell}(b)$, Rx$1$ reconstructs $x_{2,\ell}(b)$ and removes the interference caused by $x_{2,\ell}(b)$.
Moreover, Rx$1$ decodes the relay CF signal sent in channel use $b$. This decoding can be done reliably as long as
\begin{align}
N-N_r&<\ell_r \label{eq:LB_Gaussian_17} \\ 
N-\max\{N_d,N_c\}&<\ell, \label{eq:LB_Gaussian_18}\\
N_c &\neq N_d \label{eq:LB_Gaussian_19} \\ 
  \max\{N_d,N_c\}+ \ell & \neq N_r + \ell_r.\label{eq:LB_Gaussian_20}
\end{align}
While conditions \eqref{eq:LB_Gaussian_17} and \eqref{eq:LB_Gaussian_18} guarantee that the CF signals sent by the relay and the Tx with stronger channel are received higher than the noise level at Rx$1$, conditions \eqref{eq:LB_Gaussian_19} and \eqref{eq:LB_Gaussian_20} avoid an overlap between the CF sub-channels of the Tx's and the CF sub-channels of the relay and the stronger Tx.

\item DF signaling: In this strategy, the relay needs to be able to decode both signals sent by Tx's separately. Therefore, Tx$1$ and Tx$2$ have to use different sub-channels for transmitting their DF signals. Suppose that Tx$1$ and Tx$2$ use the sub-channels $\ell_1$ and $\ell_2$ to send $x_{1,\ell_1}(b)$ and $x_{2,\ell_2}(b)$ in block $b$, respectively, where $\ell_1 \neq \ell_2$. Relay is able to observe both sub-channels over the noise level as long as 
\begin{align}
N-N_s< \min\{\ell_1,\ell_2\}.% \\
%N-N_s<\ell_2.
\end{align}
In next block $(b+1)$, the relay sends  $x_{1,\ell_1}(b)$ and $x_{2,\ell_2}(b)$ in sub-channels $\ell_{r1}$ and $\ell_{r2}$ ($\ell_{r1}\neq \ell_{r2}$), respectively. 
At the receiver side, Rx's use the backward decoding. 
Supposing that decoding received signal in blocks $B,B-1,\ldots,b+1$ is done successfully, Rx$1$ knows $x_{1,\ell_1}(b)$ and $x_{2,\ell_2}(b)$ since they are sent by relay in block $b+1$. Therefore, Rx$1$ removes the interference caused by these signals before decoding the received signal in block $b$. Next, Rx$1$ decodes $x_{1,\ell_1}(b-1)$ and $x_{2,\ell_2}(b-1)$. Note that these two signals are sent both by the relay. Decoding of these signals can be done successfully, as long as
\begin{align}
N-N_r< \min\{\ell_{r1},\ell_{r2}\}. %\\
%N-N_r<\ell_{r2}.
\end{align}
\item CN signaling: Using this strategy, Tx$i$ $i\in\{1,2\}$ sends $x_{i,\ell}(b)$ and $x_{i,\ell_F}(b)$ in block $b$ over sub-channels $\ell$ and $\ell_F$, respectively. It is worth mentioning that 
\begin{align}
x_{i,\ell_F}(b) = \sqrt{\frac{p_{\ell_F}}{p_\ell}} x_{i,\ell}(b+1).
\end{align}
In other words, in block $b$ the $\ell_F$th sub-channel is used for sending the same signal as in the $\ell$th sub-channel in block $(b+1)$. 
Suppose that the decoding at the relay has been done successfully in block $1$ to $b-1$. 
Hence, relay knows $x_{1,\ell}(b) + x_{1,\ell}(b)$ in beginning of block $b$. 
Therefore, in block $b$ relay removes the interference caused by $x_{1,\ell}(b) + x_{1,\ell}(b)$ and then it decodes $\lambda_{1,\ell_F}(b) + \lambda_{2,\ell_F}(b) \mod \Lambda_c$. 
This can be done successfully as long as
\begin{align}
N-N_s<\ell_F.
\end{align}
Knowing $\lambda_{1,\ell_F}(b) + \lambda_{2,\ell_F}(b) \mod \Lambda_c$ in block $b$, the relay constructs $\lambda_{1,\ell}(b+1) + \lambda_{2,\ell}(b+1) \mod \Lambda_c$. Next, the relay sends in block $b+1$ and sub-channel $\ell_r$, the following sum
\begin{align}
x_{r,\ell_r}^n(b+1) =  -\sqrt{p_{r,\ell_r}} \left[\lambda_{1,\ell}(b+1) + \lambda_{2,\ell}(b+1) \mod \Lambda_c\right],
%x_{r,\ell_r}^n(b+1) = -\sqrt{\frac{p_{r,\ell_r}}{p_\ell}}(x^n_{1,\ell}(b+1) + x^n_{2,\ell}(b+1)),
\end{align}
where $\ell_r=\ell+N_c-N_r$. This is equivalent to $P_{r,\ell_r} = P_\ell \frac{h_c^2}{h_r^2}$. The decoding at the destination is done backward. Suppose that decoding the received signal at Rx$1$ in blocks $B,B-1,\ldots, b+1$ is done successfully. Hence, Rx$1$ knows $\lambda_{1,\ell}(b+1)$ before decoding block $b$. Knowing $\lambda_{1,\ell}(b+1)$, Rx$1$ is able to reconstruct $x_{1,\ell_F}(b)$. Therefore, Rx$1$ removes first the interference of $x_{1,\ell_F}(b)$.
Depending on the channel strength, the decoding order can be changed. 

First, suppose that $N_d<N_c$. 
In this case, Rx$1$ decodes first the sum of $h_c x_{2,\ell}(b) + h_r x_{r,\ell_r}(b)$. 
Note that $\ell_r$ is chosen such that $h_c x_{2,\ell}(b)$ and $h_r x_{r,\ell_r}(b)$ are completely aligned over the same sub-channel. Therefore, over this sub-channel, Rx$1$ observes $h_c x^n_{2,\ell}(b)+h_r x_{r,\ell_r}^n(b)$. To decode $\lambda_{1,\ell}(b)$, Rx$1$ divides this sum with $h_c\sqrt{P_\ell}$ and adds the dither vector $d_{2,\ell}$ then it calculates the quantization error with respect to $\Lambda_c$. Therefore, it obtains 
\begin{align}
\left[
\left(\lambda_{2,\ell}(b) - d_{2,\ell}\right) \mod \Lambda_c - \left(\lambda_{1,\ell}(b) + \lambda_{2,\ell}(b)  \right) \mod \Lambda_c + d_{2\ell}  \right] \mod \Lambda_c  =   - \lambda_{1,\ell}(b)  \mod \Lambda_c .
\end{align}
In this way, Rx$1$ decodes $ \lambda_{1,\ell}(b)$. Knowing, $ \lambda_{1,\ell}(b)$, Rx$1$  reconstructs $x_{1,\ell}(b)$ and removes its interference. This decoding can be done successfully, as long as
\begin{align}
\ell+N_c-N_r &= \ell_r \label{eq:cond_CN_Rx1}\\ 
N-N_c &\leq \ell \quad \text{if } \quad N_d<N_c.\label{eq:cond_CN_Rx2}
\end{align}
While the condition in \eqref{eq:cond_CN_Rx1} guarantees that the relay CN signal is received over the same sub-channel as the undesired CN signal, the condition in \eqref{eq:cond_CN_Rx2} guarantees that these two aligned signals are received above the noise level.

Now, suppose that $N_c<N_d$. In this case, Rx$1$ decodes first $\lambda_{1,\ell}(b)$ (from $x_{1,\ell}(b)$) as in the P2P channel and then it removes the interference caused by $h_c x^n_{2,\ell}(b)+h_r x_{r,\ell_r}^n(b)$ which is observed in the sub-channel $N_c-(N-\ell)$. 
This interference cancellation is done by dividing the signal in sub-channel $N_c-(N-\ell)$ by $h_c\sqrt{P_\ell}$, adding $\lambda_{1,\ell}+d_{2,\ell}$ to it and finally calculating its quantization error with respect to $\Lambda_c$. This is given as follows.
\begin{align}
\left[
\left(\lambda_{2,\ell}(b) - d_{2,\ell}\right) \mod \Lambda_c - \left(\lambda_{1,\ell}(b) + \lambda_{2,\ell}(b)  \right) \mod \Lambda_c + \lambda_{1,\ell} + d_{2\ell}  \right] \mod \Lambda_c  =   0.
\end{align}
In this case ($N_c<N_d$), following conditions need to be satisfied.
\begin{align}
\ell+N_c-N_r &= \ell_r\\ 
N-N_d & \leq \ell \quad \text{if } \quad N_c<N_d.
\end{align}
\begin{remark}
Note that when $N_c = N_d$, CN signaling cannot be used. This is due to the fact that the relay CN signal neutralizes both undesired and desired CN signals. 
\end{remark}
\end{itemize}
The schemes which are explained in the LD-IRC are combinations of private  and common signaling, in addition to CF, DF, and CN relaying scheme. By using these strategies over the sub-channels in the same manner as it is shown for the LD-IRC, we achieve the upper bound for the GDoF. Notice that while in the LD-IRC, we optimize over the number of bits which should be use by each strategy, in the Gaussian case we optimize over number of sub-channels. Moreover, while in the LD-IRC using each bit level, we achieve one, in Gaussian IRC, by using one sub-channel, we achieve $R_s=\frac{1}{2}\log\left(\frac{\delta}{4}\right)$. Notice that the parameters $n_d$, $n_c$, $n_r$, and $n_s$ in the LD-IRC are equivalent to $N_d$, $N_c$, $N_r$, and $N_s$, in the Gaussian case, respectively. In Appendix \ref{app:Example_LD-Gaussian}, we show how the achievable sum-rate for the LD-IRC can be extended to the achievable GDoF.

%%%%%%%%%%%%%%%%%%%%%%%%%%%%%%%%%%%%%%%%%%%%%%%%%%
%%%%%%%%%%%%%%%%%%%%%%%%%%%%%%%%%%%%%%%%%%%%%%%%%%
%%%%%%%%%%%%%%%%%%%%%%%%%%%%%%%%%%%%%%%%%%%%%%%%%%
%%%%%%%%%%%%%%%%%%%%%%%%%%%%%%%%%%%%%%%%%%%%%%%%%%
%%%%%%%%%%%%%%%%%%%%%%%%%%%%%%%%%%%%%%%%%%%%%%%%%%
%%%%%%%%%%%%%%%%%%%%%%%%%%%%%%%%%%%%%%%%%%%%%%%%%%

\appendices
\section{Proof of the Upper Bounds for the LD-IRC (Lemma \ref{Lemma:new_bounds})}
\label{app:UB_LD_IRC_proof}

\subsection{Proof of \eqref{eq:new_UB6_IRC_det}}
\label{app:new_UB6_det_proof}
{The proof of the bound $C_{\text{det},\Sigma}\leq \max\{n_d,\min\{n_r,n_s\}\}$ follows from the cut-set bounds. Namely, consider the following cut-set bound $R_{\mathcal{S}\to\mathcal{S}^c} \leq \max_{P(\X_1,\X_2,\X_r)} I(\X_\mathcal{S};\Y_{\mathcal{S}^c}|\X_{\mathcal{S}^c})$ where $\mathcal{S}=\{\text{Tx}1,\text{Tx}2,\text{Relay}\}$ and $\mathcal{S}^c=\{\text{Rx}1,\text{Rx}2\}$. This bound yields
\begin{align}
\label{CSB1G}
R_\Sigma \leq \max_{P(\X_1,\X_2,\X_r)} I(\X_1,\X_2,\X_r;\Y_1,\Y_2).
\end{align}
The term $I(\X_1,\X_2,\X_r;\Y_1,\Y_2)$ can be upper bounded as follows:
\begin{align}
I(\X_1,\X_2,\X_r;\Y_1,\Y_2) & = H(\Y_1,\Y_2) - H(\Y_1,\Y_2|\X_1,\X_2,\X_r)\\
&\overset{(a)}{\leq} H(\Y_1,\Y_2)\\
&= H(\Y_1)+H(\Y_2|\Y_1),
\end{align}
where in step $(a)$, we used the fact that $\Y_1$ and $\Y_2$ are deterministic functions of $\X_1$, $\X_2$, and $\X_r$. Note that $\Y_1$ and $\Y_2$ are equal for the case where $n_d=n_c$. Thus, $H(\Y_2|\Y_1)=0$. It remains to maximize $H(\Y_1)$ with respect to the input distribution. Since $\Y_1$ is a binary vector of length $\max\{n_d,n_c,n_r\}=\max\{n_d,n_r\}$, $H(\Y_1)$ is maximized when the components of $\Y_1$ are i.i.d. Bern$(\nicefrac{1}{2})$, which corresponds to inputs $\X_1$, $\X_2$, and $\X_r$ distributed also according to an i.i.d. Bern$(\nicefrac{1}{2})$ distribution. Therefore, $H(\Y_1)\leq \max\{n_d,n_r\}$ leading to
\begin{align}
\label{CSB1ndeqnc}
R_\Sigma \leq \max\{n_d,n_r\}.
\end{align}}

{Using similar steps with the cut $\mathcal{S}=\{\text{Tx}1,\text{Tx}2\}$ and $\mathcal{S}^c=\{\text{Rx}1,\text{Rx}2,\text{Relay}\}$ leads to the bound $R_\Sigma\leq \max\{n_d,n_s\}$. Namely, with this cut, we have
\begin{align}
R_\Sigma\leq \max_{P(\X_1,\X_2,\X_r)} I(\X_1,\X_2;\Y_1,\Y_2,\Y_r|\X_r).
\end{align}
Note that
\begin{align}
I(\X_1,\X_2;\Y_1,\Y_2,\Y_r|\X_r)&= H(\Y_1,\Y_2,\Y_r|\X_r)-H(\Y_1,\Y_2,\Y_r|\X_r,\X_1,\X_2)\\
&= H(\Y_1,\Y_2,\Y_r|\X_r)\\
&= H(\Y_1|\X_r)+H(\Y_2|\X_r,\Y_1)+H(\Y_r|\X_r,\Y_1,\Y_2).
\end{align}
The first term $H(\Y_1|\X_r)$ can be bounded as follows
\begin{align}
H(\Y_1|\X_r)&=H(\bS^{q-n_d}\X_1+\bS^{q-n_c}\X_2|\X_r)\\
&\leq H(\bS^{q-n_d}\X_1+\bS^{q-n_c}\X_2)\\
&\leq n_d,
\end{align}
where the first inequality follows since conditioning does not increase entropy, and the second follows since $n_d=n_c$ and since the entropy is maximized by i.i.d. Bern$(\nicefrac{1}{2})$ inputs. The second term $H(\Y_2|\X_r,\Y_1)$ is zero since $\Y_1$ and $\Y_2$ are equal given $n_d=n_c$. Finally, the last term satisfies
\begin{align}
H(\Y_r|\X_r,\Y_1,\Y_2)&=H(\Y_r|\X_r,\Y_1)\\
&=H(\Y_r|\X_r,\Y_1\oplus \bS^{q-n_r}\X_r)\\
&\leq H(\Y_r|\Y_1\oplus \bS^{q-n_r}\X_r)\\
&=H(\bS^{q-n_s}\X_1\oplus \bS^{q-n_s}\X_2|\bS^{q-n_d}\X_1\oplus \bS^{q-n_c}\X_2)\\
&\leq (n_s-n_d)^+,
\end{align}
where the first inequality follows since conditioning does not increase entropy, and the second follows since the number of unknown bits of $\bS^{q-n_s}\X_1\oplus \bS^{q-n_s}\X_2$ given $\bS^{q-n_d}\X_1\oplus \bS^{q-n_c}\X_2$ where $n_d=n_c$ is $n_s-n_d$ if $n_d\leq n_s$ and zero otherwise, and the entropy of these bits is maximized by the i.i.d. Bern$(\frac{1}{2})$ distribution. Therefore, we can write
\begin{align}
\label{CSB2ndeqnc}
R_\Sigma\leq \max\{n_d,n_s\}.
\end{align}
}
Combining \eqref{CSB1ndeqnc} and \eqref{CSB2ndeqnc}, we get
\begin{align}
R_\Sigma\leq \max\{n_d,\min\{n_r,n_s\}\},
\end{align}
which is the desired bound given in \eqref{eq:new_UB6_IRC_det} in Lemma \ref{Lemma:new_bounds}.

\subsection{Proof of \eqref{eq:new_UB4_IRC_det}}
\label{app:new_UB4_det_proof}
{This bound is in fact derived from the cut-set bound given above in \eqref{CSB1G}. As before, we have 
\begin{align}
R_\Sigma \leq \max_{P(\X_1,\X_2,\X_r)} H(\Y_1)+H(\Y_2|\Y_1),
\end{align}
where $H(\Y_1)$ is maximized by i.i.d. Bern$(\nicefrac{1}{2})$ distributed inputs, leading to $H(\Y_1)\leq \max\{n_d,n_c,n_r\}$. The last term is non-zero, contrary to the $n_d=n_c$ case. To bound it, we can use $H(\Y_2|\Y_1)= H(\Y_2\oplus \Y_1|\Y_1)$ and the property that conditioning does not increase entropy to write
\begin{align}
H(\Y_2|\Y_1) &\leq  H(\Y_2\oplus \Y_1). \label{eq:new_UB4_det_proof2}
\end{align}
Notice that $\Y_2\oplus\Y_1$ given by $\bS^{q-n_d}\X_2\oplus \bS^{q-n_c}\X_1 \oplus \bS^{q-n_d}\X_1 \oplus \bS^{q-n_c}\X_2$ has $\max\{n_d,n_c\}$ non-zero components. Thus, the maximum value of $H(\Y_2\oplus \Y_1)$ is $\max\{n_d,n_c\}$ and is achieved when $\X_1$ and $\X_2$ are i.i.d. Bern$(\nicefrac{1}{2})$. Thus, $H(\Y_2|\Y_1)\leq \max\{n_d,n_c\}$. Consequently, we get 
\begin{align}
R_\Sigma \leq \max\{n_d,n_c,n_r\} + \max\{n_d,n_c\},
\end{align}
which concludes the proof of the upper bound \eqref{eq:new_UB4_IRC_det} in Lemma \ref{Lemma:new_bounds}.}

\subsection{Proof of \eqref{eq:new_UB1_IRC_det}}
\label{app:new_UB1_det}
For establishing the upper bound \eqref{eq:new_UB1_IRC_det} in Lemma~\ref{Lemma:new_bounds}, {we use a genie-aided approach. In a general genie-aided approach, the side-information $\boldsymbol{s}_1$ and $\boldsymbol{s}_2$ is given to receivers 1 and 2, respectively. Then, using Fano's inequality, we can write
\begin{align}
\label{FanoGenie}
n (R_\Sigma-\epsilon_n) \leq I(W_1;\Y_1^n,\boldsymbol{s}_1) + I(W_2;\Y_2^n,\boldsymbol{s}_2),
\end{align}
where $\epsilon_n\to 0$ as $n\to\infty$.}

For this particular case, we use $\boldsymbol{s}_1=\bS^{q-n_r}\X_r^n$ and $\boldsymbol{s}_2=(\bS^{q-n_r}\X_r^n,\Y_1^n,W_1)$. By using the chain rule and the independence of the different messages, we can rewrite the bound as 
\begin{align}
n (R_\Sigma-\epsilon_n) \leq & I(W_1;\bS^{q-n_r}\X_r^n) + I(W_1;\Y_1^n|\bS^{q-n_r}\X_r^n) \notag \\ &+I(W_2;\bS^{q-n_r}\X_r^n|W_1)+ I(W_2;\Y_1^n|\bS^{q-n_r}\X_r^n,W_1) +  I(W_2;\Y_2^n|\Y_1^n,\bS^{q-n_r}\X_r^n,W_1), \notag \\
= & I(W_1,W_2;\bS^{q-n_r}\X_r^n) + I(W_1,W_2;\Y_1^n|\bS^{q-n_r}\X_r^n)+ I(W_2;\Y_2^n|\Y_1^n,\bS^{q-n_r}\X_r^n,W_1). \label{eq:New_UB1_det_proof1}
\end{align}
Now we consider every term in \eqref{eq:New_UB1_det_proof1} separately. The first term in \eqref{eq:New_UB1_det_proof1} can be written as
\begin{align}
I(W_1,W_2;\bS^{q-n_r}\X_r^n) & = H(\bS^{q-n_r}\X_r^n) - H(\bS^{q-n_r}\X_r^n|W_1,W_2) \notag \\ 
&\stackrel{(a)}{=}H(\bS^{q-n_r}\X_r^n)\\
&\stackrel{(b)}{\leq} n\cdot n_r, \label{eq:New_UB1_det_proof2}
\end{align}
where $(a)$ follows since $H(\bS^{q-n_r}\X_r^n|W_1,W_2)=0$, and $(b)$ follows since $H(\bS^{q-n_r}\X_r^n)$ is the entropy of $n\cdot n_r$ binary random variables, and thus it is maximized when these random variables are is i.i.d. Bern$(\nicefrac{1}{2})$ distributed. 

The second term in \eqref{eq:New_UB1_det_proof1} can be upper bounded as follows
\begin{align}
I(W_1,W_2;\Y_1^n|\bS^{q-n_r}\X_r^n)& = H(\Y_1^n|\bS^{q-n_r}\X_r^n) - H(\Y_1^n|\bS^{q-n_r}\X_r^n,W_1,W_2)  \\
& \overset{(c)}{=} H(\bS^{q-n_d}\X_1^n \oplus \bS^{q-n_c}\X_2^n |\bS^{q-n_r}\X_r^n)  \\
& \overset{(d)}{\leq} H(\bS^{q-n_d}\X_1^n \oplus \bS^{q-n_c}\X_2^n)\\
& \stackrel{(e)}{\leq} n\cdot\max\{n_d,n_c\},
\label{eq:New_UB1_det_proof3}
\end{align}
where $(c)$ follows since $H(A|B)=H(A\oplus B|B)$ and since $\Y_1^n$ is deterministic given $W_1$ and $W_2$, $(d)$ follows since conditioning does not increase entropy, and $(e)$ follows since $\bS^{q-n_d}\X_1^n \oplus \bS^{q-n_c}\X_2^n$ consists of $n\cdot \max\{n_d,n_c\}$ binary random variables, and hence the maximizing distribution is the i.i.d. Bern$(\nicefrac{1}{2})$ distribution.

Finally, the third term in \eqref{eq:New_UB1_det_proof1} can be upper bound by 
\begin{align}
I(W_2;\Y_2^n|\Y_1^n,\bS^{q-n_r}\X_r^n,W_1) =&  H(\Y_2^n|\Y_1^n,\bS^{q-n_r}\X_r^n,W_1) - H(\Y_2^n|\Y_1^n,\bS^{q-n_r}\X_r^n,W_1,W_2) \notag \\ 
\overset{(f)}{=}& H(\bS^{q-n_d}\X_2^n|\bS^{q-n_c}\X_2^n,\bS^{q-n_r}\X_r^n,W_1) \notag \\
\overset{(g)}{\leq} & H(\bS^{q-n_d}\X_2^n|\bS^{q-n_c}\X_2^n)\\
\stackrel{(h)}{\leq} & n\cdot (n_d-n_c)^+,
\label{eq:New_UB1_det_proof4}
\end{align}
where $(f)$ follows since $H(A|B,C)=H(A\oplus f(C)|B\oplus f(C),C)$ for some function $f(\cdot)$, and since $\Y_2^n$ is deterministic given $W_1$ and $W_2$, and $(g)$ follows since conditioning does not increase entropy. Step $(h)$ follows since given $\bS^{q-n_c}\X_2^n$, the top-most $n_c$ bits of $\X_2^n$ are known and can be removed from $\bS^{q-n_d}\X_2^n$. Thus, $\bS^{q-n_d}\X_2^n$ has only $n\cdot (n_d-n_c)^+$ random components (the lower-most ones), whose entropy is maximized by the i.i.d. Bern$(\nicefrac{1}{2})$ distribution. 

Now, by substituting the expressions in \eqref{eq:New_UB1_det_proof2}, \eqref{eq:New_UB1_det_proof3}, and \eqref{eq:New_UB1_det_proof4} into \eqref{eq:New_UB1_det_proof1}, we obtain
\begin{align}
n(R_\Sigma-\epsilon_n) & \leq  n\cdot (n_r + 2\max\{n_d,n_c\} - n_c).
 \label{eq:New_UB1_det_proof7}
\end{align}
By dividing the expression by $n$ and letting $n\to \infty$, we get \eqref{eq:new_UB1_IRC_det}.

\subsection{Proof of \eqref{eq:new_UB3_IRC_det}}
\label{app:new_UB3_det_proof}
This is also a genie-aided upper bound. We set $\boldsymbol{s}_1=\Y_r^n$ and $\boldsymbol{s}_2=(\Y_r^n,W_1)$. Substituting in \eqref{FanoGenie}, we can write
\begin{align}
n(R_\Sigma-\epsilon_n) \leq & I(W_1;\Y_1^ n,\Y_r^n) +  I(W_2;\Y_2^n,\Y_r^n,W_1) \label{eq:new_UB3_det_proof1}\\
=&  I(W_1;\Y_r^n)+ I(W_1;\Y_1^n|\Y_r^n) +  I(W_2;\Y_r^n|W_1) + I(W_2;\Y_2^n|\Y_r^n,W_1) \notag\\
= & I(W_1,W_2;\Y_r^n) + I(W_1;\Y_1^n|\Y_r^n)+I(W_2;\Y_2^n|\Y_r^n,W_1), \label{eq:new_UB3_det_proof2}
\end{align}
where the equalities follow from the independence of the messages and the chain rule. The first term in \eqref{eq:new_UB3_det_proof2} can be bounded as follows
\begin{align}
I(W_1,W_2;\Y_r^n) =& H(\Y_r^n)- H(\Y_r^n|W_1,W_2)\notag \\
\overset{(a)}{=}& H(\Y_r^n)\\
\overset{(b)}{\leq}& n\cdot n_s, \label{eq:new_UB3_det_proof31}
\end{align}
where $(a)$ follows since $\Y_r^n$ is a deterministic function of $W_1$ and $W_2$, and $(b)$ follows since the entropy of $\Y_r^n$ is maximized by the i.i.d. Bern$(\nicefrac{1}{2})$ distribution. The second term in \eqref{eq:new_UB3_det_proof2} satisfies
\begin{align}
I(W_1;\Y_1^n|\Y_r^n) =& H(\Y_1^n|\Y_r^n)- H(\Y_1^n|\Y_r^n,W_1)\notag \\
\stackrel{(a)}{\leq} & H(\Y_1^n|\Y_r^n)\notag \\
\overset{(b)}{=}& H(\bS^{q-n_d}\X_1^n\oplus \bS^{q-n_c}\X_2^n|\Y_r^n)\\
\overset{(c)}{\leq}& H(\bS^{q-n_d}\X_1^n\oplus \bS^{q-n_c}\X_2^n)\\ \overset{(d)}{\leq}& n\cdot\max\{n_d,n_c\},
\label{eq:new_UB3_det_proof32}
\end{align}
where $(a)$ follows from the non-negativity of mutual information, $(b)$ follows since $H(A|B)=H(A\oplus f(B)|B)$ for some function $f(\cdot)$, $(c)$ follows since conditioning does not reduce entropy, and $(d)$ follows similar to \eqref{eq:New_UB1_det_proof3}. Finally, the last term in \eqref{eq:new_UB3_det_proof2} can be bounded as follows
\begin{align}
I(W_2;\Y_2^n|\Y_r^n,W_1) =& H(\Y_2^n|\Y_r^n,W_1) -H(\Y_2^n|\Y_r^n,W_1,W_2)\notag \\
\overset{(e)}{=}& H(\bS^{q-n_d}\X_2^n|\Y_r^n,W_1)\\
\overset{(f)}{\leq}& H(\bS^{q-n_d}\X_2^n|\Y_r^n,W_1,\X_1^n,\bS^{q-n_s}\X_2^n)\\
\overset{(g)}{\leq}& H(\bS^{q-n_d}\X_2^n|\bS^{q-n_s}\X_2^n)\\
\overset{(h)}{\leq}& n\cdot(n_d-n_s)^+, \label{eq:new_UB3_det_proof33}
\end{align}
where $(e)$ follows since $\Y_1^n$, $\Y_2^n$, and $\Y_r^n$ are deterministic functions of $W_1$ and $W_2$, $(f)$ follows since knowing $W_1$ and $\Y_r^n$, $\X_1^n$ and $\bS^{q-n_s}\X_2^n$ can be constructed, $(g)$ follows since conditioning does not reduce entropy, and $(h)$ follows similar to \eqref{eq:New_UB1_det_proof4}. Substituting \eqref{eq:new_UB3_det_proof31}, \eqref{eq:new_UB3_det_proof32}, and \eqref{eq:new_UB3_det_proof33} in \eqref{eq:new_UB3_det_proof2} leads to
\begin{align}
n(R_\Sigma-\epsilon_n) \leq & n\cdot (\max\{n_d,n_c\} + \max\{n_d,n_s\}). \label{eq:new_UB3_det_proof4}
\end{align}
By dividing \eqref{eq:new_UB3_det_proof4} by $n$ and letting $n\to\infty$, we obtain the upper bound in \eqref{eq:new_UB3_IRC_det} in Lemma~\ref{Lemma:new_bounds}.

\subsection{Proof of \eqref{eq:new_UB2_IRC_det}}
\label{app:new_UB2_det_proof}
To establish the upper bound in \eqref{eq:new_UB2_IRC_det}, we use the upper bound given in \cite[Theorem 4]{ChaabanSezgin_IT_IRC}. Writing this upper bound for the LD-IRC, we obtain
\begin{align}
C_{\mathrm{det},\Sigma}\leq 2\max\{n_c,n_r,(n_d-n_c)\}+2(n_s-n_c)^+. \label{eq:new_UB2_IRC_det_proof1}
\end{align}
Now, we enhance the Rx's observation by $(n_s-n_c)^+$ bits. {Doing this, we replace $n_d$, $n_c$, and $n_r$ by 
\begin{align*}
\bar n_d &= n_d+(n_s-n_c)^+, \\
\bar n_c &= n_c+(n_s-n_c)^+, \\
\bar n_r &= n_r+(n_s-n_c)^+,
\end{align*}
respectively. We keep the source-relay channel intact, i.e., $\bar n_s = n_s$. This operation is equivalent to} reducing the noise power at the Rx's in the Gaussian IRC. {The sum-capacity of this enhanced channel is upper bounded by
\begin{align}
2\max\{\bar n_c,\bar n_r,(\bar n_d-\bar n_c)\}+2(\bar n_s-\bar n_c)^+ \label{eq:new_UB2_IRC_det_proof2}
\end{align}
according to \eqref{eq:new_UB2_IRC_det_proof1}.} Since the capacity of the enhanced channel is an upper bound for the capacity of the original channel, we get the bound 
\begin{align}
C_{\mathrm{det},\Sigma}\leq 2\max\{\bar n_c,\bar n_r,(\bar n_d-\bar n_c)\}+2(\bar n_s-\bar n_c)^+ \label{eq:new_UB2_IRC_det_proof2}
\end{align}
Notice that $(\bar{n}_s-\bar{n}_c)^+ = 0$. Now, by substituting $\bar n_d$ , $\bar n_c$, $\bar n_r$, and $\bar n_s$ into \eqref{eq:new_UB2_IRC_det_proof2}, we obtain 
\begin{align}
C_{\mathrm{det},\Sigma}\leq 2\max\{n_c+(n_s-n_c)^+,n_r+(n_s-n_c)^+,(n_d-n_c)\}. \label{eq:new_UB2_IRC_det_proof3}
\end{align}
The expression in \eqref{eq:new_UB2_IRC_det_proof3} can be rewritten as 
\begin{align}
C_{\mathrm{det},\Sigma}\leq 2\max\{n_c,n_r,n_d-\max\{n_c,n_s\}\}+2(n_s-n_c)^+, \label{eq:new_UB2_IRC_det_proof4}
\end{align}
which completes the proof of \eqref{eq:new_UB2_IRC_det} in Lemma \ref{Lemma:new_bounds}.

\subsection{Proof of \eqref{eq:new_UB5_IRC_det}}
\label{app:new_UB5_det_proof}
This is also a genie-aided upper bound. Throughout this proof, we assume that $n_r\leq n_c\leq n_d$. Here, we set $\boldsymbol{s}_1=\bS^{q-(n_c-n_r)}\X_1^n$ and $\boldsymbol{s}_2=\bS^{q-(n_c-n_r)}\X_2^n$, and substitute them into \eqref{FanoGenie} to obtain
\begin{align}
n(R_\Sigma-\epsilon_n) \leq & I(W_1;\Y_1^n,\boldsymbol{s}_1) + I(W_2;\Y_2^n,\boldsymbol{s}_2) \notag\\
=& I(W_1;\boldsymbol{s}_1) + I(W_1;\Y_1^n|\boldsymbol{s}_1) + I(W_2;\boldsymbol{s}_2)+ I(W_2;\Y_2^n|\boldsymbol{s}_2) \notag\\
\overset{(a)}{=}& H(\boldsymbol{s}_1) + H(\Y_1^n|\boldsymbol{s}_1) - H(\bS^{q-n_c}\X_2^n\oplus \bS^{q-n_r}\X_r^n|W_1) \notag \\ 
&+ H(\boldsymbol{s}_2)+ H(\Y_2^n|\boldsymbol{s}_2) - H(\bS^{q-n_c}\X_1^n\oplus \bS^{q-n_r}\X_r^n|W_2), \notag
\end{align}
where in step $(a)$, we used the fact that knowing $W_i$, $\X_i$ is deterministic. Since $n_c$ is larger than $n_r$, the top-most $n_c-n_r$ bits of interference signal is received without any overlap with the relay signal.
Therefore, we can split $\X_i^n$ in the term $H(\bS^{q-n_c}\X_i^n\oplus \bS^{q-n_r}\X_r^n|W_{j})$ ($i,j\in\{1,2\}$, $i\neq j$) into two parts: one part without overlap with $\X_r^n$ and the other part with overlap with $\X_r^n$. Doing this, we obtain 
\begin{align}
n(R_\Sigma-\epsilon_n) \leq & H(\boldsymbol{s}_1) + H(\Y_1^n|\boldsymbol{s}_1) - H(\boldsymbol{s}_2, \X_{2,[n_c-n_r+1:n_c]}^n\oplus  \X_{r,[1:n_r]}^n|W_1) \notag \\ &+ H(\boldsymbol{s}_2)+ H(\Y_2^n|\boldsymbol{s}_2) - H(\boldsymbol{s}_1, \X_{1,[n_c-n_r+1:n_c]}^n\oplus  \X_{r,[1:n_r]}^n|W_2). \notag
\end{align}
Here, we used the fact that $\boldsymbol{S}^{q-n_r}\X_r^n = \X_{r,[1:n_r]}^n$.
Next, we use chain rule and the fact that the messages are independent of each other to obtain 
\begin{align}
n(R_\Sigma-\epsilon_n) \leq & H(\boldsymbol{s}_1) + H(\Y_1^n|\boldsymbol{s}_1)- H(\boldsymbol{s}_2) - H(\X_{2,[n_c-n_r+1:n_c]}^n\oplus  \X_{r,[1:n_r]}^n|W_1,\boldsymbol{s}_2) \notag \\ 
&+ H(\boldsymbol{s}_2)+ H(\Y_2^n|\boldsymbol{s}_2) - H(\boldsymbol{s}_1) -  H(\X_{1,[n_c-n_r+1:n_c]}^n\oplus \X_{r,[1:n_r]}^n|W_2, \boldsymbol{s}_1) \notag \\ 
\overset{(b)}{\leq} & H(\Y_1^n|\boldsymbol{s}_1)  + H(\Y_2^n|\boldsymbol{s}_2),
\end{align}
where in $(b)$, we used the non-negativity of entropy. Next, we replace $\boldsymbol{s}_1$ and $\boldsymbol{s}_2$ by their values, and use the given information bits to decrease the entropy as follows
\begin{align}
n(R_\Sigma-\epsilon_n) \leq & H(\bar \X_1^n \oplus \bS^{q-n_c}\X_2^n \oplus \bS^{q-n_r}\X_r^n |\bS^{q-(n_c-n_r)}\X_1)  + H(\bar \X_2^n \oplus \bS^{q-n_c}\X_1^n \oplus \bS^{q-n_r}\X_r^n|\bS^{q-(n_c-n_r)}\X_2), \label{eq:new_UB5_det_proof4}
\end{align}
where {$\bar \X_i= \bS^{q-n_d}\X_i \oplus (\bS^T)^{q-n_d+n_c-n_r} \bS^{q-(n_c-n_r)}\X_i$.\footnote{Note that $\bS^T$ is a shift-up matrix, and thus, multiplying $\bS^{q-(n_c-n_r)}\X_i$ by $(\bS^T)^{q-n_d+(n_c-n_r)}$ aligns it with $\bS^{q-n_d}\X_i$.} Therefore, $\bar \X_i$ can be written as 
\begin{align}
\bar \X_i = \begin{bmatrix}
\boldsymbol{0}_{q-n_d+n_c-n_r} \\ X_{i,q-n_d+n_c-n_r+1} \\ \vdots \\ X_{i,q},
\end{bmatrix},
\end{align}}
where $X_{i,l}$ represents the $l$th element of the binary random vector $\X_i$. {Note that due to the assumption $n_d\geq n_c$, then $n_d-n_c+n_r>0$ and the number of random components of $\bar \X_i$ is $n_d-n_c+n_r$.} Now, similar to \eqref{eq:New_UB1_det_proof3}, we can upper bound \eqref{eq:new_UB5_det_proof4} as 
\begin{align}
n(R_\Sigma-\epsilon_n) \leq & n\cdot (2\max\{n_d-n_c+n_r,n_c,n_r\})  \label{eq:new_UB5_det_proof2}
\end{align}
Using the assumption $n_r \leq n_c$, and dividing the upper bound by $n$ and letting $n\to \infty$, we obtain 
\begin{align}
R_\Sigma \leq & 2\max\{n_d-n_c+n_r,n_c\}, \quad \text{ if } n_r\leq n_c \leq n_d
\end{align}
which completes the proof of \eqref{eq:new_UB5_IRC_det}. With this, the proof of Lemma \ref{Lemma:new_bounds} is complete.

\section{Proof of Lemma \ref{Lemma:new_UB4_2_Gauss}}
\label{app:proof_Lemma_additional}
In this appendix, we prove Lemma \ref{Lemma:new_UB4_2_Gauss} for the case that $i=1$ and $j=2$. The other case can be proved similarly. Hence, is what follows, we want to show that $h(\Gamma^n)-h(\Delta^n|W_2)$ is upper bounded by $n C\left(2+ \frac{h_c^2}{(h_c-h_r)^2}\right)$, where $ \Gamma^n = \frac{h_c}{\sqrt{Ph_r^2}}X_1^n+U_1^n$, $\Delta^n=h_cX_1^n+h_rX_r^n+Z_2^n$, and $U_1\sim \mathcal{N}(0,1)$ is i.i.d. over the time and  independent from other random variables. To do this, we write
\begin{align}
h(\Gamma^n)-h(\Delta^n|W_2) \overset{(a)}{=}& h(\Gamma^n)-h(\Delta^n|W_2) - h(U_1^n) + h(Z_2^n|W_2)  \\
=& I(X_1^n;\Gamma^n)-I(X_1^n,X_r^n;\Delta^n|W_2)  \\ 
=& I(X_1^n;\Gamma^n)-I(X_1^n;\Delta^n|W_2) -I(X_r^n;\Delta^n|W_2,X_1^n)  \\ 
\overset{(b)}{\leq}& I(X_1^n;\Gamma^n)-I(X_1^n;\Delta^n|W_2)  \\ 
%\overset{}{=}& I(X_1^n;\Gamma^n)-I\left(\frac{h_c}{\sqrt{Ph_r^2}}X_1^n;\Delta^n|W_2\right)  \\
%\overset{(c)}{=}& I(X_1^n;\Gamma^n)- I(U_1^n;\Delta^n|W_2,\frac{h_c}{\sqrt{Ph_r^2}}X_1^n+U_1^n)-I\left(\frac{h_c}{\sqrt{Ph_r^2}}X_1^n+U_1^n;\Delta^n|W_2,U_1\right)  \\ 
%\overset{}{=}& I(X_1^n;\Gamma^n)-I\left(\frac{h_c}{\sqrt{Ph_r^2}}X_1^n+U_1^n,U_1^n;\Delta^n|W_2\right)   \\ 
%\overset{}{=}& I(X_1^n;\Gamma^n)-I\left(\Gamma^n;\Delta^n|W_2\right) - I(U_1^n;\Delta^n|W_2,\Gamma^n)  \\ 
\overset{(c)}{\leq}& I(X_1^n;\Gamma^n)-I\left(\Gamma^n;\Delta^n|W_2\right)    \\
\overset{}{=}& h(\Gamma^n) - h(U_1^n) - h(\Gamma^n) + h(\Gamma^n|\Delta^n,W_2)    \\
\overset{}{=}& h(\Gamma^n|\Delta^n,W_2) -  h(U_1^n) \\
\overset{}{\leq}& h(\Gamma^n|\Delta^n) -  h(U_1^n),
\end{align}
where, step $(a)$ follows since the distribution of $U_1^n,Z_2^n\sim \mathcal{N}(0,1)$ and $Z_2^n$ is independent of $W_2$, step $(b)$ follows from the fact that mutual information is non-negative, and in $(c)$, we used the fact that $\Gamma^n$, $X_1^n$, and $\Delta^n$ form a Markov chain, i.e., $\Gamma^n\to X_1^n \to  \Delta^n$. Therefore, $I(\Gamma^n;\Delta^n)\leq I(X_1^n;\Delta^n)$. 
Now, by using \cite[Lemma 1]{AnnapureddyVeeravalli} and the fact that $U_1$ is i.i.d. over the time, we write 
\begin{align}
h(\Gamma^n)-h(\Delta^n|W_2) \leq  n\left[h(\Gamma_G|\Delta_G) - h(U_{1})\right],
\end{align}
where the subscript $G$ indicates that the inputs are i.i.d. and Gaussian distributed, i.e., $X_{i,G}\sim\mathcal{N}(0,P_i)$, where $i\in\{1,2,r\}$ and $\Gamma_G$, $\Delta_G$ are corresponding signal. Notice that $X_{r,G}$ and $X_{1,G}$ are correlated signals with a correlation coefficient $\rho_1\in[-1,1]$. Hence, we obtain 
\begin{align}
n\left[h(\Gamma_G|\Delta_G) - h(U_{1})\right]
\leq & \frac{n}{2}\log\left(\frac{\begin{vmatrix}
1+\frac{P_1h_c^2}{Ph_r^2} & \frac{P_1h_c^2}{\sqrt{Ph_r^2}} + \frac{h_ch_r\sqrt{P_1P_r}\rho_1}{\sqrt{Ph_r^2}} \\ \frac{P_1h_c^2}{\sqrt{Ph_r^2}} + \frac{h_ch_r\sqrt{P_1P_r}\rho_1}{\sqrt{Ph_r^2}} & 1+P_1h_c^2+P_rh_r^2+2h_ch_r\sqrt{P_1P_r}\rho_1
\end{vmatrix}}{1+P_1h_c^2+P_rh_r^2+2h_ch_r\sqrt{P_1P_r}\rho_1} \right) \notag \\ 
= & \frac{n}{2} \log\left(1+\frac{\frac{P_1h_c^2}{Ph_r^2}+\frac{P_1^2h_c^4}{Ph_r^2}+\frac{P_1h_c^2P_rh_r^2}{Ph_r^2}+\frac{2P_1h_c^2h_ch_r\sqrt{P_1P_r}\rho_1}{Ph_r^2}}{1+P_1h_c^2+P_rh_r^2+2h_ch_r\sqrt{P_1P_r}\rho_1}\right. \notag \\ & \left. \qquad \phantom{\log}-\frac{ \frac{P_1^2h_c^4}{Ph_r^2} +  \frac{2P_1h_c^2h_ch_r\sqrt{P_1P_r}\rho_1}{Ph_r^2}  + \frac{h_c^2h_r^2P_1P_r\rho_1^2}{Ph_r^2}  }{1+P_1h_c^2+P_rh_r^2+2h_ch_r\sqrt{P_1P_r}\rho_1}\right)\notag \\
= & \frac{n}{2} \log\left(1+\frac{\frac{P_1h_c^2}{Ph_r^2}+\frac{P_1h_c^2P_rh_r^2}{Ph_r^2}-\frac{h_c^2h_r^2P_1P_r\rho_1^2}{Ph_r^2}  }{1+P_1h_c^2+P_rh_r^2+2h_ch_r\sqrt{P_1P_r}\rho_1}\right) \notag \\
\overset{(d)}{\leq} & \frac{n}{2} \log\left(1+ \underbrace{\frac{1}{Ph_r^2}\frac{P_1h_c^2}{1+(\sqrt{P_rh_r^2}-\sqrt{P_1h_c^2})^2}}_{g_1(P_1,P_r)}+\underbrace{\frac{1}{Ph_r^2}\frac{P_1h_c^2P_rh_r^2(1-\rho_1^2)}{1+P_rh_r^2+P_1h_c^2+2h_ch_r\sqrt{P_1P_r}
\rho_1}}_{g_2(P_1,P_r,\rho_1)}\right),\label{eq:app_Lemma_4}
\end{align}
where $(d)$ follows since $\rho_1\in[-1,1]$. To upper bound the expression in \eqref{eq:app_Lemma_4}, we upper bound the functions $g_1(P_1,P_r)$ and $g_2(P_1,P_r,\rho_1)$ separately, since $\log$ function is an increasing function. First, consider $g_1(P_1,P_r)$. 
By computing sign of the derivative of $g_1(P_1,P_r)$ with respect to $P_1$, we conclude that $g_1(P_1,P_r)$ has a maximum at $P_1=\frac{(1+P_rh_r^2)^2}{P_rh_r^2h_c^2}$.
Therefore, we obtain 
\begin{align}
g_1(P_1,P_r) &\leq g_1\left(\frac{(1+P_rh_r^2)^2}{P_rh_r^2h_c^2},P_r\right) \quad \text{ for } 0 \leq \beta \\
&=\frac{1+P_rh_r^2}{Ph_r^2}\\
&\overset{(e)}{\leq} 2, \label{eq:app_Lemma_4_2}
\end{align}
where step $(e)$ is followed since $P_r\leq P$ and $1\leq Ph_r^2$. Now, consider $g_2(P_1,P_r,\rho_1)$. Supposing that $u^2=\frac{1}{P_1}$ and $v^2=\frac{1}{P_r}$, we obtain
\begin{align}
g_2(P_1,P_r,\rho_1) &= \frac{h_c^2h_r^2(1-\rho_1^2)}{Ph_r^2} \frac{1}{u^2v^2+u^2 h_r^2+v^2h_c^2+2h_c h_r uv\rho_1} \notag \\ 
& \leq  \frac{h_c^2h_r^2(1-\rho_1^2)}{Ph_r^2} \frac{1}{u^2 h_r^2+v^2h_c^2+2h_c h_r uv\rho_1}.\label{eq:app_Lemma_4_3}
\end{align}
Maximizing the expression in \eqref{eq:app_Lemma_4_3} with respect to $u$ and $v$ is equivalent to minimizing its denominator with respect to $u$ and $v$. The denominator of \eqref{eq:app_Lemma_4_3} can be rewritten as 
\begin{align}
u^2 h_r^2+v^2h_c^2+2h_c h_r uv\rho_1 = \begin{bmatrix}
u & v
\end{bmatrix}\underbrace{\begin{bmatrix}
h_r^2 & \rho_1h_ch_r \\ \rho_1h_ch_r & h_c^2
\end{bmatrix}}_{\boldsymbol{A}} \begin{bmatrix}
u\\ v
\end{bmatrix}. \label{eq:app_Lemma_4_4}
\end{align}
Since $\rho_1^2\in[0,1]$, $\boldsymbol{A}$ is a positive semi-definite matrix. Therefore, \eqref{eq:app_Lemma_4_4} is minimized by the lowest value of $u$ and $v$, i.e., $\frac{1}{\sqrt{P}}$. By substituting $u=v=\frac{1}{\sqrt{P}}$ into \eqref{eq:app_Lemma_4_3}, we obtain 
\begin{align}
g_2(P_1,P_r,\rho_1) &\leq \frac{h_c^2(1-\rho_1^2)}{h_r^2+h_c^2+2h_ch_r
\rho_1} \\ 
&\overset{(f)}{\leq} \frac{h_c^2}{(h_c-h_r)^2},\label{eq:app_Lemma_4_5}
\end{align}
where in $(f)$, we used the fact that both $h_c$ and $h_r$ are positive. 
Now, by substituting \eqref{eq:app_Lemma_4_2} and \eqref{eq:app_Lemma_4_5} into \eqref{eq:app_Lemma_4}, the proof of Lemma \ref{Lemma:new_UB4_2_Gauss} in completed.

\section{Proof of the Upper bounds for Gaussian IRC (Lemma \eqref{Lemma:UB_Gauss})}
\label{app:known_UB_Gauss_proof}
\subsection{Proof of \eqref{eq:new_UB1_IRC_Gauss}}
In what follows, we establish the upper bound $d\leq \max\{1,\min\{\beta, \gamma\}\}$ for the case that $\alpha = 1$. To do this, we establish two bounds namely $\max\{1,\beta\}$ and $\max\{1,\gamma\}$. The minimum of these two bounds gives us the bound in \eqref{eq:new_UB1_IRC_Gauss}.

To establish the bound $d\leq \max\{1,\beta\}$, consider the following cut-set bound $R_{\mathcal S\to \mathcal S^c} \leq \max_{P(X_1,X_2,X_r)} I(X_\mathcal{S};Y_{\mathcal{S}^c}|X_{\mathcal{S}^c})$, where $\mathcal{S}=\{$Tx1, Tx2, Relay$\}$ and $\mathcal{S}=\{$Rx1, Rx2$\}$. Hence, we obtain
\begin{align}
R_\Sigma \leq \max_{P(X_1,X_2,X_r)} I(X_1,X_2,X_r;Y_1,Y_2) \label{eq:proof_UB_Gauss_1_1}
\end{align}
This term can be rewritten as
\begin{align}
I(X_1,X_2,X_r;Y_1,Y_2) & = h(Y_1,Y_2) - h(Y_1,Y_2|X_1,X_2,X_r) \\ 
& \overset{(a)}{=} h(Y_1,Y_2) - h(Z_1,Z_2),
\end{align}
where in step $(a)$, we used the fact that giving $X_1$, $X_2$, and $X_r$ all randomness of $Y_1$ and $Y_2$ is caused from the additive noises $Z_1$ and $Z_2$ which are independent from all other random variables. Now, by using the chain rule, we obtain 
\begin{align}
I(X_1,X_2,X_r;Y_1,Y_2) & = h(Y_1) + h(Y_2|Y_1) - h(Z_1)-h(Z_2) \\ 
&\overset{(b)}{=} h(Y_1) + h(Y_2-Y_1|Y_1) - h(Z_1)-h(Z_2) \\
&\overset{(c)}{\leq}h(Y_1) + h(Z_2- Z_1) - h(Z_1)-h(Z_2) \\
&\overset{(d)}{\leq} h(Y_{1G}) + h(Z_2- Z_1) - h(Z_1)-h(Z_2),
\end{align}
where step $(b)$ follows since $h(A-B|B)=h(A|B)$, in step $(c)$, we used the fact that conditioning does not increase the entropy and $Y_2-Y_1=Z_2-Z_1$ since $h_d=h_c$\footnote{Note that the condition $\alpha=1$ corresponds to $h_d=h_c$ since $h_d$ and $h_c$ are both real and positive.}. 
Step $(d)$ follows due to the fact that Gaussian distribution maximizes the differential entropy given the variance \cite{CoverThomas}. The subscript $G$ indicates that the inputs are i.i.d. and Gaussian distributed, i.e., $X_{iG} \sim \mathcal N(0,P_i)$, where $i\in\{1,2,r\}$. Therefore, we obtain
\begin{align}
I(X_1,X_2,X_r;Y_1,Y_2)& \leq  h(h_d X_{1G}+ h_c X_{2G} + h_r X_{rG}+ Z_1)+ h(Z_2-Z_1) - h(Z_{1}) - h(Z_2) \notag \\
& =  \frac{1}{2}\log\left(P_1h_d^2 + P_2 h_c^2 + P_r h_r^2 + 2\rho_1 \sqrt{P_1P_r} h_d h_r + 2\rho_2 \sqrt{P_2P_r} h_c h_r + 1\right)+ \frac{1}{2}\log(2)  \\ 
& \overset{(e)}{\leq}  C\left(2Ph_d^2+ P h_r^2 + 4 P h_d h_r \right) + C(1), \label{eq:proof_UB_Gauss_1_4}
\end{align}
where the correlation coefficient between $X_i$, $X_r$ is $\rho_i\in[-1,1]$ for $i=1,2$, in step $(e)$, we used the fact that $\log(x)$ function is an increasing function with respect to $x$ and $h_d=h_c$.
Due to \eqref{eq:proof_UB_Gauss_1_1}, the sum-rate is upper bounded by the expression in \eqref{eq:proof_UB_Gauss_1_4}.
Now, by dividing the sum-rate by $\frac{1}{2}\log(Ph_d^2)$ and letting $Ph_d^2\to \infty$, we obtain
\begin{align}
d\leq \max\{1,\beta\}. \label{eq:proof_UB_Gauss_1_5}
\end{align}
Now, we need to show that $d\leq \max\{1,\gamma\}$. To establish this upper bound, we use the cut $\mathcal{S}=\{\text{Tx1, Tx2}\}$ and $\mathcal{S}^c=\{\text{Rx1, Rx2, Relay}\}$. Hence, we can write
\begin{align}
R_\Sigma \leq \max_{P(X_1,X_2,X_r)} I(X_1,X_2;Y_1,Y_2,Y_r|X_r). \label{eq:proof_UB_Gauss_1_6}
\end{align}
The mutual information term can be rewritten as follows
\begin{align}
I(X_1,X_2;Y_1,Y_2,Y_r|X_r) &= h(Y_1,Y_2,Y_r|X_r) - h(Y_1,Y_2,Y_r|X_r,X_1,X_2) \\ 
&\overset{}{=} h(Y_1,Y_2,Y_r|X_r) - h(Z_1,Z_2,Z_r) \\ 
&\overset{}{=} h(Y_1|X_r) + h(Y_2|X_r,Y_1) + h(Y_r|X_r,Y_1,Y_2) - h(Z_1) -h(Z_2)- h(Z_r). \label{eq:proof_CUT_2_7}
\end{align}
First, by keeping in mind that $h_d=h_c$, we upper bound $ h(Y_1|X_r)-h(Z_1)$ as follows
\begin{align}
h(Y_1|X_r)-h(Z_1) &\leq h(h_d(X_1+X_2)+Z_1) - h(Z_1) \\ 
& = C(h_d^2(P_1+P_2))\\
& \leq C(2h_d^2P). \label{eq:proof_CUT_2_8}
\end{align}
Then, using the fact that for $\alpha=1$, $Y_1 = h_d (X_1+X_2)+h_rX_r+Z_1$ and $Y_2 = h_d (X_1+X_2)+h_rX_r+Z_2$, we upper bound $h(Y_2|X_r,Y_1) - h(Z_2)$ and obtain
\begin{align}
h(Y_2|X_r,Y_1) -h(Z_2) & \leq h(h_d(X_1+X_2)+Z_2|h_d(X_1+X_2) + Z_1) -h(Z_2) \\
& \overset{(f)}{\leq} h(Z_2-Z_1) -h(Z_2) \\ 
&=C(1).\label{eq:proof_CUT_2_9}
\end{align}
Step $(f)$ follows since $h(A-B|B) = h(A|B)$ and conditioning does not increase the entropy.
Finally, we upper bound $h(Y_r|X_r,Y_1,Y_2)-h(Z_r)$ as follows
\begin{align}
h(Y_r|X_r,Y_1,Y_2)-h(Z_r) & = h(Y_r|X_r,Y_1-h_rX_r,Y_2)-h(Z_r) \\
& \leq h(Y_r|Y_1-h_rX_r)-h(Z_r) \\
& \leq h(h_s(X_{1}+X_{2})+Z_r|h_d(X_{1}+X_{2})+Z_1)-h(Z_r) \\
& \overset{(g)}{\leq} h(h_s(X_{1G}+X_{2G})+Z_r|h_d(X_{1G}+X_{2G})+Z_1)-h(Z_r), \label{eq:proof_CUT_2_10}
\end{align}
where step $(g)$ follows from the fact that Gaussian distribution maximizes the conditional differential entropy for a given covariance matrix \cite[Lemma1]{Thomas}. Now, we define the variable $P_{12} = P_1+P_2$ to obtain 
\begin{align}
h(Y_r|X_r,Y_1,Y_2)-h(Z_r) & \leq \frac{1}{2}\log\left(\frac{\begin{vmatrix}
P_{12}h_s^2 + 1 & h_sh_dP_{12} \\ h_sh_dP_{12}  & P_{12}h_d^2 +1
\end{vmatrix}}{h_d^2P_{12} +1}\right) \\ 
&= C\left(\frac{P_{12}h_s^2}{P_{12}h_d^2+1}\right) \label{eq:proof_CUT_2_10}\\
& \overset{(h)}{\leq} C\left(\frac{2Ph_s^2}{2P h_d^2+1}\right). \label{eq:proof_CUT_2_11}
\end{align}
Step $(h)$ follows since the function in \eqref{eq:proof_CUT_2_10} is increasing in $P_{12}$ and $\max P_{12} =2P$.
%\begin{align}
%h(Y_r|X_r,Y_1,Y_2)-h(Z_r) \leq C\left(\frac{2Ph_s^2}{2P h_d^2+1}\right) \label{eq:proof_CUT_2_11}
%\end{align}

By substituting \eqref{eq:proof_CUT_2_8}, \eqref{eq:proof_CUT_2_9}, and \eqref{eq:proof_CUT_2_11} into \eqref{eq:proof_CUT_2_7}, we obtain
\begin{align}
I(X_1,X_2;Y_1,Y_2,Y_r|X_r) \leq C(2h_d^2P) + C(1) + C\left(\frac{2Ph_s^2}{2P h_d^2+1}\right). \label{eq:proof_CUT_2_12}
\end{align}
Due to \eqref{eq:proof_UB_Gauss_1_6}, the sum-rate is upper bounded by \eqref{eq:proof_CUT_2_12}. Similar to the previous case, we divide the upper bound for the sum-rate by $\frac{1}{2}\log(Ph_d^2)$ and let $Ph_d^2\to \infty$ to obtain
\begin{align}
d\leq 1+(\gamma-1)^+ \\ 
= \max\{1,\gamma\}\label{eq:proof_CUT_2_13}
\end{align}
Now by combining \eqref{eq:proof_UB_Gauss_1_5} and \eqref{eq:proof_CUT_2_13}, we complete the proof of \eqref{eq:new_UB1_IRC_Gauss}.

\subsection{Proof of \eqref{eq:new_UB2_IRC_Gauss}}
To establish an upper bound $d\leq \max\{1,\alpha,\beta\} + \max\{1,\alpha\}$, we use the cut-set bound in \eqref{eq:proof_UB_Gauss_1_1}. Hence, write
\begin{align}
R_\Sigma \leq \max_{P(X_1,X_2,X_r)} I(X_1,X_2,X_r;Y_1,Y_2). 
\end{align}
Next, by using the chain rule, we obtain
\begin{align}
R_\Sigma &\leq \max_{P(X_1,X_2,X_r)} I(X_1,X_2,X_r;Y_1) + I(X_1,X_2,X_r;Y_2|Y_1). \label{eq:proof_CUT_3_1}
\end{align}
First, we upper bound the mutual information $I(X_1,X_2,X_r;Y_1)$ as follows
\begin{align}
I(X_1,X_2,X_r;Y_1) &= h(Y_1) - h(Y_1|X_1,X_2,X_r) \\ 
&= h(h_dX_1+h_cX_2+h_rX_r+Z_1) - h(Z_1) \\
&\leq C(P(h_d^2+h_c^2+h_r^2)) \label{eq:proof_CUT_3_2}
\end{align}
Now, we upper bound the expression $I(X_1,X_2,X_r;Y_2|Y_1)$
\begin{align}
I(X_1,X_2,X_r;Y_2|Y_1) & \leq h(Y_2|Y_1) - h(Y_2|Y_1,X_1,X_2,X_r) \\ 
&= h(Y_2|Y_1) - h(Z_2) \\ 
&= h(Y_2-Y_1|Y_1) - h(Z_2) \\ 
&\leq h((h_d-h_c)X_2 + (h_c-h_d)X_1+Z_2-Z_1) - h(Z_2) \\ 
&\leq C\left(2P(h_d-h_c)^2+1\right) \label{eq:proof_CUT_3_3}
\end{align}
Now, by substituting  \eqref{eq:proof_CUT_3_2} and \eqref{eq:proof_CUT_3_3} into \eqref{eq:proof_CUT_3_1}, we upper bound the sum-rate as follows
\begin{align}
R_\Sigma \leq C(P(h_d^2+h_c^2+h_r^2))  +  C\left(2P(h_d-h_c)^2+1\right)
\end{align}
By dividing the expression by $\frac{1}{2}\log(Ph_d^2)$ and letting $P h_d^2\to \infty$, we obtain 
\begin{align}
d\leq \max\{1,\alpha,\beta\} + \max\{1,\alpha\},
\end{align}
which completes the proof.

\subsection{Proof of \eqref{eq:new_UB3_IRC_Gauss}}
The upper bound in \eqref{eq:new_UB3_IRC_Gauss} is established by using a genie-aided approach. In genie-aided bounds, we provide the side information $s_1$ and $s_2$ to Rx$1$ and Rx$2$, respectively. Next, we can use Fano's inequality, to  upper bound $R_\Sigma$ as follows
\begin{align}
n(R_\Sigma-\epsilon_n) \leq I(W_1;Y_1^n,s_1) + I(W_2;Y_2^n,s_2),
\end{align}
where $\epsilon_n\to 0$ when $n\to \infty$.
In this case, we set $s_1=S^n$ and $s_2 = (S^n,Y_1,W_1)$, where $S^n=h_rX_r^n+Z^n$ and $Z\sim \mathcal{N}(0,1)$ is a Gaussian noise independent of all other random variables and i.i.d. over the time. 
Now, by using the chain rule and the fact that $W_1$ is independent from $W_2$, we obtain
\begin{align}
n(R_\Sigma-\epsilon_n) & \leq I(W_1;S^n) +I(W_1;Y_1^n|S^n) + I(W_2;S^n|W_1) + I(W_2;Y_1^n|S^n,W_1) + I(W_2;Y_2^n|S^n,Y_1^n,W_1) \\
& \leq I(W_1,W_2;S^n) +I(W_1,W_2;Y_1^n|S^n)+ I(W_2;Y_2^n|S^n,Y_1^n,W_1). \label{eq:new_UB1_Gauss_proof_1}
\end{align}
Next, we consider each term in \eqref{eq:new_UB1_Gauss_proof_1} separately. The first term in \eqref{eq:new_UB1_Gauss_proof_1} can be rewritten as follows
\begin{align}
I(W_1,W_2;S^n) &= h(S^n)  - h(S^n|W_1,W_2)  \\
& \overset{(a)}{\leq}  h(S^n)  - h(S^n|W_1,W_2,X_r^n)  \\
& \overset{(b)}{=} h(S^n) - h(Z^n), \label{eq:new_UB1_Gauss_proof_2}
\end{align}
where $(a)$ follows from the fact that conditioning does not increase the entropy. In step $(b)$, we used the fact that knowing $X_r^n$, all randomness of $S^n$ is caused from $Z^n$. 
Now, by using \cite[corollary to Theorem 8.6.2]{CoverThomas} and the fact that $Z^n$ is i.i.d. over the time, we obtain 
\begin{align}
I(W_1,W_2;S^n) & \leq \sum_{t=1}^n h(S[t]) - h(Z[t]).
\label{eq:new_UB1_Gauss_proof_3}
\end{align}
%where in $(c)$ we dropped the conditions since this cannot decrease the entropy and $Z[t]$ is i.i.d. over the time. 
Due to the fact that differential entropy given a variance constraint is maximized by the Gaussian distribution \cite{CoverThomas}, we upper bound \eqref{eq:new_UB1_Gauss_proof_3} as follows
\begin{align}
I(W_1,W_2;S^n) & \leq n (h(h_rX_{rG}+Z)-h(Z)) \\
&= nC(P_rh_r^2) \\ 
&\leq nC(Ph_r^2),\label{eq:new_UB1_Gauss_proof_4}
\end{align}
where $X_{rG}\sim\mathcal{N}(0,P_r)$.
Now, we upper bound the second term in \eqref{eq:new_UB1_Gauss_proof_1} as follows.
\begin{align}
I(W_1,W_2;Y_1^n|S^n) &= h(Y_1^n|S^n) - h(Y_1^n|S^n,W_1,W_2)  \\
&\overset{(c)}{\leq} h(Y_1^n-S^n) - h(Y_1^n-S^n|S^n,W_1,W_2)  \\
&\overset{}{\leq} h(Y_1^n-S^n) - h(Z_1^n-Z^n|S^n,W_1,W_2,Z^n)  \\
&\overset{(d)}{\leq} h(Y_1^n-S^n) - h(Z_1^n). \label{eq:new_UB1_Gauss_proof_5}
\end{align}
Step $(c)$ follows from the fact that $h(A|B) = h(A-B|B)\leq h(A-B)$.   
%In step $(d)$, we dropped the condition on $S^n$ since this cannot decrease the entropy, step $(e)$, we used the fact that conditioning does not increase the entropy. 
In step $(d)$, we dropped all conditions in second term since $Z_1^n$ is independent of all other random variables. 
Similar to above, by using \cite[corollary to Theorem 8.6.2]{CoverThomas}, and the fact that additive noise is i.i.d. over time, and given the variance, Gaussian distribution maximizes the differential entropy, we obtain 
\begin{align}
I(W_1,W_2;Y_1^n|S^n) &\leq n( h(Y_{1G}-S_G) - h(Z_{1}) )  \\
&= n( h(h_dX_{1G} + h_c X_{2G} + Z_{1} - Z) - h(Z_{1}) )  \\
& =n C\left(P_1h_d^2+P_2h_c^2+1\right)  \\
& \leq n C\left(P(h_d^2+h_c^2)+1\right),\label{eq:new_UB1_Gauss_proof_6}
\end{align}
where $X_{iG}\sim\mathcal{N}(0,P_i)$, $i\in\{1,2\}$.
\iffalse
Using chain rule, we rewrite \eqref{eq:new_UB1_Gauss_proof_3} as follows
\begin{align}
I(W_1,W_2;Y_1^n|S^n) &\leq \sum_{t=1}^n h(Y_1[t]-S[t]|Y_1^{t-1}-S^{t-1}) - h(Z_1[t]|Z_1^{t-1})  \\
&\overset{(g)}{\leq} \sum_{t=1}^n h(Y_1[t]-S[t]) - h(Z_1[t])  \\ 
&\overset{(h)}{\leq} n( h(Y_{1G}-S_G) - h(Z_{1}) )  \\
&= n( h(h_dX_{1G} + h_c X_{2G} + Z_{1} - Z) - h(Z_{1}) )  \\
& =n C(P(h_d^2+h_c^2)+1)\label{eq:new_UB1_Gauss_proof_6}
\end{align}
In $(g)$, we dropped the condition in first term since this cannot decrease the entropy. Moreover, we dropped the conditions in second term in $(g)$ since $Z[t]$ is i.i.d. over the time. In step $(h)$, for the first term, we used the fact that Gaussian distribution maximizes the differential entropy given the variance of random variable. Note that the noises are all Gaussian distributed. 
\fi
Finally, we consider the third term in \eqref{eq:new_UB1_Gauss_proof_1}. Similar to the previous case, we upper bound the third term in \eqref{eq:new_UB1_Gauss_proof_1} as follows.
\begin{align}
I(W_2;Y_2^n|S^n,Y_1^n,W_1) =& h(Y_2^n|S^n,Y_1^n,W_1)-h(Y_2^n|S^n,Y_1^n,W_1,W_2)  \\
\overset{}{=}& h(h_dX_2^n+h_rX_r^n+Z_2^n|S^n,h_cX_2^n+h_r X_r^n+Z_1^n,W_1) - h(h_rX_r^n+Z_2^n|S^n,Y_1^n,W_1,W_2)  \\
\overset{}{=} & h(h_dX_2^n+h_rX_r^n+Z_2^n-S^n|S^n,h_cX_2^n+h_r X_r^n+Z_1^n-S^n,W_1)  \notag\\ & - h(h_rX_r^n+Z_2^n-S^n|S^n,Y_1^n,W_1,W_2)  \\
\overset{}{=} & h(h_dX_2^n+Z_2^n-Z^n|S^n,h_cX_2^n+Z_1^n-Z^n,W_1) - h(Z_2^n-Z^n|S^n,Y_1^n,W_1,W_2)  \\
\overset{(e)}{\leq} & h(h_dX_2^n+Z_2^n-Z^n|h_cX_2^n+Z_1^n-Z^n) - h(Z_2^n-Z^n|S^n,Y_1^n,W_1,W_2,Z^n)  \\
\overset{}{=} & h(h_dX_2^n+Z_2^n-Z^n|h_cX_2^n+Z_1^n-Z^n) - h(Z_2^n).  \label{eq:new_UB1_Gauss_proof_7}
\end{align}
%In step $(i)$, we used the fact that knowing $W_1$, all randomness of $Y_j^n$ ($j\in\{1,2\}$) is from $X_r^n$, $X_2^n$, and $Z_j^n$.
%Moreover, knowing $W_1$ and $W_2$, all randomness of $Y^n_j$ ($j\in\{1,2\}$) is from $X_r^n$ and $Z_j^n$. 
In step $(e)$, we dropped some conditions from first term and included an additional condition to the second term since conditioning does not increase the entropy. 
%In $(k)$, we dropped all conditions in second term, since $Z_2^n$ is independent from all other random variables.
%Now, we define $\tilde{Y}_{2}[t] = h_dX_{2}[t]+Z_2[t]-Z[t]$ and $\tilde{Y}_{1}[t] = h_cX_{1}[t]+Z_1[t]-Z[t]$ to obtain
%\begin{align}
%I(W_2;Y_2^n|S^n,Y_1^n,W_1) = & \sum_{t=1}^n h(\tilde{Y}_2[t]|\tilde{Y}_1[t],\tilde{Y}_2^{t-1},\tilde{Y}_1^{t-1}) - h(Z_2[t]|Z_2^{t-1})\notag \\ 
%\overset{(l)}{\leq} & \sum_{t=1}^n h(\tilde{Y}_2[t]|\tilde{Y}_1[t]) - h(Z_2[t]), \label{eq:new_UB1_Gauss_proof_8}
%\end{align}
%where in step $(l)$, we dropped the conditions in first term since this cannot decrease the entropy. Moreover, since $Z_2[t]$ is i.i.d. over the time, the conditions in second term of step $(l)$ can be dropped.
%Due to the fact that Gaussian distribution maximizes the differential entropy given a covariance matrix \cite{Thomas}, we can write
Now, we define $\tilde{Y}_{2G}[t] = h_dX_{2G}[t]+Z_2[t]-Z[t]$ and ${\tilde{Y}_{1G}[t] = h_cX_{2G}[t]+Z_1[t]-Z[t]}$, where ${X_{2G}\sim \mathcal{N}(0,P_2)}$. By using Lemma 1 in \cite{AnnapureddyVeeravalli} and the fact that $Z_2$ i.i.d. over the time, we obtain
\begin{align}
I(W_2;Y_2^n|S^n,Y_1^n,W_1) \leq & n(h(\tilde{Y}_{2G}|\tilde{Y}_{1G}) - h(Z_{2}))  \\
=&\frac{n}{2}\log\left(\frac{\begin{vmatrix}
P_2h_d^2+2 & P_2h_dh_c+1 \\ P_2h_dh_c+1 & P_2h_c^2+2
\end{vmatrix}}{2+P_2h_c^2}\right)  \\
\overset{}{\leq} & \frac{n}{2}\log\left(\frac{(P_2h_d^2+2)(P_2h_c^2+2)-(1+P_2h_ch_d)^2}{2+P_2h_c^2}\right)  \\
\overset{}{=} & nC\left(1+\frac{2P_2h_d(h_d-h_c)-1}{2+P_2h_c^2}\right) \label{eq:new_UB1_Gauss_proof_9}\\
\overset{(f)}{\leq} & nC\left(1+\frac{2P h_d^2+2P\max\{ h_c^2, h_d^2\}}{2+P h_c^2}\right).  \label{eq:new_UB1_Gauss_proof_10}
\end{align}
%where we used $\tilde{Y}_{jG}$, $j\in\{1,2\}$ to represent the random variable $\tilde{Y}_{j}$ when its input $X_2$ is Gaussian distributed, i.e., ${X_2\sim \mathcal{N}(0,P_2)}$.
%In step $(m)$, we used the fact that $\log(x)$ is an increasing function in $x$ and 
Step $(f)$ follows since the function in \eqref{eq:new_UB1_Gauss_proof_9} is increasing in $P_2$ and $P_2h_dh_c \leq P_2 \max\{ h_c^2, h_d^2\} $.
Substituting \eqref{eq:new_UB1_Gauss_proof_4}, \eqref{eq:new_UB1_Gauss_proof_6}, and \eqref{eq:new_UB1_Gauss_proof_10} into \eqref{eq:new_UB1_Gauss_proof_1}, and dividing the whole expression by $n$ and letting $n\to\infty$, we obtain
\begin{align}
R_\Sigma \leq C(Ph_r^2) + C(P(h_d^2+h_c^2)+1) +C\left(1+\frac{2P h_d^2+2P\max\{ h_c^2, h_d^2\}}{2+P h_c^2}\right)  .
\end{align}
Similar to above, by dividing the expression by $\frac{1}{2}\log(Ph_d^2)$ and letting $Ph_d^2\to \infty$, we obtain the following upper bound for the GDoF
\begin{align}
d &\leq \beta + \max\{1,\alpha\}  + \max\{1,\alpha\}-\alpha %\\ 
%& = \beta + \max\{1,\alpha\}  + (1-\alpha)^+,
\end{align}
which completes the proof.
\subsection{Proof of \eqref{eq:new_UB4_IRC_Gauss}}
\label{app:new_UB2_Gauss_proof}
We use the genie-aided method to establish this upper bound. In this case, we set $s_1=Y_r^n$ and $s_2=(Y_r^n,W_1)$. Now, by using Fano's inequality, we upper bound the sum-rate as follows
\begin{align}
n(R_\Sigma-\epsilon_n) \leq & I(W_1;Y_1^ n,Y_r^n) + I(W_2;Y_2^n,Y_r^n,W_1).  
\end{align}
Then, by using the chain rule and the fact that the messages are independent from each other, we obtain
\begin{align}
n(R_\Sigma-\epsilon_n) \leq  & I(W_1;Y_r^n)+ I(W_1;Y_1^n|Y_r^n) +  I(W_2;Y_r^n|W_1) + I(W_2;Y_2^n|Y_r^n,W_1) \\
 = & I(W_1,W_2;Y_r^n) + I(W_1;Y_1^n|Y_r^n)+I(W_2;Y_2^n|Y_r^n,W_1). \label{eq:new_UB2_Gauss_proof_1}  
\end{align}
Now, we consider each term in \eqref{eq:new_UB2_Gauss_proof_1} separately. First, we write the first term as follows
\begin{align}
I(W_1,W_2;Y_r^n) &= h(Y_r^n)-h(Y_r^n|W_1,W_2) \\
&\overset{(a)}{=}h(Y_r^n)-H(Z_r^n),
\end{align}
where $(a)$ follows since knowing $W_1$, $W_2$, all randomness of $Y_r^n$ is from $Z_r^n$. Moreover, we used the fact that $Z_r^n$ is independent from all other variables.
Now, by using \cite[corollary to Theorem 8.6.2]{CoverThomas} and the fact that $Z_r^n$ is i.i.d. over the time, we obtain 
\begin{align}
I(W_1,W_2;Y_r^n) & \leq \sum_{t=1}^n h(Y_r[t])-h(Z_r[t]).
\end{align}
Using the fact that the differential entropy is maximized by the Gaussian distribution given the variance \cite{CoverThomas}, we upper bound $I(W_1,W_2;Y_r^n)$ as follows
\begin{align}
I(W_1,W_2;Y_r^n) &\leq n[ h(h_s(X_{1G}+X_{2G})+Z_{r})-h(Z_{r})]  \\ 
&=n C\left((P_1+P_2)h_s^2\right)  \\
&\leq n C\left(2Ph_s^2\right), \label{eq:new_UB2_Gauss_proof_2} 
\end{align}
where $X_{iG}\sim \mathcal{N}(0,P_i)$ for $i\in\{1,2\}$.
Next, we consider the second term in \eqref{eq:new_UB2_Gauss_proof_1}. 
\begin{align}
I(W_1;Y_1^n|Y_r^n) \overset{(b)}{=} & I(W_1;Y_1^n|Y_r^n,X_r^n)  \\ 
= & h(Y_1^n|Y_r^n,X_r^n) - h(Y_1^n|Y_r^n,X_r^n,W_1)  \\ 
\overset{(c)}{=} & h(h_dX_1^n+h_cX_2^n+Z_1^n|Y_r^n,X_r^n) - h(h_cX_2^n+Z_1^n|Y_r^n,X_r^n,W_1)  \\
\overset{(d)}{\leq} & h(h_dX_1^n+h_cX_2^n+Z_1^n) - h(h_cX_2^n+Z_1^n|Y_r^n,X_r^n,W_1,X_2^n)  \\ 
\overset{(e)}{=} & h(h_dX_1^n+h_cX_2^n+Z_1^n) - h(Z_1^n).  
\end{align}
Step $(b)$ follows since encoding  at the relay is known, hence knowing $Y_r^n$, the signal $X_r^n$ can be reconstructed,  step $(c)$ follows since knowing $X_r^n$ all randomness of $Y_1^n$ is from $X_1^n$, $X_2^n$, and $Z_1^n$ and from $W_1$, $X_1^n$ can be reconstructed. In step $(d)$, we used the fact that conditioning does not increase the entropy. In step $(e)$, we dropped the conditions since $Z_1^n$ is independent from $Y_r^n,X_r^n,W_1$, and $X_2^n$. 
%Now, similar to previous case, we use chain rule, and the facts that conditioning does not increase the entropy, $Z_1$ is i.i.d. over time, and given the variance, Gaussian distribution maximizes entropy to upper bound 
Similar to above, by using \cite[corollary to Theorem 8.6.2]{CoverThomas} and the facts that $Z_1^n$ is i.i.d. over the time, and Gaussian distribution maximizes the differential entropy for a give variance,  we obtain 
%\begin{align}
%I(W_1;Y_1^n|Y_r^n) \leq 
%%& \sum_{t=1}^n h(h_dX_1[t]+h_cX_2[t]+Z_1[t]|h_dX_1^{t-1}+h_cX_2^{t-1}+Z_1^{t-1}) - h(Z_1[t]|Z_1^{t-1}) \notag \\
%%\overset{(f)}{\leq} 
% & \sum_{t=1}^n h(h_dX_1[t]+h_cX_2[t]+Z_1[t]) - h(Z_1[t]).
%\end{align}
%In $(f)$, we dropped the conditions similar to step $(b)$. 
%Similar to previous case, since Gaussian distribution maximizes the differential entropy given the variance, we obtain 
\begin{align}
I(W_1;Y_1^n|Y_r^n) \leq & n[ h(h_dX_{1G}+h_cX_{2G}+Z_1) - h(Z_1)]  \\
= & nC(P_1h_d^2+P_2h_c^2)\\
\leq & nC(P(h_d^2+h_c^2)) \label{eq:new_UB2_Gauss_proof_3} 
\end{align}
Finally, we bound the last term in \eqref{eq:new_UB2_Gauss_proof_1} as follows
\begin{align}
I(W_2;Y_2^n|Y_r^n,W_1) \overset{(f)}{=} & I(W_2;Y_2^n|Y_r^n,W_1,X_r^n)  \\
=&  h(Y_2^n|Y_r^n,W_1,X_r^n) -h (Y_2^n|Y_r^n,W_1,X_r^n,W_2)  \\
\overset{(g)}{=} & h(h_dX_2^n+Z_2^n|h_sX_2^n+Z_r^n,W_1,X_r^n) -h (Z_2^n|Y_r^n,W_1,X_r^n,W_2)  \\
\overset{(h)}{\leq} & h(h_dX_2^n+Z_2^n|h_sX_2^n+Z_r^n) -h (Z_2^n). 
\end{align}
Step $(f)$ follows since encoding  at the relay is known hence $X_r^n$ can be reconstructed from $Y_r^n$. 
In $(g)$, we used the fact that knowing $W_1$ and $W_2$, the randomness of $X_1^n$ and $X_2^n$ can be removed. 
Step $(h)$ follows since conditioning does not increase the entropy and $Z_2^n$ is independent of all other random variables. 
By using Lemma 1 in \cite{AnnapureddyVeeravalli} and the fact that $Z_2$ i.i.d. over the time, we obtain
%
%By using the chain rule and then dropping the conditions similar to steps $(b)$ and $(f)$, we upper bound $I(W_2;Y_2^n|Y_r^n,W_1)$ as follows
%\begin{align}
%I(W_2;Y_2^n|Y_r^n,W_1) \leq \sum_{t=1}^n h(h_dX_2[t]+Z_2[t]|h_sX_2[t]+Z_r[t]) -h (Z_2[t]).
%\end{align}
%Now, by using the fact that Gaussian distribution maximizes the conditional differential entropy given the covariance matrix \cite{Thomas}, we obtain
\begin{align}
I(W_2;Y_2^n|Y_r^n,W_1) \leq & n [h(h_dX_{2G}+Z_{2}|h_sX_{2G}+Z_{r}) -h (Z_{2})]  \\
= & \frac{n}{2}\log\left(\frac{\begin{vmatrix}
P_2h_d^2+1 & P_2h_sh_d \\ P_2h_sh_d & P_2h_s^2+1
\end{vmatrix}}{P_2h_s^2+1}\right)  \\
= & n C\left(\frac{P_2h_d^2}{1+P_2h_s^2}\right) \label{eq:new_UB2_Gauss_proof_4}  \\
\overset{(i)}{\leq} & n C\left(\frac{Ph_d^2}{1+Ph_s^2}\right). \label{eq:new_UB2_Gauss_proof_5} 
\end{align}
where $X_{2G}\sim \mathcal{N}(0,P_2)$. Step $(i)$ follows since the function in \eqref{eq:new_UB2_Gauss_proof_4} is increasing in $P_2$. 

Now, by substituting \eqref{eq:new_UB2_Gauss_proof_2}, \eqref{eq:new_UB2_Gauss_proof_3}, and \eqref{eq:new_UB2_Gauss_proof_5} in \eqref{eq:new_UB2_Gauss_proof_1}, we obtain 
\begin{align}
n(R_\Sigma-\epsilon_n) \leq n \left[ C\left(2Ph_s^2\right) + C(P(h_d^2+h_c^2)) + C\left(\frac{Ph_d^2}{1+Ph_s^2}\right) \right].
\end{align}
Now, by dividing the whole expression by $n$ and letting $n\to\infty$, we obtain
\begin{align}
R_\Sigma \leq  C\left(2Ph_s^2\right) + C(P(h_d^2+h_c^2)) + C\left(\frac{Ph_d^2}{1+Ph_s^2}\right).
\end{align}
In order to get an upper bound for the GDoF, we divide the sum-rate by $\frac{1}{2}\log(Ph_d^2)$ then we let $Ph_d^2\to\infty$. Hence, we have
\begin{align}
d & \leq  \gamma + \max\{1,\alpha\} +(1-\gamma)^+\\
&= \max\{1,\alpha\} + \max\{1,\gamma\}
\end{align}
This completes the proof of Lemma \eqref{eq:new_UB4_IRC_Gauss}.

\subsection{Proof of \eqref{eq:new_UB5_IRC_Gauss}}
To establish this upper bound, we use the upper bound given in \cite[Theorem 4]{ChaabanSezgin_IT_IRC}. This theorem bounds the capacity of the Gaussian IRC as follows
\begin{align}
R_\Sigma \leq 2C\left( (|h_c|+|h_r|)^2 P+4\max\left\lbrace\frac{h_d^2P}{1+Ph_c^2},h_r^2P \right\rbrace \right) + 2C\left( \frac{h_s^2}{h_c^2} \right). \label{eq:new_UB3_Gauss_proof_1}
\end{align}
Now, suppose that the noise variance at the Rx's is reduced to $c^2 = \min \left\lbrace 1,\frac{h_c^2}{h_s^2} \right \rbrace$. This enhances the channel. Therefore, an upper bound for the capacity of the new channel is an upper bound for the original channel. Reducing the noise variance to $c^2$ is equivalent to increasing the channel strength $h_d$, $h_c$, and $h_r$ while the noise variance is 1. Therefore, the upper bound in \eqref{eq:new_UB3_Gauss_proof_1} is upper bounded by 
\begin{align}
R_\Sigma \leq 2C\left( (|\bar h_c|+|\bar h_r|)^2 P+4\max\left\lbrace\frac{\bar h_d^2P}{1+P\bar h_c^2},\bar h_r^2P \right\rbrace \right) + 2C\left( \frac{\bar h_s^2}{\bar h_c^2} \right). \label{eq:new_UB3_Gauss_proof_2}
\end{align}
where $\bar h_d$, $\bar h_c$, $\bar h_r$, and $\bar h_s$ represent the channel values for the enhanced IRC. They are defined as follows
\begin{align}
\bar{h_d} = h_d \max \left\lbrace 1,\frac{h_s}{h_c}\right\rbrace, \quad \bar{h_c} = h_c \max \left\lbrace 1,\frac{h_s}{h_c}\right\rbrace,
\quad \bar{h_r} = h_r \max \left\lbrace 1,\frac{h_s}{h_c}\right\rbrace, \quad \bar{h_s} = h_s. \label{eq:new_UB3_Gauss_proof_3}
\end{align}
%Notice that since we enhanced the channel by reducing the noise variance only at Rx's, in equivalent setup $h_s$ is not changed. 
Now, we divide the expression in \eqref{eq:new_UB3_Gauss_proof_2} by $\frac{1}{2}\log_2(Ph_d^2)$ and then we let $Ph_d^2\to \infty$. Then, we obtain the following upper bound for the GDoF
\begin{align}
d \leq \lim_{m_d\to\infty} \frac{2\max\{\bar m_c, \bar m_r,(\bar m_d-\bar m_c)^+ \} + 2(\bar m_s-\bar m_c)^+}{m_d}, \label{eq:new_UB3_Gauss_proof_4}
\end{align}
where $\bar m_d$, $\bar m_c$, $\bar m_r$, and $\bar m_s$ are defined as follows
\begin{align}
\bar m_d = m_d+(m_s-m_c)^+, \quad \bar m_c = m_c+(m_s-m_c)^+, \quad\bar m_r = m_r+(m_s-m_c)^+, \quad \bar m_s = m_s. \label{eq:new_UB3_Gauss_proof_5}
\end{align}
By substituting \eqref{eq:new_UB3_Gauss_proof_5} into \eqref{eq:new_UB3_Gauss_proof_4}, we obtain 
\begin{align}
d \leq \lim_{m_d\to\infty} \frac{2\max\{m_c+(m_s-m_c)^+, m_r+(m_s-m_c)^+,( m_d- m_c)^+ \}}{m_d}. \label{eq:new_UB3_Gauss_proof_6}
\end{align}
Note that $(\bar m_s - \bar m_c)^+ = 0$. Now, we rewrite \eqref{eq:new_UB3_Gauss_proof_6} as follows 
\begin{align}
d \leq \lim_{m_d\to\infty} \frac{2\max\{m_c, m_r,m_d-\max\{m_c,m_s\} \}+2(m_s-m_c)^+}{m_d}. \label{eq:new_UB3_Gauss_proof_7}
\end{align}
Now, by using the definition of $m_d$, $m_c$, $m_r$, and $m_s$, we obtain the following upper bound for the GDoF of the IRC
\begin{align}
d \leq  2\max\{\alpha, \beta,1-\max\{\alpha,\gamma\} \} + 2(\gamma-\alpha)^+, %\label{eq:new_UB3_Gauss_proof_8}
\end{align}
which completes the proof of \eqref{eq:new_UB5_IRC_Gauss}.

\subsection{Proof of \eqref{eq:new_UB6_IRC_Gauss}}
\label{app:new_UB4_Gauss_proof}
To establish this upper bound, we use the genie-aided method with $s_1^n = S_1^n$ and $s_2^n=S_2^n$, where $S_1^n=\frac{ h_c}{\sqrt{Ph_r^2}}X_1^n+U_1^n$ and $S_2^n=\frac{ h_c}{\sqrt{Ph_r^2}}X_2^n+U_2^n$. Here $U_1$ and $U_2$ are both $\mathcal{N}(0,1)$ distributed. Moreover, they are independent of all other random variables and i.i.d. over time. Now, we use Fano's inequality to upper bound the sum-rate as follows
\begin{align}
n(R_\Sigma-\epsilon_n) \leq & I(W_1;Y_1^n,S_1^n) + I(W_2;Y_2^n,S_2^n) \notag \\
\overset{(a)}{=}& I(W_1;S_1^n) + I(W_1;Y_1^n|S_1^n) + I(W_2;S_2^n) + I(W_2;Y_2^n|S_2^n), \label{eq:new_UB4_Gauss_proof_1}
\end{align}
where step $(a)$ follows from the chain rule. Now, we proceed as follows
\begin{align}
n(R_\Sigma-\epsilon_n)  
%= & h(S_1^n) - h(S_1^n|W_1) + h(Y_1^n|S_1^n) - h(Y_1^n|S_1^n,W_1) + h(S_2^n) - h(S_2^n|W_2) + h(Y_2^n|S_2^n)- h(Y_2^n|S_2^n,W_2) \notag \\
\overset{(b)}{\leq} & h(S_1^n) - h(U_1^n) + h(Y_1^n|S_1^n) - h(h_cX_2^n + h_rX_r^n+Z_1|S_1^n,W_1) \notag \\ 
&+ h(S_2^n) - h(U_2^n) + h(Y_2^n|S_2^n)- h(h_cX_1^n + h_rX_r^n+Z_2|S_2^n,W_2) \notag \\ 
\overset{(c)}{=} & h(S_1^n) - h(U_1^n) + h(Y_1^n|S_1^n) - h(h_cX_2^n + h_rX_r^n+Z_1|W_1) \notag \\ 
&+ h(S_2^n) - h(U_2^n) + h(Y_2^n|S_2^n)- h(h_cX_1^n + h_rX_r^n+Z_2|W_2). \label{eq:new_UB4_Gauss_proof_2}
\end{align}
In step $(b)$ we used the fact knowing $W_i$, where $i\in\{1,2\}$, all randomness of $S_i^n$ is caused from $U_i^n$. Step $(c)$ follows since $U_i^n$ is independent of all other random variables.
Now, due to Lemma \ref{Lemma:new_UB4_2_Gauss}, the sum-rate is upper bounded as follows
\begin{align}
n(R_\Sigma-\epsilon_n)  
\leq &  2nC\left(2+ \frac{h_c^2}{(h_c-h_r)^2}\right)- h(U_1^n) + h(Y_1^n|S_1^n) - h(U_2^n) + h(Y_2^n|S_2^n).\label{eq:new_UB4_Gauss_proof_3}
\end{align}
Next, we use \cite[Lemma 1]{AnnapureddyVeeravalli} and the fact that $U_i^n$ is i.i.d. over the time, to write
\begin{align}
n(R_\Sigma-\epsilon_n)  
\leq &  n\left[2C\left(2+ \frac{h_c^2}{(h_c-h_r)^2}\right)- h(U_1) + h(Y_{1G}|S_{1G}) - h(U_{2}) + h(Y_{2G}|S_{2G})\right],
\end{align}
where the subscript $G$ indicates that the inputs are i.i.d. and Gaussian distributed, i.e., $X_{i,G}\sim\mathcal{N}(0,P_i)$, where $i\in\{1,2,r\}$ and $S_{1,G}$, $S_{2G}$, $Y_{1,G}$, and $Y_{2G}$ are corresponding signals. In what follows, we upper bound the expression $h(Y_{1G}|S_{1G})-h(U_1)$. Similarly, we can bound $h(Y_{2G}|S_{2G})-h(U_2)$. To this end, we write 
\begin{align}
h(Y_{1G}|S_{1G})-h(U_1) &= h(h_dX_{1G}+h_c X_{2G}+h_r X_{rG}+Z_1|\frac{h_c}{\sqrt{Ph_r^2}}X_{1G}+U_1) - h(U_1) \\ 
&= h(h_dX_{1G}+h_c X_{2G}+h_r X_{rG}+Z_1|h_dX_{1G}+ \frac{h_d \sqrt{Ph_r^2}}{h_c}U_1) - h(U_1)\\
&\overset{(d)}{\leq} h(h_c X_{2G}+h_r X_{rG}+Z_1- \frac{h_d \sqrt{Ph_r^2}}{h_c}U_1) - h(U_1)\\
&=C\left( P_2h_c^2 + P_r h_r^2  + \frac{h_d^2Ph_r^2}{h_c^2} + 2\rho_2 h_ch_r\sqrt{P_2P_r} \right),\label{eq:new_UB4_Gauss_proof_5}
\end{align}
where the parameter $\rho_2\in[-1, 1]$ is the correlation coefficient between $X_2$ and $X_r$. 
Step $(d)$ follows since $h(A-B|B) = h(A|B)$ and conditioning does not increase the entropy.
Since $\rho_2\in[-1 ,1]$ and $\log(x)$ is an increasing function in $x$, the expression in \eqref{eq:new_UB4_Gauss_proof_5} is upper bounded by
\begin{align}
h(Y_{1G}|S_{1G})-h(U_{1}) &\leq C\left( Ph_c^2 + P h_r^2  + \frac{h_d^2h_r^2}{h_c^2}P +  2h_ch_rP \right).\label{eq:new_UB4_Gauss_proof_6}
\end{align}
Similarly, we upper bound $h(Y_{2G}|S_{2G})-h(U_{2})$. Doing this, the sum rate is upper bounded by
\begin{align}
n(R_\Sigma-\epsilon_n) \leq & 2n\left[C\left(2+ \frac{h_c^2}{(h_c-h_r)^2}\right) + C\left( Ph_c^2 + P h_r^2  + \frac{h_d^2h_r^2}{h_c^2}P +  h_ch_rP \right)\right].\label{eq:new_UB4_Gauss_proof_7}
\end{align}
Now, by dividing the expression by $n$ and letting $n\to\infty$, we obtain 
\begin{align}
R_\Sigma \leq & 2\left[C\left(2+ \frac{h_c^2}{(h_c-h_r)^2}\right) + C\left( Ph_c^2 + P h_r^2  + \frac{h_d^2h_r^2}{h_c^2}P +  h_ch_rP \right)\right].\label{eq:new_UB4_Gauss_proof_8}
\end{align}
To obtain an upper bound for the GDoF, we divide the sum-rate by $\frac{1}{2}\log(Ph_d^2)$ and then we let $Ph_d^2 \to \infty$. Hence, we get 
\begin{align}
d \leq 2\max\{\alpha,\beta,(1+\beta-\alpha)^+\}.
\end{align}
Since in \eqref{eq:new_UB6_IRC_Gauss}, we have the condition $\beta \leq \alpha<1$, the upper bound is rewritten as 
\begin{align}
d \leq 2\max\{\alpha,\beta+1-\alpha\},
\end{align}
which completes the proof of \eqref{eq:new_UB6_IRC_Gauss}.

\section{Extension of the achievable sum-rate from LD-IRC to the achievable GDoF}
\label{app:Example_LD-Gaussian}

Suppose that by using a scheme in an LD-IRC, we achieve a following linear combination of $n_d$, $n_c$, $n_r$, and $n_s$
\begin{align}
R_\Sigma = k_d n_d + k_c n_c + k_r n_r + k_s n_s,
\end{align}
where $k_i \in \mathbb{Z}$ for $i\in\{d,c,r,s\}$. By using the sub-channel allocation in the Gaussian IRC as the rate allocation for the LD-IRC, and keeping in mind that $\frac{1}{2}\log\left(\frac{\delta}{4}\right)$ is achieved by using each sub-channel, we achieve
\begin{align}
R_\Sigma &= \frac{1}{2} \log\left(\frac{\delta}{4}\right)(k_d N_d + k_c N_c + k_r N_r + k_s N_s) \\ 
&=\frac{1}{2} \log(\delta)(k_d N_d + k_c N_c + k_r N_r + k_s N_s) -(k_d N_d + k_c N_c + k_r N_r + k_s N_s). \label{eq:exam1_1}
\end{align}
Using the definition of $N_d$, $N_c$, $N_r$, and $N_s$, we conclude that the following sum-rate is achievable as long as \eqref{eq:exam1_1} is achievable
\begin{align}
R_\Sigma =& \frac{1}{2} \log(\delta)  \left(k_d \frac{\log(Ph_d^2)}{\log(\delta)} + k_c \frac{\log(Ph_c^2)}{\log(\delta)} + k_r \frac{\log(Ph_r^2)}{\log(\delta)} + k_s \frac{\log(Ph_s^2)}{\log(\delta)} - (|k_d|+|k_c|+|k_r|+|k_s|) \right) \notag \\  &-\left(k_d \frac{\log(Ph_d^2)}{\log(\delta)} + k_c \frac{\log(Ph_c^2)}{\log(\delta)} + k_r \frac{\log(Ph_r^2)}{\log(\delta)} + k_s \frac{\log(Ph_s^2)}{\log(\delta)}\right) . \label{eq:exam1_2}
\end{align}
Now, by dividing the sum-rate by $\frac{1}{2}\log(Ph_d^2)$ and using the definition in \eqref{eq:param_def}, we can write
\begin{align}
\frac{R_\Sigma}{\frac{1}{2}\log(Ph_d^2)} =& \ \left(k_d  + k_c \alpha + k_r \beta + k_s \gamma\right) \left(1-\frac{1}{\log(\delta)}\right) -  (|k_d|+|k_c|+|k_r|+|k_s|)  \frac{\log(\delta)}{\log(Ph_d^2)}. \label{eq:exam1_3}
\end{align}
{To obtain the achievable GDoF, we need to let $Ph_d^2 \to \infty $ in \eqref{eq:exam1_3}. For a fixed $h_d^2$, this is equivalent to $P\to \infty$. Therefore, the term $\frac{1}{\log(\delta)} = \frac{N}{\log(P)}\to 0 $ for a fixed $N$. Therefore, we obtain 
\begin{align}
\lim_{Ph_d^2\to \infty}  \frac{R_\Sigma}{\frac{1}{2}\log(Ph_d^2)} &= \left(k_d  + k_c \alpha + k_r \beta + k_s \gamma\right) - (k_d+k_c+k_r+k_s)  \frac{\log(\delta)}{\log(Ph_d^2)} \\ 
&= \left(k_d  + k_c \alpha + k_r \beta + k_s \gamma\right) - (k_d+k_c+k_r+k_s)  \frac{\frac{1}{N}}{1+ \frac{\log(h_d^2)}{\log(P)}}.
 \label{eq:exam1_4}
\end{align}
For a fixed $h_d^2$, the term $\frac{\log(h_d^2)}{\log(P)} \to 0$. Notice that $N$ is a constant which can be chosen arbitrarily. Therefore, by choosing $N$ sufficiently large, the second term in  \eqref{eq:exam1_4} is negligible and we obtain the following achievable GDoF 
\begin{align}
d &= \left(k_d  + k_c \alpha + k_r \beta + k_s \gamma\right). \label{eq:exam1_5}
\end{align}}

\bibliography{myBib}		% use data in file "biblio.bib"
\end{document}